\documentclass[12pt]{article}
\usepackage[margin=2cm]{geometry}
\usepackage{amsmath,amssymb,amsfonts,amsthm}
\usepackage{helvet}
\usepackage{bm,bbm}
\usepackage{graphicx}
\usepackage{enumerate}
\usepackage{enumitem}
\usepackage{natbib}
\usepackage{lscape}
\usepackage{multirow}
\usepackage{rotating}

\usepackage{url}
\usepackage[dvipsnames]{xcolor}
\usepackage[acronym,nogroupskip,nonumberlist,nopostdot,toc]{glossaries}
\usepackage{authblk}
\usepackage{IEEEtrantools}
\usepackage{algpseudocode}
\usepackage{algorithm,setspace}
\usepackage{mathrsfs}
\usepackage[mathscr]{eucal}
\usepackage{changepage}
\usepackage{booktabs}     
\usepackage{subcaption}
\usepackage[percent]{overpic}
\usepackage[dvipsnames]{xcolor}

\usepackage{xr}
\makeatletter
\newcommand*{\addFileDependency}[1]{
\typeout{(#1)}
\@addtofilelist{#1}
\IfFileExists{#1}{}{\typeout{No file #1.}}
}\makeatother

\newtheorem{definition}{Definition}
\newtheorem{theorem}{Theorem}
\newtheorem{lemma}{Lemma}
\newtheorem{example}{Example}

\makeatletter
\@addtoreset{example}{section}
\makeatother

\newtheorem{proposition}{Proposition}
\newtheorem{assumption}{Assumption}
\newtheorem{remark}{Remark}
\newcommand{\W}{\ensuremath{\mathcal{W}}}
\newcommand{\suppW}{\ensuremath{\text{\normalfont supp}(\PR_{\bm
      X/\lVert\bm X\rVert})}}
\newcommand{\suppWclean}{\ensuremath{\text{\normalfont supp}(\PR_{\bm W})}}
\newcommand{\suppX}{\ensuremath{\text{\normalfont supp}(\PR_{\bm X})}}

\newcommand{\zL}{\ensuremath{z_{-}}}
\newcommand{\zR}{\ensuremath{z_{+}}}

\newcommand{\nuR}{\ensuremath{\nu_{\{+1\}}}}
\newcommand{\LambdaR}{\ensuremath{\Lambda_{\{+1\}}}}
\newcommand{\xiR}{\ensuremath{\xi_{\{+1\}}}}

\newcommand{\nuL}{\ensuremath{\nu_{\{-1\}}}}
\newcommand{\LambdaL}{\ensuremath{\Lambda_{\{-1\}}}}
\newcommand{\xiL}{\ensuremath{\xi_{\{-1\}}}}
\newcommand{\rGP}{\ensuremath{\text{rGP}}}
\newcommand{\rExp}{\ensuremath{\text{rExp}}}

\newcommand{\HMRV}{\ensuremath{H_{\textsf{MRV}}}}
\newcommand{\Rstar}{\ensuremath{\RR^d\setminus\{\Orig\}}}
\newcommand{\QQ}{\ensuremath{\mathsf{Q}}}
\newcommand{\fL}{\ensuremath{f_{[L]}}}
\newcommand{\OO}{\ensuremath{O}} \newcommand{\oo}{\ensuremath{o}}
\newcommand{\QS}{\ensuremath{\mathcal{Q}}}
\newcommand{\QSq}{\ensuremath{\QS_q}}

\newcommand{\rquant}{\ensuremath{r_{\QSq}}}
\newcommand{\rW}{\ensuremath{r_{\W}}}

\newcommand{\rquants}[1]{\ensuremath{r_{\QSq,#1}}}
\newcommand{\loc}{\ensuremath{\mathcal{M}}}
\newcommand{\B}{\ensuremath{\mathcal{B}}}

\newcommand{\LL}{\ensuremath{\mathcal{L}}}
\newcommand{\G}{\ensuremath{\mathcal{G}}}
\newcommand{\rG}{\ensuremath{r_{\G}}}
\newcommand{\rGs}[1]{\ensuremath{r_{\G,#1}}}
\newcommand{\gG}{\ensuremath{g_{\G}}}

\newcommand{\dd}{\ensuremath{\mathrm{d}}}

\newcommand{\RR}{\ensuremath{\mathbb{R}}}

\newcommand{\PR}{\ensuremath{\mathbb{P}}}
\newcommand{\EE}{\ensuremath{\mathbb{E}}}

\newcommand{\Espace}{\ensuremath{\mathbb{E}}}
\newcommand{\Emult}{\ensuremath{\overline{\Espace}^d}} 
\newcommand{\obsdatbold}{\ensuremath{\bm o}}
\newcommand{\obsdat}{\ensuremath{o}}
\newcommand{\obsrvbold}{\ensuremath{\bm O}}

\newcommand{\dat}{\ensuremath{\bm y}}

\newcommand{\SSS}{\ensuremath{\mathbb{S}}}
\newcommand{\RS}{\ensuremath{\mathsf{X}}}
\newcommand{\Q}{\ensuremath{\QQ}} \newcommand{\qfirst}{\ensuremath{u_1}} \newcommand{\qsecond} {\ensuremath{u_2}}
\newcommand{\Orig}{\ensuremath{\bm 0}}
\newcommand{\pfirst} {\ensuremath{p_1}}  \newcommand{\roc} {\ensuremath{u}}
\newcommand{\afirst} {\ensuremath{a_1}} \newcommand{\asecond}
{\ensuremath{a_2}} \newcommand{\lfirst} {\ensuremath{l_1}}
\newcommand{\lsecond} {\ensuremath{l_2}}
\newcommand{\radd}{\ensuremath{+}} 
\newcommand{\rmult}{\ensuremath{\cdot}}
\newcommand{\rdiv}{\ensuremath{/}}
\newcommand{\Mh}{\ensuremath{\text{M}_1}}
\newcommand{\Mg}{\ensuremath{\text{M}_2}}
\newcommand{\Mgh}{\ensuremath{\text{M}_3}} 
\newcommand{\vol}[1]{\ensuremath{\lvert #1 \rvert}}

\newcommand{\setExc}{\ensuremath{\mathcal{Y}}}
\newcommand{\sfR}{\ensuremath{\{+1\}}}
\newcommand{\sfL}{\ensuremath{\{-1\}}}
\definecolor{bleudefrance}{rgb}{0.19, 0.55, 0.91}
\definecolor{darkgreen}{rgb}{0.0, 0.2, 0.13} 
\usepackage{hyperref}
\hypersetup{colorlinks=true,linkcolor=darkgreen,citecolor=MidnightBlue}

\author[1]{Ioannis Papastathopoulos \thanks{\small
    i.papastathopoulos@ed.ac.uk}} \author[1]{Lambert {De Monte}
  \thanks{\small l.demonte@ed.ac.uk}} \author[2]{Ryan Campbell
  \thanks{\small r.campbell3@lancaster.ac.uk }} \author[3]{H{\aa}vard
  Rue \thanks{\small haavard.rue@kaust.ed.sa}} \affil[1]{{\small
    School of Mathematics and Maxwell Institute for Mathematical
    Sciences, University of Edinburgh, Edinburgh, EH9 3FD, Scotland}}
\affil[2]{{\small School of Mathematical Sciences, Lancaster
    University, Lancaster, LA1 4YF, England}} \affil[3] {{\small
    Statistics Program, King Abdullah University of Science and
    Technology, Thuwal 23955-6900, Saudi Arabia}}
\newcommand{\blind}{1}

\newcommand{\wk}{\ensuremath{\,\overset {\mathrm {w}
    }{\rightarrow}\,}}
\newcommand{\cind}{\ensuremath{\,\overset {\mathrm {d}
    }{\rightarrow}\,}}

\newcommand{\vg}{\ensuremath{\,\overset {\mathrm {v}
    }{\rightarrow}\,}}

\makenoidxglossaries \newglossaryentry{lset} { name
  ={\ensuremath{\mathcal{K}_d}}, description={family of non-empty
    compact subsets of $\RR^d$}, sort={A} }
\newglossaryentry{limitset} { name=$G$, description={a limit set,
    assumed to have positive area $\vol{\G}>0$}, sort={Gl} }
\newglossaryentry{gauge} { name=$\gG$, description={continuous gauge
    function describing a non-singular limit set $\G$}, sort={g} }
\newglossaryentry{simplex} { name={\ensuremath{\SSS^{d-1}}},
  description={$(d-1)$--dimensional sphere
    $\{\bm w\in \RR^d\,:\, \lVert \bm w\rVert = 1\}$ with respect to a
    norm $\lVert\, \cdot\, \rVert$.}, sort={Simplex} }
\newglossaryentry{polar} { name={\ensuremath{\Theta}},
  description={angular component of the polar coordinate system used
    for illustrations.\ A bijective map from $\Psi=[-\pi,\pi)$ to
    $\SSS$ is given by
    $\Psi \ni \psi \mapsto
    \mathcal{P}(\psi)=(\cos\psi,\sin\psi)\in \SSS$.},
  sort={SimplexPolar} }
\newglossaryentry{spherical} {
  name={\ensuremath{\Theta\times\Phi}}, description={angular component
    of the spherical coordinate system used for illustrations.\ A
    bijective map from
    $\Psi\times\Phi=[-\pi,\pi)\times[-\pi/2,\pi/2]$ to $\SSS^2$ is
    given by
    \[ \Psi\times\Phi \ni (\psi,\varphi) \mapsto
      \mathcal{S}(\psi,\varphi)=(\cos\psi\sin\varphi,\sin\psi\sin\varphi,\cos\varphi)\in
      \SSS^2.\]}, sort={SimplexSpherical} }
\newglossaryentry{spherical_d} {
  name={\ensuremath{\Phi}}, description={space of angles in spherical coordinates.\ A
    bijective map from
    $\Phi=[0,\pi]^{d-2}\times[-\pi,\pi)$ to $\SSS^{d-1}$ is
    given by
    \[ \Phi \ni \bm \varphi =(\varphi_1,\ldots,\varphi_{d-1}) \mapsto
      \mathcal{S}(\bm \varphi)\in
      \SSS^{d-1},
      \]
      where
      $\mathcal{S}(\bm \varphi)=(\cos\varphi_1,
      \cos\varphi_2\sin\varphi_1, \ldots,
      \cos\varphi_{d-1}\prod_{i=1}^{d-2}\sin\varphi_i,\prod_{i=1}^{d-1}\sin\varphi_i)$.\
    }, sort={SimplexSpherical} } \newglossaryentry{Lapvector} {
    name={\ensuremath{\bm X_L}}, description={random vector with
      standard Laplace marginal distributions}, sort={Radius} }
  \newglossaryentry{radius} { name={\ensuremath{R}},
    description={radial part $\lVert \bm X_L \rVert$ of
      $\bm X_L\in \RR_+$ with respect to a norm
      $\lVert\, \cdot\, \rVert$.}, sort={Radius} }
  \newglossaryentry{angle} { name={\ensuremath{\bm W}},
    description={angular part
      $\bm X_L/\lVert \bm X_L \rVert \in \SSS^{d-1}$ with respect to a
      norm $\lVert\, \cdot\, \rVert$.}, sort={angle} }
  \newglossaryentry{classM} { name={\ensuremath{M_p(\EE)}},
    description={collection of all Radon point measures on a
      \textit{nice} metric space $\EE$}, sort={angle} }
  \newglossaryentry{radialsum} { name={\ensuremath{x \radd y}},
    description={When $x$ and $y$ are elements of $\RR^{d}$,
      $x \radd y$ denotes the radial sum of $x$ and $y$, defined to be
      the usual vector sum of $x$ and $y$ when $x$ and $y$ are
      contained through a line through the origin $o$, and $o$
      otherwise}, sort={radialsum} } \newglossaryentry{littleo} {
    name={\ensuremath{\mathcal{o}}}, description={ADD descro},
    sort={o} } \newglossaryentry{bigO} {
    name={\ensuremath{\mathcal{O}}}, description={ADD descro},
    sort={O} } \newglossaryentry{linesegment} { name={\ensuremath{[a :
        b]}}, description={Closed line segment \ensuremath{\{(1-t)a +
        t b\,:\, 0 \leq t\leq 1\}}} between points $a\in\RR^d$ and
    $b\in\RR^d$, sort={linesegment} }

  \title{\bf\Large Statistical inference for radial generalised Pareto
    distributions and return sets in geometric extremes}

\allowdisplaybreaks

\begin{document}
\date{}
\maketitle

\def\spacingset#1{\renewcommand{\baselinestretch}%
  {#1}\small\normalsize}
\spacingset{1.15}

\begin{abstract}
  We use a functional analogue of the quantile function for
  probability measures on $\RR^d$ to characterize a novel limit
  Poisson point process for radially recentred and rescaled random
  vectors under a radial-directional decomposition. This limit process
  yields new multivariate distributions, including \textit{radial
    generalised Pareto distributions}, exhibiting stability for
  extrapolation to extremal sets along any direction. We show that the
  normalising functions leading to the limit Poisson point process
  correspond to a novel class of sets visited with fixed probability,
  with geometric properties determined by the conditional distribution
  of the radius given the direction and the Radon-Nikodym derivative
  of the directional probability distribution relative to reference
  spherical measures. This leads to return sets, defined by the
  complement of these probability sets and expressed by their return
  period. We identify an important member, the \textit{isotropic
    return set}, where all directions of exceedances outside the set
  are equally likely. Building on the limit Poisson point process
  likelihood, we develop parsimonious statistical models leveraging
  links between limit distribution parameters, with novel diagnostics
  for assessing convergence to the limiting distribution. These models
  enable Bayesian inference for return sets with arbitrarily large
  return periods and probabilities of unobserved extreme events,
  incorporating directional information from observations outside
  probability sets. The framework supports efficient computations in
  dimensions $d=2$ and $d=3$. We demonstrate the utility of the methods
  through simulations and case studies involving hydrological and
  oceanographic data, showcasing potential for robust and
  interpretable analysis of multivariate extremes.
\end{abstract}
\noindent%


\section{Introduction}
\label{sec:intro}
The multivariate nature of extreme events casts a shadow of
potentially devastating consequences upon ecosystems, infrastructures,
as well as financial, economic, and insurance sectors.\ Knowledge of
the frequency and magnitude of extreme events is crucial in enhancing
planning strategies and adaptation efforts.\ The statistical
properties of univariate extremes are
well-established~\citep{balkema1974residual, pick75, davismit90}, but
statistical inference for multivariate random processes is much more
intricate:\ one must analyse how random processes interact with and
influence each other.\ A common way to describe the extremal
dependence structure of a real-valued random vector
$\bm X=\left(X_1,\dots,X_d\right)$ is through the coefficient of tail
dependence
\begin{equation}
  \label{eq:chi-def}
  \chi_q(A) = \PR\Big[\,\bigcap_{j\in A}\{F_j(X_{j})>q\}\,\Big]\big/(1-q),\quad q\in(0,1), \quad A\subseteq\left\{1,\dots,d\right\}, \left|A\right|>1,
\end{equation}
where $F_j$ is the cumulative distribution function of $X_j$.\ When
$\lim_{q\rightarrow1}\chi_q(A)=0$, the variables in $A$ are unlikely
to grow large together and we say that they exhibit asymptotic
independence.\ Conversely, when $\lim_{q\rightarrow1}\chi_q(A)>0$, the
variables in $A$ are likely to grow large together and the variables
in $A$ exhibit asymptotic dependence.\ The different dependence
structures that can be present within subgroups of the marginal
variables of $\bm X$ can make inference for rare events challenging
and extrapolation inaccurate.\

Classical approaches to multivariate extreme events often rely on the
framework of multivariate regular variation (MRV)~\citep{haanresn77},
which posits that the point processes of exceedances of a random
vector over a high threshold, when suitably renormalised, converge in
distribution to a non-degenerate non-homogenous Poisson point
process~\citep{haan84}.\ This provides a framework for understanding
the joint behaviour of extreme observations and leads to meaningful
limit distributions that can be used for statistical inference of
multivariate extremes.\
In practice, MRV is applied in a way that does not adequately describe
relationships among asymptotically independent random
variables~\citep{nolde2021linking}.\ This is due to the type of the
renormalisation that is employed:\ given a random sample of $n$
$d$-dimensional observations, all components are normalised by the
same amount.\ This leads to considering the dependence structure only
in a single direction in $\RR^d$ (see Figure~\ref{fig:directions} and
Section~\ref{sec:background_univ_extr}).\

Another drawback of MRV is the limited set of directions in the
multivariate space in which one can extrapolate the model.\ Under MRV,
the probability of lying in an extremal set is estimated by shifting
the extremal set and performing empirical probability estimation on
the translated set.\ When the translated set does not contain
observations from the initial dataset, the estimate the probability of
interest is inevitably 0.\ To correct the joint rate of tail decay in
the case of weaker extremal dependence, the notion of hidden regular
variation (HRV) was introduced
\citep{ledtawn96,ledtawn97,maulik2004characterizations}.\ However, HRV
also suffers from the drawback that it does not allow extrapolation
along directions where not all variables are simultaneously large.\ To
extrapolate to a wider range of extremal regions with a wider array of
dependence structures, the frameworks of conditional
extremes~\citep{hefftawn04} and angular dependence \citep{wadtawn13}
have been introduced, but statistical methodology based on these
frameworks suffers from drawbacks.\ Despite its wide applicability and
widespread adoption, the conditional extremal inference method
of~\cite{hefftawn04} is based on composite likelihood methods and on
gluing separate models post-fit, making statistical inference and
computations unwieldy.\ While the angular dependence method
of~\cite{wadtawn13} permits extrapolation in regions where variables
are not simultaneously extreme, it is only useful for joint survival
regions.\
\begin{figure}[t!]
  \centering
  \includegraphics[scale=1, trim= 50 100 50 100]{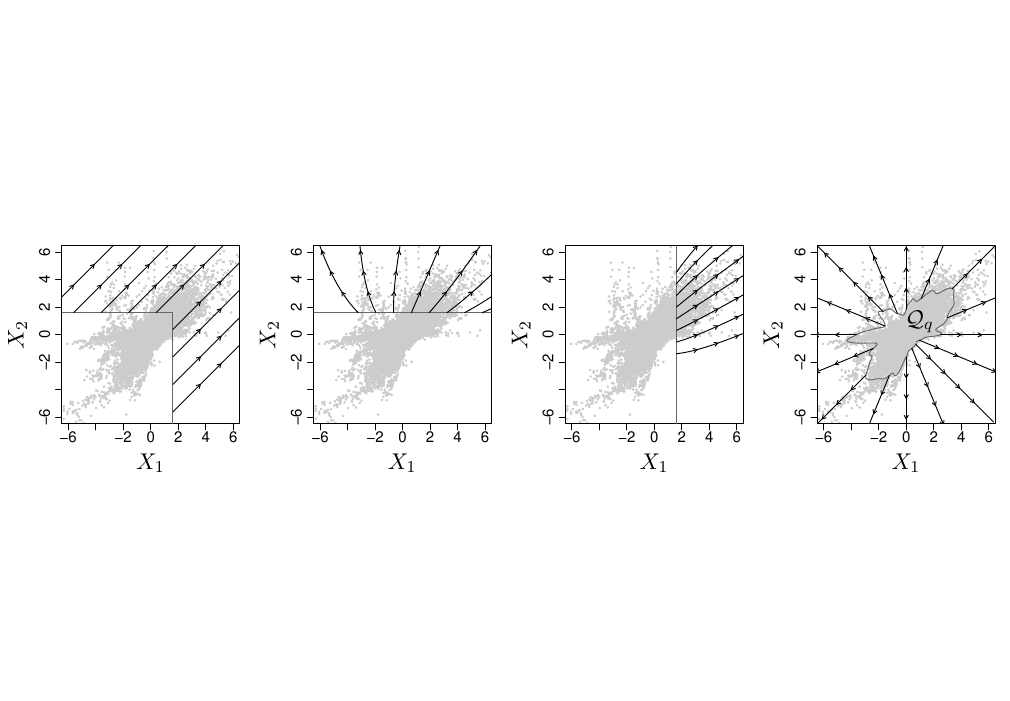}
  \caption{Directions along which MEVT frameworks allow extrapolation
    to tail regions of a bivariate Thames tributary river flow dataset transformed so that each marginal distribution is
    standard Laplace:\ $(a)$ MRV, $(b)$ and $(c)$ conditional extremes
    given $X_2$ and $X_1$ are large, respectively, and $(d)$ geometric
    extremes, with $\QS_q$ illustrating the posterior mean of the
    quantile set ($q=0.95$).\ The support of the distribution of exceedances is inscribed by the region containing the
    arrows.\ }
  \label{fig:directions}
\end{figure}
Recently, a characterisation of extremal dependence through the
limiting geometry of observations from $\bm{X}$ has become of
interest.\ The \emph{limit set}, whose boundary arises as the limiting
hull of appropriately scaled sample clouds, provides insight into the
extremal dependence structure of $\bm X$.\ The \emph{gauge function},
whose unit level set is in one-to-one correspondence with the boundary
of the limit set, has been shown to connect several known coefficients
describing extremal dependence of known
copulas~\citep{nolde2014geometric,nolde2021linking}.\
\cite{wadsworth2022statistical} proposed a framework for performing
statistical inference for multivariate extremes using this geometric
approach.\ Using a radial-directional decomposition, this framework
treats the gauge function as a rate parameter of a left-truncated
gamma model for the distribution of radii conditioned on angles on the
unit simplex.\ Inference for the gauge function and its associated
limit set is based on parametric models derived from known copulas in
$d$-dimensions and standard exponential margins, and a maximum
likelihood approach is implemented within the rate parameter of the
truncated Gamma distribution.\ The result is a new statistical
inference method in estimating extremal probabilities with great
flexibility relative to state-of-the-art methods in multivariate
extremes.\ Also in exponential margins and in a bivariate setting,
\cite{simptawn22bivariate} estimate the gauge function via the
generalised Pareto (GP) distribution~\citep{pick75} for the
conditional distribution of the excess radii given angles on the
simplex, and showing improvements in the estimation of extremal
dependence coefficients derived from the estimated gauge
\citep{nolde2021linking}.
\begin{figure}[t!]
  \centering \includegraphics[width=.3\textwidth, trim= 40 40 40
  100]{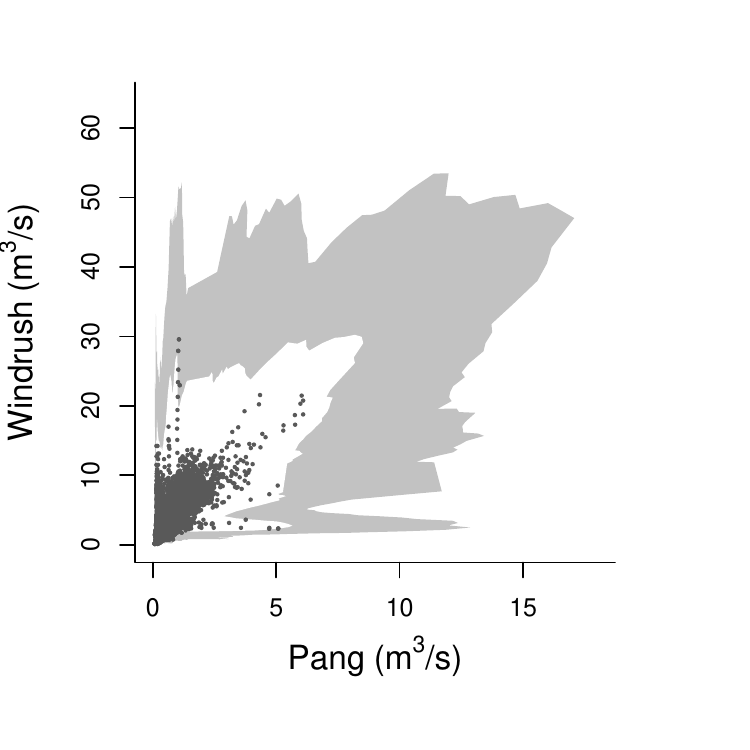} \hspace{0.13\textwidth}
  \includegraphics[width=.3\textwidth, trim= 40 40 40
  40]{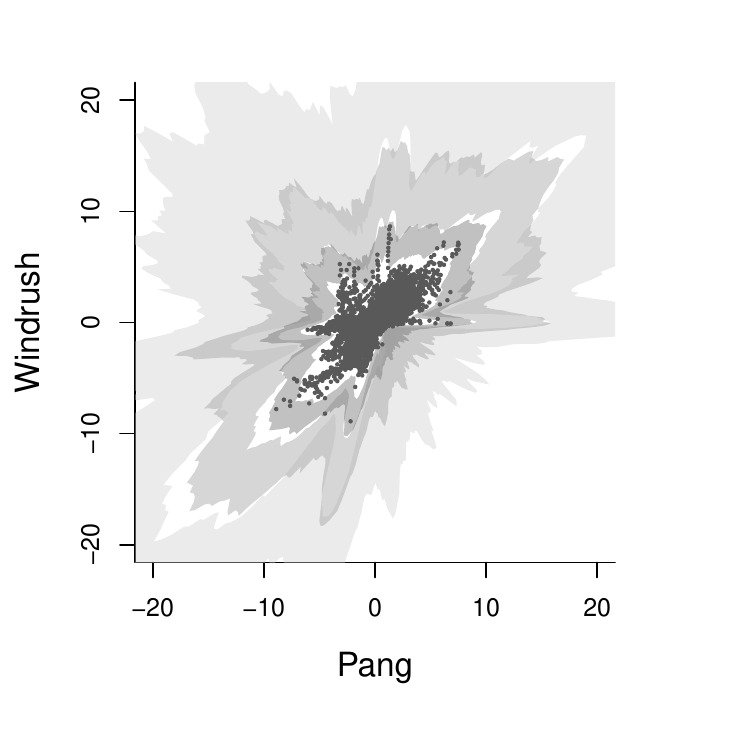}
  \caption{Simultaneous prediction intervals for the boundary of river flow
    return level-sets of the Thames tributaries.
    \textit{Left:} original margins, return period of $T=10^3$ days.\
    \textit{Right:} standard Laplace margins, return periods from dark
    to light grey, $T=10^3$, $10^5$, and $T=10^9$ days.}
    \label{fig:river-returnlevel}
\end{figure}

In this paper, we develop a framework that leverages the structure of
weak limits in suitably radially renormalized multivariate exceedances
in $\RR^d$.\ Multivariate exceedances are observations that lie within
a \textit{return set}, defined as the complement of a
\textit{probability set}.\ Among the infinite class of probability
sets that are visited with a prespecified probability, we identify two
key types:\ quantile sets and isotropic probability sets.\ These sets
are characterized by the distribution of directions associated with
extreme events.\ Specifically, the quantile set is defined via the
conditional quantile of the radius, ensuring that the directional
distribution of exceedances matches the directional distribution of
all events.\ In contrast, the isotropic probability set is defined
such that the directional distribution of exceedances is uniform.\ The
quantile set derives its name from its role as a $d$-dimensional
generalisation of quantiles, addressing challenges in multivariate
extreme value theory that arise from the absence of natural ordering
in $\RR^d$ \citep{barnett76}.\ This concept parallels the quantile
regions introduced by \cite{hallin2021distribution}, but our
definition is specifically designed for extrapolation beyond the
observed data.\ Likewise, the isotropic return set derives its name
from its key property, which balances exceedances across all
directions.

By specifying an appropriate sequence of quantile sets, we show under
mild conditions that the framework of MRV be suitably extended to the
cone $\Rstar$.\ We further show that similar extensions are also
available for light-tailed and bounded multivariate distributions.\ A
key property of our framework is that it is applicable to any
$d\geq 1$.\ In the multivariate setting ($d>1$), our framework enables
the modelling of the entire joint tail, accommodating scenarios where
subsets of components are extreme.\ This approach can capture
behaviours across the entire spectrum of multivariate space and reveal
hidden dependencies, thereby bridging the gap between theory and
practice.\ Using radially recentred and rescaled exceedances over high
quantile sets, we characterise non-trivial limit distributions on
$\RR^d$, termed \textit{radially-stable generalised Pareto
  distributions}.\ The radial stability properties of these families
of distributions permit extrapolation beyond the range of observed
data along any direction (see Figure~\ref{fig:directions}) and lead to
the notion of \emph{return sets}, a geometric $d$-dimensional
extension of the uni-dimensional (upper-tail) return
level~\citep{cole01}.\ This yields an interpretable way to communicate
the risk associated with extreme multivariate events allowing
decision-makers, policymakers, and the general public to understand
the likelihood of experiencing impact from \textit{joint extreme
  events}.\

Throughout our work, inference is done in a Bayesian manner,
allowing us us to obtain inferences for any functional of the joint
tail distribution.\ Our methods account for multiple sources of
uncertainty in the estimation, including that of uncertainty in the
estimation of the marginal tails as well as the multivariate
threshold, which is typically not accounted for in previous
statistical methods for multivariate extremes.\ We use use prediction
intervals for functionals of interest, allowing us for example to
quantify uncertainty via simultaneous prediction bands when the
functional of interest consists of the boundary of a
distribution-dependent set.\ 
The flexibility and accuracy of our modelling approach is
illustrated in Figure~\ref{fig:river-returnlevel}, where the
posterior return level sets are shown on a bivariate dataset of
river flow measurements (${m^3}/{s}$) displaying a complicated
dependence structure.\
Simultaneous prediction intervals for posterior estimates of the
lower-boundary of the return level sets for various return periods are
shown with growing uncertainty as the return period increases, which
is expected.\ Our definition of the return level sets are linked with
the findings in \cite{hallin2021distribution}, which is based on
optimal transportation theory.\ However, we are able to use our model
to obtain these same curves far beyond the domain of our data with
uncertainty estimated coherently through the posterior distribution.\
We note the similarities with \cite{hallin2021distribution} by
estimation of return level curves on a series of bivariate Gaussian
mixture distributions, presented in Supplementary
Material~\ref{sec:biv-gauss-mix}.\

The paper is organised as follows.\ In Section~\ref{sec:theory}, we
initiate our findings through classical univariate extreme value
theory and multivariate regular variation, and introduces the main
results, which include the weak convergence of radially recentred and
rescaled exceedances to a Poisson point process.\ Our findings lead to
a novel class of limit multivariate distributions that are presented
in Section~\ref{sec:mult_rad_stab_distr} and to novel return
level-sets presented in Section~\ref{sec:quantile_return_sets}.\
Section~\ref{sec:inference} details methodology for statistical
inference of extremes using hierarchical Bayesian models with latent
Gaussian random effects on Euclidean spheres.\ 
while information on
Finally, Section~\ref{sec:wave}
illustrates the merits of our approach on a 3-dimensional data set of
extreme sea levels in Newlyn, England.\
Publicly available code can be accessed via our \texttt{R}
package \texttt{geometricExtremes} at \texttt{GitHub} to use the methodology discussed herein.\ 

\section{Theory}
\label{sec:theory}
\subsection{Notation and background}
\textit{Notation linked with measures:} Let $\mu$ be a measure on a
topological space $\Espace$ and let $h\,:\,\Espace \to \RR$ be a
measurable function.\ The support of a measure $\mu$ is denoted by
$\text{supp}(\mu)$ and is the smallest closed subset of $\Espace$ such
that $\mu(\Espace\setminus \text{supp}(\mu))= 0$.\ The integral of $h$
with respect to $\mu$ is denoted by $\mu (h)=\int_{\Espace} h d\mu$.\
The Dirac measure centred at $x\in \Espace$, denoted by $\delta_x$, is
defined by $\delta_x(h) = h(x)$.\ A measure $\nu$ is said to be
absolutely continuous with respect to another measure $\mu$, denoted
as $\nu \ll \mu$, if for every measurable set $A$, $\mu(A) = 0$
implies $\nu(A)=0$.\ We write $M_+(\Espace)$ for the space of
nonnegative Radon measures on Borel subsets of $\Espace$,
$M_1(\Espace)$ for the space of Radon probability measures on Borel
subsets of $\Espace$ and $M_p(\mathcal{\Espace})$ for the space of all
locally finite point measures on $\Espace$.\ Given a standard
probability space $(\Omega, \mathcal{F}, \PR)$, we let a random point
measure $P$ on $\Espace$ be a measurable map
$P\,:\, \Omega \to M_p(\Espace)$, such that
$P = \sum_{i=1}^{N}\delta_{X_i}$ where
$N\,:\,\Omega \to \mathbb{N}\cup\{\infty\}$ is a discrete random
variable, representing the number of atoms of the point measure and
given $\{N=n\}$, $X_i\,:\,\Omega \to \Espace$, $i=1,\dots,n$, are
random variables representing the location of the atoms.\ For
notational simplicity, when subscripts become cumbersome, we also
write $\delta[x]$ to denote the Dirac measure $\delta_x$.\ The
indicator random variable $\mathbbm{1}_A\,:\,\Omega \to \RR$ is
defined for $A\in \mathcal{F}$ by $\mathbbm{1}_A(\omega) = 1$ if
$\omega \in A$ and $\mathbbm{1}_A(\omega)=0$ otherwise.\ For a
probability measure $\mu$ defined on a subset $\mathcal{E}$ of a
topological space $\Espace$, then $\mu(B)=\mu(B \cap \mathcal{E})$ for
all $B\in \mathcal{B}(\Espace)$, that is, the measure is always
interpreted as its extension to the entire ambient space $\Espace$.\

\textit{Notation linked with vague convergence:} For a sequence of
measures $\mu_n$ and $\mu$ on a nice space $\Espace$, we say that
$\mu_n$ converges vaguely to $\mu$ in $M_+(\Espace)$, denoted by
$\mu_n\vg \mu$, if $\mu_n(B)\to \mu(B)$ for all relatively compact
sets $B\subset \Espace$ for which the boundary of $B$ has $\mu$
measure zero, that is, $\mu(\partial B)=0$.\ An equivalent definition
which is based on probing the behaviour of measures by integrating
them along compactly supported continuous test functions is given in
Appendix~\ref{sec:vg}.\ Additionally, we write
$\Espace_\xi = \{x\in\RR\,:\,1+\xi x > 0\}$ for $\xi\in\RR$, and
$\overline{\Espace}_{\xi}$ denotes
${\Espace}_{\xi} \cup \sup \Espace_{\xi}$, where
$\sup \Espace_{\xi}:=\sup \{{x\in\RR\,:\, x\in\Espace}_{\xi}\}$, that
is, the right closure of the interval $\Espace_\xi$.\

\textit{Notation linked with weak convergence of a random point
  measure to a Poisson point process}: A sequence of random point
measures $P_n$ converges weakly in $M_p(\Espace)$ to a
Poisson point process $P$ with intensity measure $\mu(dx)$ if the
sequence of Laplace functionals of suitable test functions converges
to the Laplace functional of the Poisson point process, \textit{i.e.}, if
for any $h\in C_K^+(\Espace)$,
\begin{IEEEeqnarray*}{rCl}
  &&\lim_{n\to\infty}\EE \left[\exp\{- P_{n} (h)\}\right] =
  \exp\left\{-\int_{\Espace} [1- \exp\{-h(x)\}] \mu(d
    x)\right\},
\end{IEEEeqnarray*}
where $C_K^+(\Espace)$ denotes the space of continuous functions
with compact support on $\Espace$.\

\color{black} \textit{Notation linked with distributions and measure
  transportation} \color{black} A random element
$Z\,:\,\Omega \to \Espace$ follows a distribution $\PR_{Z}$ and write
$Z \sim \PR_Z$ whenever $\PR_Z(A)=\PR(Z\in A)$.\ Push-forward of
measures, cyclically monotonic, When $\Espace=\mathbb{R}^d$ with
$d\geq 1$, then we let
$F_Z(\bm x) =\PR_Z((-\infty, x_1]\times \cdots \times (-\infty,
x_d])$, $\bm x=(x_1, \dots, x_d)\in \mathbb{R}^d$, denote the
corresponding cumulative distribution function of $Z$.\ For a
nondecreasing function $F\,:\,\mathbb{R}\to \mathbb{R}$, the
left-continuous inverse is defined by
$F^{-1}(y) = \inf\{x\,:\, F(x)\geq y\}$.\ Two cumulative distribution
functions $G_1, G_2 \,:\, \mathbb{R}^d\to [0,1]$ are said to be of the
same type if there exist constants $\bm a>0$ and $\bm b$ such that
$G_1(\bm a \bm x + \bm b) = G_2(\bm x)$ for all
$\bm x\in\mathbb{R}^d$, with vector algebra interpreted as
elementwise.\ \color{black}

\textit{Notation linked with starshaped sets:} A set $\B\subset \RR^d$
is called starshaped at $\bm x\in \B$ if every point $\bm y\in \B$ is
visible from $\bm x$ in the sense that the closed line segment
$\{(1-t)\bm x + t \bm y\,:\, t \in [0,1]\}\subset \B$.\ If $\B$ is
starshaped at $\bm 0$ and compact in $\RR^d$, then we call $\B$ a
star-body.\ The collection of all star-bodies is denoted by
$\bigstar$.\ A star-body $\B$ is in one-to-one correspondence with its
radial function $r_{\B}:\SSS^{d-1}\to [0,\infty)$ given by
$r_{\B}(\bm w) = \sup\{\lambda > 0\,:\, \lambda \bm w \in \B\}$.\
Star-bodies are endowed with well-defined algebraic operations through
their radial functions.\ For instance, for $\mathcal{A},\B\in\bigstar$, $\mathcal{A}\radd\B$ and $\mathcal{A}\rmult\B$ are interpreted as the star-bodies with radial functions $r_{\mathcal{A}+\B}=r_{\mathcal{A}}+r_{\B}$ and $r_{\mathcal{A}\rmult\B}=r_{\mathcal{A}}\rmult r_{\B}$.\ Additional information about
starshaped sets and radial functions can be found in Supplementary
Materials~\ref{appendix:star_background} and~\ref{appendix:radial}.\
\subsection{Univariate extremes}
\label{sec:background_univ_extr}
A complete characterisation of all possible non-trivial limit laws of
extremes of a univariate \mbox{random} variable $X$ with probability
distribution $F_X$ relies on regular variation and its extensions,
postulating location and scale constants $a_{n} > 0 $ and $b_{n}$, and
$\alpha_n >0$ and $\beta_n\in \RR$, for $n = 1, 2, \dots$, alongside
shape parameters $\xiR, \xiL\in\RR$ and non-degenerate limit measures
$\nuR$ and $\nuL$, such that
\begin{IEEEeqnarray}{rCl}
  && n \mathbb{P}\left[\{X-b_{n}\}/a_{n} \in \cdot\right] \vg
  \nuR(\cdot) \quad \text{and} \quad n
  \mathbb{P}\left[\{(-X)-\beta_n\}/\alpha_n \in \cdot\right] \vg
  \nuL(\cdot)
  \label{eq:RV_univariate_standard_RL} 
\end{IEEEeqnarray}
as $n\to\infty$, in $M_+(\overline{\Espace}_{\xiR})$ and
$M_+(\overline{\Espace}_{\xiL})$, respectively.\ Standard references
include \cite{resn87} and \cite{haanferr06}.\ The location and scale
renormalisations in~\eqref{eq:RV_univariate_standard_RL} imply that
the \textit{exponent functions} of the limit measures, defined by
\[
  \LambdaR(z):=\nuR(z,\sup \Espace_{\xiR}] \quad \text{and} \quad
  \LambdaL(z):=\nuL(z,\sup \Espace_{\xiL}],
\]
satisfy the functional equations
\begin{IEEEeqnarray*}{rCl}
  &&n \,\LambdaR(a_n + b_n z) = \LambdaR(z)\qquad \text{and} \qquad n \,
  \LambdaL(\alpha_n + \beta_n z) = \LambdaL(z),
\end{IEEEeqnarray*}
whose unique solutions are given, up-to-type, by 
\[
  \text{$\LambdaR(z)= \left[1+\xiR z\right]_+^{-1/\xiR}$ and
    $\LambdaL(z)= \left[1+\xiL z\right]_+^{-1/\xiL}$},
\]
where $(x)_+ = \max(x,0)$, respectively characterizing the limit
measures $\nuR$ and $\nuL$.\

The connection between the vague convergences
\eqref{eq:RV_univariate_standard_RL} and the limiting laws of extremes
of a sequence $\{X_i\,:\, i=1,2,\dots\}$ of independent and
identically distributed random variables from $F_X$, comes from the
fact that these convergences constitute necessary and sufficient
conditions for the sequences of random point measures
\[
  \sum_{i=1}^n\delta\left[\left(\frac{X_i - b_n}{a_n},
      \frac{i}{n+1}\right)\right] \quad \text{and} \quad
  \sum_{i=1}^n\delta\left[\left(\frac{(-X_i) - \beta_n}{\alpha_n},
      \frac{i}{n+1}\right)\right]
\]
to respectively weakly converge in
$M_p(\overline{\Espace}_{\xiR}\times (0,1))$ and
$M_p(\overline{\Espace}_{\xiL}\times (0,1))$, to Poisson point
processes with intensity measures $\nuR(dz)dt$ and $\nuL(dz)dt$, as
$n\to \infty$.\

A key aspect of the convergences~\eqref{eq:RV_univariate_standard_RL}
is that these can be unified under a single theme, allowing us to
consider upper and lower extremes simultaneously.\
Theorem~\ref{thm:univariate_vc} illustrates this unification, which is
contingent upon explicitly using the directions along which extremes
occur, and shows that the right and left extremes from a random sample
from $\PR_{\bm X}$ become ultimately independent as the sample size
approaches infinity.\ Explicit in the construction below is the
transformation of $X$ into polar coordinates, which has the usual
geometric interpretation that the sign of $X$ is viewed as
encoding the direction $X/|X|$, while the magnitude $|X|$ represents
the distance of $X$ from 0.
\begin{theorem}
  \label{thm:univariate_vc}
  Let $X$ be a real-valued random variable and let $S$ be the set of
  all points of the real line, any open neighbourhood of which
  receives a positive probability mass.\ Suppose that
  $0 \in \text{\normalfont supp}(\PR_{X})$, $F_{X}(\{0\})=0$, and convergences~\eqref{eq:RV_univariate_standard_RL} hold true.\ Let
  $\pi = \PR(X > 0)=\PR(X/|X| = 1)$,
  $r_n^a(w) = a_{n,\{+1\}} \mathbbm{1}_{\{w\in\{+1\}\}} + a_{n,\{-1\}}
  \mathbbm{1}_{\{w\in\{-1\}\}} $ and
  $r_{n}^b(w) = b_{n,\{+1\}} \mathbbm{1}_{\{w\in\{+1\}\}} + b_{n,\{-1\}}
  \mathbbm{1}_{\{w\in\{-1\}\}}$.\ Let
  \[
    P_n = \sum_{i=1}^n \delta\left[\left(\frac{|X_i| -
          r_{n}^a(X_i/|X_i|)} { r_n^b(X_i/|X_i|)}, \frac{X_i}{|X_i|},
        \frac{i}{n+1}\right)\right].
  \]
  \item[$(i)$]
  There exist location and scale constants $a_{n, \sfR}>0$ and
  $b_{n, \sfR} \in \RR$, and $a_{n, \sfL} > 0$ and
  $b_{n, \sfL} \in \RR$, for $n = 1, 2, \dots$, such that as
  $n\to\infty$,
  \begin{IEEEeqnarray}{rCl}
    n \mathbb{P}\left[(|X|-b_{n, \sfR})/a_{n, \sfR} \in \cdot \mid X >
      0\right] &\vg& \nuR(\cdot), \qquad \text{in }
    M_+(\overline{\Espace}_{\xiR}),
    \label{eq:RV_univariate_R_cond}\\ \lefteqn{\text{and}} \nonumber\\
    n \mathbb{P}\left[(|X|-b_{n, \sfL})/a_{n, \sfL} \in \cdot \mid X <
      0\right] &\vg& \nuL(\cdot), \qquad \text{in }
    M_+(\overline{\Espace}_{\xiL}).
    \label{eq:RV_univariate_L_cond}
  \end{IEEEeqnarray}
  \begin{itemize}
  \item[$(ii)$] As $n\to \infty$,
    \begin{equation}
      n \PR\left[\left(\frac{|X| -
            r_{n}^b(X/|X|)}{ r_n^a(X/|X|)}, X/|X|\right) \in \cdot \right] \vg
      \nu(\cdot),   
      \label{eq:RV_univariate_RL}
    \end{equation}
    in $M_+(\overline{\Espace}^1)$, where
    $\overline{\Espace}^1:=(\overline{\Espace}_{\xiL} \times \{-1\})
    \cup (\overline{\Espace}_{\xiR} \times \{+1\})$ and
    \begin{equation} \nu = (1-\pi) \nuL \, \delta_{-1} + \pi \nuR \,
      \delta_{+1};
      \label{eq:Lambda_univariate}
    \end{equation}
  \item[$(iii)$]
    As $n\to \infty$,
    \[
      P_n \wk P, \qquad \text{in
        $M_p(\overline{\Espace}^1\times (0,1))$},
    \] where $P$ is Poisson point process with intensity measure
    \begin{equation}
      \left[\left\{(1+\xiL z)_+^{-1-1/\xiL} \,dz \right\} (1-\pi)\,
        \delta_{-1}(d w) + \left\{(1+\xiR z)_+^{-1-1/\xiR} \,dz
        \right\} \pi \,\delta_{+1}(d w) \right] \, dt.
      \label{eq:Lambda_intensity_univariate}
    \end{equation}
  \end{itemize}
\end{theorem}
A proof is given in Appendix~\ref{sec:min-max-stable}.\ The unification
presented in Theorem~\ref{thm:univariate_vc} has a simple
interpretation in terms of superposition of random point
measures.\ Under this framework, the sequences
\begin{IEEEeqnarray*}{rCl}
  &&P_n^+=\sum_{i=1}^n \delta\left[\left(\frac{|X_i| - b_{n,
          \sfR}}{a_{n, \sfR}}, \frac{i}{n+1}\right)\right] \mathbbm{1}_{\{X_i >
    0\}} \text{ and } P_n^-=\sum_{i=1}^n
  \delta\left[\left(\frac{|X_i| - b_{n, \sfL}}{a_{n, \sfL}},
      \frac{i}{n+1}\right)\right] \mathbbm{1}_{\{X_i < 0\}},
\end{IEEEeqnarray*}
weakly converge to independent Poisson point processes $P^+$ and $P^-$
with intensity measures $\pi \nuR(dz)dt$ and $(1-\pi)\nuL(dz)dt$ on
$(\overline{\Espace}_{\xiR} \times \{+1\})\times (0,1)$ and
$(\overline{\Espace}_{\xiL} \times \{-1\})\times (0,1)$,
respectively.\ This follows directly
from~\eqref{eq:RV_univariate_R_cond}
and~\eqref{eq:RV_univariate_L_cond} and has a straightforward sampling
interpretation.\ First, we sample $n$ independent observations from
each renormalised conditional distribution.\ Then, we apply a thinning
operation to the resulting independent random point measures whereby
points are probabilistically removed based on the likelihood of
observing a positive or negative value of $X$, conditionally on
retaining $n$ points in total.\ Consequently, $P_n$ is seen as the
superposition $P_n^- + P_n^+$ of the two renormalised random point
measures, leading to $P_n \wk P$ as $n\to \infty$, where $P$ is a
Poisson point process on the product space between the disjoint union
$(\overline{\Espace}_{\xiL} \times \{-1\}) \cup
(\overline{\Espace}_{\xiR} \times \{+1\})$ and $(0,1)$, with intensity
measure~\eqref{eq:Lambda_intensity_univariate}.\

Although Proposition~\eqref{thm:univariate_vc} assumes
$0\in \text{\normalfont supp}(\PR_{X})$, we can more generally
consider directions relative to any point
$m\in \text{\normalfont supp}(\PR_{X})$, with $|X-m|$, $(X-m)/|X-m|$,
and $X-m\lessgtr 0$, respectively replacing $|X|$, $X/|X|$, and
$X\lessgtr 0$.\ The assumption
$0\in \text{\normalfont supp}(\PR_{X})$, however, can be typically met
by shifting the origin.\ While the definition of right and left
extremes may vary depending on the application, for generic purposes
it may be natural to take $m$ as a point in
$\text{\normalfont supp}(\PR_{X})$ which minimises a well defined
measure of extremity relative to any other point in
$\text{\normalfont supp}(\PR_{X})$.\ For absolutely continuous
distributions, whereby $F_X(\{m\})=0$ for any
$m\in \text{\normalfont supp}(\PR_{X})$, there is always a unique such
choice of $m$ (the median of the univariate distribution) which
uniquely maximises the population version of halfspace depth.\
Additionally, the assumption $F_{X}(\{m\}) = 0$ can also be relaxed at
the cost of adjusting the probabilities $\pi$ and $(1-\pi)$ to new
values $\pi_+>0$ and $\pi_->0$, while ensuring that
$\pi_+ + F_X(\{m\}) + \pi_- = 1$.\ In this case, if the median is not
well defined, we can select a unique point such as the centroid of the
set of values maximizing halfspace depth.\

While such generalisations introduce no technical difficulties, here
and throughout we opt for the simpler cases for clarity and ease of
presentation.\ Thus, assuming $0\in \text{\normalfont supp}(\PR_{X})$ with
$F_X(\{0\})=0$, the connection between non-degenerate limit laws that
facilitate a simultaneous description of right and left extremes
relative to a centre, can be understood through the convergence of the
joint distribution of the appropriately renormalised sample minima and
maxima.\

Let the sample minima and maxima be defined by
$m_n=\min_{i=1,\dots,n}(X_i)$ and $M_n=\max_{i=1,\dots, n}(X_i)$,
respectively, and fix $z_{-}$ and $z_{+}$ such that
$-z_-\in \overline{\Espace}_{\xi_{\sfL}}$ and
$z_+\in \overline{\Espace}_{\xi_{\sfR}}$.\ First, we have the relation
\begin{IEEEeqnarray*}{rCl}
  &&\PR\left( \frac{m_n + b_{n, \sfL}}{a_{n, \sfL}} > z_{-}, \frac{M_n
      - b_{n, \sfR}}{a_{n, \sfR}} < z_{+} \right) = \PR\left(a_{n,
      \sfL} z_{-} - b_{n, \sfL} < X < a_{n, \sfR} z_{+} + b_{n,
      \sfR}\right)^n.
\end{IEEEeqnarray*}
The right-hand side of this relation can be rewritten as
\begin{IEEEeqnarray}{rCl}
  && \left[1 - \frac{1}{n}n\left\{1-\PR\left(a_{n, \sfL} z_{-} - b_{n,
          \sfL} < X < a_{n, \sfR} z_{+} + b_{n,
          \sfR}\right)\right\}\right]^n=: \left[1- \frac{1}{n}n
    \overline{\PR}_n\right]^n,
  \label{eq:min_max_prelimit}
\end{IEEEeqnarray}
and thus, if $n \overline{\PR}_n$ converges to a non-trivial limit
function, then the joint distribution of the renormalised sample
minima and maxima converges weakly to a
non-degenerate limit distribution.\

Assuming convergences~\eqref{eq:RV_univariate_R_cond} and
\eqref{eq:RV_univariate_L_cond} hold, the justification given in
Appendix~\ref{sec:univariate_supplementary} shows that
\begin{equation}
  \lim_{n\to \infty}\PR\left( \frac{m_n + b_{n, \sfL}}{a_{n, \sfL}} >
    z_{-}, \frac{M_n - b_{n, \sfR}}{a_{n, \sfR}} < z_{+} \right) =
  G_{0}(-\zL, \zR \mid 1-\pi, \pi),\quad \text{as $n \to \infty$},
  \label{eq:ETT_limit}
\end{equation}
at continuity points of the limit distribution $G_0$, where
\begin{equation}
  G_{0}(-\zL, \zR\mid \pi)=\exp[-\left\{(1-\pi)\LambdaL(-\zL) +
    \pi\LambdaR(\zR) \right\}\mid  1-\pi, \pi]
  \label{eq:ETT_limit_form}.
\end{equation}
Here, we note that one always have the flexibility to change the
probability mass function $(1-\pi, \pi)$ in~\eqref{eq:ETT_limit_form}
to any (not necessarily probability) mass function, through suitable
changes in the renormalisation.\ This can be seen from the fact that
for any $\lambda_-, \lambda_+ > 0$, we have
\[
  G_0(-\zL, \zR\mid \pi, 1-\pi)=G_0\left( -\frac{\zL-\mu_-(1-\pi,
      \lambda_-)} {\sigma_-(1-\pi,\lambda_-)}, \frac{\zR-\mu_{+}(\pi,
      \lambda_+)}{\sigma_+(\pi,\lambda_+)}~\Big|~ \lambda_-,
    \lambda_+\right)
\]
where
$\mu_\text{sgn}(x,y)=\{(y/x)^{\xi_{\text{sgn}}}-1\}/\xi_{\text{sgn}}$
and $\sigma_\text{sgn}(x,y)=(y/x)^{\xi_{\text{sgn}}}$, for $x,y>0$
and $\text{sgn}\in\{-,+\}$.\ Thus, if~\eqref{eq:ETT_limit} holds true,
then a suitable adaptation of the normalizing constants yields
\begin{equation}
  \lim_{n\to \infty}\PR\left( \frac{m_n + b_{n, \sfL}^\star}{a_{n, \sfL}^\star} >
    z_{-}, \frac{M_n - b_{n, \sfR}^\star}{a_{n, \sfR}^\star} < z_{+} \right) =
  G_{0}(-\zL, \zR \mid \lambda_-, \lambda_+),\quad \text{as $n \to \infty$},
\end{equation}
where
\begin{equation}
  b_{n,\text{sgn}}^\star = b_{n, \text{sgn}} + \mu_\text{sgn}(\pi_{\text{sgn}}, \lambda_{\text{sgn}})  a_{n,\text{sgn}} \quad \text{and} \quad a_{n,\text{sgn}}^\star = \sigma_\text{sgn}(\pi_{\text{sgn}}, \lambda_{\text{sgn}})  a_{n,\text{sgn}},
  \label{eq:nfs_change_of_probability}
\end{equation}
with $\pi_-=1-\pi$ and $\pi_{+}=\pi$.\ These steps, alongside the
conditions set out in Theorem~\ref{thm:univariate_vc}, lead to an
extremal types theorem for left and right extremes, presented below by
Proposition~\ref{prop:min_max_stable}.\ \color{black}
\begin{proposition}
  \label{prop:min_max_stable}
  \color{black} If for some constants $c_n$ and $d_n>0$, and
  $\gamma_n$ and $\delta_n>0$, we have
  \begin{IEEEeqnarray*}{rCl}
    \PR\left(\frac{m_n + \gamma_n}{\delta_{n}} > \zL, \frac{M_n -
        c_{n}}{d_{n}} < \zR\right)
    &\wk& G(-\zL, \zR)
    , \quad \text{as $n \to \infty$},
    \label{eq:max_min_convergence}
  \end{IEEEeqnarray*}
  for some distribution $G$, with non-degenerate margins, then $G$ is
  of the same type as $G_0$ with $\pi = 1/2$, that is, there exist
  location and scale parameters $\mu_+$ and $\sigma_+>0$, and $\mu_-$
  and $\sigma_->0$, such that
  $G(-\zL, \zR)=G_0(-(\zL - \mu_-)/\sigma_-, (\zR-\mu_+)/\sigma_+\mid
  1/2)$.\ Conversely, a distribution function of the same type as $G_0$ with
  $\pi=1/2$ can appear as such a limit.
\end{proposition}
The proof of Proposition~\ref{prop:min_max_stable} follows closely the
proof of the classical extremal types theorem presented in Section 1.4
in \cite{leadling83} and is omitted for brevity.\ However, the form of
$G$ in Proposition~\ref{prop:min_max_stable} has an instructive
interpretation: since $m_n = -\max_{i=1,\dots,n}(-X_i)$, it follows at
once that $G$ arises as the limiting cumulative distribution function
of the appropriately renormalised component-wise sample-maxima of the
random sample $\{(-X_i, X_i)\,:\,i=1,\dots, n\}$.\ Thus, $G$ is a
bivariate max-stable distribution, which is confirmed by the identity
\begin{equation}
  G(-\zL, \zR)^s = G\left(-\frac{\zL -
      \mu_{-,s}}{\sigma_{-,s}},\frac{\zR -
      \mu_{+,s}}{\sigma_{+,s}}\right), \qquad \text{for any $s > 0$},
 \label{eq:max_stability_minmax}
\end{equation}
where
$\mu_{\text{sgn},s} = \mu_\text{sgn} +
\{(s^{\xi_\text{sgn}}-1)/\xi_\text{sgn}\}\sigma_\text{sgn} $ and
$\sigma_{\text{sgn}, s} = s^{\xi_{\text{sgn}}} \sigma_{\text{sgn}}$.\
Furthermore, Theorem~\ref{prop:min_max_stable} shows that ultimately,
the renormalised sample minima and maxima become independent since
$G(-z_-, z_+)= G_-(-z_-) G_+(-z_+)$, where the marginals of $G$,
defined by $G_+(z_+) := G(\infty, z_+)$ and
$G_-(z_-) := G(z_-, \infty)$, are both generalised extreme value
distributions, corresponding to the limiting non-degenerate
1-dimensional distributions of the appropriately renormalised sample
maxima drawn from the distribution of $X$ and $-X$, respectively.\

\begin{figure}[htbp!]
    \centering
    \begin{overpic}[width=\textwidth]{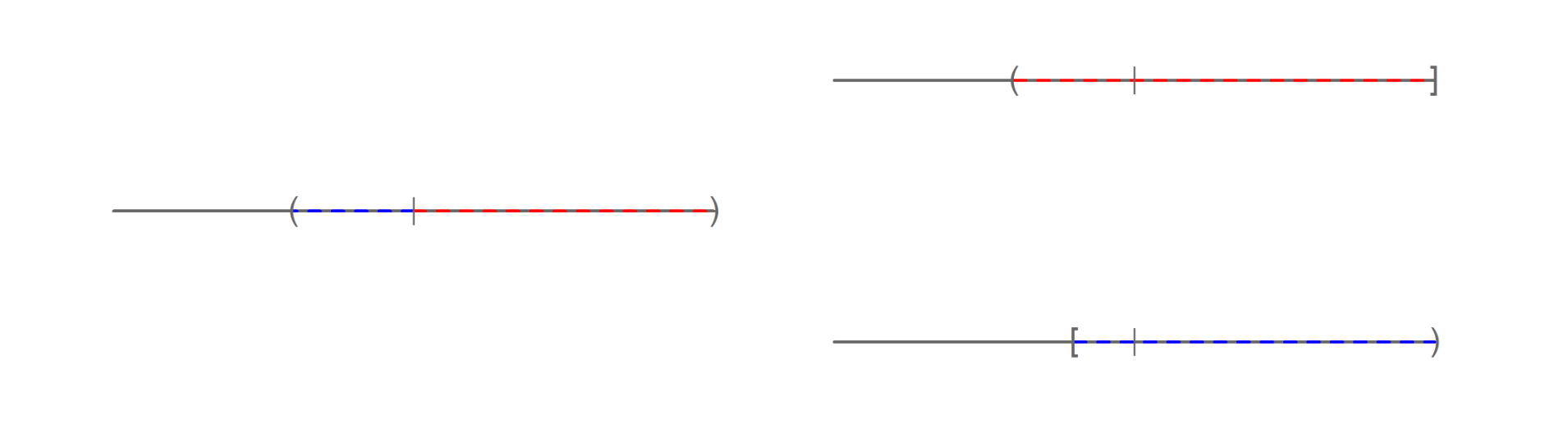}
    \put (5,10) {{\small $-\infty$}}
    \put (25.9,10) {{\small $0$}}
    \put (44.5,10) {{\small $\infty$}}

    \put (51.5,18.5) {{\small $-\infty$}}
    \put (72,18.5) {{\small $0$}}
    \put (60,18.5) {{\small $-1/\xi_{\{+1\}}$}}
    \put (90.5,18.5) {{\small $\infty$}}
    \put (73,24) {{\small \color{red}$\overline{\Espace}_{\xi_{\{+1\}}}\times \{\bm +1\}$\color{black}}}

    \put (51.5,2) {{\small $\infty$}}
    \put (72,2) {{\small $0$}}
    \put (62,2) {{\small $-1/\xi_{\{-1\}}$}}
    \put (90.5,2) {{\small $-\infty$}}
    \put (73,7.5) {{\small \color{blue}$\overline{\Espace}_{\xi_{\{-1\}}}\times \{\bm -1\}$\color{black}}}
    \end{overpic}
    \caption{Illustration of
      $\overline{\Espace}_{\xi_{\{+1\}}}\times \{+1\}$ and
      $\overline{\Espace}_{\xi_{\{-1\}}}\times \{-1\}$ using a
      probability distribution on $\mathbb{R}$ associated with a
      random variable $X$.\ In this example, the probability
      distribution is assumed to have a bounded left lower tail
      ($\xi_{\{-1\}} < 0$) and a heavy upper tail ($\xi_{\{+1\}}>0$).\
      The support of $\PR_{X}$, containing $0$, is the union of the
      coloured segments shown in the left panel.\ The red and blue
      segments, corresponding to positive and negative values of $X$,
      are mapped to $\overline{\Espace}_{\xi_{\{+1\}}}$ and
      $\overline{\Espace}_{\xi_{\{-1\}}}$, respectively, and are
      oriented according to $\{+1\}$ and $\{-1\}$ in the right panel.}
    \label{fig:T_S_to_C}
\end{figure}
While the limiting Poisson point process in Proposition
\ref{thm:univariate_vc} guarantees the convergence of the joint
distribution of renormalised sample maxima and minima, as illustrated
in Proposition~\ref{prop:min_max_stable}, it also yields non-degenerate
limiting distributions for suitably rescaled excesses above
direction-dependent thresholds.\ Specifically, we have that as $q\to 1$
\[
  \PR\left[\frac{|X| - r_{1/(1-q)}^a(X/|X|)}{r_{1/(1-q)}^b(X/|X|)} <
    z, X/|X|=w~\Big|~ |X| > r_{1/(1-q)}^a(X/|X|)\right] \wk H_w(z)
  f_{X/|X|}(w), \quad z > 0,    
\]
where $H_w(y) = 1-(1+\xi_w z)_+^{-1/\xi_w}$ is the cumulative
distribution function of the standard generalised Pareto distribution
with shape parameter $\xi_w$ and
$f_{X/|X|}(w)=\pi \mathbbm{1}_{w\in\{+1\}} +
(1-\pi)\mathbbm{1}_{w\in\{-1\}}$.\


We conclude this section by discussing a practical method for
identifying candidates for the functions $r_n^a$ and $r_n^b$, and a
connection between them, and the transformation
$x \mapsto 2 F_X(x) - 1$ based on the \textit{centre-outward
  distribution function} $2 F_X - 1$, which pushes $\PR_X$ forward to
the uniform distribution on $(-1,1)$, and can be naturally generalised
to multivariate distributions \citep{hallin2021distribution}; see
Section~\ref{sec:quantile_return_sets} for a discussion of this
transformation to the multivariate setting.\

Let $R=|X|$ and $W=X/|X|$.\ A necessary and sufficient condition for
the convergence of each summand in~\eqref{eq:RCsets_at_inf} is that
the functions $r_n^a$ and $r_n^b$ in Theorem~\ref{thm:univariate_vc}
satisfy
\begin{equation} \lim_{n\to \infty} \frac{U_w(n x) -
r_n^b(w)}{r_n^a(w)} = \frac{x^{\xi_w}-1}{\xi_w}, \quad x > 0,
  \label{eq:NS_condition}
\end{equation} where $U_w(1/(1-q))=F_{R ~|~W}^{-1}(q \mid w)$ for $q
\in (0,1)$.\ When the function $U_w$
satisfies~\eqref{eq:NS_condition}, then it follows that we can take
without loss of generality,
\begin{IEEEeqnarray}{l} r_n^b(w) =U_w(n), \quad \text{and} \quad
  r_n^a(w)=U_w(n(1+\xi_w)^{1/\xi_w})-U_w(n) \quad \text{whenever
    $\xi_w \geq -1$}.
  \label{eq:normings_univariate}
\end{IEEEeqnarray} This provides with a practical method for finding
these functions and neatly shows how information from the directional
variable is encoded in the normalizing constants through the
conditional quantile function, which captures the extremity of a point
away from $0\in S$, and the probability mass function of the
direction, which reflects the propensity to observe right or left
extremes.\

These two elements carry essential information about $\PR_X$, which
can also be expressed via the centre-outward distribution function.\
By applying total probability, we obtain the relation
\begin{equation} 2 F_X(x) - 1 = \pi \left[2 F_{R\mid W} (x\mid
1)-1\right] + (1-\pi) \left[2 \left\{1-F_{R\mid W}(-x\mid -1)\right\}
- 1\right],
  \label{eq:CO_1d}
\end{equation} explicitly showing the one-to-one
correspondence between the centre-outward distribution function and
$F_{R\mid W}$ and $f_W$, or equivalently, $F_{R\mid W}^{-1}$ and
$f_W$.\ The centre-outward quantile function $Q_{\pm}\,:\, (-1,1)\to
\RR$ is obtained by inverting $2 F_X(x) -1 = u$, with $u\in (-1,1)$.\
Assuming that the median of the distribution exists and is unique,
with $0$ corresponding to this median, so that $\pi=1/2$,
then~\eqref{eq:CO_1d} gives
\[ Q_{\pm}(u) = [F_{R\mid W}^{-1}(|u| \mid u/|u|)] (u/|u|).
\] The sets
\[ C_q = [-F_{R\mid W}^{-1}(q \mid -1), +F_{R\mid W}^{-1}(q \mid
+1)],\qquad \text{$q\in(0,1)$},
\] have the interpretation of \textit{quantile regions} in the sense
that $\PR(X \in C_q)=q$ and form a sequence of closed, connected, and
nested sets.\ Similarly, the boundaries of these quantile regions
  \[ \partial C_q = \{-F_{R\mid W}^{-1}(q \mid -1), +F_{R\mid
W}^{-1}(q \mid +1)\},\qquad \text{$q\in(0,1)$},
\] have the interpretation of \textit{quantile contours}.\

The key advantage of the centre-outward distribution and its
associated quantile function compared to traditional quantiles based
on the distribution function $F_X$, is that it does not rely on the
total ordering of $\RR$.\ This makes the centre-outward approach much
more flexible, as it naturally extends to higher dimensions, where no
such total order exist, and is particularly useful for generalizing
concepts like quantiles and probability integral transformations in
multivariate settings \citep{hallin2021distribution}.\ As we move to
higher dimensions, however, the relationship of the centre-outward
distribution and the conditional quantile function becomes less
direct, with additional complexities arising due to the structure of
probability distributions on higher-dimensional spaces.\ 
Both approaches, working with the conditional distribution of
$R\mid W$ and the directional variable $W$, or with the multivariate
analogue of the centre-outward distribution function, offer
complementary insights and each viewpoint has its own strengths.\ The
former emphasises how extremes behave relative to fixed directions
alongside knowing the frequency of the directions, while the latter
provides a broader geometric interpretation of data spread and balance
via measure transportation to a target reference distribution.\

In Section~\ref{sec:quantile_return_sets}, we explore these
perspectives in more detail, highlighting their relative strengths and
weaknesses in the context of extreme values.\ By adopting the approach
of transforming a random element $\bm X$, arising from a distribution
$\PR_{\bm X}$, into polar coordinates relative to the centre~$\bm 0$
which is assumed to be an element of $\text{supp}(\PR_{\bm X})$, and
by using the conditional quantile function, we introduce a novel
definition of quantile sets that closely parallels the quantile
regions and quantile contours of \cite{hallin2021distribution}, but
has the added benefit of being easily amenable to extrapolation beyond
the observed data.\ Specifically, our definition facilitates the
development of new refined limit theory for multivariate extremes via
the powerful convergence to the Poisson point process, in a manner
akin to the univariate case presented in this section.\ In
Section~\ref{sec:PP_characterisation} below, we present this theory
and establish key properties of the distribution of the limit Poisson
point process in association with the renormalisation of the radial
component.\ \color{black}


\subsection{Weak convergence to a Poisson point process}
\label{sec:PP_characterisation}

In the multivariate setting, a natural extension of convergences
\eqref{eq:RV_univariate_R_cond} and~\eqref{eq:RV_univariate_L_cond}
can be obtained by requiring location and scale functions of the
direction $b_n\,:\,\mathbb{S}^{d-1}\to \RR_+$ and
$a_n\,:\,\mathbb{S}^{d-1}\to \RR_+$, respectively, and a function
$\xi\,:\, \mathbb{S}^{d-1}\to \RR$, such that as $n\to\infty$,
\begin{equation}  \nu_{n, \bm w}(\cdot):= n \, \PR\left(\frac{\lVert
      \bm X\rVert-b_n(\bm X/\lVert \bm X\rVert)}{a_n(\bm X/\lVert \bm
      X\rVert)}\in\cdot~\Big|~ \bm X/\lVert \bm X\rVert=\bm w \right)
  \vg \nu_{\bm w}(\cdot)\quad \text{in
    $M_+(\overline{\Espace}_{\xi(\bm w)})$}, \quad \PR_{\bm
    X/\lVert\bm X\rVert}-a.e.,
  \label{eq:vg_cond}
\end{equation} where the limit measure $\nu_{\bm w}$ satisfies
$\Lambda_{\xi(\bm w)}(z) = [1 + \xi(\bm w) z]_+^{-1/\xi(\bm w)}$, with
$\Lambda_{\xi(\bm w)}(z):=\nu_{\bm w}(z, \sup \Espace_{\xi(\bm w)}]$.\
Unlike the univariate setting, however, where the set of permissible
directions is finite -- hence the interchange between the limit as
$n\to\infty$ and the sum over the set of directions in the proof of
Theorem~\ref{thm:univariate_vc} is legitimate -- in the multivariate
setting we have a continuum of directions and convergence
\eqref{eq:vg_cond} alone is not sufficient to guarantee that the mean
measure
\begin{equation} \nu_n(\cdot) := n \, \PR\left[\left(\frac{\lVert \bm
X\rVert-b_n(\bm X/\lVert \bm X\rVert)}{a_n(\bm X/\lVert \bm X\rVert)},
\bm X/\lVert \bm X\rVert\right) \in \cdot\right]
\end{equation} vaguely converges in $M_+(\Emult)$ to a non-degenerate
limit mean measure $\nu$, a condition that is necessary and sufficient
for the sequence of renormalised random point measures
\begin{IEEEeqnarray}{rCl} P_n := \sum_{i=1}^n\delta\left[ \frac{\lVert
\bm X_i\rVert - b_n(\bm X_i/\lVert\bm X_i\rVert)}{a_n(\bm
X_i/\lVert\bm X_i\rVert)}, \bm X_i/\lVert\bm X_i\rVert,
\frac{i}{n+1}\right],\quad n=1,2,\dots,
  \label{eq:PPconvergence}
\end{IEEEeqnarray} to weakly converge in $M_p(\Emult)$ to a
non-degenerate Poisson point process $P$.\ Here, the space $\Emult$ is
defined in a similar manner to the space $\overline{\mathbb{E}}$ in
Theorem~\ref{thm:univariate_vc}, through the disjoint union
\[
  \Emult := \bigcup_{\bm w\in \suppW}(\overline{\Espace}_{\xi(\bm
    w)}\times\{\bm w\}) \subset (-\infty, \infty] \times \SSS^{d-1}
\]
of the family of sets
$(\overline{\Espace}_{\xi(\bm w)}\,:\, \bm w \in \suppW)$, and each
space
$\overline{\Espace}_{\xi(\bm w)}\times \{\bm w\} = \{(z, \bm
w)\,:\,z\in\overline{\Espace}_{\xi(\bm w)}\}$ is interpreted as a
subset of the extended real line, passing through the origin and
oriented in the direction of the unit vector $\bm w$.\ For simplicity
and to avoid overloading the notation, we do not explicitly show the
dependence of the space $\Emult$ on the distribution
$\mathbb{P}_{\bm X}$, although this dependence is inherent in its
definition.\ An illustration of this construction is presented in
Figure~\ref{fig:T_S_to_C}.
\begin{figure}[htbp!]  \centering
    \begin{overpic}[width=\textwidth]{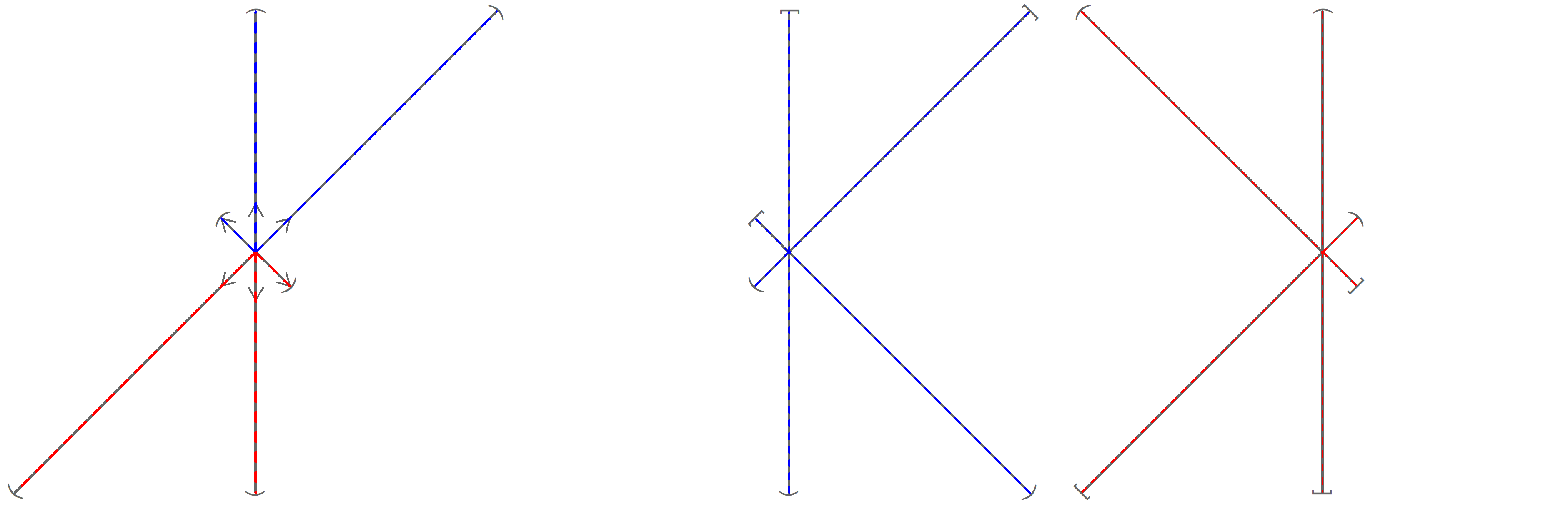} \put
(29.5,16.5) {{\small $x_1$}} \put (13.5,31) {{\small $x_2$}} \put
(18.5,17.5) {{\color{blue}\tiny $\bm w_1$\color{black}}} \put
(16.5,19) {{\color{blue}\tiny $\bm w_2$\color{black}}} \put (12,19)
{{\color{blue}\tiny $\bm w_3$\color{black}}} \put (11,14.5)
{{\color{red}\tiny $-\bm w_1$\color{black}}} \put (13,12.5)
{{\color{red}\tiny $-\bm w_2$\color{black}}} \put (18.5,13)
{{\color{red}\tiny $-\bm w_3$\color{black}}}

    \put (19,21) {\rotatebox{45}{\color{blue}\tiny
$\text{supp}(\lVert \bm X\rVert \mid \bm w_1) = (0,\infty)$}} \put
(0,1.5) {\rotatebox{45}{\color{red}\tiny $\text{supp}(\lVert \bm X\rVert
\mid \bm -w_1) = (0,\infty)$}}

    \put (41,19) {{\color{blue}\tiny $-1/\xi(\bm w_3)$\color{black}}}
\put (62,0.5) {{\color{blue}\tiny $-\infty$\color{black}}} \put
(54.5,22.5) {\rotatebox{45}{\color{blue}\tiny
$\overline{\Espace}_{\xi(\bm w_1)}\times \{\bm w_1\}$\color{black}}}
\put (51.5,31.5) {{\color{blue}\tiny $\infty$\color{black}}} \put
(51.5,0.5) {{\color{blue}\tiny $-\infty$\color{black}}} \put (51,21)
{\rotatebox{90}{\color{blue}\tiny $\overline{\Espace}_{\xi(\bm
w_2)}\times \{\bm w_2\}$\color{black}}} \put (41,13)
{{\color{blue}\tiny $-1/\xi(\bm w_1)$\color{black}}} \put (63,31.5)
{{\color{blue}\tiny $\infty$\color{black}}} \put (54.5,13)
{\rotatebox{-45}{\color{blue}\tiny $\overline{\Espace}_{\xi(\bm
w_3)}\times \{\bm w_3\}$\color{black}}}

    \put (87,19) {{\color{red}\tiny $-1/\xi(-\bm w_1)$\color{black}}}
\put (69.5,0.5) {{\color{red}\tiny $\infty$\color{black}}} \put
(72,6.5) {\rotatebox{45}{\color{red}\tiny
$\overline{\Espace}_{\xi(-\bm w_1)}\times \{-\bm w_1\}$\color{black}}}
\put (81.5,31.5) {{\color{red}\tiny $\infty$\color{black}}} \put
(81,0.5) {{\color{red}\tiny $\infty$\color{black}}} \put (84.5,20)
{\rotatebox{90}{\color{red}\tiny $\overline{\Espace}_{\xi(-\bm
w_2)}\times \{-\bm w_2\}$\color{black}}} \put (87,13)
{{\color{red}\tiny $1/\xi(-\bm w_3)$\color{black}}} \put (69.5,31.5)
{{\color{red}\tiny $-\infty$\color{black}}} \put (72,29)
{\rotatebox{-45}{\color{red}\tiny $\overline{\Espace}_{\xi(-\bm
w_3)}\times \{-\bm w_3\}$\color{black}}}
    \end{overpic}
    \caption{Illustration of a subset of $\overline{\Espace}^2$ using
      a probability distribution on $\mathbb{R}^2$ associated with a
      random variable $\bm X$ whose support contains~$\bm 0$.\ The
      subset of $\overline{\Espace}^2$ is based on the set of
      directions
      $\{\bm w_1,\bm w_2, \bm w_3, -\bm w_1, -\bm w_2, -\bm w_3\}$,
      shown in the left panel.\ In this example, the conditional
      probability distribution of
      $\lVert\bm X\rVert \mid \{\bm X/\lVert\bm X\rVert=\bm w\}$ is
      assumed to have a bounded upper tail ($\xi(\bm w) < 0$) for
      $\bm w \in \{-\bm w_1, -\bm w_3\}$, a light upper tail
      ($\xi(\bm w)=0$) for $\bm w\in\{\bm w_2, -\bm w_2\}$ and a heavy
      upper tail $(\xi(\bm w)>0)$ for
      $\bm w\in\{\bm w_1, -\bm w_1\}$.\ The red and blue segments,
      corresponding to directions $\bm w$ and $-\bm w$ are mapped to
      $\overline{\Espace}_{\xi(\bm w)}$ and
      $\overline{\Espace}_{\xi(-\bm w)}$, respectively, and are
      oriented according to $\bm w$ and $-\bm w$ in the centre and
      right panels.}
\end{figure}

In Assumption~\ref{ass:main_convergence_assumptions} below, we
strengthen convergence~\eqref{eq:vg_cond} by assuming that the total
variation of the signed measure
\begin{equation} \mu_{n, \bm w}(\cdot) := n \, \PR\left(\frac{\bm
X-b_n(\bm X/\lVert \bm X\rVert)}{a_n(\bm X/\lVert \bm X\rVert)}\in
\cdot ~\Big|~\bm X/\lVert \bm X\rVert = \bm w \right) - \nu_{\bm
w}(\cdot),\quad \bm X\sim \PR_{\bm X},
  \label{eq:total_variation_measure}
\end{equation} is bounded on compact intervals by a function that is
integrable with respect to $\PR_{\bm X/\lVert\bm X\rVert}$.
\begin{assumption}
  \label{ass:main_convergence_assumptions} The distribution $\PR_{\bm
X}$ satisfies~\eqref{eq:vg_cond} and there exists a
$\PR_{\bm X/\lVert\bm X\rVert}$-integrable function $\Delta \,:\,
\mathbb{S}^{d-1}\to \RR$, such that for any $-\infty< r_{\inf} \leq
r_{\sup} \leq \infty$, there exists $n_0=n_0(r_{\inf}, r_{\sup})$ such
that
  \[ \int_{r_{\inf}}^{r_{\sup}} \left\lvert \mu_{n, \bm
w}\right\rvert(dz) \leq \Delta(\bm w), \quad \text{for all $n \geq
n_0$,}
  \] where $\lvert \mu_{n, \bm w} \rvert(\cdot)=\mu_{n, \bm
w}^+(\cdot) + \mu_{n, \bm w}^-(\cdot)$ denotes the total variation
measure of $\mu_{n, \bm w}$ defined
in~\eqref{eq:total_variation_measure}.
\end{assumption} The bounding function $\Delta$ introduced in
Assumption~\ref{ass:main_convergence_assumptions} serves as dominating
function, ensuring sufficient regularity for the vague convergence of
$\nu_n$ to $\nu$ in $M_+(\Emult)$, or equivalently, the weak
convergence of $P_n$ to $P$ in $M_p(\Emult\times (0,1))$, as
established below in Theorem~\ref{thm:PPconvergence}.\ The function
$\Delta$ need not in general be bounded and may exhibit singularities,
meaning that the convergence of $\nu_{n}$ need not be uniform on
$\suppW$.\ On the other hand, when $\Delta$ is bounded then the
convergence is uniform on $\suppW$.\

With Assumption~\ref{ass:main_convergence_assumptions} in place, we
are now in position to state our main theorem, namely
Theorem~\ref{thm:PPconvergence}, which concerns the asymptotic
behaviour of the radial and directional variables, expressed in terms
of a limit Poisson point process, after suitably renormalizing the
radius using location and scale functions of the direction,
while preserving the directional variable.
\begin{theorem}
  \label{thm:PPconvergence} Let $\{\bm X_{i}\,:\,i=1,\dots, n\}$ be a
random sample from $\PR_{\bm X}$, where $\PR_{\bm X}$ satisfies
Assumption~\ref{ass:main_convergence_assumptions} with $\bm 0
\in\text{\normalfont supp}(\PR_{\bm X})$.
  \begin{itemize}
  \item[$(i)$] As $n \to \infty$, $\nu_n(\cdot) \vg \nu (\cdot)$ in
    $M_+(\Emult)$, where the limit measure $\nu$ satisfies
    \[ \nu(h) = \iint\limits_{\Emult} h(z, \bm w) \nu_{\bm w}(d z)
\PR_{\bm X/\lVert\bm X\rVert}(d \bm w),
    \]
  \item[$(ii)$] As $n\to\infty$,
    \begin{equation} P_n \wk P \quad \text{\normalfont in $M_p(
\Emult\times (0,1))$},
      \label{eq:PntoP}
    \end{equation} where $P$ is a Poisson point process with intensity
measure
    \[ \left([1+\xi(\bm w) z]_+^{-1-1/\xi(\bm w)} d z \PR_{\bm
X/\lVert\bm X\rVert}(d \bm w)\right) d t.
    \]
  \end{itemize}
\end{theorem}
A proof is given in Appendix~\ref{sec:PP_justification}.\ 
Theorem~\ref{thm:PPconvergence}, which covers Theorem
\ref{thm:univariate_vc} as a special case, provides a foundation for
new asymptotic theory associated with the probabilistic behaviour of
multivariate extremes, leading to new insights that stem from limit
laws which parallel those typically appearing in extreme value theory,
namely the limit laws associated with the renormalised component-wise
maxima, $r$-largest order statistics and exceedances above high
thresholds.\ 

Building on the argument presented in Appendix
\ref{sec:univariate_supplementary}, Proposition
\ref{prop:avoidance_probability} establishes a new limit law from the
the void probability of the limit Poisson point process $P$ in
Theorem~\ref{thm:PPconvergence} and naturally generalises Proposition
\ref{prop:min_max_stable}, covering the situation where in the limit
as $n\to \infty$, all atoms of the renormalised random point measure
$P_n$ fall in a set.\ 
\begin{proposition}
  \label{prop:avoidance_probability} Let
  $\{\bm X_i\,:\,i=1,\dots, n\}$ be a sequence of independent and
  identically distributed random variables from $\PR_{\bm X}$, where
  $\PR_{\bm X}$ satisfies
  Assumption~\ref{ass:main_convergence_assumptions}.\ Then
  \begin{equation} \PR\left[\left\{\left(\frac{\lVert\bm
X_i\rVert-b_n(\bm X_i/\lVert\bm X_i\rVert)}{a_n(\bm X_i/\lVert\bm
X_i\rVert)}, \bm X_i/\lVert \bm X_i\rVert\right)\in B\
\,:\,i=1,\dots,n\right\}\right] \wk e^{-\Lambda(B)}\quad \text{in
$M_1(\Emult)$},
    \label{eq:conv_max_stab}
  \end{equation} as $n\to\infty$, where $\Lambda(B):=\nu\big((\Emult
\setminus B)\times (0,1)\big)$.
\end{proposition} Unlike the univariate case, the limit law in
Proposition~\ref{prop:avoidance_probability} is not solely tied to the
convergence of renormalised \textit{component-wise} maxima or minima,
but encompasses a broader range of scenarios.\ It can be seen as a
special case of Proposition~\ref{prop:r_largest}, which covers the
situation where in the limit as $n\to\infty$, all but finitely many
atoms of the renormalised random point measure $P_n$ fall in a set.\
\begin{proposition}
  \label{prop:r_largest} Let $\{\bm X_i\,:\,i=1,\dots, n\}$ be a
  sequence of independent and identically distributed random variables
  from $\PR_{\bm X}$, where $\PR_{\bm X}$ satisfies
  Assumption~\ref{ass:main_convergence_assumptions}.\ Then for any
  $r=1,2,\dots$,
  \begin{equation} \PR\left[\sum_{i=1}^n1\left \{\left(\frac{\lVert\bm
X_i\rVert-b_n(\bm X_i/\lVert\bm X_i\rVert)}{a_n(\bm X_i/\lVert\bm
X_i\rVert)}, \bm X_i/\lVert \bm X_i\rVert\right)\in B
\right\}=n-r+1\right] \wk
e^{-\Lambda(B)}\sum_{k=0}^{r-1}\frac{\Lambda(B)^k}{k!} \text{ in
$M_1(\Emult)$},
  \end{equation} as $n\to \infty$, where $\Lambda(B):=\nu\big((\Emult
\setminus B)\times (0,1)\big)$.
\end{proposition} 

Finally, Proposition~\ref{prop:thresh_exceedances} focuses on
exceedances above directionally-dependent thresholds.\ It establishes
weak convergence in $M_1(\Emult_+)$, the subset of $\Emult$
restricted to positive elements, 
that is,
\[
  \Emult_+ := \bigcup_{\bm w\in\suppW} \left((0, \sup
    \overline{\Espace}_{\xi(\bm w)} ]\times\{\bm w\}\right)\subset
  \Emult.
\]
Note that for any $\xi \in \mathbb{R}$, $0$ is an interior point of
$\overline{\Espace}_{\xi}$ and thus, $\Emult_+$ is non-empty a.s..
This result highlights the role of thresholds
$b_{1/(1-q)}(\bm X/\lVert\bm X\rVert)$ and demonstrates how the
Poisson point process characterisation extends to such exceedances.
\begin{proposition}
  \label{prop:thresh_exceedances} Let $\bm X\sim\PR_{\bm X}$ where
$\PR_{\bm X}$ satisfies
Assumption~\ref{ass:main_convergence_assumptions}.\ Then
  \begin{IEEEeqnarray*}{rCl} &&\PR\left[\left(\frac{\lVert \bm X\rVert
- b_{1/(1-q)}(\bm X/\lVert \bm X\rVert)}{a_{1/(1-q)}(\bm X/\lVert \bm
X\rVert) }, \bm X/\lVert \bm X\rVert\right) \in A ~\Big|~ \lVert \bm
X\rVert > b_{1/(1-q)}\left(\bm X/\lVert \bm X\rVert\right)\right]\wk
\Lambda(\Emult \setminus A) \text{ in $M_1(\Emult_+)$},
    \label{eq:MRVgauge}
  \end{IEEEeqnarray*} as $q\to 1$, where $\Lambda(\Emult \setminus
A):=\nu\big(A\times (0,1)\big)$.
\end{proposition}
We note that the limit distribution that appears in
Proposition~\ref{prop:thresh_exceedances} is directly expressed in
terms of $\Lambda$ and not in terms of the ratio
\[ \frac{\Lambda((\Emult \setminus A))}{\Lambda(\Emult \setminus
    \Emult_+)}= \frac{\nu(A\times (0,1))}{\nu(\Emult_+\times (0,1))},
\] because the restriction of $\nu$ on $\Emult_+\times (0,1)$ is a
probability measure, satisfying $\nu(\Emult_+\times(0,1))=1$.\

These propositions provide a cohesive narrative.\ They demonstrate how
the Poisson point process framework captures multivariate extremes
through void probabilities, order statistics, and threshold
exceedances.\ Each result expands upon the classical extreme value
theory, offering new insights into the structure and behaviour of
multivariate extremes.

In a similar manner to that explained in
Section~\ref{sec:background_univ_extr}, one has the flexibility to
change the measure $\PR_{\bm X/\lVert\bm X\rVert}$ in the definition
of $\nu$ in Theorem~\ref{thm:PPconvergence}.\ Here, the intuition is
that by applying admissible perturbations in the norming constants
$b_{1/(1-q)}$ and $a_{1/(1-q)}$---adjustments that preserve weak
convergence under a convergence-to-types argument---alternative
norming sequences $b_{1/(1-q)}^\star$ and $a_{1/(1-q)}^\star$ can be
selected.\ Such modifications act on the limit point measure $\nu$ by
thinning it at loci of high directional density and densifying it at
loci of low directional density, so that the directional cross-section for
the resulting radial-directional intensity is constant.\ This is
explained below in Remark~\ref{rem:alternative_renormalization}.\


\begin{remark}
  \label{rem:alternative_renormalization} The probability measure
  $\PR_{\bm X/\lVert\bm X\rVert}$ in the definition of $\nu$ can be
  changed to some other spherical (not necessarily probability)
  measure, through suitable changes in the renormalisation.\ For
  example, 
  if $\PR_{\bm X/\lVert\bm X\rVert}$ is mutually absolutely continuous
  with respect to a probability measure $\mu$ on $\suppX$, then the sequence of
  renormalised random point measures
  \[ P_n^\star := \sum_{i=1}^n\delta\left[ \frac{\lVert \bm X_i\rVert
        - b_{n}(\bm X_i/\lVert\bm X_i\rVert\,;\, \mu)}{a_n(\bm
        X_i/\lVert\bm X_i\rVert\,;\, \mu)}, \bm X_i/\lVert\bm X_i\rVert,
      \frac{i}{n+1}\right], \quad n=1,2,\dots,
  \] converges weakly to $P^\star$ in $M_p(\Emult\times (0,1))$, with
  $P^\star$ a Poisson point process having intensity measure
\[
  ([\{1+\xi(\bm w) z\}_+^{-1-1/\xi(\bm w)} d z] \mu(d \bm w)) d
t,
\] as $n\to\infty$, whenever
\begin{IEEEeqnarray}{rCl}\label{eq:b_n_star} b_n(\bm w\,;\,\mu) &\sim&
  b_n(\bm w) + [\{\varrho_\mu(\bm w)^{\xi(\bm w)}-1\}/\xi(\bm w)] a_n(\bm w), \quad \text{and} \\
  \nonumber \\ a_n(\bm w\,;\, \mu) &\sim&
  \label{eq:a_n_star} \varrho_\mu(\bm w)^{\xi(\bm w)} a_n(\bm w),      
\end{IEEEeqnarray}
as $n\to\infty$, where $\varrho_\mu:=d\PR_{\bm W}/d\mu$ denotes the
Radon--Nikodym derivative of $\PR_{\bm W}$ with respect to $\mu$.
\end{remark}
Building on Remark~\ref{rem:alternative_renormalization}, below in
Remark~\ref{rem:uniform_over_the_ball} we construct a non-linear
renormalisation of the exceedances, using the 1-dimensional
probability integral transform, producing a limit uniform distribution
over the unit ball.\ This distribution is defined by the product of
the uniform distribution over the unit interval of distances to the
origin with the uniform distribution over the unit sphere.
\begin{remark}
  \label{rem:uniform_over_the_ball} 
  If $\lVert\,\cdot\,\rVert$ denotes the Euclidean norm and
  $\PR_{\bm X/\lVert\bm X\rVert}$ is absolutely continuous with
  respect to the spherical Lebesgue measure, with positive density
  $f_{\bm X/\lVert\bm X\rVert}(\bm w)$, then
  \begin{IEEEeqnarray*}{rCl} (\bm X/\lVert \bm X\rVert) ~\big|~
    \{\lVert \bm X\rVert > b_{1/(1-q)}^\star\left(\bm X/\lVert \bm
      X\rVert\right)\}
    \label{eq:uniform_angles}
  \end{IEEEeqnarray*} converges in distribution as $q\to 1$ to a
  random variable that follows the uniform distribution on
  $\SSS^{d-1}$.\ Here,
  $b_n^\star(\bm w):=b_n(\bm w\,;\,\mu^\star$ with
  $\mu^\star$ being the uniform
  probability measure on $\SSS^{d-1}$.\

  Additionally, let $H_{\bm w}$ denoting the cumulative distribution
  function of the standard generalised Pareto distribution with shape
  parameter $\xi(\bm w)$, and define
  $a_n^\star(\bm w)=a_n(\bm w\,;\, \mu^\star)$.\ Then, 
  \[
    \left(H_{\bm X/\lvert\bm X\rVert}\left[\frac{\lVert \bm X\rVert -
          b_{1/(1-q)}^{\star}(\bm X/\lVert \bm
          X\rVert)}{a_{1/(1-q)}^{\star}(\bm X/\lVert \bm X\rVert)
        }\right], \bm X/\lVert \bm X\rVert\right) ~\Big|~ \{\lVert \bm
    X\rVert > b_{1/(1-q)}^{\star}\left(\bm X/\lVert \bm
      X\rVert\right)\}, \quad \text{as } n\to \infty,
  \]
  converges in distribution to a random variable following the uniform
  distribution over the unit ball.\ 
\end{remark} 

Similarly to the normalising functions \eqref{eq:normings_univariate}
in Section~\ref{sec:background_univ_extr}, the functions $a_n$ and
$b_n$ satisfy
\[
  b_n(\bm w) = U_{\bm w}(n), \quad \text{and}\quad a_n(\bm w) = U_{\bm
    w}(n(1+\xi(\bm w))^{1/\xi(\bm w)})-U_{\bm w}(n)\quad
  \text{whenever $\xi(\bm w) > -1$}.
\]
where $U_{\bm w}(1/(1-q))=F_{R~|~\bm W}^{-1}(q \mid \bm w)$, for
$\bm w \in \SSS^{d-1}$.\

In the subsequent
Sections~\ref{sec:quantile_return_sets}
and~\ref{sec:scaling-limit-sets}, we discuss these sequences and their
link to novel sets offering insightful risk representations in more
details.\

\subsection{Correspondence between norming functions and probability
  sets}
\label{sec:quantile_return_sets} Extending the concept of quantiles
to probability distributions on spaces beyond the ordinary real line
is still being recognised as a major challenge, with apparent
difficulties in its formulation.\ \cite{cheretal2017} ingeniously
used measure transportation maps between a distribution of interest
on $\RR^d$ and a reference distribution on the $d$-dimensional unit
ball, to define a new depth concept, termed the
\textit{Monge-Kantorovich} depth, which specialises to halfspace
depth when $d=1$.\ Based on this approach, and by leveraging the
theory of \cite{McCann95}, \cite{hallin2021distribution} proposed a
\textit{centre-outward} definition of multivariate distributions and
quantile functions, along with their empirical counterparts, which
is free of any moment assumptions.\ Central to the aforementioned
line of research is the existence of a $\PR_{\bm X}$-a.e.\ unique
gradient of a convex function that pushes the distribution
$\PR_{\bm X}$ forward to the uniform distribution over the unit
ball, this distribution is defined in
Remark~\ref{rem:uniform_over_the_ball}.\ That gradient yields a
homeomorphism and its inverse naturally qualifies as a quantile
function of a probability distribution in $\RR^d$.\


One important consideration
behind these approaches can be explained when a well-defined centre in $\suppX$ is
available for the distribution of $\bm X$, allowing for discussions of
its radius and direction, defined by
$(R,\bm W)=(\lVert\bm X\rVert,\bm X/\lVert\bm X\rVert)$, relative to a
suitable norm $\lVert \,\cdot \,\rVert$.\ Specifically, the proposed
measure transportation maps of \cite{cheretal2017} and
\cite{hallin2021distribution} do not necessarily preserve the
direction $\bm W$ since a pair $(R, \bm W)$ is mapped to the pair
$(U_{(0,1)}, \bm U_{\SSS^{d-1}})$, where
$U_{(0,1)}\sim \text{Unif}((0,1))$,
$\bm U_{\SSS^{d-1}}\sim \text{Unif}(\SSS^{d-1})$, with $U_{(0,1)}$
independent of $\bm U_{\SSS^{d-1}}$.\ Thus, this approach may not be
directly applicable to certain important metric spaces of interest where the Lebesgue measure may not even exist, and hence, where the target
uniform distribution is not well-defined.\ Another limitation in
the context of rare event probability estimation is the
non-parametric nature of the associated statistical methods, which
currently prevent extrapolation beyond the observed data.\ This is
similar to the sample quantile, which cannot estimate quantiles
outside the observed range, see \textit{e.g.\ }\cite{hallin2024multivariate}.\


The transformation that we propose maps $\bm X$ to $U\bm W$, where
$U=F_{R\mid \bm W}(R\mid \bm W)$ is uniform over the unit interval
of distances to the origin and independent of $\bm W$, and where
$F_{R\mid \bm W}$ denotes the conditional cumulative
distribution function of $R\mid \bm W$.\ Unlike the
approaches of \cite{cheretal2017} and \cite{hallin2021distribution},
which are indifferent to the radius and direction of $\bm X$, our
transformation requires a choice of a norm measuring the
size of vectors.\

We highlight that the independence of $U$ and $\bm W$, and the
optimality of the transport plan $F^{-1}_{R\mid \bm W}(U\mid \bm w)$
for each fixed direction $\bm w$ under the squared Euclidean distance
\citep{brenier1991polar}, pushing forward the distribution of $U$ to
the conditional distribution of $R \mid \bm W=\bm w$, imply that the
map $U \bm W\to F^{-1}_{R\mid \bm W}(U\mid \bm W) \bm W$ is also an
optimal transport plan under the squared Euclidean distance, pushing
forward the distribution of $U\bm W$ to the distribution of $\bm X$.\
Hence, it entirely characterises $\PR_{\bm X}$.\ However, unlike
typical transport plans, which move a distribution to a known target
distribution with full access to its properties, this transport plan
relies on having knowledge of the distribution of $\bm W$ to optimally
transport the distribution of $U\bm W$ to that of $\bm X$.\ This
reliance on the distribution of $\bm W$ is in fact a key strength, as
it enables the derivation of transport plans for a wide range of
distributions using only one-dimensional transformations, and
facilitates inference since the distribution of $\bm W$ can be
estimated directly from data.\ 

Based on these observations, we introduce a novel definition of
\textit{probability sets} and \textit{return sets} for a given
probability measure $\PR_{\bm X}$.\ Consider a $\sigma$-finite measure
$\mu$ such that
\begin{equation}
  \mu \ll \PR_{\bm W} \quad \text{and} \quad \PR_{\bm W} \ll \mu.
  \label{eq:abs_cont}
\end{equation}
Let
\begin{equation}
\label{eq:q_l}
  q_l(\PR_{\bm W} \,\|\, \mu) := 1-\inf\left\{\frac{d \PR_{\bm W}}{d
      \mu}(\bm w) \,:\, \bm w\in\text{\normalfont supp}(\PR_{\bm
      W})\right\},
\end{equation}
where $d \PR_{\bm W}/d\mu$ denotes the Radon--Nikodym derivative of
$\PR_{\bm W}$ with respect to $\mu$, which is strictly positive on
$\text{\normalfont supp}(\PR_{\bm W})$, due to assumption
\eqref{eq:abs_cont}.\ The range of $d \PR_{\bm W}/d\mu$ contains
values strictly less than 1 if $\PR_{\bm W}\neq\mu$, reflecting the
fact that at some neighbourhoods on
$\text{\normalfont supp}(\PR_{\bm W})$ the measure $\PR_{\bm W}$
assigns less mass than $\mu$, while at others it assigns more mass,
and is equal to unity if $\PR_{\bm W} = \mu$, reflecting that the two
measures are identical.\ Thus, when \eqref{eq:abs_cont} holds, it
follows that $q_l(\PR_{\bm W} \,\|\, \mu) \in (0,1)$.\ Define, for all
$q\in (q_l(\PR_{\bm W} \,\| \,\mu) , 1)$, the radial function
\begin{equation}
  r_{\QSq(\PR_{\bm W} \,\| \,\mu)}(\bm w):=\inf\left\{r \in \RR\,:\,
    F_{R\mid \bm W}(r\mid \bm w) \geq 1-(1-q) \frac{d\mu}{d \PR_{\bm
        W}}(\bm w)\right\},\quad \bm w \in \text{\normalfont
    supp}(\PR_{\bm W}).
  \label{eq:r_probability_set}
\end{equation}
\begin{definition}
  \label{defn:mult_quantile}
  Suppose that \eqref{eq:abs_cont} holds.\ A $q$-probability set of
  $\PR_{\bm X}$ relative to $\mu$ is defined by
  \[
    \QSq(\PR_{\bm W} \,\|\, \mu) := \bigcup_{\bm w \in \text{\normalfont
        supp}(\PR_{\bm W})} \{(1-t) \,\bm 0 + t\,
    r_{\QSq(\PR_{\bm W} \,\|\, \mu)}(\bm w)\,\bm w)\,:\, t\in[0,1]\},\quad
    \text{for $q\in (q_l(\PR_{\bm W} \,\| \,\mu) , 1)$}.
  \]
  A return set of $\PR_{\bm X}$ relative to $\mu$ with return period
  $T$ is defined by
  \[
    \RS_{T}(\PR_{\bm W} \,\|\, \mu) = \suppX \setminus
    \mathcal{Q}_{1-1/T}(\PR_{\bm W} \,\|\, \mu), \quad T >
    1/\{1-q_l(\PR_{\bm W} \,\| \,\mu)\}.
  \]

\end{definition}
The justification of the term $q$-probability set in Definition
\ref{defn:mult_quantile} comes from the property
$\PR[\bm X\in \QSq(\PR_{\bm W} \,\|\, \mu)]=q$, which is a direct
consequence of Proposition~\ref{prop:properties_probability_set}.
\begin{proposition}
  \label{prop:properties_probability_set}
  Suppose that condition \eqref{eq:abs_cont} holds and
  $\inf\{\xi(\bm w)\,:\, \bm w\in \suppWclean\} > -1$. Then, 
  \begin{itemize}
  \item[$(i)$] for all $q > q_l(\PR_{\bm W} \,\|\, \mu)$,
    \begin{equation}
      \PR[R > r_{\QS_q(\PR_{\bm W}\,\|\, \mu)}(\bm W)] = 1-q \quad
      \text{and} \quad \PR[\bm W \in B \mid R > r_{\QSq(\PR_{\bm W}
        \,\|\, \mu)}(\bm W)] = \mu(B);
      \label{cor:equality_of_W_distributions}
    \end{equation}
    and
  \item[$(ii)$] the $b_{1/(1-q)}(\bm w \,;\,\mu)$ defined by
    expression \eqref{eq:b_n_star} satisfies
    \[
      b_{1/(1-q)}(\bm w \,;\,\mu) \sim r_{\QS_{q}(\PR_{\bm W}\,\|\,
        \mu)}(\bm w), \quad \bm w\in\suppWclean,\qquad \text{as
        $n\to\infty$},
    \]
    and the function $a_{1/(1-q)}(\bm w \,;\,\mu)$ defined by
    expression \eqref{eq:a_n_star} satisfies
    \[
      a_{1/(1-q)}(\bm w \,;\,\mu) \sim r_{G_q(\PR_{\bm W}\,\|\,
        \mu)}(\bm w), \quad \bm w\in\suppWclean,\qquad \text{as
        $n\to\infty$},
    \]
    where
    \[
      r_{G_q(\PR_{\bm W}\,\|\, \mu)}(\bm w):= r_{\QS_{q(\bm
          w)}(\PR_{\bm W}\,\|\, \mu)}(\bm w) - r_{\QS_{q}(\PR_{\bm
          W}\,\|\, \mu)}(\bm w) \quad \text{and} \quad q(\bm w) =
      1-(1-q)\{1+\xi(\bm w)\}^{-1/\xi(\bm w)}.
    \]
  \end{itemize}
\end{proposition}
A proof is given in Appendix~\ref{sec:radon_nikodym}.
\begin{definition}
  A \textit{quantile set} with probability $q\in(0,1)$ and a
  \textit{return set} with return period \color{black} $T>0$ are
  respectively defined by
  \begin{equation}
    \QS_q:=\QS_q(\PR_{\bm W}\,\|\,\PR_{\bm W})
    \qquad \text{and} \qquad \RS_{T} := \RS_{T}(\PR_{\bm W}\,\|\,\PR_{\bm W}).
    \label{eq:Rq}
  \end{equation}
  \color{black} Suppose that condition \eqref{eq:abs_cont} holds with
  $\mu=\mu^\star$, where $\mu^\star$ denotes the uniform probability
  measure on ${\text{\normalfont supp} (\PR_{\bm W})}$.\ Then for
  $q>q_l(\PR_{\bm W}\,\|\, \mu^\star)$, an \textit{isotropic
    probability set} with probability $q$ and an \textit{isotropic
    return set} with return period \color{black}
  $T>1/\{1-q_l(\PR_{\bm W}\,\|\,\mu^\star)\}$ are respectively defined
  by
  \begin{equation}
    \QS_q^\star:=\QS_q(\PR_{\bm W}\,\|\,\mu^\star)
    \qquad \text{and} \qquad \RS_{T}^\star := \RS_{T}(\PR_{\bm W}\,\|\,\mu^\star).
    \label{eq:Rq}
  \end{equation}
  
\end{definition}

\subsection{Scaling and limit sets} 
\label{sec:scaling-limit-sets}
Assumption~\ref{ass:main_convergence_assumptions} ensures sufficient
regularity of $\PR_{\bm X}$ for the vague convergence of $\nu_n$ to
$\nu$ in $M_+(\Emult)$.\ However, verifying this assumption directly
can be challenging.\ Proposition~\ref{prop:RV1} provides easily
verifiable sufficient conditions for this convergence.\ These
conditions, which hold for a wide range of copulas, especially those
with suitably transformed common marginal distributions, rely on the
existence of a Lebesgue density $f_{\bm X}$ for $\PR_{\bm X}$, and
guarantee the convergence in distribution of the radially renormalized
exceedances, as described in
Proposition~\ref{prop:thresh_exceedances}.\ Our conditions here assume
that $\suppW = \SSS^{d-1}$ and a constant shape function
$\xi(\bm w)= \xi\in \RR$ for all $\bm w\in\SSS^{d-1}$.\ Additionally,
our conditions require that the conditional distribution of the radius
given the direction is in the domain of \textit{uniform local
  attraction} \citep{Sweeting85}, and that the dominating function
$\Delta$ in Assumption~\ref{ass:main_convergence_assumptions} is
bounded.\ This is a consequence of the uniform convergence that is
assumed in Proposition \ref{prop:RV1} and can be easily seen to hold
when the functional that determines the rate of convergence of the
density function to the limit gauge function, presented in
Supplementary~\ref{sec:rof}, is bounded on $\SSS^{d-1}$.
\begin{proposition} 
  \label{prop:RV1} 
  Suppose that the random vector $\bm X$ is absolutely continuous with
  respect to the Lebesgue measure on $\mathbb{R}^d$, admitting a
  density $f_{\bm X}$.\ Let $\G\in\bigstar$ be described by continuous
  1-homogeneous gauge function $\gG\,:\,\RR^d\to \RR_+$.\ Suppose that
  one of the following conditions holds:
  \begin{description}
  \item[$(i)$] The limit $\QS_{1}:=\lim_{q\to 1} \QS_q$ exists in
    $\bigstar$, $f_{\bm X}(\bm x) > 0$ when $\bm x\in \G^\circ$,
    $f_{\bm X}(\bm x) = 0$ when $\bm x\in \G'$, and there exists
    $\psi_{\textsf{B}}\,:\,\RR_+\to\RR_+$, and a $\xi < 0$ such that
    as $t\to \infty$,
    \begin{equation}
      \frac{f_{\bm X}[\{1-(\lVert t \bm y\rVert)^{-1}\}r_{\QS_1}(t \bm y)(t \bm y)]}{\psi_{\textsf{B}}(t)}\to
      g_{\G}(\bm y)^{1/\xi+1}, \qquad \bm y\in\Rstar.\ 
      \label{eq:RVneg}
    \end{equation}
  \item[$(ii)$] There exists $\psi_{\textsf{L}}\,:\,\RR_+\to \RR_+$,
    and a $\rho > 0$ such that as $t\to\infty$,
    \begin{equation}
      -\frac{\log f_{\bm X}(t \bm y)}{\psi_{\textsf{L}}(t)} \to \gG(\bm y)^\rho, \qquad \bm y\in\Rstar.
      \label{eq:RV1}
    \end{equation}
  \item[$(iii)$] There exists $\psi_{\textsf{H}}\,:\,\RR_+\to \RR_+$,
    and $\xi > 0$, such that as $t\to\infty$,
    \begin{equation}    
      \frac{f_{\bm X}(t \bm y)}{\psi_{\textsf{H}}(t)}\to
      \gG(\bm y)^{-(d+\xi^{-1})}, \qquad \bm y\in\Rstar.
      \label{eq:RVpos}
    \end{equation}
  \end{description}
  Further, if the convergence is uniform on $\SSS^{d-1}$, then
  Proposition~\ref{prop:thresh_exceedances} holds with $b_{1/(1-q)} := r_{\QSq}$ and
  \begin{itemize}
  \item[$(i)$] $a_{1/(1-q)} := r_{\G_q}(\bm w)$ where
    $r_{\G_q}(\bm w)= -\xi \, [r_{\QS_1}(\bm w) - r_{\QS_q}(\bm w)]$
    when $\xi < 0$;
  \item[$(ii)$] $a_{1/(1-q)} := r_{\G_q}(\bm w)$ where
    $r_{\G_q}(\bm w)=\rG(\bm w)^{\rho}/\psi_{\textsf{L}}'\{\rquant(\bm
    w)\}$ when $\xi = 0$;
  \item[$(iii)$] $a_{1/(1-q)} := r_{\G_q}(\bm w)$ where
    $r_{\G_q}(\bm w)=\xi r_{\QS_q}(\bm w)$ when $\xi > 0$.\
  \end{itemize}
\end{proposition}
A proof is given in Appendix~\ref{sec:proof_prop1}.\ 
\begin{remark}
  For the light-tailed case corresponding to case $(ii)$,
  Assumption~\ref{prop:RV1} guarantees that a suitable sequence of
  scaling constants $r_n$ can be found so that the random point-set
  $N_n=\left\{\bm X_1/r_n, \dots, \bm X_n/r_n\right\}$ converges in
  probability onto a limit set $\G$ having radial function $r_{\G}$
  \citep{nolde2021linking}.\ For the heavy tailed-case $(iii)$, our
  assumption is in a similar spirit to~\cite{de1987regular} and
  guarantees that the distribution of $\bm X$ is multivariate
  regularly varying on $\Rstar$.\
  Similarly to the classical theory of univariate extremes, where
  bounded-tail behaviour is handled by employing power-law tails for
  the normalised variable through appropriate modifications of the
  location sequence near the upper endpoint and the scaling sequence,
  here the multivariate case follows a similar structure.\
  Specifically, the assumption in case $(i)$ corresponds to that in
  case $(iii)$ after a suitable transformation. 
\end{remark}


\subsection{Probability density function of directions}
\label{sec:density_angles}
In this section, we introduce useful properties of the direction
variable $\bm W = \bm X/\lVert \bm X \rVert\in\SSS^{d-1}$ in light of
its central role in the limiting Poisson Process of appropriately
centred and scaled radii, as presented in
Section~\ref{sec:PP_characterisation}.\ The framework presented in
that section has the important advantage that it yields a distribution
of direction of exceedances of $\QSq$ equal to that of directions
found in the original variable $\bm X$ as exemplified by
Corollary~\ref{cor:equality_of_W_distributions}.\ Below we present
additional properties of the directional variable $\bm W$, including a
connection between the distribution of directions and star-bodies,
even in cases where $F_{\bm W}$ partially depends on the limit set
$\G$.\

For simplicity, we assume throughout
that $F_{\bm W}$ is absolutely continuous with respect to the
spherical Lebesgue measure and that it admits a probability density
$f_{\bm W}$.\ In this setting, we observe that any continuous
probability density function on the sphere is in one-to-one
correspondence with a strongly starshaped set $\mathcal{W}$ with
radial function $r_{\W}=f_{\bm W}$ defined by
\[
  \W := \bigcup_{\bm w \in \suppW}
  [0\,:\, r_{\W}(\bm w)\, \bm w] \in \bigstar, \quad \bm
  w\in\SSS^{d-1}.
\]
Because $\W\in\bigstar$, we have
$r_{\W^{1/d}}(\bm w)^d = r_{\W}(\bm w)$.\ Thus, from the definition of
the volume given in equation~\eqref{eq:vol} in Supplementary
Material~\ref{appendix:radial}, for $f_{\bm W}$ to be a valid density,
$\W$ must satisfy
\begin{equation}
  1=\int_{\SSS^{d-1}}f_{\bm W}(\bm w)\,\dd \bm w = d\left[\frac{1}{d}\int_{\SSS^{d-1}}r_{\W^{1/d}}(\bm w)^d\,\dd \bm w\right] = d\vol{\mathcal{W}^{1/d}}.
  \label{eq:W_density_constraint}
\end{equation}
Hence, to construct any set $\mathcal{W}$ satisfying
condition~\eqref{eq:W_density_constraint} that $r_{\W}$, it suffices
to consider
\begin{IEEEeqnarray}{rCl}
  \W&=&\LL^d\rdiv (d \vol{\LL}),
  \label{eq:W_from_L}
\end{IEEEeqnarray}
for some star-body $\LL$, which follows from the formula for the
volume of a starshaped set.\

There are several possible forms for $\LL$, and hence, for $\W$.\ A
first possibility for the form of the set $\LL$ is motivated from the
properties of probability density functions that are homothetic with
respect to $\G$~\citep{balknold10}.\ Recall that a density is termed
homothetic if it has level-sets that are scaled copies of $\G$.\ In
the setting where $f_{\bm X}$ is homothetic,
Proposition~\ref{prop:angle_distribution} explores the one-to-one
correspondence between the density of directions and $\G$.\
\begin{proposition}
  Suppose that $f_{\bm X}(\bm x)=f_0\left(r_\G(\bm
    x)^{-1}\right)$ for a decreasing, positive, continuous function
  $f_0:[0,\infty)\to[0,\infty)$ and a radial function
  $r_{\G}$ characterizing a set
  $\G\in\bigstar$.\ Then $\mathcal{L}=\G$.
  \label{prop:angle_distribution}
\end{proposition}
\noindent A proof of Proposition~\ref{prop:angle_distribution} based
on the integration by parts method of~\citet[Section
3.1,][]{balknold10} is given in Appendix~\ref{sec:proof_W_density}.\
We remark that Proposition~\ref{prop:angle_distribution} is a
probabilistic proof for the volume of a star-body in polar
coordinates~\citep{klain1997invariant}.\ \color{black}
As a result, under the homothetic framework, an analytic form of the
density of directions can be obtained.\ Consider the three examples
below, all of which have homothetic densities $f_{\bm{X}}$ with common
radial function,
$r_{\G}(\bm x) = (\bm x^\top \mathsf{Q} \bm x)^{-1/2}$ for
$\mathsf{Q}$ is positive, but different generators, $f_0$:\
\begin{example} Suppose
  $\bm X\sim {\mathcal{N}_d}(\bm 0,\mathsf{Q}^{-1})$.\ The margins of
  $\bm X$ are standard-normally distributed.\ The density $f_{\bm X}$
  is homothetic with respect to $r_{\G}$ and
  $f_0(r) = C_1 \exp(-r^{-2}/2)$, $r >0$, where
  $C_1=[\lvert\mathsf{Q}\rvert^{1/2}/(2\pi)^{d/2}]$.\
  \label{ex:MVN}
\end{example}
\begin{example} Suppose $\bm X\sim \text{\normalfont MVL}(\Orig,
  \mathsf{Q}^{-1})$.\ The margins of $\bm
  X$ are standard Laplace distributed.\ The density $f_{\bm X}(\bm
  x)$ is homothetic with respect to $r_{\G}$ and $f_0(r) = C_2
  r^{-\nu}K_{\nu}\left(r^{-1}\right)$,
  $r>0$, where \color{black}
  $C_2=[\lvert\mathsf{Q}\rvert^{1/2}/(2\pi)^{d/2}]$\color{black},
  $K_\nu$ is the modified Bessel function of the second kind
  \citep{gradshteyn2014table} and $\nu = (2-d)/2$.\ 
\label{ex:MVL}
\end{example}
\begin{example} Suppose
  $\bm X\sim \text{t}_{\nu}(\bm 0, \mathsf{Q}^{-1})$.\ The margins of
  $\bm X$ are Student-$t_{\nu}$ distributed.\ The density
  $f_{\bm X}(\bm x)$ is homothetic with respect to $r_{\G}$ and
  \smash{$f_0(r) = C_3 \left(1+r^{-2}/\nu\right)^{-(\nu + d)/2}$},
  $r>0$, where
  $C_3=[\Gamma\{(\nu+d)/2\}\lvert\mathsf{Q}\rvert^{1/2}/\{(\nu\pi)^{d/2}\Gamma(\nu/2)\}]$.\
\label{ex:t}
\end{example}
\noindent In each of the above examples, Proposition~\ref{prop:angle_distribution} gives an analytic form for the star
set $\mathcal{W}$.\ The set $\mathcal{W}$
is common to all examples as it is defined by $r_{\G}^{d}/(d\vol{\G})$ for the same $r_{\G}$ in all examples.\
\begin{figure}[t!]
    \centering
    \includegraphics[width=\textwidth]{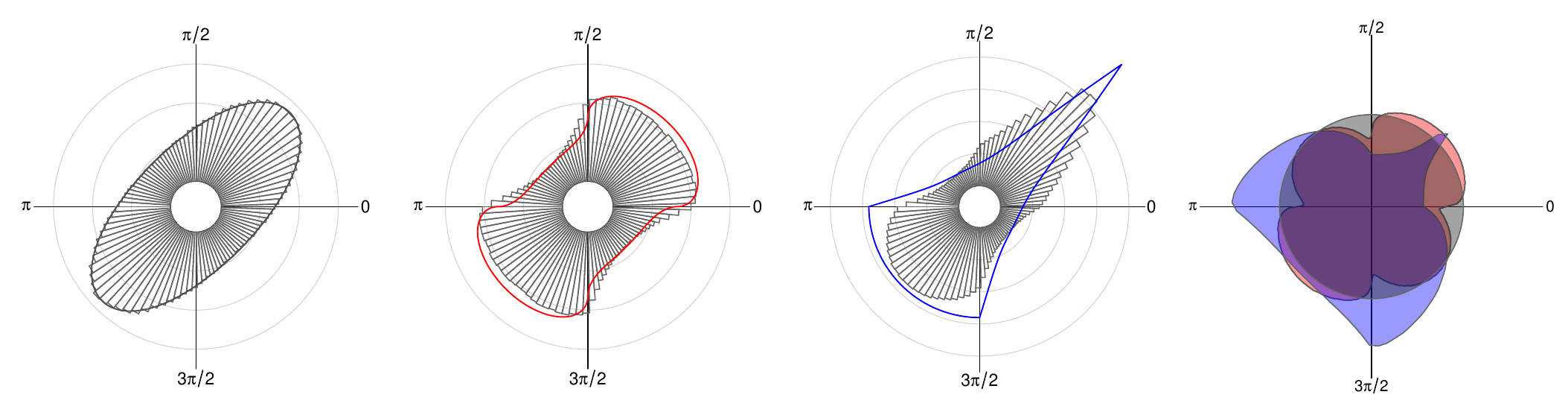}
    \caption{\textit{Left to right}:\ Circular histogram of
      $10^6$ directions sampled from a standard bivariate normal
      distribution (correlation
      $\rho=0.5$) with standard normal margins, standard Laplace
      margins, and from a bivariate max-stable logistic distribution
      with Laplace margins (dependence parameter
      $\theta=0.5$).\ A concentric circle corresponds to a leap of 0.1
      in the density of directions.\ Solid lines correspond to the boundary of
      $\G^{d}/(d\vol{\G})$, for
      $\G$ of the respective distributions.\ Plot 4:\ Sets
      $[\G^d/(d\vol{\G})]\cdot
      \W^{-1}$ corresponding to Plots 1--3.\ }
    \label{fig:f_W_from_g}
\end{figure}

The class of homothetic densities, although rich, serves at best as an
idealistic setting.\ A second option for the form of $\LL$ is revealed
when removing the effect of marginal scale to studying the extremal
dependence structure of $\bm X$.\ This is standard practice in extreme
value analysis, and is typically achieved by standardising the margins
of $\bm X$ to a common distribution (see
Section~\ref{sec:standardisation}).\ Consider the map
$\RR^d \ni \bm x \mapsto T(\bm x) =
\big(\Psi^{\leftarrow}(F_{X_i}(x_i)):i=1,\ldots,d\big) \in \RR^d$
performing a transformation of the marginal distributions of $\bm X$,
so that the $i$th element of $T(\bm X)$ follows $\Psi$, a continuous
cumulative distribution function on $\RR$.\ The transformation
preserves the cardinality of the set of input vectors
$\bm X/\lVert \bm X\rVert$, ensuring a one-to-one correspondence with
the output vector $T(\bm X)/\lVert T(\bm X)\rVert$.\ From a geometric
perspective, this means that the densities of
$\bm X/\lVert \bm X\rVert$ and $T(\bm X)/\lVert T(\bm X)\rVert$ can be
defined over the same set of
points.
\ The distribution of directions, however, may change when $T$
introduces nonlinearities.\ 
For example, suppose that $\bm X_N$ is
distributed according to a multivariate normal distribution with zero
mean and covariance matrix $\mathsf{Q}^{-1}$ where $\Q$ is such that
each marginal follows the standard normal distribution.\ Let
$\bm X = F_L^{-1}[\Phi(\bm X_N)]$, were $F_L$ and $\Phi$ denote the
cumulative distribution functions of standard Laplace and standard
normal random variables, respectively.\ While the density of $\bm X_N$ is
homothetic, thus having uniform rate of convergence is across
directions, the density of $\bm X$ is not
homothetic.

It is worth noting that the asymptotic analysis with Laplace marginals
may get complicated, for example through the behaviour of the joint
density function along the coordinate axes.\ Specifically, for the
multivariate normal copula in Laplace margins, the standard Mill's ratio
approximation breaks down, leading to singularities in the convergence
rate.\ These singularities, however, are resolved when convergence can
be established locally uniformly on $\RR^d\setminus\{0\}$.\ A refined
analysis in Supplementary Material~\ref{sec:MVN_supplementary} for the
multivariate normal copula example demonstrates this and highlights
that marginal standardisation may not always adversely affect the
convergence.

For the multivariate normal copula in Laplace margins, we prove that
convergence in $(ii)$ of Proposition~\ref{prop:RV1} holds uniformly,
with $\G$ determined by
\begin{equation}
  r_{\G}(\bm w)= \left[\left\{\text{sgn}(\bm w) \vert \bm
      w\rvert^{1/2}\right\}^\top \mathsf{Q} \, \left\{\text{sgn}(\bm
      w) \vert \bm w\rvert^{1/2}\right\}\right]^{-1}, \qquad \bm
  w\in\SSS^{d-1}.
  \label{eq:MVN_Laplace}
\end{equation}
where $\text{sgn}(\bm w) |\bm w|^{1/2}=(\text{sgn}(w_i)
|w_i|^{1/2}\,:\,i=1,\dots,d)$ and $\text{sgn}(x)=x/\lvert x\rvert$ denotes the signum function.\ In this case, the star set $\mathcal{W}$
describing the density of $\bm X/\lVert \bm X\rVert$ is no longer
a constant scale multiple of $\G$, but instead a radial product of
$\G$ with another star-body.\ 
Figure~\ref{fig:f_W_from_g} also shows the true density (histogram) of
$\bm X/\lVert \bm X\rVert$ and $r_{\G}(\bm w)^d/\{d\vol{\G}\}$
(solid curve).\ The star-body in panel four of the same figure is the
ratio between $\W$ and $\G^d/\{d\vol{\G}\}$.\ 

The analysis of the marginal transformation and its effect on the distribution of the directions reveals a case where $\mathcal{W}$ can depend
both on $\G$ and on some additional star-body $\B$ independent of $\G$.\ The star-body $\B$ captures residual directional
variation that is not explained by a homothetic density.\ For
instance, assume that $\bm X$ follows a bivariate max-stable logistic
distribution in standard Laplace margins.\
From Figure~\ref{fig:f_W_from_g}, we see that
$\W\neq\G^d/(d\vol{\G})$.\ For the practical case where both $r_{\W}$
and $r_{G}$ are strictly positive on $\SSS^{d-1}$, then
$\LL=\B\rmult\G$, $ \B, \G \in \bigstar$ with $\B$ independent of
$\G$.\ Hence, in the specific context where $f$ is homothetic with
respect to $\rG$, $\B=B_{1}(\Orig)$.\ 

A third possibility for the set $\LL$, and consequently the set $\W$,
is to set $\LL=\B$ for some set $\B$ independent of $\G$, implying a
case where $\bm W$ is independent of $\G$.\ This can occur within the
class of radial generalised Pareto distributions, introduced in the
following section.\ In Supplementary Material~\ref{supp:angle}, we
derive closed-form expressions of the distribution of directions
$f_{\bm{W}}$ for a number of distributions when the marginal
distributions are common and prespecified.\

\subsection{Radial generalised Pareto distributions}
\label{sec:mult_rad_stab_distr}
We now present a novel family of multivariate distributions, termed
radial generalised Pareto ($\rGP$) distributions and detail some of
their stability properties making them suitable for extrapolation.\
They arise as the only non-trivial limits of \textit{radially
  renormalised} exceedances above a threshold $\QS_q$, as detailed in
Section~\ref{sec:PP_characterisation}, and enable the modelling of
rare events in a much wider set of joint-tail regions of $\RR^d$ than
other well-established frameworks of extreme value theory.\ Radial generalised Pareto
distributions are parameterised via members of
the class of star bodies $\bigstar$.\


\begin{definition}[Multivariate radial generalised Pareto
  distributions]
  \label{def:mult_rad_stab_GP}
  A random vector $\bm Z\in \RR^d$ is said to follow the multivariate
  radial generalised Pareto distribution with location
  $\loc$, scale $\Sigma$, and shape $\xi:\SSS^{d-1}\to \RR$
  and directional shape
  $\W$, if for any Borel set
  $K\subseteq\RR^{d}\setminus\loc$,
  \begin{equation}
    \PR\left[\bm Z\in K\right]= \int\limits_{\SSS^{d-1}}\int\limits_{]\Orig:\bm w)\cap K}
      \left\{1+\xi(\bm w) \frac{r-r_{\loc }(\bm w)}{r_{\Sigma}(\bm w)}\right\}_+^{-1/\xi(\bm w)}\,\rW(\bm w) \,\dd r\,\dd \bm w.
    \label{eq:RingStab_gP}
  \end{equation}
\end{definition}
The radial generalised Pareto distribution presented in
Definition~\ref{def:mult_rad_stab_GP} possesses stability properties
detailed in Propositions~\ref{prop:MRS_stability_GP}
and~\ref{prop:MRS_stability_exp}, where Proposition~\ref{prop:MRS_stability_exp} is concerned with the special case of the when $\xi(\bm{w})=0$ for all $\bm{w}\in\mathbb{S}^{d-1}$, corresponding to a radial exponential distribution.
\begin{proposition}
  Let $\bm Z\in\RR^d$ follow a radial generalised Pareto distribution with location
  $\loc$, scale $\Sigma$, shape $\xi:\SSS^{d-1}\to \RR$,
  and directional shape
  $\W$, then the quantile and isotropic probability sets of $\bm Z$ are respectively
  \begin{IEEEeqnarray}{rCl}
        \label{eq:Qq_rGP}  
        \QSq &=& \loc \radd \xi^{-1}\{B_{1-q}(\Orig)^{-\xi}-B_{1}(\Orig)\}\rmult\Sigma, \quad q\in(0,1), \\
        \QS_q^{\star} &=& \QSq + \xi^{-1}[\{\W\rmult B_{\vol{\SSS^{d-1}}}(\bm 0)\}^{\xi} -B_{1}(\Orig)]\rmult \Sigma, \quad q\in(q_l(\PR_{\bm Z/\lVert \bm Z \rVert}\,\|\,\mu^\star),1),
        \label{eq:Qqstar_rGP}      
    \end{IEEEeqnarray}
  where $q_l$ is defined in expression~\eqref{eq:q_l}, and where $\B^\xi$ for $\B\in\bigstar$ is interpreted as $r_{\B}(\bm w)^{\xi(\bm w)}$, for all $\bm w \in \text{\normalfont supp}(\PR_{\bm Z/\lVert \bm Z \rVert})$.\ Further, if $\xi$ is constant, $\bm Z$ satisfies the radial stability property
  \begin{IEEEeqnarray}{rCl}
  \label{eq:MRS_stability_GP}
    \PR\left[\bm Z\in \{\loc {\radd}B_{r_1+r_2}(\Orig){\rmult} \Sigma\}^\prime\mid \bm Z \in \{\loc{+}B_{r_1}(\Orig){\rmult} \Sigma\}^\prime\right] {=} \PR\left[\bm Z \in \left\{\loc {+} \frac{B_{r_2}(\Orig){\rmult}\Sigma}{B_1(\Orig){\radd} B_{\xi}(\Orig){\rmult} B_{r_1}(\Orig)}\right\}^\prime\right].\qquad
  \end{IEEEeqnarray}
  \label{prop:MRS_stability_GP}
\end{proposition}
    \begin{proposition}
    \label{prop:MRS_stability_exp}
    Let $\bm Z\in\RR^d$ follow a radial exponential distribution with location
    $\loc$, scale $\Sigma$, and directional shape
    $\W$, then the quantile and isotropic probability sets of $\bm Z$ are respectively
    given by
    \begin{IEEEeqnarray}{rCl}
        \label{eq:Qq_rExp}
        \QSq &=& \loc \radd B_{-\log(1-q)}(\Orig)\rmult\Sigma, \quad q\in(0,1), \\ 
        \QS_q^{\star} &=& \QSq + \log\{\W\rmult B_{\vol{\SSS^{d-1}}}(\bm 0) \}\rmult \Sigma, \quad q\in(q_l(\PR_{\bm Z/\lVert \bm Z \rVert}\,\|\,\mu^\star),1).
        \label{eq:Qqstar_rExp}      
    \end{IEEEeqnarray}
    Further, $\bm Z$ satisfies the radial memoryless property that
    \begin{equation}
    \PR\left[\bm Z\in [\loc \radd B_{r_1+r_2}(\Orig)\rmult \Sigma]^\prime\mid \bm Z \in \{\loc +B_{r_1}(\Orig)\rmult \Sigma\}^\prime\right]= \PR\left[\bm Z\in \{\loc \radd B_{r_2}(\Orig)\rmult\Sigma\}^\prime\right].
    \label{eq:MRS_stability_exp}
    \end{equation}
\end{proposition}
Proofs of stability properties~\eqref{eq:Qq_rGP}
and~\eqref{eq:Qq_rExp} are given in Appendix~\ref{proof:stab_exp}
and~\ref{proof:stab_gP}.\ We note that in
equation~\eqref{eq:MRS_stability_exp}, equality holds when replacing
the location $\loc$ by any $K\in\bigstar$ with $K\subseteq\loc$.\

Given a probability set $\QS_q$ of a random vector $\bm Z$, the stability properties established above allow extrapolation to probability sets with probability $q^\prime\geq q$.\ If $\bm Z$ follows a radial generalised Pareto distribution as in Proposition~\ref{prop:MRS_stability_GP}, then $\bm Z\mid \bm Z \notin \QSq$ follows a radial generalised Pareto distribution with scale $\Sigma^\prime = \Sigma + \xi\rmult(\QSq - \loc)$, and for $ q^\prime\geq q$,
\begin{IEEEeqnarray}{rCl}
\label{eq:rs_stability_gP}
\QS_{q^\prime} &=& \QSq \radd \{B_{(1-q)/(1-q^\prime)}(\Orig)^{\xi}-B_{1}(\Orig)\}\rmult\Sigma^\prime\rdiv \xi, 
\end{IEEEeqnarray}
and $\QS_{q^\prime}^\star$ is obtained according to equality~\eqref{eq:Qqstar_rGP}.\ Similarly, if $\bm Z$ follows a radial exponential distribution as in Proposition~\ref{prop:MRS_stability_exp}, then $\bm Z\mid \bm Z \notin \QSq$ follows a radial exponential distribution with same scale $\Sigma$, and for $ q^\prime\geq q$,
\begin{IEEEeqnarray}{rCl}
\label{eq:rs_stability_exp}
  \QS_{q^\prime} &=& \QSq \radd B_{[\log\{(1-q)/(1-q^\prime)]\}}(\Orig)\rmult\Sigma,
\end{IEEEeqnarray}
and $\QS_{q^\prime}^\star$ obtained according to equality~\eqref{eq:Qqstar_rExp}.\

In Section~\ref{sec:inference}, we propose a statistical inference method
utilising the family of radial generalised Pareto distributions to extrapolate to extreme regions lying beyond the range of observed data.\ 


\section{Statistical inference}
\label{sec:inference}
\subsection{Standardisation of margins}
\label{sec:standardisation}

Suppose that $\obsdatbold_1, \dots, \obsdatbold_n$, are observations
drawn randomly from the distribution of the random variable
$\obsrvbold = (O_1, \dots, O_d)^\top$.\
When the tails of the original distribution decay exponentially, we
have the flexibility to model a broader range of extremal dependence
structures, not only allowing for both asymptotic independence and
asymptotic dependence, but also allowing for more complex types of
dependencies such as when some coordinates exhibit positive extremal
dependence while others exhibit negative extremal dependence
\citep{keefpaptawn13, nolde2021linking}.\ When $\bm{O}$ does not have
exponential decay, the transformation
$X_j = F_L^{-1}(\widehat{F}_j(O_j))$ must be applied for
$j=1,\dots,d$, where $F_L^{-1}$ is the distribution function of the
standard Laplace distribution, and
\begin{equation}
  \widehat{F}_j(o) =
  \begin{cases}
    \Big[1-\left\{1 - {\xi}_{j, -} \left(\dfrac{o - u_{j,-}}{{\sigma}_{j,-}}\right)\right\}_+^{-1/{\xi}_{j,-}}\Big] \widetilde{F}_j(u_{j,-})& o \leq u_{j,-}\\
    \widetilde{F}_j(o)& u_{j,-} < o \leq u_{j,+}\\
    1-[1-\widetilde{F}(u_{j,+})]\left\{1 + {\xi}_{j,+} \left(\dfrac{o - u_{j,+}}{{\sigma}_{j,+}}\right)\right\}_+^{-1/{\xi}_{j,+}}& o > u_{j,+}
  \end{cases} .
  \label{eq:PIT_transformation}
\end{equation}
$\widetilde{F}_j(o) = (n+1)^{-1}\sum_{i=1}^n \mathbbm{1}[\obsdat_{ij}
\leq o]$, and $(\widehat{\sigma}_1, \widehat{\xi}_1)$ and
$(\widehat{\sigma}_{j,-}, \widehat{\xi}_{j,-})$ and
$(\widehat{\sigma}_{j,+}, \widehat{\xi}_{j,+})$ are scale and shape
parameters of generalised Pareto distributions for the lower and upper
tail of the $j^{\text{th}}$ margin, which are used to model the tail
decay below and above the thresholds $u_{j,-}$ and $u_{j,+}$,
respectively.

In practice, the thresholds $u_{j,-}$ and $u_{j,+}$ can be modelled as
quantiles of the marginal distribution.\ By adopting a Bayesian
approach, we can assign a posterior distribution to the joint
distribution of these quantiles.\ This posterior enables us to sample
candidate thresholds, above which separate GP models are implemented
for the tails.\ This approach naturally accounts for uncertainty in
the marginal distributions' parameters, as the posterior distribution
of the quantiles propagates this uncertainty forward into the tail
modelling process.\ Adopting a Frequentist approach, we can estimate
GP parameters using maximum likelihood to the exceedances above
empirical quantiles, and propagate the uncertainty in the quantile
estimate and the corresponding GP parameters forward using a bootstrap
approach.\ 




\subsection{Quantile regression}
\label{sec:QR}
Quantile regression methods are typically implemented using a pinball
loss function~\citep{koenker2005quantile} and often without a
distributional model for the density.\ Although~\cite{yumoy01} propose
the use of the asymmetric Laplace distribution for the model density
due to the equivalence of the negative log density of the
asymmetric Laplace with the pinball loss function, naively treating the
asymmetric Laplace as an adequate model for the data is precarious in
a Bayesian setting.\ For example, \cite{waldmann2013bayesian} show
that the resulting posterior prediction intervals have poor
frequentist calibration properties, and this is especially pronounced
for tail quantiles, which are essential in our setup.\ Second, the
scale parameter of the asymmetric Laplace distribution is arbitrary in
a Bayesian framework and even maximum likelihood based estimators are
known to lead to inaccurate quantile estimates,
see~\cite{fasiolo2021fast}.\

We adopt a \textit{generalised linear model} based approach for
quantile regression, requiring an adequate distributional model for
the density of $R\mid \bm W$.\ In particular, we assume that
$R\mid \bm W=\bm w$ follows a Gamma distribution and model the
logarithm of its conditional $q$-quantile $\log \rquant(\bm w)$ using
a a finite-dimensional continuously specified Gaussian process prior
on $\SSS^{d-1}$ (see Supplementary Material~\ref{supp:Posterior}).\
\color{black} We model the logarithms of the radial function
$\rquant$ by approximations of Mat\'ern Gaussian fields on
$\SSS^{d-1}$ using the stochastic partial differential equation
approach by \cite{lindetal11}, with $\alpha=2$ (see their
Equation~(2)) which is also the default option in the \texttt{R-INLA}
package
(\texttt{www.r-inla.org}).\ 
\color{black}

There are at least two strengths behind this choice.\ First, due to
Proposition~\ref{prop:rof}, the Gamma distribution serves as a good
approximation for the density in the tail region of the distribution
of $R\mid \bm W$.\ Since estimators of high quantiles are not
influenced by the bulk of the distribution, a likely misspecification
between our choice and the true density of $R\mid \bm W$ in the body
of the distribution is not of concern.\ Second, our choice exploits
the form of the decay of the conditional density and allows for
Bayesian inference for the conditional quantile under a model that can
adequately describe the behaviour in the tail.\ In this paper, we
adopt a Gamma quantile regression models due to the simplicity that it
affords, but remark that further improvements may be achieved by
using, for example, a truncated Gamma distribution, so as to further
eliminate effects from the body of the distribution of $R\mid \bm W$
\citep{wadsworth2022statistical}.\

Given i.i.d.\ observations $\bm o_1,\ldots,\bm o_n$ from a random
vector $\bm O$, we obtain the standardised data
$\bm x = \{\bm x_1,\ldots,\bm x_n\}$ via
$\bm x_i
=(F^{-1}(\widetilde{F}_1(o_{i,1})),\ldots,F^{-1}(\widetilde{F}_d(o_{i,d})))$
as described in Section~\ref{sec:standardisation}.\ To infer the
quantile set $\QSq$ of $\bm X$, we treat
$\dat = \{(r_i,\bm w_i)\,:\,i=1,\dots,n\}$, with
$(r_i,\bm w_i)=(\lVert \bm x_i\rVert, \bm x_i/\lVert \bm x_i\rVert)$,
as observations from $(R,\bm W)$ and apply the Gamma quantile
regression method detailed above.

\subsection{Conditional likelihood function}
\label{sec:likelihood}
Given that $\bm X$ has standard Laplace margins, we assume that
exceedances of $\QSq$ follow a radial GP distribution with location
$\QSq$, scale $\G$, radial shape $\xi$, and directional shape $\W$.\
Below we describe a likelihood-based statistical inference approach
for these parameters.\ 

We recall from equation~\eqref{eq:W_from_L} that the set $\W$ can be
expressed as $f_{\bm W}(\bm w)= r_{\LL}(\bm w)^d/(d\vol{\LL})$, where
$\LL$ is a star-body.\ Based on Section~\ref{sec:density_angles}, we
motivate three models for $\W$ given by
\begin{IEEEeqnarray}{rCl}
  \Mh:\LL=\B, \quad \quad \Mg:\LL=\G, \quad \quad
  \Mgh:\LL=\B\rmult\G,
\end{IEEEeqnarray}
for $\B \in \bigstar$ independent of $\G$.\ We note that
the latent star-bodies $\B$ in models $\Mh$ and $\Mgh$ do not have the
same interpretation, but the slight abuse of notation allows to
simplify the notation for the parameter space in later sections
without loss of interpretation.\ 
The nested structure of model $\Mg$ within the parameter space of
model $\Mgh$ translates into a bias-variance trade-off as the latter
offers additional flexibility at the cost of a possibly increased
variance for the latent set $\G$.\ A similar trade-off occurs for $\G$
between models $\Mh$ and $\Mgh$ since the former ignores possible
information contained in the observed directions.\

Conditionally on the quantile set $\QSq$ and the data $\dat$, the
likelihood of $\bm \theta = (\G,\LL,\xi)$, is given by
$\smash{L}(\bm \theta \mid \QSq, \dat):= \prod_{i=1}^n f_{R, \bm
  W}(r_i, \bm w_i\mid \bm \theta,\QSq)$.\ Letting
$\mathcal{S}_q:=\{i\in \{1,\dots,n\}\,:\, r_i > \rquant(\bm w_i)\}$
denote the set of random indices corresponding to exceedances of
$\QSq$, and
$\mathcal{S}_q^\prime:=\{1,\ldots,n\}\backslash \mathcal{S}_q$,
$\smash{L}$ can be expressed in terms of contributions of the radii
$\{r_1,\ldots,r_n\}$ above and below $\rquant$ via
\begin{IEEEeqnarray}{rCl}
  && L(\bm \theta\mid \QSq, \dat) = \left[\prod_{i\in
      \mathcal{S}_q^\prime}f_{R\mid R\leq \rquant(\bm W),\bm
      W}(r_i)f_{\bm W}(\bm w_i)\right]\left[\prod_{i \in
      \mathcal{S}_q}f_{R\mid R> \rquant(\bm W),\bm W}(r_i)f_{\bm
      W}(\bm w_i)\right].
  \label{eq:likelihood_alldata}
\end{IEEEeqnarray}
Under the assumption that $\dat$ is a random sample,
Corollary~\ref{cor:equality_of_W_distributions} suggests that all
directions $\{\bm w_1,\ldots,\bm w_n\}$ may be used in the likelihood for
the inference of $f_{\bm W} = r_{\LL}^d/(d\vol{\LL})$ used to model
$\bm W\mid R>r_{\QS_q}(\bm W)$.\ In particular, when models $\Mh$ or
$\Mgh$ are used, substantial gains in the inference for $\G$ can be
attained by also including in inference the directions at which
non-exceedances occur, as illustrated in the simulation study of
Supplementary Material~\ref{sec:sim_study}.\ Based on
Theorem~\ref{thm:PPconvergence}, non-exceedance radii
$\{r_i:i\in\mathcal{S}_q^{C}\}$ are assumed not to carry
information about $\G$ and $\LL$; we pose that $\LL$ is constant with
respect to them.\ Denoting by $\mathcal{S}_{\bm w}$ the set of indices
of at least all exceedances and at most all observations---or
$\mathcal{S}_{q}\subseteq \mathcal{S}_{\bm
  w}\subseteq\{1,\ldots,n\}$---the likelihood thus reduces to
\begin{IEEEeqnarray}{rCl}
  L(\bm \theta\mid \QSq,\dat) &\propto&
  \prod_{i\in\mathcal{S}_q} f_{R\mid R> \rquant(\bm W),\bm W}(r_i\mid \bm w_i) \prod_{i\in\mathcal{S}_{\bm w}} f_{\bm W}(\bm w_i)\nonumber\\
    &=& \exp\left\{-\lvert \mathcal{S}_{\bm w} \rvert \log
      (d\vol{\LL}) \right\}\prod_{i\in \mathcal{S}_q} f_{R_E\mid \bm
      W}\left((r_i - \rquant(\bm w_i))/ {\rG}_q(\bm
      w_i)\mid\bm w_i\right) \prod_{i\in\mathcal{S}_{\bm
        w}}r_{\LL}(\bm w_j)^d,\quad
  \label{eq:likelihood}
\end{IEEEeqnarray}
where
$f_{R_E\mid \bm W}(z\mid \bm w) = [1+\xi(\bm w) z]_+^{-1/\xi(\bm w)-1}/\rG(\bm w)$.\ 
The likelihood function~\eqref{eq:likelihood} is amenable to standard
likelihood based inference using either frequentist or Bayesian
methods when parametric models are selected for $\mathcal{W}$ and
$\G$.\ When interest is in semi-parametric models, a key complication
arises when using the likelihood in expression~\eqref{eq:likelihood}.\
Evaluating the likelihood function~\eqref{eq:likelihood} requires
computing the constant $d\vol{\LL}$ which ensures the density of
$\bm W$ integrates to one.\ This normalising constant is in several
cases difficult to compute exactly, which makes inference difficult.\
However, using the \textit{Poisson transform} \citep{bake94}, we can
map the likelihood into an \textit{equivalent} likelihood function
$L( \bm \theta,\beta \mid \QSq,\dat)$ of a Poisson point process
defined in an expanded space given by
\begin{align}\label{eq:likelihood-v2}
  L(\bm \theta,\beta)
  = \exp\left[-\lvert \mathcal{S}_{\bm w} \rvert
  e^\beta (d\vol{\LL})\right] \prod_{i\in
  \mathcal{S}_q}f_{R_E\mid \bm W}\left((r_i - \rquant(\bm w_i))/ {\rG}_q(\bm
      w_i)\mid\bm w_i\right) \prod_{i\in\mathcal{S}_{\bm w}}e^\beta r_{\LL}(\bm w_i)^d,
\end{align}
where the latent variable $\beta$ estimates the normalising constant
$d\vol{\LL}$; it is inferred as another parameter at no loss of
information \citep{barthelme2015poisson}.\ It is worth noting that
through this approach, we still need to compute the volume of $\LL$.\
However, since our parametrisation assumes that $\log r_\LL$ is linear
within, $\vol{\LL}$ can be estimated numerically in a stable and
efficient manner 
based on the method introduced by \cite{simpsetal16}, see also
\cite{yuanetal17,Lindgren_ExpLik_2023} and
\cite{papastathopoulos2023bayesian} for further applications of this
method.\ \color{black} 
Inference can be performed either using frequentist methods or in a
fully Bayesian manner, that is, by assigning suitable prior
distributions on $\G$, $\LL$, and $\beta$
\citep{Lindgren_ExpLik_2023}.\ 
More details on statistical inference for the latent variables and on
how the fitted models can be used to perform rare event probability
estimation are found in Supplementary
Material~\ref{sup:Statistical_inference}.

\subsection{Assessment of convergence to limit distribution}
\label{sec:validation_selection}

In this section, we outline a novel diagnostic, allowing assessment of
the weak convergence in Theorem~\ref{thm:PPconvergence}.\ Following
Remarks~\ref{rem:alternative_renormalization} and
\ref{rem:uniform_over_the_ball}, observe that under the conditions of
Theorem~\ref{thm:PPconvergence} we have
\begin{IEEEeqnarray*}{rCl}
  &&H_{\bm W}\left(\frac{R - r_{\QS_{q}^\star}(\bm
      W)}{r_{\G_{q}^\star}(\bm W)}\right)^{1/d} \bm W ~\Big|~ \left\{R
    > r_{\QS_{q}^\star}\left(\bm W\right)\right\}\cind \bm U_{B_1(\bm
    0)}, \quad \text{as $q\to 1$},
\end{IEEEeqnarray*}
where $\bm U_{B_1(\bm 0)}$ is a random variable with a uniform
distribution \textit{in} the $d$-dimensional unit ball ${B_1(\bm 0)}$,
that is,
$\PR_{\bm U_{B_1(\bm 0)}}(A)= \vol{A\cap B_1(\bm 0)}/\vol{\SSS^{d-1}}$
for any Borel set $A\subseteq B_1(\bm 0)$.\ Here,
$r_{\QS_{q}^\star}:=b_{1/(1-q)}^\star $ and
$r_{\G_{q}^\star}:=a_{1/(1-q)}^\star$, where $b_{1/(1-q)}^\star$ and
$a_{1/(1-q)}^\star$ are given by the functions~\eqref{eq:b_n_star}
and~\eqref{eq:a_n_star}, that is,
\[
  r_{\QS_{q}^\star}(\bm w) = r_{\QS_{q}}(\bm w) + r_{\G_{q}}(\bm
  w)[\{\vol{\SSS^{d-1}}r_{\W}(\bm w)\}^{\xi} -1]/\xi  \quad \text{and}\quad  r_{\G_{q}}^\star(\bm w)= r_{\G_{q}}(\bm
  w)\{\vol{\SSS^{d-1}}r_{\W}(\bm w)\}^{\xi}, 
\]
for $\bm w \in \SSS^{d-1}$ and $q$ such that $\inf\{r_{\QS_{q}^\star}(\bm w)\,:\,\bm w \in \SSS^{d-1}\} \geq 0$.\

We note the power $1/d$ used in the transformation of the renormalised
radius as this transformation results in the aforementioned limit
distribution, possessing particularly convenient features that allow
us to build a simple diagnostic based on the second order properties
\citep{Ripley1976,Ripley1977modelling} of a stationary random point
measure.\ 
Here, stationarity means that the distribution of $P^\star$ is invariant under rotations in $\RR^d$.\ 
In particular, the stationary random point
measure
\begin{IEEEeqnarray}{rCl}
  P^\star := \sum_{i=1}^n\delta\left[ H_{\bm W_i}\left(\frac{R_i -
        r_{\QS_{q}^\star}(\bm W_i)}{r_{\G_{q}^\star}(\bm
        W_i)}\right)^{1/d}\bm W_i\right]\mathbbm{1}_{R_i >
    r_{\QS_{q}^\star}\left(\bm W_i\right)},
  \label{eq:stationary_point_measure}
\end{IEEEeqnarray}
has size
$n^\star:=\sum_{i=1}^n\mathbbm{1}_{\lVert \bm X_i\rVert >
  r_{\QS_{q}^\star}\left(\bm X_i/\lVert \bm X_i\rVert\right)}$.\ In
what follows, we develop a method to assess if it is statistically
distinguishable from a random point measure with constant intensity on
$B_{1}(\bm 0)$.\

For a random point measure $\Pi_{n^\star}:=\sum_{i=1}^{n^\star}\delta_{\bm U_i}$ on
$B_{1}(\bm 0)$,
we define the
normalised reduced second-order measures $K_{B}$ and $K_{C}$ by
\begin{IEEEeqnarray}{rCl}
  \label{eq:lambda_K}
  &&K_{B}(r) =\lambda^{-1}\mathbb{E} \bigg[\sum_{j\neq i}
  \mathbbm{1}_{\bm U_{j}}\{B_r(\bm U_{i})\} \bigg] \quad\text{and} \quad
  K_{C}(r) = \lambda^{-1}\mathbb{E} \bigg[\sum_{j\neq i}
  \mathbbm{1}_{\bm U_{j}}\{C_r(\bm U_{i})\} \bigg],
\end{IEEEeqnarray}
where $\lambda$
is the average intensity of the
point measure over $B_1(\bm 0)$, 
$B_r(\bm U_j)$ is the ball of radius $r$ centred at $\bm U_j$,
and $C_r(\bm U_j)$ is the spherical sector with apex at~$\bm 0$,
extending outwards in the direction of $\bm U_j$, with half aperture angle
$r$.\ 
The specific geometries of $B_r$ and $C_r$ are chosen such that they reveal regions of deviation from constant intensity of an underlying random point measure occurring in radial-directional and strictly directional neighbourhoods, respectively.\ 
In particular, if $\{\bm U_i\,:\,i=1,\dots,n^\star\}$ is a random sample of uniformly
distributed directions arising from the distribution of
$U_{B_1(\bm 0)}$, then we have
\begin{IEEEeqnarray}{rCl}
\label{eq:K_uniform}
  &&K_{B}(r) = \int_{B_1(\bm 0)} \frac{\vol{B_r(\bm x)\cap B_1(\bm
      0)}}{\vol{B_1(\bm 0)}} d\bm x \quad\text{and} \quad K_{C}(r) =
  \int_{B_1(\bm 0)} \frac{\vol{C_r(\bm x)\cap B_1(\bm
      0)}}{\vol{B_1(\bm 0)}} d\bm x.
\end{IEEEeqnarray}

Hence, given an observed point pattern
$P^\star$ obtained via
transformation~\eqref{eq:stationary_point_measure} of a random sample
from $\PR_{\bm
  X}$ and a sample from the posterior distribution of
$\QS_q^\star$ and $\G_q^\star$, we can compare the functions
$K_{B}$ and
$K_{C}$ given by expression~\eqref{eq:K_uniform} with the empirical
estimates of the measures~\eqref{eq:lambda_K},
\begin{IEEEeqnarray}{rCl}
  \label{eq:k_hat}
  && \widehat{K}_{B}(r\mid P^\star) =\frac{\vol{B_1(\bm
      0)}}{(n^\star)^2}\sum\limits_{\substack{\bm u_i, \bm u_j\in P^\star\\ \bm u_i\neq \bm u_j}}\mathbbm{1}_{\bm u_j}\{B_r(\bm u_i)\} \quad \text{and}\quad
  \widehat{K}_{C}(r\mid P^\star) =\frac{\vol{B_1(\bm
      0)}}{(n^\star)^2}\sum\limits_{\substack{\bm u_i, \bm u_j\in P^\star\\ \bm u_i\neq \bm u_j}}\mathbbm{1}_{\bm u_j}\{C_r(\bm u_i)\}. \qquad 
\end{IEEEeqnarray}


To evaluate the variability of the function
$\widehat{K}_S(\cdot\mid \Pi_{n^\star})$, where $S$ represents the
sets $B$ or $C$, on random point measures $\Pi_{n^\star}$ with
constant intensity on $B_1(\bm 0)$, we construct a
$(1-\alpha)$-envelope denoted by
$\mathcal{E}_{S, n^\star}(r)=[\mathcal{E}_{S,n^\star}^{\text{low}}(r),
\mathcal{E}_{S,n^\star}^{\text{low}}(r)]$ which satisfies, for $I_B=[0,2]$ and $I_C=[0,\pi]$, that
\begin{IEEEeqnarray}{rCl}
\label{eq:envelopes}
&& \PR\left[\left\{ \mathcal{E}_{S,n^\star}^{\text{low}}(r) \leq
  \widehat{K}_{S}\left(r\mid \Pi_{n^\star}\right)\leq
  \mathcal{E}_{S,n^\star}^{\text{upp}}(r)\,:\,r\in I_S\right\}\right]
  =
  1-\alpha.
\end{IEEEeqnarray}

The envelope bounds are estimated using Monte-Carlo
simulation.\ Specifically, we draw $m$ independent samples of size
$n^\star$ from the distribution $\PR_{U_{B_1(\bm 0)}}$, say $\{\Pi_{n^\star}^{(j)}:j=1,\ldots,m\}$, each
corresponding to a realisation of a random point measure with constant
intensity in $B_1(\bm 0)$.\ For each sample $\Pi_{n^\star}^{(j)}$, we compute the discretised
estimates $\widehat{K}_{S}(\cdot\mid \Pi_{n^\star}^{(j)})$ at a predefined sequence of $r$ values.\
The lower and upper bounds of the envelope are defined in a manner
akin to \cite{Bolin_excursions_2018}, that is, 
as the empirical $(\rho/2)$- and $(1-\rho/2)$-quantiles of the $m$ values of $\{\widehat{K}_S(r\mid \Pi_{n^\star}^{(j)}):j=1,\ldots,m\}$, where $\rho\in(0,1)$ is chosen
such that $100(1-\alpha)\%$ of the $m$ discretised function estimates
lie entirely within the envelope.\
Then, evidence against the assumption that $P^\star$ has constant intensity---observed whenever the estimates
$\widehat{K}_B(\cdot\mid P^\star)$ or $\widehat{K}_C(\cdot\mid P^\star)$ do not lie entirely within their $(1-\alpha)$-envelopes---suggests evidence against the validity of the approximation of the distribution of the observed extremes
via the limit distribution of Theorem~\ref{thm:PPconvergence}.\ 


More details on model selection and validation, on classical goodness-of-fit methods, and on simultaneous
uncertainty bands are found in the Supplementary
Material~\ref{sup:Validation}.\

\section{Newlyn wave heights data analysis}
\label{sec:wave}


We apply our methodology to a dataset of dimension
$d=3$ consisting of hourly measurements of wave height $O_H$, in
meters, wave period $O_P$, in seconds, and surge $O_S$, in meters,
measured over the period 1971--1977 at the Newlyn port in south-west
England.\ The dataset was first analysed
in~\cite{coles1994statistical} in the case where asymptotic dependence was
assumed between all three variables.\ Although a typical assumption, \cite{wadsworth2022statistical} show that assymptotic indepence is a more reasonable assumption.\ Here, we revisit this data with a more flexible approach also enabling the modelling of negative dependence.\ Following previous literature,
we analyse componentwise maxima over 15-hour periods, resulting in a
dataset of $n=2,894$ observations.\ The margins $O_H, O_P, O_S$
of the data are unknown and we thus standardise them to standard Laplace using methods
from Section~\ref{sec:standardisation}, resulting in observations $\widetilde{\bm x} := \{\bm x_i = (x_{H,i},x_{P,i},x_{S,i})\,:\,i=1,\ldots,n\}$ interpreted as random draws from $\bm X=(X_H,X_P,X_S)$.\

We begin by fitting to $\widetilde{\bm x}$ the hierarchical Bayesian gamma quantile
regression model for the quantile set $\QSq$ at the $q=0.9$ probability level.\ 
For each of $20$ samples $\{\QS_{q,i}\,:\, i=1,\ldots,20\}$ from the posterior distribution of $\QS_q$, we assume that the exceedances of $\QS_{q,i}$ follow a radial exponential distribution exactly, and fit the models~$\Mh$, $\Mg$, and $\Mgh$ to them.\ For each model and $\QS_{q,i}$, we sample $50$ observations $\{(\G_{j,i},\W_{j,i})\,:\, j=1,\ldots,50\}$ from the posterior distribution of $\G,\W\mid\QS_{q,i}$.\  
It is deemed via model checks that
$\Mg$ fitted on exceedance directions only provides the best overall joint
fit (see Supplementary Material~\ref{sec:wave-additional-plots}).\

Figure~\ref{fig:wave-fit-figs}
shows the posterior mean of $\QS_q$, $\G$ and $\W$ obtained from model $\Mg$, as well as the posterior means of extrapolated $\QS_{0.95}$ and $\QS_{0.95}^\star$, corresponding to complement return sets with return period $T=300$-hour.\ The probability sets $\QS_{0.95}$ and $\QS_{0.95}^\star$ were inferred via expressions~\eqref{eq:rs_stability_exp}, using the assumption that a radial exponential distribution holds exactly above $\QSq$.\ 

The model captures the extremal dependence structure of
our dataset well as observed through the correspondence between the
posterior mean of $\G$ and the scaled sample cloud.\
Figure~\ref{fig:wave-fit-figs} also shows the directions of exceedances of the posterior means of $\QS_{0.95}$ and $\QS_{0.95}^{\star}$ which, as expected, seem to respectively follow $\W$ and the uniform density on $\SSS^2$ well.\ 

\begin{figure}[t!]
    \centering
    \begin{overpic}[width=0.31\textwidth,trim=0 60 0
    0]{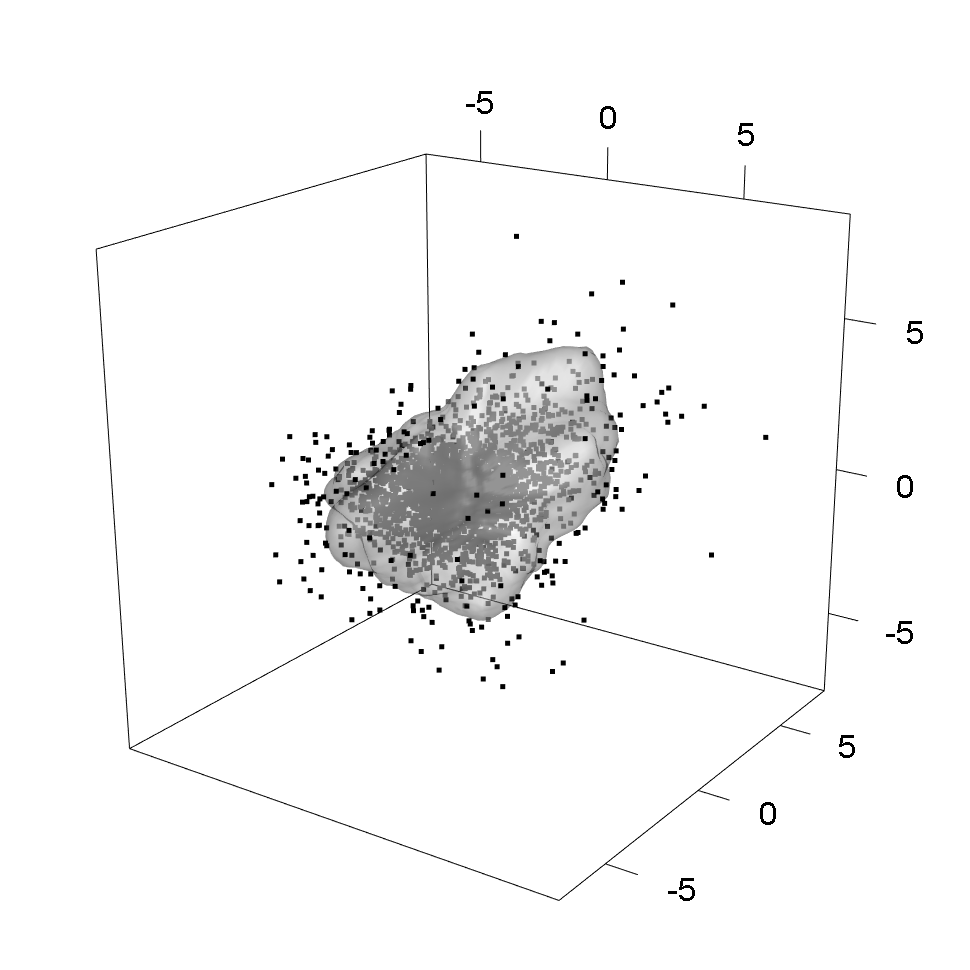}
    \put (63,85) {\scriptsize$X_H$}
    \put (96,42) {\scriptsize$X_S$}
    \put (82,8) {\scriptsize$X_P$}
    \end{overpic}
    \hspace{1.5em}
    \begin{overpic}[width=0.31\textwidth,trim=0 60 0
    0]{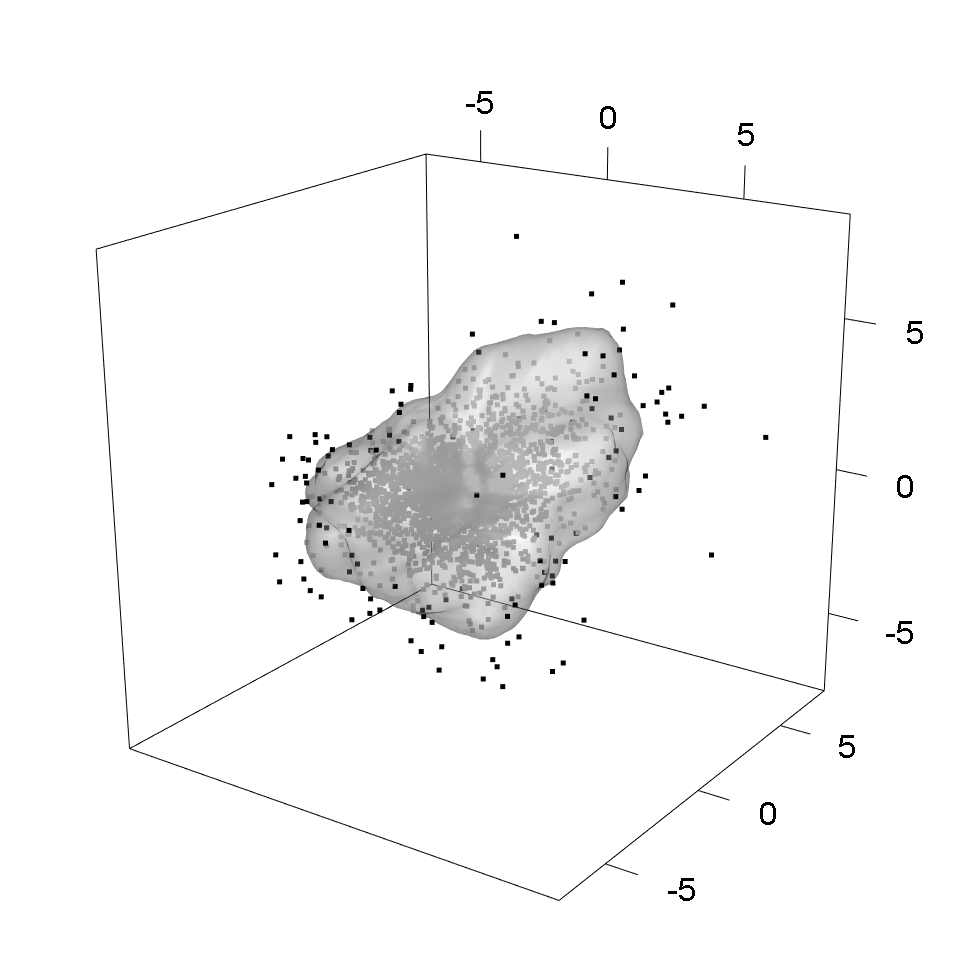}
    \put (63,85) {\scriptsize$X_H$}
    \put (96,42) {\scriptsize$X_S$}
    \put (82,8) {\scriptsize$X_P$}
    \end{overpic}
    \hspace{0.5em}
    \begin{overpic}[width=0.30\textwidth,trim=0 60 0
    0]{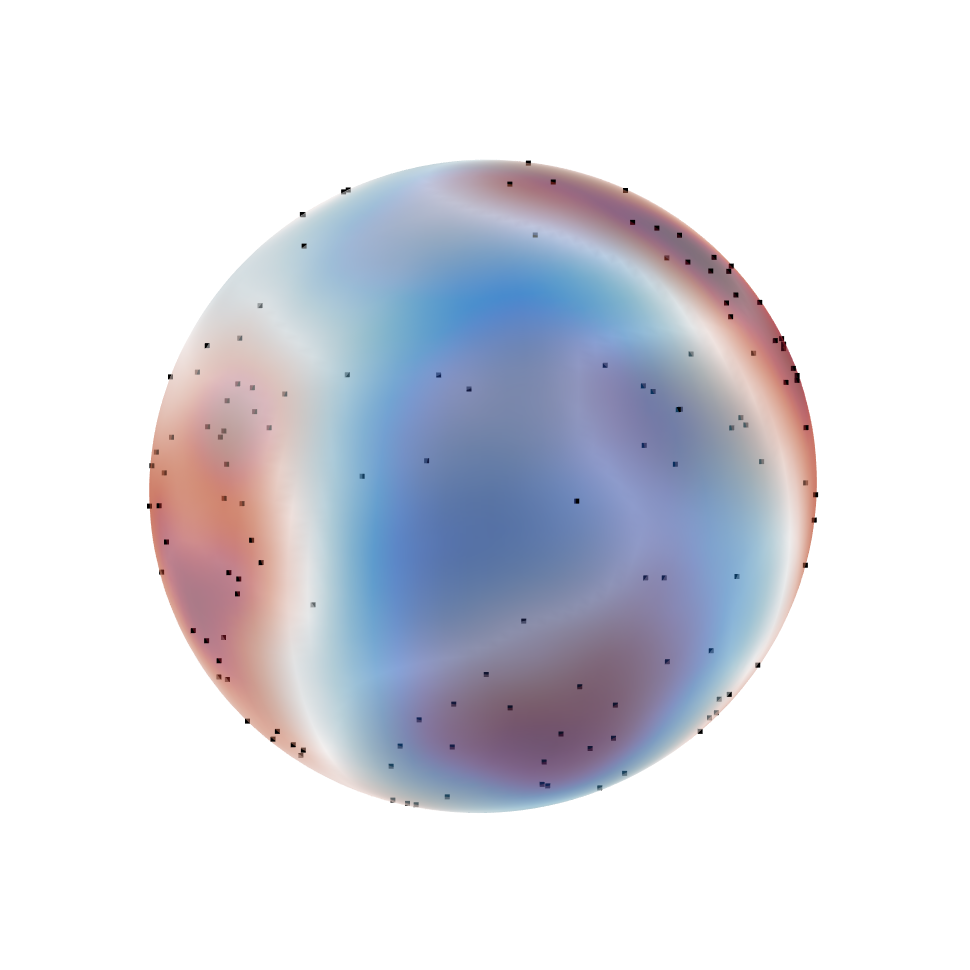}
    \put (73,73) {\small$\SSS^2$}
    \put (14,10) {\scriptsize $\color{BrickRed}\blacksquare\color{black}$ \tiny$-4.17$}
    \put (14,6) {\scriptsize $\color{RoyalBlue}\blacksquare\color{black}$ \tiny$-1.17$}
    \end{overpic}
    \hspace{-3em}
    \begin{overpic}[width=0.31\textwidth,trim=0 60 0
    0]{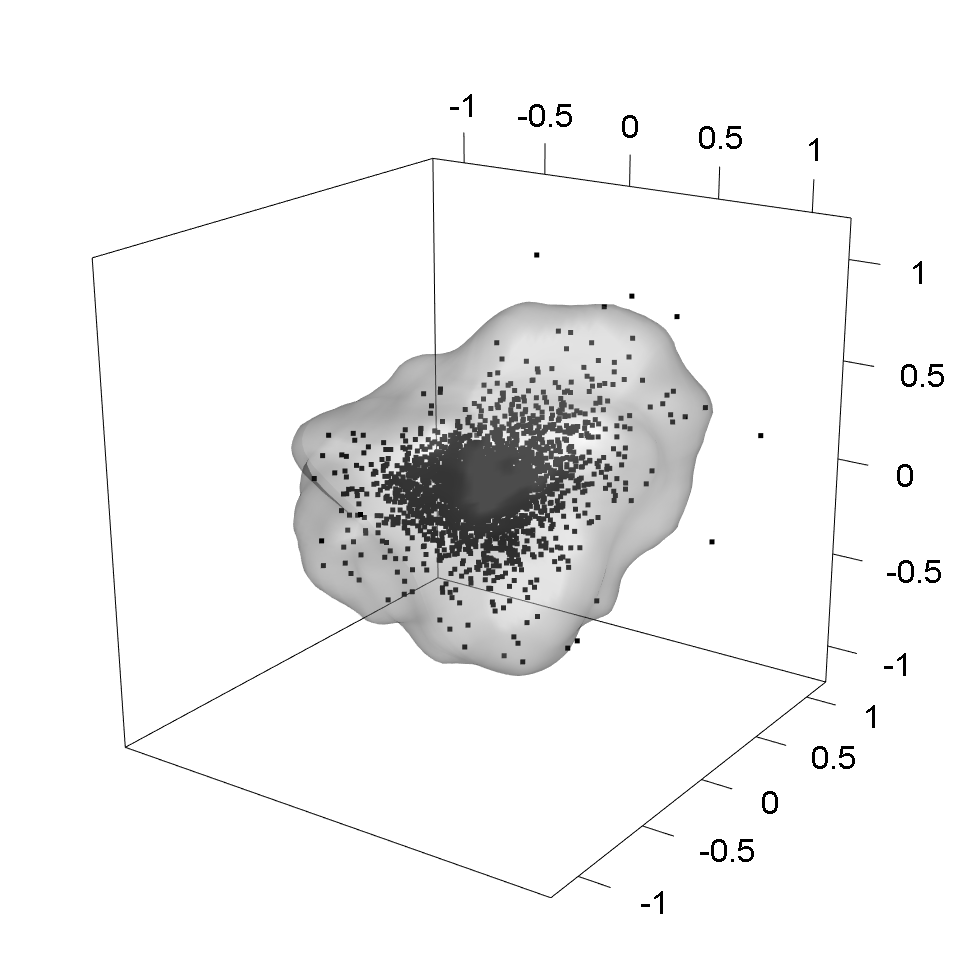}
    \put (63,85) {\tiny$X_H/\log(n/2)$}
    \put (96,44) {\tiny$X_S/$}
    \put (82,8) {\tiny$X_P/\log(n/2)$}
    \end{overpic}
    \hspace{1.5em}
    \begin{overpic}[width=0.31\textwidth,trim=0 60 0
    0]{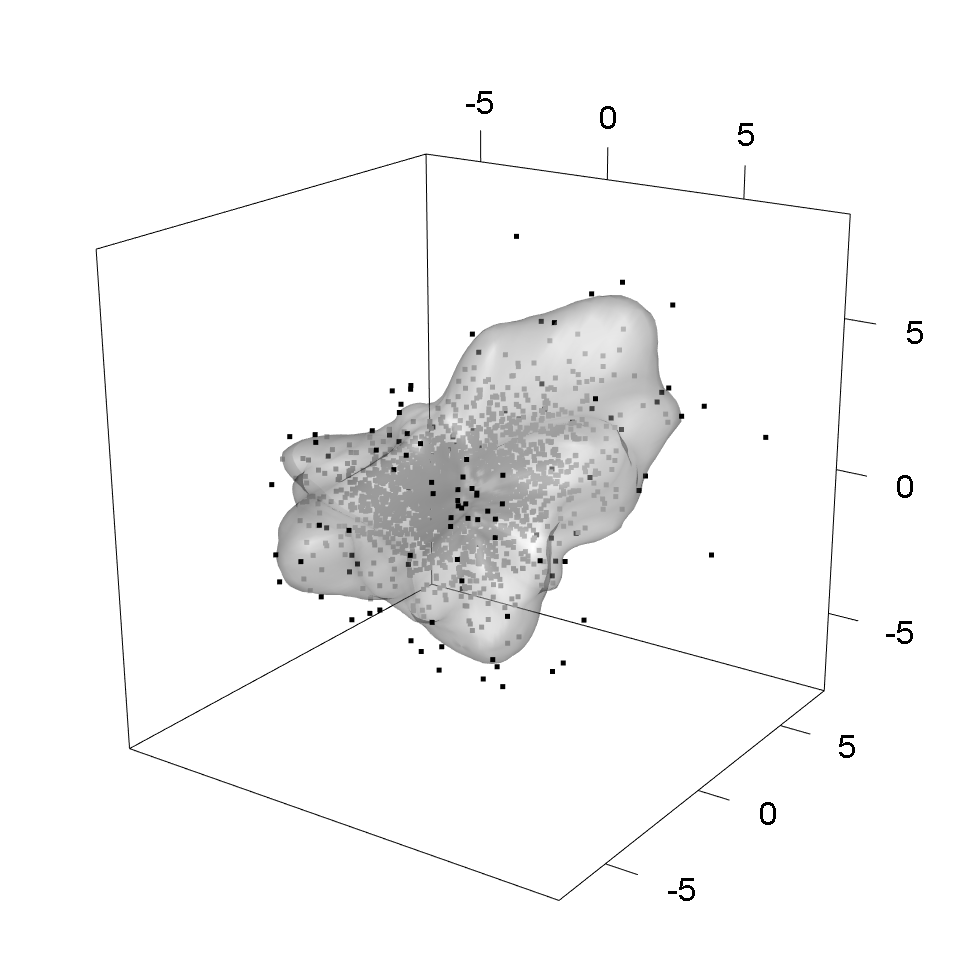}
    \put (-11,44) {\tiny$\log(n/2)$}
    \put (63,85) {\scriptsize$X_H$}
    \put (96,42) {\scriptsize$X_S$}
    \put (82,8) {\scriptsize$X_P$}
    \end{overpic}
    \hspace{0.5em}
    \begin{overpic}[width=0.30\textwidth,trim=0 60 0
    0]{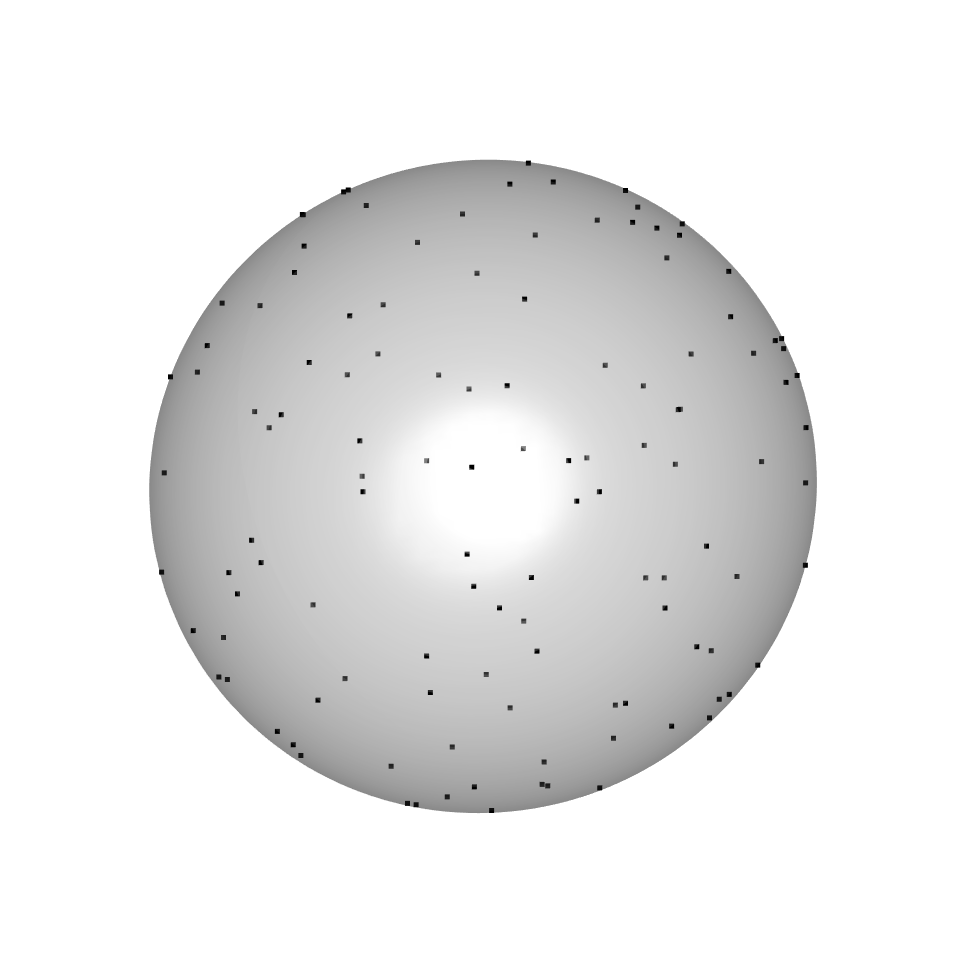}
    \put (73,73) {\small$\SSS^2$}
    \put (14,7) {\scriptsize $\color{gray}\blacksquare\color{black}$ \tiny$-\log(4\pi)$}
    \end{overpic}
    \caption{ Posterior means of model variables for the
      Newlyn wave data.\ \textit{Top-left}:\ quantile set $\QS_{0.9}$.\ \textit{Bottom-left}:\ boundary of the scaling set
      $\G$.\ \textit{Top-centre}:\ quantile set $\QS_{q}$.\ \textit{Bottom-centre}:\ isotropic
      probability set $\QS_{q}^\star$.\
      \textit{Top-right}:\ $\log$-density of $\bm W\mid R>r_{\QS_{q}}$ plotted over       $\SSS^2$ with observations in $\RS_{1/(1-q)}$ projected onto $\SSS^2$.\
      \textit{Bottom-right}:\ logarithm of the density of the uniform distribution on $\SSS^2$ with observations in $\RS_{1/(1-q)}^\star$ projected onto
      $\SSS^2$.\ The probability level $q$ is set to $0.95$.
      }
    \label{fig:wave-fit-figs}
\end{figure}

To assess model adequacy, we use the method developed in Section~\ref{sec:validation_selection}.\ We apply transformation~\eqref{eq:stationary_point_measure} to $\widetilde{\bm x}$ for each of the posterior samples $(\QS_{q,i},\G_{i,j},\W_{i,j})$, with $i=1,\ldots,20$ and $j=1,\ldots,50$, resulting in 1000 associated observed point patterns $P^\star_{i,j}$.\ Variability in the posterior samples entails varying numbers of atoms $n^\star_{i,j}$ in each $P^{\star}_{i,j}$; to enable a global comparison under a single 0.95-envelope, we thin each $P^{\star}_{i,j}$ by removing uniformly at random the number of atoms needed to reach $\min_{i,j} n_{i,j}^\star$.\ Each thinned $P^{\star}_{i,j}$ is then used to compute $\widehat{K}_B(r\mid P^\star_{i,j})$ and $\widehat{K}_C(r\mid P^\star_{i,j})$, the estimates~\eqref{eq:k_hat}, at a sequence of
distances and half aperture angles $r$.\ These are shown in Figure~\ref{fig:wave-kplot-returnlvl}.\ 
A vast
majority of these function estimates stay entirely within the the 0.95-envelopes created from a uniform sample of size $\min_{i,j}n^\star_{i,j}$, thus providing evidence that the weak
convergence given in Theorem~\ref{thm:PPconvergence} may hold well for
the data at hand.\ Furthermore, the QQ plots in
Figures~\ref{fig:wave-qq} in Supplementary Material
\ref{sec:wave-additional-plots} show that there is good agreement with
the empirical and model-based estimates in their abilities to
extrapolate exceedances to extreme values.\
Figure~\ref{fig:wave-pp-angle} in Supplementary Material
\ref{sec:wave-additional-plots} shows PP plots for directions
generated from the posterior distribution of $r_{\W}$ have good
agreement with the empirical distribution of the directions.\ Via
plots of posterior mean and prediction intervals for $\chi_q(A)$,
$A\subseteq\left\{H,P,S\right\}$, Supplementary~\ref{sec:wave-additional-plots} also shows the strong ability
of our model to estimate joint tail probabilities.

\begin{figure}[t!]
    \centering
    \includegraphics[width=0.28\textwidth]{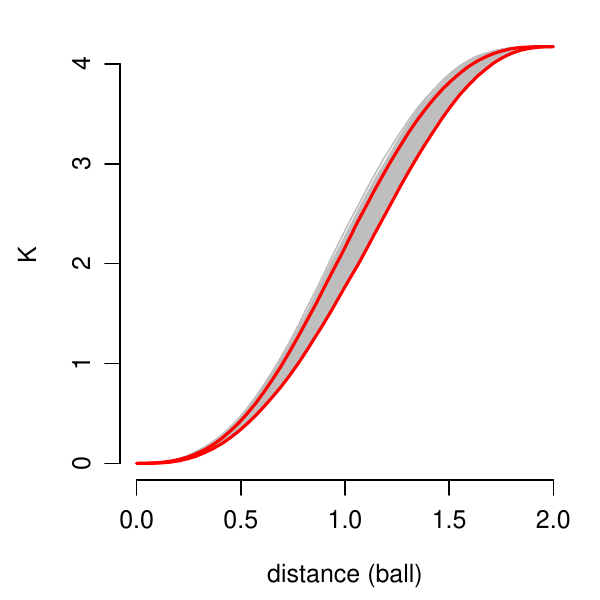}
    \includegraphics[width=0.28\textwidth]{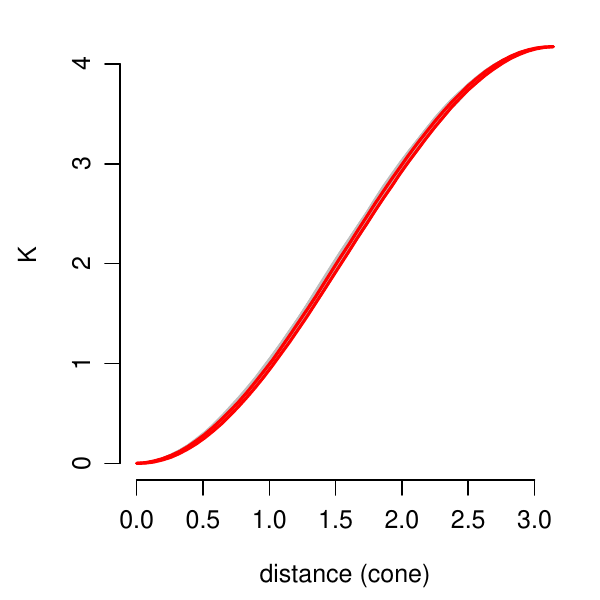}
    \includegraphics[width=0.32\textwidth,trim=0 25 0 0]{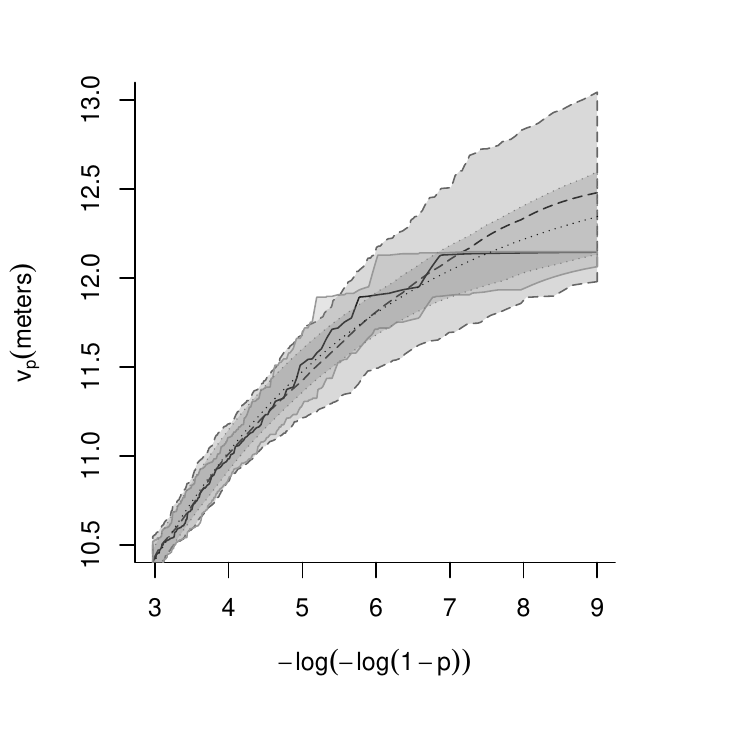}
    \caption{\textit{Left and centre}: Score values $K_B(r)$ and $K_C(r)$ from
      1000 posterior values of $r_{\QS^{\star}}(\bm{w})$ and
      $r_{\G^{\star}}(\bm{w})$ for increasing distances $r$, with
      simultaneous 95\% prediction intervals in red.\ \textit{Right}: Estimated
      return levels for sea-wall height.\ Presented are results from
      empirical fits (solid black line), GP (dotted black line), and
      our semi-parametric method (dashed black line).\ Grey regions
      correspond to 95\% confidence bands.\ }
      \label{fig:wave-kplot-returnlvl}
\end{figure}

\cite{coles1994statistical} introduce a structure variable
$Q(v;\bm{O})$, interpreted as the volume of water (in cubic meters,
$m^3)$ overtopping the sea-wall per unit length (in meters, $m$) over
a fixed duration (in seconds, $s$), and measured in
$m^3 s^{-1} m^{-1}$ for a sea-wall $v$ meters in height.\ 
More precisely,
\[
Q(v;\bm{O}) = a_1 O_S O_P \exp\{{a_2\left(v-O_S-l\right)}/{(O_P
  {O_H^\ast}^{{1}/{2}})}\}.
\]
The wave height component $O_H^\ast$ is a
calibration of the wave height marginal variable $O_H$ to approximate
the off-shore wave height, since measurements are taken
on-shore.\
We estimate the sea-wall height structure variable $v_p$ (in meters,
$m$) for which the discharge rate value $Q(v_p;\bm{O})$ is expected to
exceed the design standard of $0.002 m^3 s^{-1} m ^{-1}$ with
probability $p$.\ Setting $V=Q^{-1}(0.002; \bm{O})$, $v_p$ is the solution to
$\PR(V>v_p) = p$.\ 
As
in~\citet{Bortotetal01}, we fix the sea-wall design feature constants
to $a_1=0.25$, $a_2=26$, and tidal level relative to the seabed to
$l=4.3$.\ 

We obtain three separate estimates for $v_p$ at a range of
$p$ values.\ First, an empirical estimate of $v_p$ is obtained by
using the empirical quantile function on values of the structure
variable $V$ computed from the dataset.\ Second, a GP model is fitted
using the dataset's computed values of $V$ obtained from the dataset,
and quantiles are then obtained for $v_p$.\ Finally, we compare these with prediction intervals for $v_p$ obtained from our fitted model~$\Mg$ via the procedure described below.\

The method consists in sampling a collection $\{\widetilde{\bm O}_k = \{\bm O_{1,k},\ldots,\bm O_{n,k}\}\,:\,k=1,\ldots,l\}$ of new datasets, each comprising $n=2,894$ observations, 
to collect a sample of wall heights $\{v_{p,k}\,:\,k=1,\ldots,l\}$ where each $v_{p,k}$ is the $p$-quantile of $\{V_{j,k}=Q^{-1}(0.002; \bm{O}_{j,k})\,:\,j=1,\ldots,n\}$.\ For a sequence of $p\in(0.9,1)$, we consider the $p$-wise mean as well as the 95\% prediction intervals based on equal-tailed quantiles of the sample of sea wall heights.\ 

To generate a new data set $\widetilde{\bm O}_k$, we randomly sample
one of the 20 realisations from the posterior distribution of $\QS_q$,
label it $\QS_{q,k}$, as well as one of the 50 realisations from the
conditional distribution of $\G,\W\mid\QS_{q,k}$, and label it
($\G_k,\W_k$).\ The new sample $\widetilde{\bm X}_k$ is then
constructed by sampling with replacement $nq$ observations from the
original observations falling in $\QS_{q,k}$, and $n(1-q)$
observations from a radial exponential distribution with location
$\QS_{q,k}$, scale $\G_k$, and directional shape $\W_k$, before
transforming the sample to original margins using the probability
integral transform via the standard Laplace distribution function and
then inverse transform via the inverse marginal model specified in
Section~\ref{sec:standardisation}.\


Figure~\ref{fig:wave-kplot-returnlvl} includes the resulting estimates
of $v_p$ plotted against $-\log(-\log(1-p))$ for a range of $p$ values.\ Compared to the empirical and GP distribution fit
approach, our method accurately estimates the sea-wall height variable
$v_p$ across all values of $p\in(0,0.10)$.\ The larger prediction
intervals corresponding to our method can be attributed to our more holistic account of uncertainty via a joint model for extreme events, in contrast with a structured variable approach; our method hence reveals that risk may have been underestimated by such previous methods.\


\section{Discussion}
This work introduces a framework for defining probability sets and
return sets in multivariate extremes, emphasizing the role of these
sets in capturing the geometry of extreme events.\ Central to this
framework is the identification of the isotropic return set under a
radial-directional decomposition, which balances exceedances across
all directions and provides a natural representation of the geometry
of a sample cloud.\ Importantly, in our framework these return sets
arise through equivalence classes of normalising functions that lead
to the weak convergence of radially renormalised sample clouds to a
novel Poisson point process, further refining recent results on the
convergence of scaled sample clouds and allowing for extrapolation of
these sets.

Our construction of probability sets shares conceptual similarities to
other constructions, such as the total probability sets and
environmental contours \citep{Mackay2021contours}, which have been
used in structural reliability and risk assessment. These sets have
been motivated by the multivariate probability integral transform of
\cite{Rosenblatt52}, which may be burdensome to implement in
high-dimensional settings due to its reliance on sequential
conditioning, and recent efforts have focussed on methods that bypass
this transformation, see for example \cite{Mackay2023contours}. In a
similar spirit, our framework simplifies implementation by relying on
a one-dimensional probability integral transform under a
radial-directional decomposition, making it more computationally
tractable in multivariate extremes.\ This is already evidenced by the
work of \cite{campbell2024piecewiselinear}, which leverage piece-wise
linear models for the radial functions of the limit set, the quantile
set and the directional set, to infer probabilities of rare events in
dimensions $d=4$ and $d=5$, with statistical guarantees.\
Additionally, our approach shares conceptual similarities with the
framework of \cite{hallin2021distribution} for optimal transport-based
center-outward ranks, where a probability integral transform that does
not rely on sequential conditioning, pushes the multivariate
distribution forwards to the uniform distribution over unit
ball. However, while this method focusses on center-outward ordering,
our work leverages related ideas in the context of extremes to define
isotropic return sets that can reflect the geometry of extreme
multivariate events and facilitates their extrapolation beyond the
observed data.

In parallel to our work, \cite{simpson2024inference} developed a pair
of environmental contours under a radial-directional decomposition,
termed ``new environmental contours''.\ While their paper acknowledges
that our framework appeared on \textit{arXiv} during the review of
their manuscript, it is important to clarify that this approach
defines only two contours, one of which coincides with the
\textit{quantile set} introduced in our framework, which also appeared
in \cite{mackay2024modelling}. In contrast, here we provide a general
definition of an infinite class of return sets, which includes the
isotropic return set, exhibiting particularly desirable properties by
aligning closely with the structure of the underlying sample. By
integrating a new probabilistic theory with geometric insights, our
work establishes a foundation that advances existing state-of-the-art
theory, opening pathways for innovation in statistical inference and
risk assessment.


\subsection*{Acknowledgments}
We thank Jennifer L.\ Wadsworth for insightful discussions on the
framework of geometric extremes, and Finn Lindgren for guidance with
the \texttt{inlabru} and \texttt{excursions} packages in the
\textsf{R} computing language.\ Part of this paper was written during
a visit of IP and LDM at King Abdullah University of Science and
Technology and during a visit of RC at the University of
Edinburgh.\ 
LDM acknowledges funding by the School of Mathematics, University of
Edinburgh.\ RC acknowledges funding by the EPSRC DTP EP/W523811/1 fund
at Lancaster University, with additional funding provided by NSERC PGS
D (Canada), and the FRQNT doctoral fund (Qu\'{e}bec).\

\bibliographystyle{agsm}

\clearpage
\appendix
\begin{center}
\Large{\textbf{APPENDIX}}
\end{center}
\renewcommand{\theequation}{A.\arabic{equation}}
\setcounter{equation}{0}
\section{Proofs}
\subsection{Proof of Theorem~\ref{thm:univariate_vc}}
\label{sec:min-max-stable}
The fact that~\eqref{eq:RV_univariate_standard_RL} imply
\eqref{eq:RV_univariate_R_cond} and~\eqref{eq:RV_univariate_L_cond}
follows from Assumption $0\in S$, which implies that
$\pi=\PR(X>0)\in(0,1)$ and $1-\pi=\PR(X<0)\in(0,1)$, and convergence
of types \citep[Theorem 1.4.2 in][]{bill08}, with normalizing
sequences satisfying
\begin{IEEEeqnarray*}{rCl}
  b_{n,\sfR} \sim (b_n-m) + [(\pi^{-\xiR}-1)/\xiR]a_n& \quad
  \text{and} \quad &a_{n,\sfR} \sim a_n \pi^{-\xiR},\\
  b_{n,\sfL} \sim (\beta_n-m) + [((1-\pi)^{-\xiL}-1)/\xiL]\alpha_n&
  \quad \text{and} \quad &a_{n,\sfL} \sim \alpha_n
  (1-\pi)^{-\xiL}
\end{IEEEeqnarray*}
as $n\to \infty$.\ To show that the sequence of mean measures in
expression~\eqref{eq:RV_univariate_RL} converges vaguely in
$M_+((\overline{\Espace}_{\xiL} \times \{-1\}) \cup
(\overline{\Espace}_{\xiR} \times \{+1\})\color{black})$, we need to
show that for any
$h\in C_K^+((\overline{\Espace}_{\xiL} \times \{-1\}) \cup
(\overline{\Espace}_{\xiR} \times \{+1\})\color{black})$
\begin{IEEEeqnarray}{rCl}
  &\lim_{n\to\infty}&\int_{(\overline{\Espace}_{\xiL} \times
    \{-1\}) \cup (\overline{\Espace}_{\xiR} \times \{+1\})} h(z,
  w) n \PR\left[\left(\frac{|X| - r_{n}^a(X/|X|)}{ r_n^b(X/|X|)},
      X/|X|\right) \in dz \times dw \right]\label{eq:vg_RL_sequence}\\\nonumber\\
  &=& \int_{(\overline{\Espace}_{\xiL} \times \{-1\}) \cup
    (\overline{\Espace}_{\xiR} \times \{+1\})} h(z,
  w)\Lambda(dz\times dw),\nonumber
\end{IEEEeqnarray}
with $\Lambda$ defined in expression
\eqref{eq:Lambda_univariate}.\ Note that because
$(\overline{\Espace}_{\xiL} \times \{-1\}) \cup
(\overline{\Espace}_{\xiR} \times \{+1\})$ is a disjoint union,
then $h(z, w) = h_{\sfL}(z) \mathbbm{1}_{\{w=-1\}} + h_{\sfR}(z) \mathbbm{1}_{\{w=1\}}$,
with $h_{\sfL} \in C_K^+(\overline{\Espace}_{\xiL})$ and
$h_{\sfR} \in C_K^+(\overline{\Espace}_{\xiR} )$.\ Using total
probability, the integral in expression~\eqref{eq:vg_RL_sequence}
equals
\begin{IEEEeqnarray*}{rCl}
  & (1-\pi)&\int_{\overline{\Espace}_{\xiL}} h_{\sfL}(z) n
  \PR\left[\left(\frac{|X| - r_{n}^a(X/|X|)}{ r_n^b(X/|X|)} \in
      dz ~\Big|~ X/|X|=-1\right)\right] +\\\\
  &&+ \pi\int_{\overline{\Espace}_{\xiR}} h_{\sfR}(z) n
  \PR\left[\left(\frac{|X| - r_{n}^a(X/|X|)}{ r_n^b(X/|X|)} \in dz
      ~\Big|~ X/|X|=+1\right)\right].
\end{IEEEeqnarray*}
Convergences~\eqref{eq:RV_univariate_R_cond} and
\eqref{eq:RV_univariate_L_cond} guarantee that as $n\to \infty$, the
latter expression converges to
\begin{IEEEeqnarray*}{rCl}
  && (1-\pi)\int_{\overline{\Espace}_{\xiL}} h_{\sfL}\, d \nuL
  + \pi\int_{\overline{\Espace}_{\xiR}} h_{\sfR}\, d\nuR
  =\int_{(\overline{\Espace}_{\xiL} \times \{-1\}) \cup
    (\overline{\Espace}_{\xiR} \times \{+1\})} h \,d \Lambda,
\end{IEEEeqnarray*}
completing the proof.
\subsection{Convergence in distribution of renormalised sample minimum
  and maximum}
\label{sec:univariate_supplementary}
By replacing $1$ with $\pi + (1-\pi)$ in $1 - \PR_n:=\overline{\PR}_n$
and by applying the law of total probability to $\PR_n$, we average
the conditional probabilities given $\{X\geq 0\}$ and $\{X < 0\}$ over
$\pi$ and $1-\pi$ and note that, due to conditions~\eqref{eq:RV_univariate_R_cond} and~\eqref{eq:RV_univariate_L_cond},
which guarantee that both $a_{n, \{+1\}} z_{+} + b_{n, \{+1\}}$ and
$a_{n, \{-1\}} (-z_{-}) + b_{n, \{-1\}} $ are ultimately positive, we
find that
\begin{IEEEeqnarray}{rCl}
  n \overline{\PR}_n&=&(1-\pi) n \PR\left(\frac{|X|-b_{n, \sfL}}{a_{n,
        \sfL}} > - z_{-}~\Big| X < 0\right) + \pi n
  \PR\left(\frac{|X|-b_{n, \sfR}}{a_{n, \sfR}} > z_{+}~\Big| X >
    0\right),
  \label{eq:RCsets_at_inf}
\end{IEEEeqnarray}
for all sufficiently large
$n$.\ Thus, the convergences~\eqref{eq:RV_univariate_R_cond} and
\eqref{eq:RV_univariate_L_cond} lead to
\begin{equation}
  \lim_{n\to \infty}\PR\left( \frac{m_n + b_{n, \sfL}}{a_{n, \sfL}} >
    z_{-}, \frac{M_n - b_{n, \sfR}}{a_{n, \sfR}} < z_{+} \right) =
  G_{0}(-\zL, \zR \mid \pi),\quad \text{as $n \to \infty$},
\end{equation}
at continuity points of the limit distribution
$G_0$ given by expression~\eqref{eq:ETT_limit_form}.

\color{black}
\if0\blind{
\subsection{Proof of Corollary~\ref{prop:min_max_stable}}
Let $\pi=\PR(X/|X| = 1) = \PR(X >
0)$.\ It can be verified that for any $z_{-}, z_{+}\in\RR$
\begin{IEEEeqnarray}{rCl}
  &&\PR\left( \frac{-m_n + b_{n, \textsf{L}}}{a_{n, \textsf{L}}} \leq
    z_{-}, \frac{M_n - b_{n, \sfR}}{a_{n, \sfR}} \leq
    z_{+} \right) = \PR\left(-a_{n, \textsf{L}} z_{-} + b_{n,
      \textsf{L}} < X \leq
    a_{n, \sfR} z_{+} + b_{n, \sfR}\right)^n = \nonumber\\
  &&\quad = \left[1 - \frac{1}{n}n\left\{1-\PR\left(-a_{n, \textsf{L}}
        z_{-} + b_{n, \textsf{L}} < X \leq a_{n, \sfR}
        z_{+} + b_{n, \sfR}\right)\right\}\right]^n.
  \label{eq:min_max_prelimit}
\end{IEEEeqnarray}
Thus, it suffices to show that
\begin{IEEEeqnarray*}{rCl}
  \lim_{n\to \infty} n\left[1-\PR\left(-a_{n, \textsf{L}} z_{-} +
      b_{n, \textsf{L}} < X \leq a_{n, \sfR} z_{+} + b_{n,
        \sfR}\right)\right] = (1-\pi)(1-\xiL \zL)_+^{-1/\xiL} +
  \pi (1+\xiR \zR)_+^{-1/\xiR} \in [0,\infty],
\end{IEEEeqnarray*}
since then, expression~\eqref{eq:min_max_prelimit} can be written as
\[
  [1- \{(1-\pi)(1-\xiL \zL)_+^{-1/\xiL} + \pi (1+\xiR
  \zR)_+^{-1/\xiR}\}/n + o(1/n)]^n,
\]
so that convergence at continuity points of the limit distribution
follows at once.\ We have
\begin{IEEEeqnarray}{rCl}
  &&n\left[1-\PR\left(-a_{n, \textsf{L}} z_{-} + b_{n, \textsf{L}} < X \leq
      a_{n, \sfR} z_{+} + b_{n, \sfR}\right)\right] = \nonumber\\
  && n (1-\pi) -n(1-\pi) \PR\left(-a_{n, \textsf{L}} z_{-} + b_{n, \textsf{L}}
    < X \leq a_{n, \sfR} z_{+} + b_{n, \sfR}~\Big| X \leq 0\right)+\nonumber\\
  &&+ n\pi -n\pi \PR\left(-a_{n, \textsf{L}} z_{-} + b_{n, \textsf{L}} < X \leq
    a_{n, \sfR} z_{+} + b_{n, \sfR}~\Big| X > 0\right),
    \label{eq:prelim_eqn_Prop1}
  \end{IEEEeqnarray}
  where the equality follows from total probability and by expressing
  the constant $1$ as $\pi + (1-\pi)$.\ Due to the condition $0\in S$
  from Assumption~\ref{thm:univariate_vc}, the sequences
  $-a_{n,\textsf{L}} \zL + b_{n,\textsf{L}}$ and $a_{n,\sfR} \zR + b_{n,\sfR}$ are
  necessarily tending to the upper and lower end points of the
  distribution of $X$, respectively.\ Thus, these two sequences are
  eventually negative and positive, respectively.\ Therefore,
  expression~\eqref{eq:prelim_eqn_Prop1} reduces to
\begin{IEEEeqnarray*}{rCl}
  &&n (1-\pi) -n(1-\pi) \PR\left(\frac{X-b_{n, \textsf{L}}}{a_{n,
        \textsf{L}}} > -z_{-}~\Big| X \leq 0\right)+ n\pi -n\pi
  \PR\left(\frac{X-b_{n, \sfR}}{a_{n, \sfR}} \leq
    z_{+}~\Big| X > 0\right)\\
  &&= (1-\pi) n \PR\left(\frac{X-b_{n, \textsf{L}}}{a_{n, \textsf{L}}} \leq
    - z_{-}~\Big| X \leq 0\right) + \pi n \PR\left(\frac{X-b_{n,
        \sfR}}{a_{n, \sfR}} > z_{+}~\Big| X > 0\right),
\end{IEEEeqnarray*}
The result follows from Assumption~\ref{thm:univariate_vc}.}  \fi

\color{black}

\subsection{Convergence to Poisson point process}
\label{sec:PP_justification}
\color{black}
\begin{proof}[Proof of Theorem~\ref{thm:PPconvergence}]
  To prove the weak convergence of the random point measure to the
  stated limit Poisson random measure, we need to show that for any
  $h \in C_K^+(\color{black}\cup_{\bm w \in
    \SSS^{d-1}}\overline{\Espace}_{\xi(\bm w)}\times \{\bm
  w\}\color{black})$
  \begin{IEEEeqnarray*}{rCl}
    && \nu_n(h) := \int_{\mathbb{S}^{d-1}}\int_{\color{black}
      \overline{\Espace}_{\xi(\bm w)}\color{black}} h(z, \bm w)\,
    \nu_{n, \bm w}(d z) \PR_{\bm X/\lVert\bm X\rVert}(d \bm w) \to
    \int_{\SSS^{d-1}}
    \int_{\color{black}\overline{\Espace}_{\xi(\bm
        w)}\color{black}} h(z, \bm w)\, \nu_{\bm w} (dz) \PR_{\bm
      X/\lVert\bm X\rVert}(d \bm w)=: \nu(h),
  \end{IEEEeqnarray*}
  as $n\to \infty$.\ It suffices to show
  that $\lvert \nu_n(h) - \nu(h)\rvert \to 0$ as $n\to \infty$.\ Since
  $h$ is continuous on a compact set, say $K_h\subset \Emult$, we can
  find a compact cover
  $\mathbb{S}^{d-1}\times [r_{\inf}, r_{\sup}]\supset K_h$, with
  $r_{\inf} = \inf\{\lVert \bm x\rVert\,:\, \bm x \in K_h\}$ and
  $r_{\sup} = \sup\{\lVert \bm x\rVert\,:\, \bm x \in K_h\}$, such
  that
  \begin{IEEEeqnarray*}{rCl}
    \lvert \nu_n(h) - \nu(h) \rvert &\leq& \int_{\mathbb{S}^{d-1}}
    \lVert h\rVert_\infty(\bm w) \lvert \mu_{n, \bm w} \rvert
    ([r_{\inf}, r_{\sup}]) \PR_{\bm X/\lVert \bm X\rVert}(\bm w),\\
    & \leq& \int_{\mathbb{S}^{d-1}} \lVert h\rVert_\infty(\bm w)
    \Delta(\bm w) \PR_{\bm X/\lVert\bm X\rVert}(d \bm w)\quad \text{for
      all $n > n_0$},
  \end{IEEEeqnarray*}
  where
  $\lVert h\rVert_{\infty}(\bm w)=\sup\{h(z, \bm w)\,:\, z\in
  \overline{\Espace}_{\xi(\bm w)}\} < \infty$.\ Since
  \[
    \lim_{n\to\infty}(\nu_{n,\bm w}(h)-\nu_{\bm w}(h)) = 0,
  \]
  dominated convergence gives that
  $\lvert \nu_n(h) - \nu(h) \rvert \to 0$ as $n \to \infty$.
\end{proof}
\color{black}

\subsection{Properties of probability sets}
\label{sec:radon_nikodym}
\begin{proof}[Proof of Proposition~\ref{prop:properties_probability_set}]
  \color{black} $(i)$ Using total probability and the change of the
  integrating measure $\PR_{\bm W}$ to $\mu$, we have
  \begin{IEEEeqnarray*}{rCl}
    \PR(R > r_{\QS_q(\PR_{\bm W}\,\|\, \mu)}(\bm W))&=& \int_{\suppW}
    \PR(R > r_{\QS_q(\PR_{\bm W}\,\|\, \mu)}(\bm w) \mid \bm W = \bm
    w) \,\frac{d\PR_{\bm W}}{d \mu}(\bm w)\, \mu(d \bm w)
  \end{IEEEeqnarray*}
  Multiplication and division by $1/(1-q)$ gives
  \begin{IEEEeqnarray*}{rCl}
    \PR(R > r_{\QS_q(\PR_{\bm W}\,\|\, \mu)}(\bm
    W))&=&(1-q)\int_{\suppWclean} \frac{\PR(R > r_{\QS_q(\PR_{\bm
          W}\,\|\, \mu)}(\bm w) \mid \bm W = \bm w)}{1-q}
    \,\frac{d\PR_{\bm W}}{d \mu}(\bm w)\, \mu(d \bm w).
  \end{IEEEeqnarray*}
  From the definition \eqref{eq:r_probability_set} of
  $r_{\QS_q(\PR_{\bm W}\,\|\, \mu)}(\bm w)$, we have that for any
  $q > q_l(\PR_{\bm W}\,\|\,\mu^\star)$,
  \[
    \frac{\PR(R > r_{\QS_q(\PR_{\bm W}\,\|\, \mu)}(\bm w) \mid \bm W =
      \bm w)}{1-q} = \left(\frac{d\PR_{\bm W}}{d \mu}(\bm
      w)\right)^{-1}.
  \]
  Thus, since $\mu$ is a probability measure on $\suppWclean$,
  \[
    \PR[R > r_{\QS_q(\PR_{\bm W}\,\|\, \mu)}(\bm W)] =
    (1-q)\int_{\suppWclean} \mu(d\bm w) = (1-q).
  \]

  Next, we consider the conditional distribution of $\bm W$ given
  $R >r_{\QS_q(\PR_{\bm W}\,\|\, \mu)}(\bm W)$.\ For any measurable set
  $B\subset \suppWclean$, we have
  \begin{IEEEeqnarray*}{rCl}
    \PR\left(\bm W\in B \mid R >r_{\QS_q(\PR_{\bm W}\,\|\, \mu)}(\bm
      W)\right)&=& \PR\left(R >r_{\QS_q(\PR_{\bm W}\,\|\, \mu)}(\bm
      W), \bm W\in B\right)/\PR(R >r_{\QS_q(\PR_{\bm W}\,\|\, \mu)}(\bm
    W))
  \end{IEEEeqnarray*}
  Because $\PR[R > r_{\QS_q(\PR_{\bm W}\,\|\, \mu)}(\bm W)]= (1-q)$,
  it follows that for any $q>q_l(\PR_{\bm W}\,\|\,\mu^\star)$,
  \begin{IEEEeqnarray*}{rCl}
    \PR\left(\bm W\in B \mid R >r_{\QS_q(\PR_{\bm W}\,\|\, \mu)}(\bm
      W)\right)&=& (1-q)^{-1}\int_B \PR\left(R >r_{\QS_q(\PR_{\bm
          W}\,\|\,
        \mu)}(\bm w) \mid \bm W=\bm w\right) \PR_{\bm W}(d\bm w)\\\\
    &=& (1-q)^{-1}\int_B (1-q) \frac{d\mu}{\PR_{\bm W}}(\bm
    w) \PR_{\bm W}(d\bm w)\\ \\
    &=&\int_{B} \mu(\bm w) = \mu(B),
  \end{IEEEeqnarray*}
  which completes the proof of $(i)$.

  \noindent $(ii)$ 
  For each $z$ in some open interval $(a,b)$ that is a continuity
  point of the limit measure $\nu_{\bm w}$, we have by Assumption that
  \ref{ass:main_convergence_assumptions}
  \begin{equation}
    \lim_{n\to\infty}n \PR\left[\frac{R-b_n(\bm W)}{a_n(\bm
        W)}>z~\Big|~\bm W=\bm w\right] \varrho_\mu(\bm w) = [1+\xi(\bm w)z]_+^{-1/\xi(\bm
      w)}\varrho_\mu(\bm w), 
    \label{eq:convergence_survival}
  \end{equation}
  for $\bm w\in\suppWclean$, where
  $\varrho_\mu = d\PR_{\bm W}/d\mu>0$.\ Convergence
  \eqref{eq:convergence_survival} is equivalent to
  $\lim f_n(z\,;\,\bm w) = g(z\,;\, \bm w)$ for $\bm w\in\suppWclean$,
  where
  \begin{IEEEeqnarray}{rCl}
    f_n(z\,;\,\bm w)&:=&1\Big/\left(n\PR\left[\frac{R-b_n(\bm W)}{a_n(\bm
          W)}>z~\Big|~\bm W=\bm w\right] \varrho_\mu(\bm w)\right) \\ \nonumber\\
    g(z\,;\,\bm w)&=&[1+\xi(\bm w)z]_+^{1/\xi(\bm w)}/\varrho_\mu(\bm w),
    \label{eq:convergence_survival_reciprocal}
  \end{IEEEeqnarray}
  for $\bm w\in\suppWclean$.\ Since a sequence of non-decreasing
  functions in $z$ converges to a limit function that is
  non-decreasing in $z$, Lemma 1.1.1 in \cite{haanferr06} applies and
  for each $y\in (g(a\,;\,\bm w), g(a\,;\,\bm w))$ that is a
  continuity point of $g^{-1}$, we have
  $\lim_{n \to \infty}f_n^{-1}(y \,;\, \bm w) = g^{-1}(y\,;\,\bm w)$,
  where $f_n^{-1}$ and $g^{-1}$ denote the left-continuous inverses of
  $f_n$ and $g$, respectively.\ That is,
  \[
    \lim_{n \to \infty}\frac{V_{\bm w}(n y \,;\, \mu) - b_n(\bm w)}{a_n(\bm
      w)} = \frac{\{y \varrho_\mu(\bm w)\}^{\xi(\bm w)}-1}{\xi(\bm
      w)},\quad \bm w\in \suppWclean,
  \]  
  where
  \[
    V_{\bm w}(y \,;\, \mu) = \inf\left\{z\in\RR \,:\, \left(n \PR(R>z
        \mid \bm W=\bm w)\varrho_\mu(\bm w)\right)^{-1}\geq y\right\},
  \]
  is defined for $y > 1$.\ Note that the function
  $V_{\bm w}(n \,;\, \mu)$ satisfies
  $V_{\bm w}(1/(1-q) \,;\, \mu) = r_{\QS_{q}(\PR_{\bm W}\,\|\,
    \mu)}(\bm w)$.

  Since $\inf\{\xi(\bm w)\,:\, \bm w\in\suppW\} > -1$, it suffices to
  show that
  \begin{equation}
    \lim_{n \to \infty}\frac{V_{\bm w}(n y \,;\, \mu) - V_{\bm
        w}(n\,;\,\mu)}{V_{\bm w}(n \{1+\xi(\bm w)\}^{1/\xi(\bm w)}
      \,;\, \mu) - V_{\bm w}(n \,;\, \mu)} = \frac{y^{\xi(\bm
        w)}-1}{\xi(\bm w)},\quad \bm w\in \suppWclean.
    \label{eq:to_prove}
  \end{equation}
  The left-hand side of expression \eqref{eq:to_prove} equal
  \begin{IEEEeqnarray*}{rCl}
    &&\lim_{n \to \infty}\cfrac{\cfrac{V_{\bm w}(n y \,;\,
        \mu)-b_n(\bm w)}{a_n(\bm w)} - \cfrac{V_{\bm w}(n \,;\,
        \mu)-b_n(\bm w)}{a_n(\bm w)}}{\cfrac{V_{\bm w}(n (1+\xi(\bm
        w))^{1/\xi(\bm w)} \,;\, \mu)-b_n(\bm w)}{a_n(\bm w)} -
      \cfrac{V_{\bm w}(n\,;\, \mu)-b_n(\bm w)}{a_n(\bm
        w)}}\\ \\
    &&=\cfrac{\cfrac{\{y \varrho_\mu(\bm w)\}^{\xi(\bm w)}-1}{\xi(\bm
        w)} - \cfrac{\varrho_\mu(\bm w)^{\xi(\bm w)}-1}{\xi(\bm
        w)}}{\cfrac{\{(1+\xi(\bm w))^{1/\xi(\bm w)} \varrho_\mu(\bm
        w)\}^{\xi(\bm w)}-1}{\xi(\bm w)}- \cfrac{\varrho_\mu(\bm
        w)^{\xi(\bm w)}-1}{\xi(\bm w)}}\\\\
    &&=\cfrac{\rho_\mu(\bm w)^{\xi(\bm w)}(y^{\xi(\bm w)}-1)/\xi(\bm
      w)}{\varrho_\mu(\bm w)^{\xi(\bm w)}} = \cfrac{y^{\xi(\bm
        w)}-1}{\xi(\bm w)}, \qquad \bm w \in \suppWclean.
  \end{IEEEeqnarray*}
\end{proof}

\subsection{Density convergence of rescaled radial excess
  variable}\label{sec:exp-excesses}
\label{sec:proof_prop1}
Below, we use Lemma~\ref{lemma:Potter_bounds}, a direct
consequence of the Corollary 3.3 of \cite{de1987regular}.
\begin{lemma}[Potter bounds]
  \label{lemma:Potter_bounds}
  Let $V:\RR^d\to \RR$ and suppose there exist $\rho
  \in\RR$, $h:\RR_+\to\RR_+$, and
  $\G\in\bigstar$ such that $\lim_{t\to\infty}V(t\bm x)/h(t)=g_\G(\bm
  x)^{\rho}$ uniformly on
  $\SSS^{d-1}$.\ Then, for any
  $\varepsilon>0$, there exists
  $t_0=t_0(\varepsilon)>0$ such that for all $t>t_0$,
  \begin{IEEEeqnarray*}{rCl}
  \label{eq:Potter_bounds}
   (1-\varepsilon)\lVert \bm x \rVert^{-\varepsilon}g_\G(\bm x)^{\rho} &\leq& \frac{V(t\bm x)}{h(t)}\leq (1+\varepsilon)\lVert \bm x \rVert^{\varepsilon}g_\G(\bm x)^{\rho},\qquad \forall\,\bm x\in\Rstar.
  \end{IEEEeqnarray*}
\end{lemma}
\color{black}

\noindent \begin{proof}[Proof of Proposition~\ref{prop:RV1}] Consider
  the following relation
  \begin{equation}
    \frac{\partial }{\partial z}\, \PR\left[\frac{R-\rquant(\bm w)}{r_{\G_q}(\bm
        w)}\leq z ~\Big|~R> \rquant(\bm W), \bm W=\bm w\right] = \frac{ r_{\G_q}(\bm w)z_q(\bm w)^{d-1}f_{\bm X}(z_q(\bm w)\bm w)}{\int_{\rquant(\bm w)}^{r_{\QS_1}(\bm w)} f_{R,\bm W}(s, \bm
      w)~\dd s},
    \label{eq:dens_ratio}
  \end{equation}
  where $z_q(\bm w)=r_{\QS_q}(\bm w) + r_{\G_q}(\bm w) z$, $z\in[0,\{r_{\QS_1}(\bm w)-r_{\QS_q}(\bm w)\}/r_{\G_q}(\bm w)]$, and $\bm w\in\SSS^{d-1}$.\
  
\item[$(i)$] Write the integral in expression~\eqref{eq:dens_ratio} using the change of variable $s=\{1-(t \lVert \bm w\rVert)^{-1}\}\,r_{\QS_1}(\bm w)$, with $L_q(\bm w):= r_{\QS_1}(\bm w)/\{r_{\QS_1}(\bm w)-\rquant(\bm w)\}$, as
    \begin{IEEEeqnarray*}{rCl}
      \int_{\rquant(\bm w)}^{r_{\QS_1(\bm w)}} f_{R, \bm W}(s, \bm
      w)~\dd s
    &=& r_{\QS_1}(\bm w)^{d}\int_{L_q(\bm w)}^{\infty}\left(1-1/t\right)^{d-1}t^{-2}\{f_{\bm X}[(1-1/t)r_{\QS_1}(\bm w)\bm w]/\psi_{\textsf{B}}(t)\}\psi_{\textsf{B}}(t)\,\dd t.
    \end{IEEEeqnarray*}
    Due to convergence $(i)$, the map
    $(0,\infty) \ni s \mapsto f_{\bm X}[(1-1/s)r_{\QS_1}(\bm w)\bm w]/\psi_{\textsf{B}}(s) \in \RR^+$ is slowly varying at~$\infty$.\ Additionally, convergence $(i)$ implies that
    $\psi_{\mathsf{B}}(s)\in \text{RV}_{1+1/\xi}^\infty$, hence
    there exists $l_{\textsf{B}}\in \text{RV}_0^\infty$ such that
    \smash{$\psi_{\mathsf{B}}(s)=l_{\mathsf{B}}(s) s^{1+1/\xi}$}.\
    Thus, by Karamata's theorem~\citep[Theorem
    2.1][]{resnick2007heavy},
    \[
      \frac{\int_{\rquant(\bm w)}^{r_{\QS_1(\bm w)}} f_{R, \bm W}(s, \bm
      w)~\dd s} {\{r_{\QS_1}(\bm w)-\rquant(\bm w)\}f_{R, \bm W}[\rquant(\bm w), \bm
      w]}= -\xi\{1+o(1)\},\qquad \text{as $q\to 1$}.
    \]
    Using the assumption that $f_{\bm X}(r\bm w)\sim \psi_{\mathsf{B}}[r_{\QS_1}(\bm w)/\{r_{\QS_1}(\bm w)-r\}]g_{\G}(\bm w)^{1+1/\xi}$ as $r\to r_{\QS_1}(\bm w)$, and setting $r_{\G_q}(\bm w)=-\xi\{r_{\QS_1}(\bm w)-\rquant(\bm w)\}$, expression~\eqref{eq:dens_ratio} converges to
    \begin{IEEEeqnarray*}{rCl}
      & & \frac{r_{\G_q}(\bm
      w)z_q(\bm w)^{d-1}f_{\bm X}[z_q(\bm w)\bm w]}{-\xi\{r_{\QS_1}(\bm w)-\rquant(\bm w)\}\rquant(\bm w)^{d-1}f_{\bm X}[\rquant(\bm w)\bm
      w]} \{1+o(1)\} = \\\\
      &&=\frac{z_q(\bm w)^{d-1}}{\rquant(\bm w)^{d-1}}\left[\frac{r_{\QS_1}(\bm w)-z_q(\bm w)}{r_{\QS_1}(\bm w)-\rquant(\bm w)}\right]^{-1-1/\xi}\{1+o(1)\}\to(1 + \xi
    z)_+^{-1-1/\xi},\qquad \text{as $q\to 1$}.
  \end{IEEEeqnarray*}
\item[$(ii)$] Consider first the integral term in expression~\eqref{eq:dens_ratio}.\ The integral is equal to
    \begin{IEEEeqnarray*}{rCl}
    &\int_{\rquant(\bm w)}^\infty& s^{d-1}f_{\bm X}(s\bm w)~\dd s=
        \int_{\rquant(\bm w)}^\infty s^{d-1}
        \exp[-\psi_{\textsf{L}}(s) \{- \{\log f_{\bm X}(s \bm
        w)\}/\psi_{\textsf{L}}(s)\}]\dd s.
    \end{IEEEeqnarray*}
    Due to Lemma~\ref{lemma:Potter_bounds} with
    $V(\bm x) = -\log f_{\bm X}(\bm x)$ and $h=\psi$, for all
    $\varepsilon>0$ there exists $q_0>q$ such that
    \begin{IEEEeqnarray*}{rCl}
    \int_{\rquant(\bm w)}^\infty s^{d-1}
        \exp[-(1+\varepsilon)\psi_{\textsf{L}}(s)g_\G(\bm w)^\rho]\dd s &\leq &\int_{\rquant(\bm w)}^\infty s^{d-1}f_{\bm X}(s\bm w)~\dd s\leq
        \int_{\rquant(\bm w)}^\infty\exp[-(1-\varepsilon)\psi_{\textsf{L}}(s)g_\G(\bm w)^\rho]\dd s.
    \end{IEEEeqnarray*}
    Since $\varepsilon$ can be made arbitrarily small, we have    
    \begin{IEEEeqnarray*}{rCl}
        \int_{\rquant(\bm w)}^\infty s^{d-1}f_{\bm X}(s\bm w)~\dd s
        &\sim&
        \int_{\rquant(\bm w)}^\infty s^{d-1}\exp[-\psi_{\textsf{L}}(s)g_\G(\bm w)^\rho]\dd s
    \end{IEEEeqnarray*}
    Without loss, $\psi_{\textsf{L}}$ can be taken smooth~(Lemma 1.4,~\cite{sene06}) so that $d^k \psi_{\textsf{L}}(t) /d t^k$ exists for all $t>0$ and $k\in \mathbb{N}$.\ Therefore,
    \begin{IEEEeqnarray*}{rCl}
        \int_{\rquant(\bm w)}^\infty s^{d-1}f_{\bm X}(s\bm w)~\dd s
        &\sim&
        \int_{\rquant(\bm w)}^\infty s^{d-1}\left[\frac{\frac{d}{ds}\exp[-\psi_{\textsf{L}}(s)g_\G(\bm w)^\rho]}{-g_\G(\bm w)^\rho\psi^\prime_{\textsf{L}}(s)}\right]\dd s \\\\
        &=& r_{\QS_q}(\bm w)^{d-1}\left[\frac{\exp[-\psi_{\textsf{L}}(r_{\QS_q}(\bm w))g_\G(\bm w)^\rho]}{g_\G(\bm w)^\rho\psi^\prime_{\textsf{L}}(r_{\QS_q}(\bm w))}\right]\{1+o(1)\}, \qquad \text{as $q\to1$}.
    \end{IEEEeqnarray*}
    Using Lemma~\ref{lemma:Potter_bounds} again, we have that equation~\eqref{eq:dens_ratio} converges to
    \begin{IEEEeqnarray*}{rCl}
        &&\frac{r_{\G_q}(\bm w)z_q(\bm w)^{d-1}}{r_{\QS_q}(\bm w)^{d-1}}g_\G(\bm w)^{\rho}\psi^\prime_{\textsf{L}}(r_{\QS_q}(\bm w))\exp[-g_\G(\bm w)^\rho\{\psi_{\textsf{L}}(z_q(\bm w))-\psi_{\textsf{L}}(r_{\QS_q}(\bm w))\}]\{1+o(1)\}\\
        &&= r_{\G_q}(\bm w)g_\G(\bm w)^{\rho}\psi^\prime_{\textsf{L}}(r_{\QS_q}(\bm w))\exp[-g_\G(\bm w)^\rho\{\psi_{\textsf{L}}(z_q(\bm w))-\psi_{\textsf{L}}(r_{\QS_q}(\bm w))\}]\{1+o(1)\}, \qquad \text{as $q\to1$}.
    \end{IEEEeqnarray*}
    A Taylor expansion of $\psi_{\textsf{L}}(z_q(\bm w))$ about
    $\rquant$ and local uniform convergence give
    $\psi_{\textsf{L}}(z_q(\bm w)) = \psi_{\textsf{L}}(r_{\QS_q}(\bm
    w))+\psi^\prime_{\textsf{L}}(r_{\QS_q}(\bm w))r_{\G_q}(\bm
    w)+o(1)$.\ Thus,
    \begin{IEEEeqnarray*}{rCl}
        &&r_{\G_q}(\bm w)g_\G(\bm w)^{\rho}\psi^\prime_{\textsf{L}}(r_{\QS_q}(\bm w))\exp[-g_\G(\bm w)^\rho\{\psi^\prime_{\textsf{L}}(r_{\QS_q}(\bm w))r_{\G_q}(\bm w)+o(1)\}]\{1+o(1)\}\\
        && = r_{\G_q}(\bm w)g_\G(\bm w)^{\rho}\psi^\prime_{\textsf{L}}(r_{\QS_q}(\bm w))\exp[-g_\G(\bm w)^\rho\{\psi^\prime_{\textsf{L}}(r_{\QS_q}(\bm w))r_{\G_q}(\bm w)]\}]\{1+o(1)\}\to \exp(-z),\qquad \text{as $q\to 1$},
    \end{IEEEeqnarray*}
    whenever
    $r_{\G_q}(\bm w)\sim\{\psi^\prime_{\textsf{L}}(r_{\QS_q}(\bm
    w))g_\G(\bm w)^\rho\}^{-1}$.\
\item[$(iii)$] Consider first the integral term in expression~\eqref{eq:dens_ratio}.\ The integral is equal to
  \begin{IEEEeqnarray*}{rCl}
    \int_{\rquant(\bm w)}^\infty f_{R, \bm W}(s, \bm w)~\dd s&=&
    \int_{\rquant(\bm w)}^\infty s^{d-1} \{f_{\bm X}(s \bm
    w)/\psi_{\textsf{H}}(s)\}\{\psi_{\textsf{H}}(s)/s^{1-d}\} ~\dd s.
  \end{IEEEeqnarray*}
  Due to convergence $(iii)$, the map
  $(0,\infty) \ni s \mapsto \{f_{\bm X}(s \bm
  w)/\psi_{\textsf{H}}(s)\} \in \RR^+$ is slowly varying at $\infty$
  for any $\bm w\in \SSS^{d-1}$.\ Further, convergence $(iii)$ implies
  that $\psi_{\mathsf{H}}(s)\in \text{RV}_{-(d+\xi^{-1})}^\infty$,
  hence there exists $l_{\textsf{H}}\in \text{RV}_0^\infty$ such that
  \smash{$\psi_{\mathsf{H}}(s)=l_{\mathsf{H}}(s) s^{-
      \{d+(1/\xi)\}}$}.\ By Karamata's theorem~\citep[Theorem
  2.1][]{resnick2007heavy},
  \[
    \frac{ \int_{\rquant(\bm w)}^\infty f_{R, \bm W}(s, \bm w)~\dd s}
    {\rquant(\bm w) f_{R, \bm W}(\rquant(\bm w), \bm w)}= \xi\,
    \{1+o(1)\},\qquad \text{as $q\to 1$}.
  \]
  Thus, when $r_{\G_q}(\bm w) = \xi \, r_{\QS_q}(\bm w)$ and as $q\to 1$,
  expression~\eqref{eq:dens_ratio} converges to
  \begin{IEEEeqnarray*}{rCl}
    &\frac{1}{\xi}& \, \frac{r_{\G_q}(\bm w)}{\rquant(\bm
      w)}\left[\frac{z_q(\bm w)}{r_{\QS_q}(\bm w)}\right]^{d-1} \frac{f_{\bm X}
      [z_q(\bm w) \bm w]}{f_{\bm
        X}[\rquant(\bm w)\,\bm w]} \{1+o(1)\} = \left[1+\frac{r_{\G_q}(\bm w)
        z}{r_{\QS_q}(\bm w)}\right]^{-1-1/\xi} \{1+o(1)\}\to (1 + \xi
    z)_+^{-1-1/\xi}.
  \end{IEEEeqnarray*}
\end{proof}

\subsection{Stability of the radial Exponential distribution}
\label{proof:stab_exp}
\begin{proof}[Proof of stability equation~\eqref{eq:MRS_stability_exp}]
  Let the random vector $\bm Z\in\RR^d$ follow an $\rExp$ distribution
  with location $\loc$, scale $\Sigma$, and directional component
  $\W$.\ Then,
\begin{IEEEeqnarray*}{rCl}
  &&\PR\left[\bm Z\in [\loc \radd B_{r_1+r_2}(\Orig)\rmult \Sigma]^\prime\mid \bm Z \in \{\loc +B_{r_1}(\Orig)\rmult \Sigma\}^\prime\right] = \frac{\PR\left[\bm Z\in \{\loc +B_{r_1+r_2}(\Orig)\rmult\Sigma\}^\prime\right]}{\PR\left[\bm Z \in \{\loc +B_{r_1}(\Orig)\rmult\Sigma\}^\prime\right]}\\
  && \quad = \frac{\PR\left[\{\lVert\bm Z\rVert-r_{\loc}\left(\bm
        Z/\lVert \bm Z\rVert\right)\}/r_{\Sigma}\left(\bm Z/\lVert \bm
        Z\rVert\right)>r_1+r_2\right]}{\PR\left[\{\lVert\bm
      Z\rVert-r_{\loc}\left(\bm Z/\lVert \bm
        Z\rVert\right)\}/r_{\Sigma}\left(\bm Z/\lVert \bm
        Z\rVert\right)>r_1\right]} = \exp\{-r_2\} = \PR\left[\bm Z \in
    \{\loc +B_{r_2}(\Orig)\}^\prime\rmult\Sigma\right].
\end{IEEEeqnarray*}
\end{proof}

\subsection{Stability of the radial generalised Pareto distribution}
\label{proof:stab_gP}
\begin{proof}[Proof of stability equation~\eqref{eq:MRS_stability_GP}]
  Let the random vector $\bm Z\in\RR^d$ follow an $\rGP$ distribution
  with location $\loc$, scale $\Sigma$, shape $\xi\in\RR$, and
  directional component $\W$.\ Then,
\begin{IEEEeqnarray*}{rCl}
  &&\PR\left[\bm Z\in [\loc \radd B_{r_1+r_2}(\Orig)\rmult \Sigma]^\prime\mid \bm Z \in \{\loc +B_{r_1}(\Orig)\rmult \Sigma\}^\prime\right] = \frac{\PR\left[\bm Z\in \{\loc +B_{r_1+r_2}(\Orig)\rmult\Sigma\}^\prime\right]}{\PR\left[\bm Z \in \{\loc +B_{r_1}(\Orig)\rmult\Sigma\}^\prime\right]}\\
  && \quad =\frac{\PR\left[\{\lVert\bm Z\rVert-r_{\loc}\left(\bm
        Z/\lVert \bm Z\rVert\right)\}/r_{\Sigma}\left(\bm Z/\lVert \bm
        Z\rVert\right)>r_1+r_2\right]}{\PR\left[\{\lVert\bm
      Z\rVert-r_{\loc}\left(\bm Z/\lVert \bm
        Z\rVert\right)\}/r_{\Sigma}\left(\bm Z/\lVert \bm
        Z\rVert\right)>r_1\right]}
  = \frac{\left[1+\xi(r_1+r_2)\right]_+^{-1/\xi}}{\left[1+\xi r_1\right]_+^{-1/\xi}}= \left[1+\xi\frac{r_2}{1+\xi r_1}\right]^{-1/\xi}\\
  && \quad = \PR\left[\bm Z \in \left\{\loc +
      \frac{B_{r_2}(\Orig)\rmult\Sigma}{B_1(\Orig)\radd
        B_{\xi}(\Orig)\rmult B_{r_1}(\Orig)}\right\}^\prime\right].
\end{IEEEeqnarray*}
\end{proof}

\subsection{Marginal density of directional variable under homothetic density}
\label{sec:proof_W_density}
\begin{proof}[Proof of Proposition~\eqref{prop:angle_distribution}]
  Suppose that $f(\bm x) = f_0(\gG(\bm x))$, $\bm x\in \RR^d$.\ Then,
  \begin{IEEEeqnarray}{rCl}\label{eq:angle-density-derivation}
    f_{\bm W}(\bm w)&=&\int_0^\infty f_{R, \bm W}(r, \bm w)\, \dd r =
    \frac{1}{\gG(\bm w)^d}\int_0^\infty s^{d-1}f_0(s)\, \dd s =
    \frac{1}{d\,\vol{\G}\,\gG(\bm w)^d},
  \end{IEEEeqnarray}
  where the last equality follows from~\citet[][see Section
  3.1]{balknold10}.
\end{proof}

\newpage
\appendix
\renewcommand{\thesection}{\arabic{section}}
\title{\noindent \bf\Large Supplementary Material for ``Statistical inference for radially-stable generalized Pareto distributions and return level-sets in geometric extremes"}

\maketitle
\section{Technical background}
\label{supp:Technical_background}

\subsection{Vague convergence}
\label{sec:vg}
For a detailed exposition of the concept of vague convergence, see
\cite{resnick2007heavy}.\ Let $\mathcal{E}$ be a locally compact metric space
with countable base (e.g., $\RR^d$) and let $M_+(\mathcal{E})$ be the class of
nonnegative Radon measures on Borel subsets of $\mathcal{E}$.\ If
$\mu_n\in M_+(\mathcal{E})$ for $n\geq 0$, then we say that $\mu_n$ converges
vaguely to $\mu$ (presented as $\mu_n \vg \mu$) if 
\[
  \lim_{n\to \infty} \int_{\mathcal{E}} f d\mu_n = \int_{\mathcal{E}} f
  d\mu.
\]
for all bounded continuous functions $f$ with compact support.

\subsection{Multivariate regular variation}
\label{sec:background_MRV}
To understand radial stability properties in a $d$-dimensional setup,
we revisit classical multivariate extreme value theory through
consideration of multivariate regular variation (MRV).\ Given a
  random vector $\bm{X}\in\mathbb{R}^d$, a common assumption is that
  its margins follow the same univariate distribution.\ If this is not
  guaranteed, standardisation is required through the probability
  integral transform with the true distribution function and the
  desired target quantile function.\ When the true distribution
  function is not known, the empirical distribution function is used
  in the bulk of the data, while the Generalised Pareto distribution
  function with fitted parameters is used outside of high left and
  right thresholds to avoid issues the arise when using the empirical
  distribution function in the tail \citep{coletawn91}.
Under the choice of standard exponential margins,
a random vector $\bm X$ is multivariate regularly varying on the cone
$\mathcal{E}=[-\infty,\infty]^d \setminus\{-\infty\}^d$ if there
exists a limit measure $\nu_{\textsf{MRV}}$ such that as $n\to \infty$
\begin{equation}
  n\,\PR[\bm X - (\log n)\bm 1 \in \cdot] \vg \nu_{\textsf{MRV}}(\cdot), \quad \text{in $M_+(\mathcal{E})$},
  \label{eq:MRV}
\end{equation}
A consequence of convergence~\eqref{eq:MRV} is that
$\nu_{\textsf{MRV}}(t \bm 1 + B)=\exp(-t)\nu_{\textsf{MRV}}(B)$,
implying domain of attraction properties for the distribution of
$\bm X$.\ \color{black} For instance, after appropriate recentering,
the distribution function of the componentwise maxima
$\bm M_n := (\max_{i=1,\dots, n} X_{ij} \,:\, j=1,\dots, d)$ of a
random sample $\bm X_i=(X_{ij}\,:\,j=1,\dots, d)$, $i=1,\dots,n$, from
the distribution of $\bm X$, converges weakly to
\[
  \PR\left(\bm M_n - (\log n)\bm 1\leq \bm z\right)=\PR\left\{\bm X_i
    - (\log n)\bm 1 \in [-\bm \infty, \bm z]\,:\,i=1,\dots,n\right\}
  \wk \exp\left[-\Lambda_{\textsf{MRV}}(\bm z)\right],
\]
where the exponent function
$\Lambda_{\textsf{MRV}}(\bm z):=\nu_{\textsf{MRV}}([-\infty, \bm
z]')$.\ The measure $\Lambda_{\textsf{MRV}}$ is in one-to-one
correspondence with a Radon measure $H_{\textsf{MRV}}$ on the
unit-simplex
$\Delta_+^{d-1}=\{\bm \omega \in \RR^{d}_+ \,:\, \lVert \bm w
\rVert_1=1\}$, where
$\lVert \bm x \rVert_p = (\sum_{i=1}^d \lvert x_i\rvert^p)^{1/p}$ for
$\bm x = (x_1,\ldots,x_d)\in\RR^{d}$, termed the spectral measure,
satisfying the mass--moment constraint $\HMRV(\Delta_+^{d-1})=d$ and
$\int_{\Delta_+^{d-1}} w_i H(d \bm w) = 1$, for $i=1,\dots,d$.\ The
correspondence is given by
\[
  \Lambda_{\textsf{MRV}}(B) = \int_{\Delta_+^{d-1}} \max\left[\bm w \exp(-\bm
    x)\right] H_{\textsf{MRV}}(d\bm w),
\]
where vector algebra is interpreted component-wise.\ Similarly, the
vector of renormalised threshold exceedances converges in distribution
to a multivariate Pareto distribution.\ That is,
\[
  \lim_{r\to \infty}\PR\left[\bm X -r\bm 1\in [-\bm \infty,\bm x] \mid \max \bm
    X > r\right] = \frac{\Lambda_{\textsf{MRV}}(\bm x \wedge \bm
    0)-\Lambda_{\textsf{MRV}}(\bm x)}{\Lambda_{\textsf{MRV}}(\Orig)},
\]
at continuity points $\bm x \in\mathcal{E}$ of the limit distribution
function of the multivariate generalised Pareto distribution, where
$\Lambda_{\textsf{MRV}}(\bm x)=\nu_{\textsf{MRV}}([-\infty, \bm
x]')$.\ Multivariate generalised Pareto distributions arise as the
only non-trivial limit distributions when renormalised multivariate
exceedances are considered~\citep{roo+ta:06}.\ These distributions
possess appealing theoretical and practical properties
\cite{engehitz19} and are often thought to be the
multivariate analogue of the generalised Pareto distribution
\cite{pick75} by being the only threshold-stable multivariate
distributions~\citep{kiriliouk2019peaks}.\ The practical importance of
their threshold-stability lies in the preservation of the model's
distributional type at higher levels, which is advantageous for
extrapolation into the joint tail region of the
distribution.\ 

Both component-wise maxima and multivariate threshold exceedance
approaches can be unified under a single theme.\ Let $r_q$ be the
$q$-th quantile of the distribution of $\lVert \bm X \rVert_1$.\ Then
the overarching assumption that unifies these approaches is that as
$q\to 1$, the point process of recentered exceedances
$\bm X_i- r_q\bm 1 \mid \lVert \bm X_i\rVert_1 > r_q$, converges in
distribution to a Poisson point process on
$(0,\infty\,]\times \Delta_+^{d-1}$ with mean measure
\begin{IEEEeqnarray}{rCl}
  \lim_{q\to 1}\PR\left[\lVert \bm X\rVert_1 - r_q>z, \frac{\bm
      X}{\lVert \bm X\rVert_1} \in B~\Big|~ \lVert \bm X\rVert_1 >
    r_{q}\right] = \exp(-z) \HMRV(B),\quad \text{$z>0$,
    $B\subset \Delta_+^{d-1}$.}
  \label{eq:MRV_polar}
\end{IEEEeqnarray}
In practice, convergence~\eqref{eq:MRV_polar} is typically assumed to
hold with $\HMRV$ placing mass in the interior of the unit-simplex
$\Delta_+^{d-1}$,
$(\Delta_+^{d-1})^\circ:=\{\bm \omega\in\RR_+^d\,:\, \min_i\omega_i >
0, \lVert\bm \omega\rVert_1=1\}$.\ Although useful for theoretical
investigations, $\HMRV$ placing mass in $(\Delta_+^{d-1})^\circ$ as a
\textit{dependence assumption} can be inadequate in practice since it
imposes a strong form of dependence in the joint tail, which may not
be supported by the data at hand, see for example
Section~\ref{sec:wave}.\

\subsection{Background on starshaped sets}
\label{appendix:star_background}
An expository review and survey of material about the fundamental
mathematical notion of starshaped sets, emphasising their
geometric, analytical, combinatorial, and topological properties, is
given by~\cite{hansen2020starshaped} and references therein.

If $\bm a$ and $\bm b$ are different elements of $\RR^d$, then
$[\bm a:\bm b]=\{(1-t)\bm a + t \bm b\,:\, 0 \leq t\leq 1\}$ denotes the segment with
endpoints $\bm a$ and $\bm b$, $[\bm a:\bm b)=\{(1-t)\bm a + t \bm b\,:\, t \geq 0\}$ denotes
the half-line with origin $\bm a$ through $\bm b$,
$(\bm a:\bm b)=\{(1-t)\bm a + t \bm b\,:\, 0 \leq t\in \RR\}$ denotes the line through
$\bm a$ and $\bm b$, $]\bm a:\bm b]=\{(1-t)\bm a + t \bm b\,:\, 0 < t\leq 1\}$ denotes the
segment $[\bm a:\bm b]$ without its endpoint $\bm a$,
$[\bm a:\bm b[=\{(1-t)\bm a + t \bm b\,:\, 0 \leq t < 1\}$ denotes the segment
$[\bm a: \bm b]$ without its endpoint $\bm b$ and
$]\bm a:\bm b)=\{(1-t)\bm a + t \bm b\,:\, t > 0\}$ denotes the ray $[\bm a, \bm b)$ without
its origin $\bm a$.\ A set $G\subset \RR^d$ is called starshaped if there
exists $\bm x\in G$ such that every point $\bm y\in G$ is visible from $\bm x$ in
the sense that $[\bm x:\bm y]\subset G$.\ If $G$ is a set and $\bm x\in G$, the
star of $\bm x$ in $G$ is the set of all points in $G$ that are visible
from $\bm x$, that is, $\text{st}(\bm x:G)=\{\bm y\in G\,:\, [\bm x:\bm y]\subset G\}$.\ A
set $G$ is starshaped if and only if there exists $\bm x\in G$ such that
$\text{st}(\bm x:G)=G$.\ The kernel of a set $G$, denoted by
$\text{ker}\,G$, is the convex set of all of its star centres, that
is, $\text{ker}\,G= \{\bm x\in G\,:\, \text{st}(\bm x:G)=G\}$.\ A star set is
called a star-body if it is a compact set on a regular domain, that is if $G$
is a compact starshaped set such that $G^\circ$ is connected and
$G=\overline{G^{\circ}}$.\ Let $G$ be a closed proper subset of
$\RR^d$ with non-empty interior.\ Then $G$ is said to be strongly
starshaped at $\bm a$ whenever there exists $\bm a\in\RR^d$ such that for
every $\bm w\in\SSS^{d-1}$, the halfline $\bm a+[\Orig: \bm w)$ does not
intersect the boundary $\partial G$ of $G$ more than once.\ A strongly
starshaped set at $a$ is a starshaped set with
$a \in \text{ker}\, G$.\ A set $G$ is said to be strongly starshaped
if it is strongly starshaped at $\Orig$.\ We write $\Delta_{\bm w}$ for
the half-line through the origin $[\Orig:\bm w)$ passing through
$\bm w\in\SSS^{d-1}$.\ If $G$ is starshaped at $\bm m\in \text{ker}\,G$,
we say that $\bm z\in G$ is the last point of a ray $[\bm m: \bm x)$ in $G$ if
$\bm z\in[\bm m:\bm x)$ and there is no point $\bm y\in[\bm m: \bm x)\cap G$ such that
$\bm y\notin [\bm m: \bm z]$.\ Let $G$ be a starshaped set such that
$(\text{ker}\, G)^\circ \not= \varnothing$.\ If
$\bm m\in (\text{ker}\, G)^\circ$ and $\bm x\in \overline{G}$, then
$[\bm m: \bm x[\subset G$.\ If $G$ is a closed starshaped set such that
$(\text{ker}\, G)^\circ\not= \varnothing$, then $G^\circ$ is connected
and $G=\overline{(G^\circ)}$, whence a regular domain.


\label{appendix:background}
\subsection{Radial addition and gauges of starshaped sets}
\label{appendix:radial}
If $G$ is a closed starshaped set, $\bm m\in \text{ker}\,G$ and
$\bm w\in\SSS^{d-1}$, then there are two cases.\ There exists a last
point $\bm p$ of $\bm m+\Delta_{\bm w}$ in $G$ or
$\bm m + \Delta_{\bm w} \subset G$.\ In the first case, let
\[
  \rho_{\bm m,G}(\bm w) = \sup\{\lambda \in \RR\,:\, \bm m + \lambda \bm w \in G\}.
\]
The \textit{radial function} of $G$ at $\bm m$ is the function
$r_{\bm m,G}\,:\, \SSS^{d-1}\to \RR_+$ defined by
\[
  r_{\bm m,G} (\bm w)=
  \begin{cases}
    \rho_{\bm m,G}(\bm w)& \bm m+ \Delta_{\bm w} \not\subseteq G\\
    +\infty& \text{$\bm m+\Delta_{\bm w}\subseteq G$.}
  \end{cases}
\]
We have the equivalence
$\bm m \in (\text{ker}\,G)^\circ \Leftrightarrow r_{\bm m, G} > 0$.\ Also, if
$G$ is compact and $\bm m \in (\text{ker}\,G)^\circ$, then $r_{\bm m,G}$ is a
Lipschitz function.

Let $\bigstar$ denote the space of star-bodies in $\RR^d$ and
$\radd,\rmult\,:\,\bigstar \times \bigstar \to \bigstar$ be the
operators of \textit{radial addition} and \textit{radial
  multiplication} respectively, defined by means of radial functions
as as
\begin{IEEEeqnarray}{rCl}
  \rho_{G_1 \radd G_2} (\bm w) &=& \rho_{G_1}(\bm w) + \rho_{G_2}(\bm
  w), \qquad G_1, G_2 \in\bigstar,\\\\
  \rho_{G_1 \rmult G_2} (\bm w) &=& \rho_{G_1}(\bm w) \rho_{G_2}(\bm w),
  \qquad G_1, G_2 \in\bigstar,
\end{IEEEeqnarray}
as illustrated in Figure~\ref{fig:set_operations}.\ Multiplication of a set $G \in \bigstar$ by a nonnegative scalar $c$
is equivalent to radial multiplication of $G$ by $B_c(\Orig)$, that is,
$c G:=B_c(\Orig)\rmult G$, $c \geq 0$.\ The set $\bigstar$ equiped with the binary
operation $\radd$ is a commutative monoid.\ In other words,
$(\bigstar, \radd)$ satisfies the following.
\begin{description}
\item[Identity element:] There exists an element $\Orig\in\bigstar$ such
  that $\Orig \radd G = G \radd \Orig = G$ $\forall G \in\bigstar$.
\item[Commutative law:] For $G_1$ and $G_2$ in $\bigstar$,
  $G_1 \radd G_2 = G_2 \radd G_1$.
\end{description}
Additionally, we have
\begin{description}
\item[Scalar Multiplication Closure:] The result of $\rmult$
  between a compact star set and a non-negative scalar remains within
  the space of compact star sets.
  
\item[Distributive Laws:] The operations $\radd$ and $\rmult$ obey
  distributive laws similar to those in a ring.\ For example,
  $(G_1 + G_2) \rmult B_c(\Orig) = G_1 \rmult B_c(\Orig) + G_2 \rmult B_c(\Orig)$.
  
\item[Associativity:] Both $\radd$ and
  $\rmult$ are associative operations.\ For example, $(G_1 + G_2)
  + G_3 = G_1 + (G_2 + G_3)$.
  
\end{description}

The monoid $(\bigstar, \radd)$ is endowed with its algebraic
preordering $\leq$ defined by $G_1 \leq G_2$ whenever there exists
$H\in\bigstar$ such that $G_1 \radd H = G_2$.\ In this case, we define
radial subtraction $H = G_2 - G_1$ through ordinary difference of
radial functions.\

The algebraic structure on the space of compact star sets with the
operations $\radd$ and $\rmult$ provides a mathematical
framework to explore relationships, transformations, and compositions
of these sets.\ This algebraic approach facilitates the development of
consistent theories and methodologies when dealing with operations
involving compact star sets.\ 

Last, let $G_1$ and $G_2$ be two star-bodies such that $\rho_{S_i}$ is
continuous on $\SSS^{d-1}$.\ Then the dual Brunn-Minkowski inequality
\citep{lutwak1988intersection} yields
\[
  \lvert G_1 \radd G_2\rvert^{1/d} \leq \lvert G_1\rvert^{1/d} + \lvert G_2\rvert^{1/d},
\]
and equality holds if $d=1$ or $d\geq 2$ and $G_2 = r G_1$, for some
$r>0$.\ If $G$ is compact and strongly starshaped, then the volume of
$G$ can be expressed by the formula
\begin{equation}
  \vol{G}=\frac{1}{d}\int\limits_{\SSS^{d-1}} \rG(\bm w)^d d\bm w,
  \label{eq:vol}
\end{equation}
see~\cite{klain1997invariant}.\ The reciprocal of the radial function
of a set $G$ starshaped at $\Orig$ is called the gauge function
$g_G\,:\,\Rstar \to \RR_+$, defined by
$\gG(\bm x) = \inf\{t\geq 0\,:\,\bm x\in t G\}$ for $\bm x\in\Rstar$.\

\begin{figure}[t!]
\centering
\includegraphics[width=0.9\textwidth]{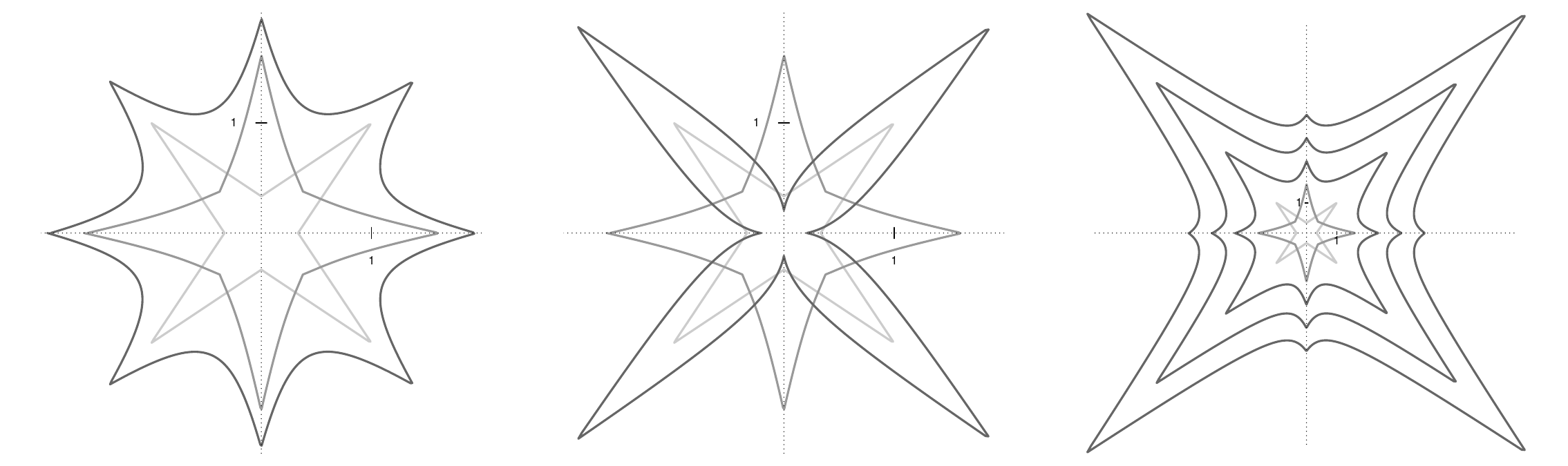}
\caption{Left:\ $\loc \radd \G$.\ Centre:\ $\loc \rdiv \G$.\ Right:
  $\loc \radd B_{-\log(1-q)}(\Orig)\rmult\G$, for $B_r(\bm x)$ the ball of radius $r$ centred at $\bm x$, and $q\in\{0.9,0.99,0.999\}$.\
  $\loc$ in lightest grey, $\G$ in mid grey, result of operations in
  dark grey.\ \label{fig:set_operations}}
\end{figure}

\section{Remainder term and quality of convergence}
\label{appendix:proofs}
\subsection{Rates of convergence}
\label{sec:rof}
Let
\begin{equation}
  \fL(t \bm x)=\exp\{-g_{\G_L}(t\bm x) + \text{remainder}\}(1+\oo(1)), \quad \text{as $t\to \infty$},
  \label{eq:rof}
\end{equation}
where the remainder is interpreted as the evaluation $\roc(t \bm x)$
of a function $\roc\,:\, \mathbb{K}\to \RR$ at
$t\bm x \in \mathbb{K}$.\ The domain of the function $\roc$ is assumed
to be a solid cone $\mathbb{K}$ in $\Rstar$.\ A characterisation of
the map $\roc$ under second-order conditions~\citep{de1996generalized}
is given in Section~\ref{sec:rof}.\ A functional
$\roc\,:\,\Rstar\to\RR$ is said to determine the rate of convergence
to $\gG$ in expression~\eqref{eq:RV1} if it satisfies
\begin{equation}
  \fL(t\bm x) = \exp\left(-\left[t \gG(\bm x)
      \{1+\oo(1)\}\right]\right)=\exp\left[-\left\{\gG(t\bm x) +
      \roc(t\bm x)\right\}\right]\{1+\oo(1)\},
  \quad\bm x \in \Rstar,
  \label{eq:rof_defn}
\end{equation}
as $t\to \infty$.\ Two functionals $\roc_1$ and $\roc_2$ that
determine the rate of convergence to $\gG$ are asymptotically
equivalent, that is,
$\lim_{t\to \infty} \roc_1(t \bm x)/\roc_2(t \bm x) = 1$, $\forall \bm
x\in\Rstar$.\ As such, the set of all functionals
$\roc$ satisfying expression~\eqref{eq:rof_defn} forms an equivalence
class.\ To better understand the properties of the members of this
class, it is helpful to consider the special case where the leading
order term of
$\roc$ is equal to the product of a positive scaling function
$\lfirst(t)$ and a function $\qfirst(\bm x)$, that is, when $\roc(t
\bm x) = \OO(\lfirst(t)\qfirst(\bm x))$ as $t\to \infty$, for $\bm
x\in\Rstar$.\ Assumption~\ref{ass:RV2} delineates a sufficient
condition subject to which the rate of convergence has such a
behaviour.
\begin{assumption}
  \label{ass:RV2}
  Suppose that condition $(ii)$ of Proposition~\ref{prop:RV1} holds.\
  There exists a non-decreasing function
  $\lfirst\,:\,(0,\infty) \to (0,\infty)$ and a function
  $\qfirst\,:\,\Rstar\to \RR$ which is not everywhere zero, such that
  \begin{equation}
    \lim_{r\to\infty}\frac{\{-\log \fL(t \bm x)/t\} - \gG(\bm x)}{\afirst(t)} =
    \qfirst(\bm x),
    \label{eq:RV2}
  \end{equation}
  where $\afirst(t)=\lfirst(t)/t$ for $t>0$.
\end{assumption}
Due to condition $(ii)$ of Proposition~\ref{prop:RV1} we have that the
scale sequence $\afirst(t) = \oo(1)$ as $t\to\infty$.\ This implies
that $\lfirst(t) = \oo(t)$ as $t\to \infty$.\ Also
Assumption~\eqref{ass:RV2} requires that $-\log \fL$ is a
\textit{multivariate extended regularly varying} function in the sense
of Definition 1 of~\cite{pappau23}.\ To be precise,
convergence~\eqref{eq:RV2} is equivalent to assuming that there exists
$\pfirst\,:\,\Rstar\to \RR$ which is not everywhere zero, such that
\begin{equation}
  \lim_{t\to\infty}\frac{[-\log \fL(t\bm x)/\gG(t\bm x)] - [-\log \fL(t\bm
    1)/\gG(t\bm 1)]}{\afirst(t)} = \pfirst(\bm x), \qquad \bm x\in\Rstar.
  \label{eq:ERVd}
\end{equation}
The correspondence between the two limit functions in~\eqref{eq:RV2}
and~\eqref{eq:ERVd} is
$\pfirst(\bm x)=[\qfirst(\bm x)/\gG(\bm x)]-[\qfirst(\bm 1)/\gG(\bm
1)]$.\ From this, it immediately follows that $p_1(\bm 1)=0$.\
Also, it is known that the only possible non-trivial limit functions
that can arise from convergence~\eqref{eq:ERVd} are those that satisfy
the functional equation
\[
  [\pfirst(t \bm x)-\pfirst(t \bm 1)]/t^{\gamma - 1} =\pfirst(\bm x) ,
  \qquad t > 0,\qquad \bm x\in\Rstar,
\]
for some $\gamma \in \RR$, see Theorem 2 of~\cite{pappau23} for a
proof.\ These observations lead to Proposition~\ref{prop:q1} which
identifies the key properties of the scaling function
$\afirst(t)=\lfirst(t)/t$ and of the limit function $\qfirst(\bm x)$
under Assumption~\eqref{ass:RV2}.
\begin{proposition}
  \label{prop:q1}
  Suppose that Assumption~\ref{ass:RV2} holds.\ There exists
  $\gamma < 1$ such that $\afirst(t)\in\textsf{RV}_{\gamma-1}^\infty$
  and $\qfirst$ is $\gamma$--homogeneous, that is,
  $\qfirst(t \bm x)=t^\gamma \qfirst(\bm x)$ $\forall t>0$ and
  $\bm x\in\Rstar$.
\end{proposition}
From Proposition~\ref{prop:q1} we obtain
$\lfirst(t)\in \text{RV}_{\gamma}^\infty$.\ Because Assumption~\ref{ass:RV2} requires $\lfirst(t)$ to be a non-decreasing function,
it suffices to consider the following two cases.\ First, the case when
$\lfirst(t)$ converges to a constant $K\in (0,\infty)$ and second, the
case when $\lfirst(t)$ grows without bound.\ For the former case we
consider without loss of generality the case $\lfirst(t)=1$
$\forall t>0$.\ This is because positive real constants can be
absorbed in the limit function $\qfirst$.\ Below, we treat the two
cases in turn.

First, when $\lfirst(t)=1$, the leading order term of $\roc(t\bm x)$
is constant in $t$ which implies that the rate of convergence is
determined by the function $\qfirst$, that is,
$\roc(t \bm x) = \qfirst(t \bm x) + \oo(1)$, as $t\to \infty$.\ Also,
because $\lfirst\in \textsf{RV}_{0}^\infty$, Proposition~\ref{prop:q1}
yields that $\qfirst$ is a $0$--homogeneous function.\ Consequently,
\begin{equation}
  \fL(t \bm x) = \exp\left[-\left\{t \gG(\bm x) + \qfirst(\bm x)\right\}\right] \{1+\oo(1)\}, \qquad
  \text{as $t\to \infty$}.
  \label{eq:approx_g0}
\end{equation}
This special case gives a characterisation of the rate of convergence
which is independent of the distance from the origin $t$.\ This
behaviour is specific to the case $l(t)=1$ and is waived when
$\lfirst(t)$ grows without bound because in the latter case, the rate
of convergence may have non-negligible remainder terms.\

Using the same working principle as in Assumption~\ref{ass:RV2}, below
in Assumption~\ref{ass:RV3} we outline a sufficient condition subject
to which the rate of convergence exhibits a so-called
\textit{second-order} behaviour~\citep{de1996generalized, pappau23}.
\begin{assumption}
  \label{ass:RV3} Suppose that Assumption~\ref{ass:RV2} holds with
  $\lim_{t\to\infty} \lfirst(t) = \infty$.\ There exists a
  non-decreasing function $\lsecond$ and a function
  $\qsecond\,:\,\Rstar\to \RR$ which is not equal to the product of
  $\qfirst$ and a $0$--homogeneous function on
  $\Rstar$, such that for $\asecond(t)=\lsecond(t)/\lfirst(t)$
  \begin{equation}
    \lim_{t\to\infty} \frac{([\{-\log f_{L}(t \bm x)/t\} - \gG(\bm x)]/\afirst(t)) - \qfirst(\bm x)}{\asecond(t)} =
    \qsecond(\bm x).
    \label{eq:RV3}
  \end{equation}
\end{assumption}
It is clear that when Assumption~\ref{ass:RV3} holds with
$\lsecond(t)=1$, then
\begin{equation}\label{eq:gu1u2-generic}
  \fL(t \bm x) = \exp[-\{t \gG(\bm x) + \lfirst(t)\,\qfirst(\bm x) + \qsecond(\bm x)\} \{1+\oo(1)], \qquad \text{as $t\to \infty$}.
\end{equation}
More generally, we can construct an asymptotic expansion of
$-\log \fL(t\bm x)$ of the form
\begin{equation}
  -\log \fL(t\bm x) = \gG(t \bm x) + \sum_{j=1}^{\infty}l_j(t)
  u_j(\bm x), \quad\text{as $t\to\infty$}.
  \label{eq:expansion}
\end{equation}
for a sequence of positive scale functions $\{l_j\}_{j=1}^M$ with
$M\in\mathbb{N}\cup\{\infty\}$, satisfying $l_j(t)=\oo(l_{j+1}(t))$ as
$t\to \infty$, and a sequence of functions $\{u_j\}_{j=1}^M$
mapping $\Rstar$ into $\RR$.\ The function $u_j$ captures the
contribution of $-\log\fL(t\bm x)$ at a given order of the series and
$l_{j}$ is the scaling factor associated with this term.\ The $j$-th
term $l_j(t)u_j(\bm x)$ in the sum represents the $j$-th
approximation to $-\log \fL(t\bm x)$ as $t\to \infty$, and the scaling
factors $l_j$ determine the rate of convergence to $u_{j-1}$,
with $u_0(\bm x):=\gG(\bm x)$.

\begin{figure}[htbp!]
  \centering
  \includegraphics[scale=.2, trim= 100 0 100 0]{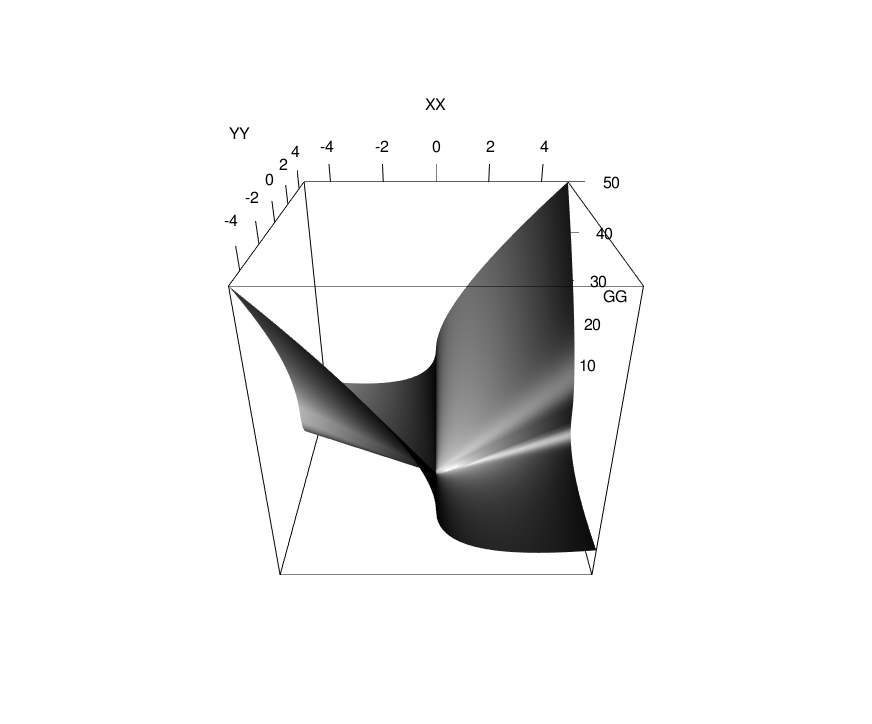}
  \includegraphics[scale=.2, trim= 100 0 100 0]{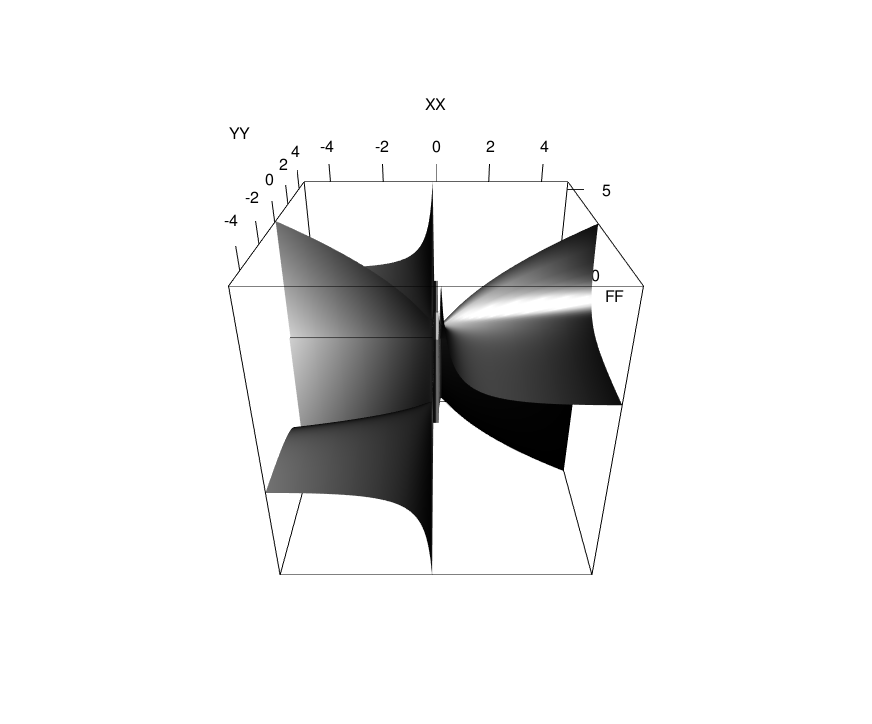}
  \includegraphics[scale=.2, trim= 100 0 100 0]{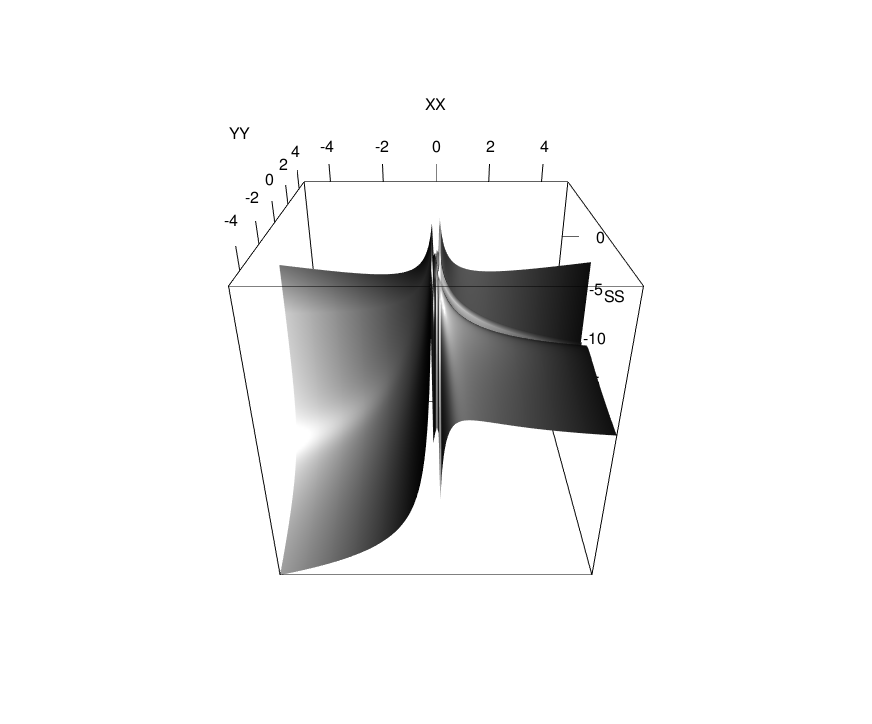}
  \includegraphics[scale=.2, trim= 100 0 100 0]{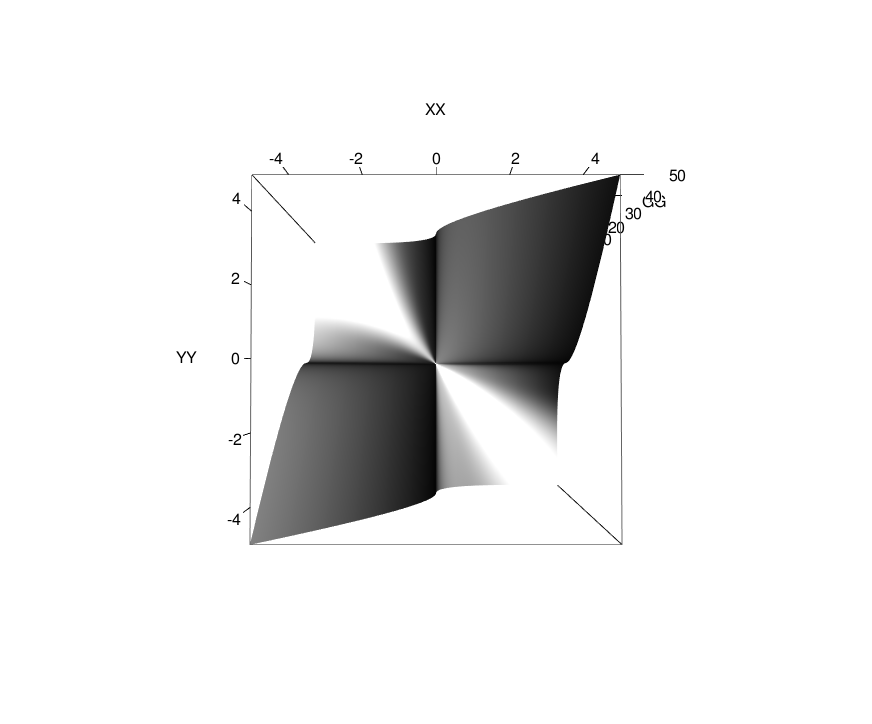}
  \includegraphics[scale=.2, trim= 100 0 100 0]{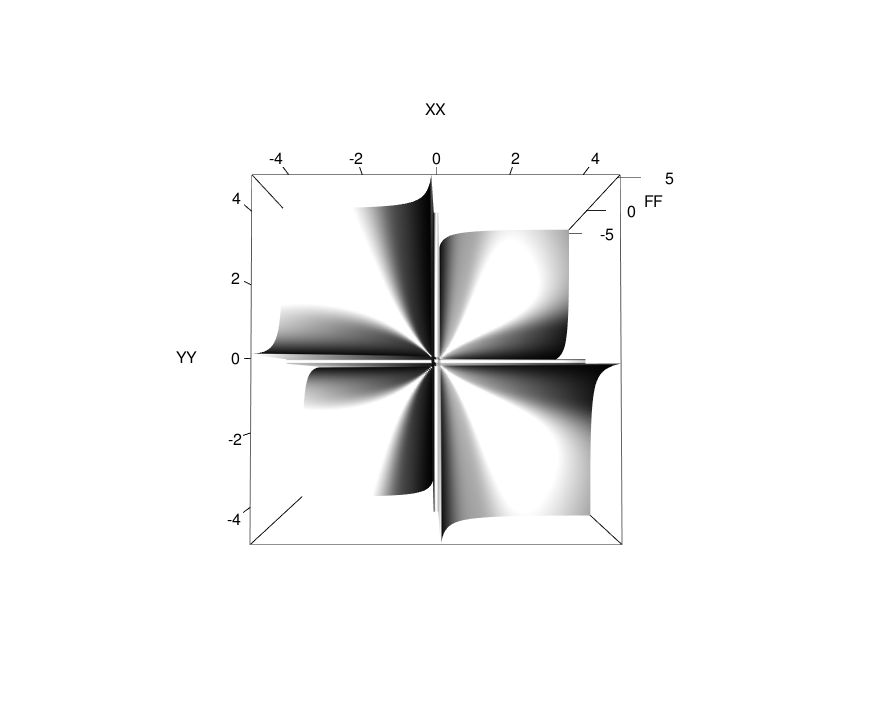}
  \includegraphics[scale=.2, trim= 100 0 100 0]{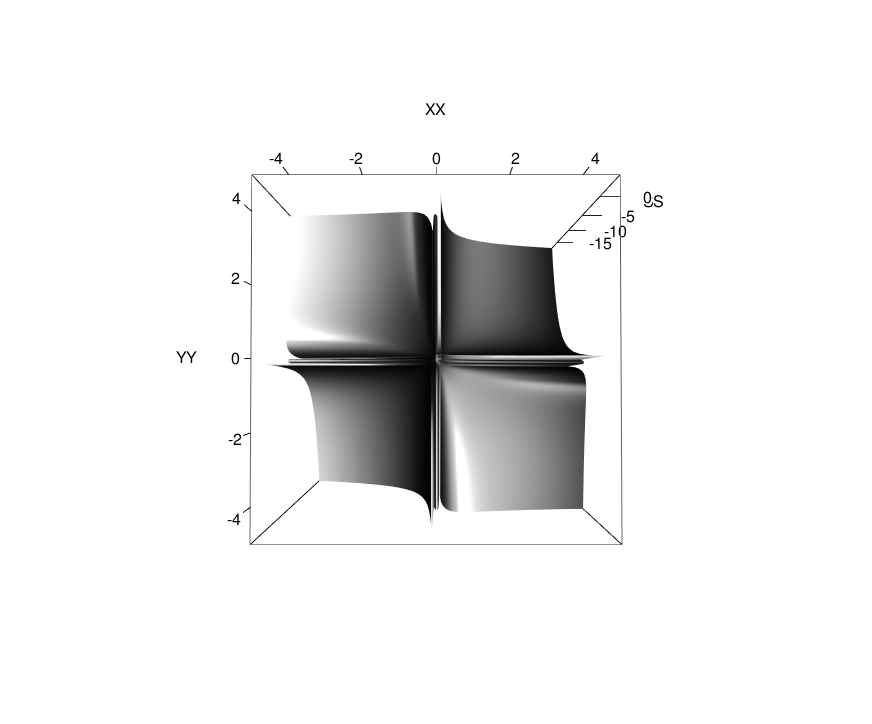}
  \includegraphics[scale=.2, trim= 100 0 100 0]{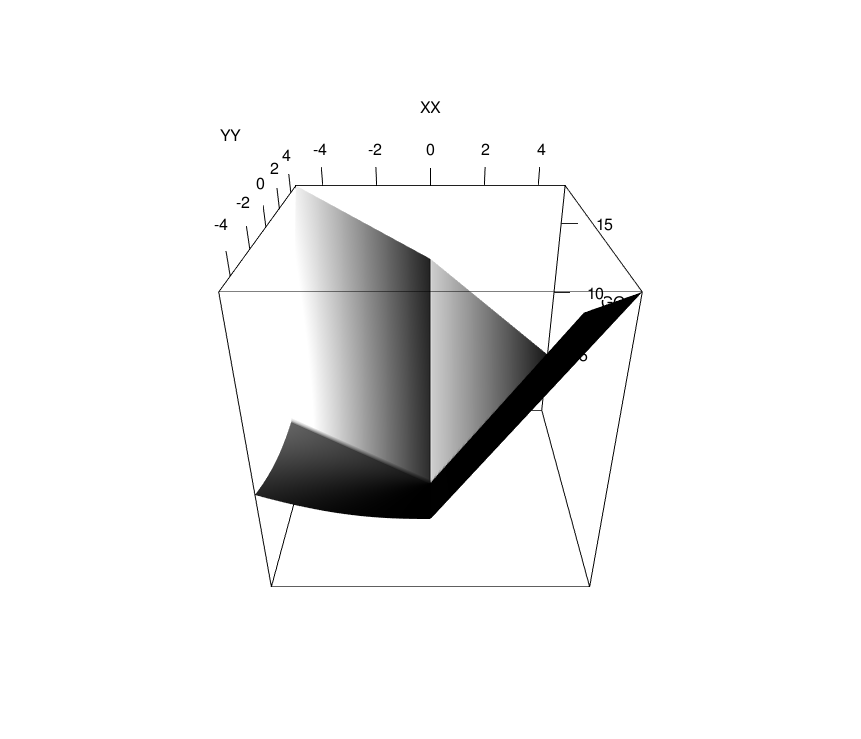}
  \includegraphics[scale=.2, trim= 100 0 100 0]{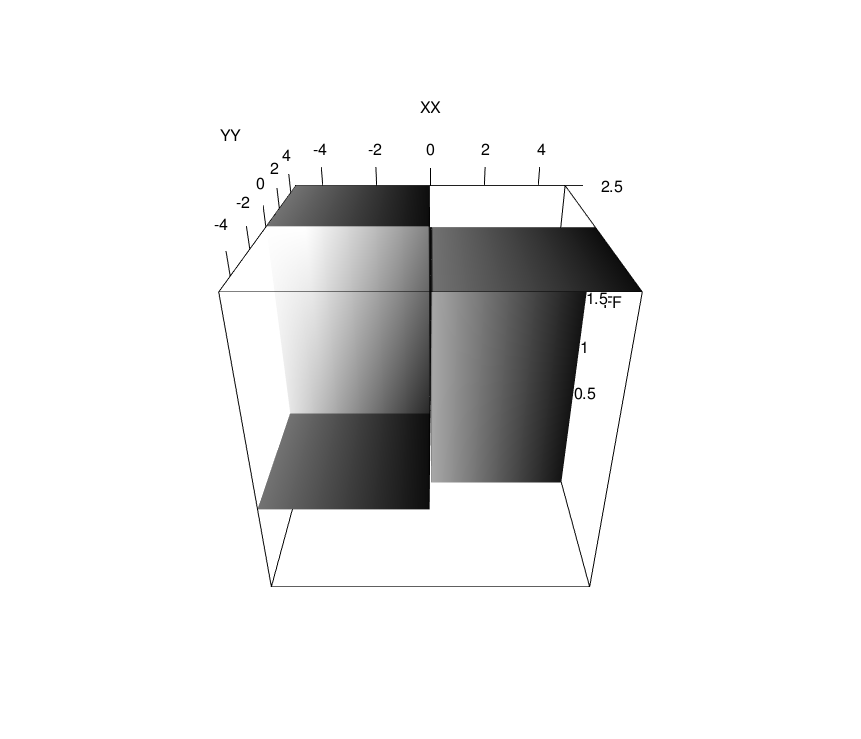}
  \includegraphics[scale=.2, trim= 100 0 100 0]{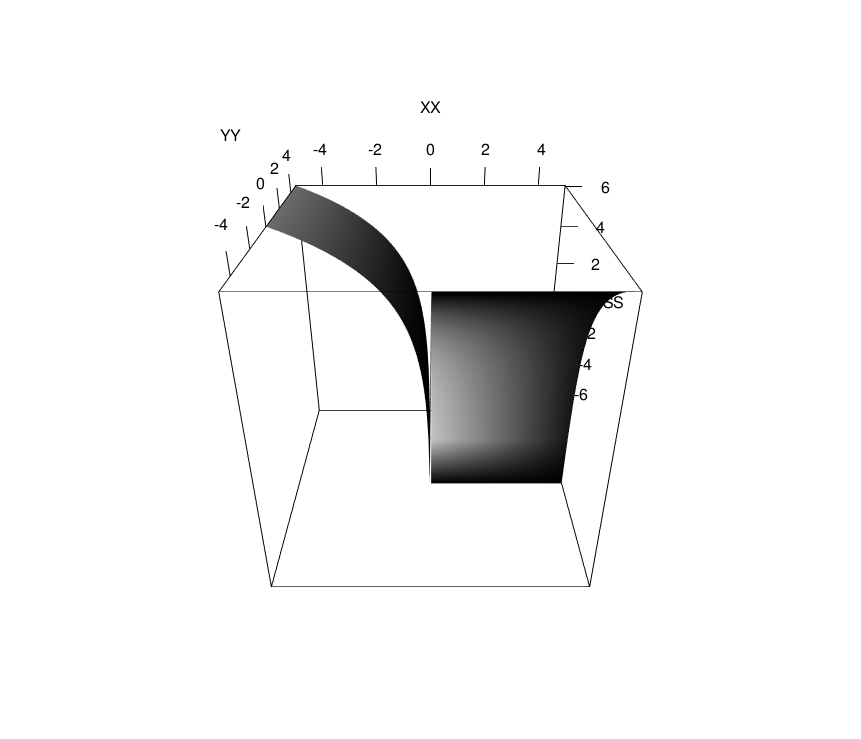}
  \includegraphics[scale=.2, trim= 100 0 100 0]{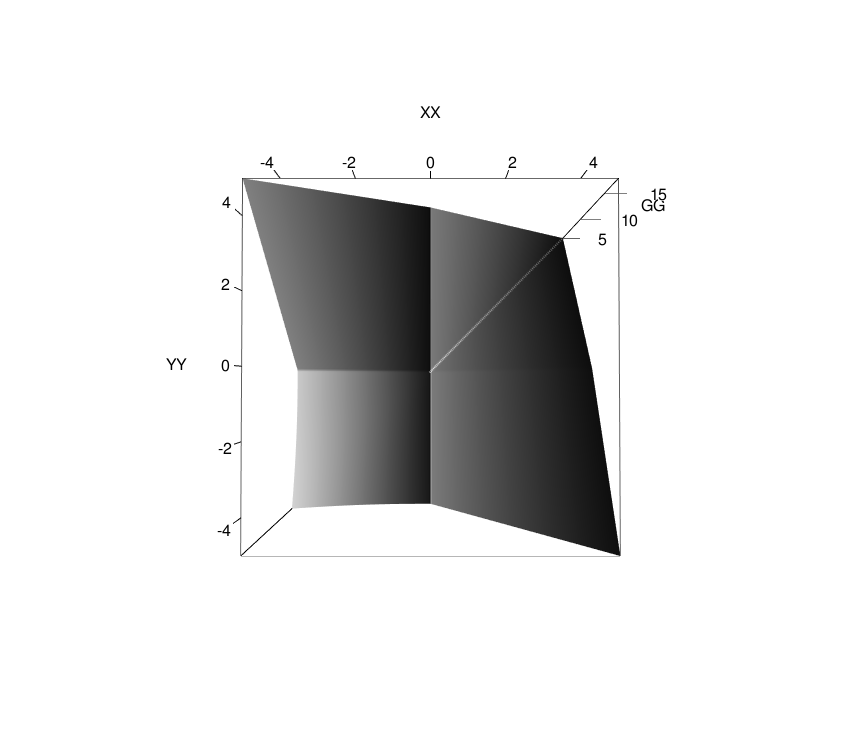}
  \includegraphics[scale=.2, trim= 100 0 100 0]{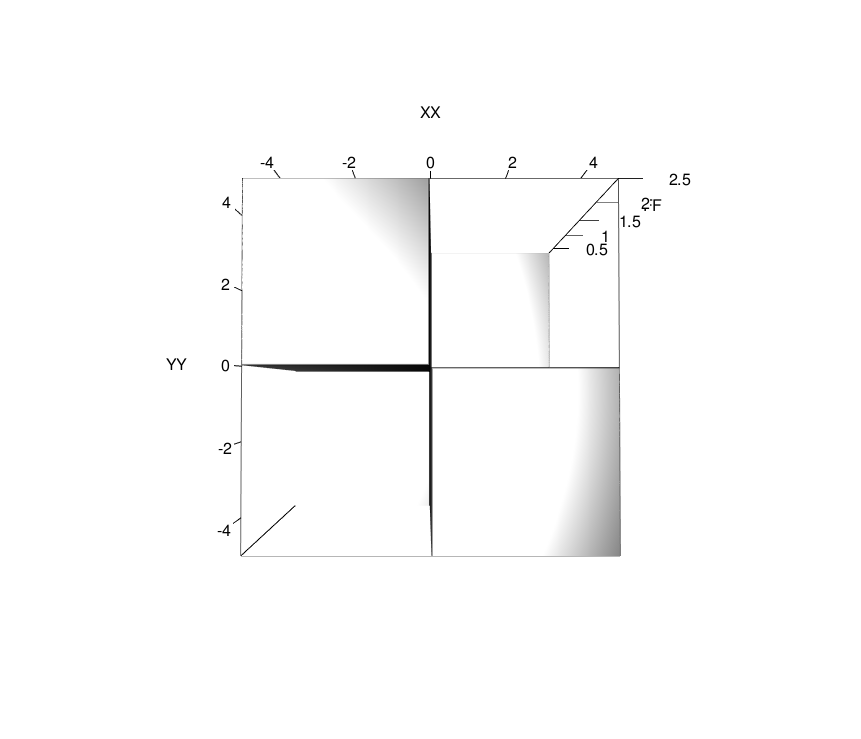}
  \includegraphics[scale=.2, trim= 100 0 100 0]{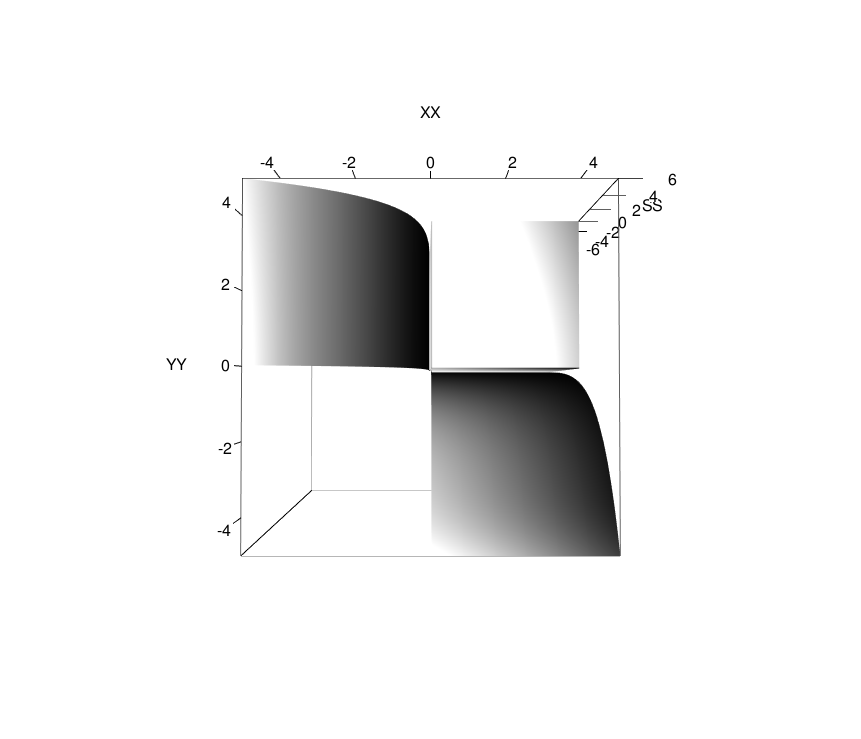}
  \caption{Limit functions in asymptotic expansion of the logarithm of
    the density of the multivariate normal copula with Laplace
    marginals.\ \textit{Left:} gauge function $\gG$ of limit set
    $\G$.\ \textit{Centre:} limit function $\qfirst$.\ \textit{Right:}
    Limit function $\qsecond$.\ \textit{Bottom:} alternative viewing
    angle obtained by rotating the images shown in top row by
    $60^\circ$ about $x$ axis}
  \label{fig:rof}
\end{figure}

\begin{figure}[htbp!]
  \centering
  \includegraphics[scale=.22, trim= 100 0 100 0]{log_grof.png}
  \includegraphics[scale=.22, trim= 100 0 100 0]{loglambda1rof.png}
  \includegraphics[scale=.22, trim= 100 0 100 0]{loglambda2rof.png}
  \includegraphics[scale=.22, trim= 100 0 100 0]{log_grof_rot.png}
  \includegraphics[scale=.22, trim= 100 0 100 0]{loglambda1rof_rot.png}
  \includegraphics[scale=.22, trim= 100 0 100 0]{loglambda2rof_rot.png}
  \caption{Limit functions in asymptotic expansion of the logarithm of
    the density of the bivariate max-stable copula with logistic
    dependence and Laplace marginal distributions.\ \textit{Left:}
    gauge function $\gG$ of limit set $\G$.\ \textit{Centre:} limit
    function $\qfirst$.\ \textit{Right:} Limit function $\qsecond$.\
    \textit{Bottom:} alternative viewing angle obtained by rotating
    the images shown in top row by $60^\circ$ about $x$ axis}
  \label{fig:rof2}
\end{figure}
\begin{proposition}
  \label{prop:rof} Suppose that Assumption~\ref{ass:RV3} holds for the
  random vector $\bm X_L$.\ There exists $\gamma < 1$ and $\rho\leq 0$
  such that 
  \[
    \qsecond (t \bm x) = t^\gamma \qsecond(\bm x) + c
    t^\gamma\frac{t^\rho - 1}{ \rho} \qfirst(\bm x).
  \]
  \label{rem:single_case}
  and the solution of this functional equation is
  \begin{equation}
    \qsecond (\bm x) = \lVert \bm x \rVert^{\gamma+\rho}
    h\left(\frac{\bm x}{\lVert \bm x\rVert}\right) + c \lVert \bm x
    \rVert^{\gamma} \log \lVert \bm x
    \rVert \qfirst\left(\frac{\bm x} {\lVert \bm x
        \rVert}\right), \qquad \bm x\in\Rstar.
    \label{eq:g0star_sol}
  \end{equation}
  for a measurable function $h$ that maps $\SSS^{d-1}$ into $\RR$.
\end{proposition}

Proposition~\ref{prop:rof} provides key insight into the rate of
convergence to the limit function in Proposition~\ref{prop:RV1}
(\textit{ii}).\ This result is shown to hold in
Appendix~\ref{sec:limit-behav} for certain known copulas in
Examples~\ref{sec:MVN_supplementary}, \ref{sec:MVL_supplementary}, and~\ref{sec:Student-t_Lap}.\ We note that a similar approach can be used
to study the rate of convergence in Proposition~\ref{prop:RV1}
(\textit{i}) and (\textit{iii}) for light- and heavy-tailed
distributions.\

\begin{proof}[Proof of Proposition~\ref{prop:q1}]  
  The homogeneity of $\qfirst$ follows directly from the
  correspondence between $\qfirst$ and $\pfirst$ and the fact that
  $\pfirst$ is $(\gamma-1)$--homogeneous.\ It remains to show that
  $\afirst\in\text{RV}_{\gamma-1}^\infty$.\ For any $a>0$ and
  $\bm x\in\Rstar$ we have 
  \begin{IEEEeqnarray*}{rCl}
    \frac{[-\log \fL\{t (a \bm x)\}/t] - \gG(a \bm x)}{\afirst(t)} &=&
    a \frac{[-\log \fL\{(a t) \bm x\}/(at)] -  \gG(\bm x)}{\afirst(a t)}
    \frac{{\afirst(a t)} }{{\afirst(t)}  }\cdot  
  \end{IEEEeqnarray*}
  By Assumption~\ref{ass:RV2}, the LHS converges to $\qfirst(a\bm x)$
  as $t\to \infty$ and the second factor of the RHS converges to
  $\qfirst(\bm x)$ as $t\to \infty$.\ Hence, the third factor of RHS
  must converge as $t\to \infty$ which implies that there exists
  $\gamma \in \RR$ such that $\afirst\in\text{RV}_{\gamma-1}^\infty$.\
  The case $\gamma > 1$ is not admissible due to homogeneity of $\gG$ as
  the latter ensures that the rate of convergence $\roc(t\bm x)$ is
  not growing faster than $\gG(t\bm x)$, that is,
  $\lim_{t\to \infty} \roc(t\bm x)/\gG(t\bm x) = 0$ for all
  $\bm x\in\Rstar$.\ When $\gamma =1$ we have
  $\afirst\in\text{RV}_{-1}^\infty$, which implies that
  $\lfirst(t)\in\text{RV}_{0}^\infty$.\ Thus, the case $\gamma=1$ is
  not admissible because it requires $\lfirst(t)\to 0$ as
  $t\to \infty$ for $\roc(t\bm x)/\gG(t\bm x)$ to converge to zero for
  all $\bm x\in\Rstar$, which contradicts the assumption that
  $\lfirst$ is a non-decreasing function.
\end{proof}


\begin{proof}[Proof of Proposition~\ref{prop:rof}]
  Since for any $a>0$ we have
  \begin{IEEEeqnarray*}{rCl}
    &&\cfrac{\cfrac{-t^{-1}\log f_{[L]} \{t\, (a \bm x)\} -\gG (a \bm
        x)}{\afirst(t)} - \qfirst(a \bm x) } {\asecond(t)} 
        =\cfrac{a
      \cfrac{ -(at)^{-1}\log f_{[L]} \{(a t) \,\bm x\} - \gG (\bm
        x)}{\afirst(a t)} \cfrac{\afirst(at)}{\afirst(t)} - a^{\gamma}
      \qfirst(\bm x)}{\asecond(a t)} \cfrac{\asecond(a
      t)}{\asecond(t)}
    \\\\
    && =\cfrac{\asecond(a t)}{\asecond(t)} a^\gamma\cfrac{ \cfrac{
        -(at)^{-1}\log f_{[L]} \{(a t) \,\bm x\} - \gG (\bm x)}{\afirst(a
        t)} - a \qfirst(\bm x)}{\asecond(a t)} \\
        && \hspace{1cm}+ a
    \cfrac{\Bigg(\cfrac{\afirst(at)}{\afirst(t)} -
      a^{\gamma-1}\Bigg)}{\asecond(t)} \cfrac{ -(at)^{-1}\log f_{[L]} \{(a
      t) \,\bm x\} - \gG (\bm x)}{\afirst(a t)}\cdot,
  \end{IEEEeqnarray*}
  it follows that
  \[
    \qsecond(a \bm x) = \lim_{t \to \infty} \frac{\asecond(a
      t)}{\asecond(t)} a^\gamma \qsecond(\bm x) + a
    \lim_{t\to \infty} \cfrac{\Bigg(\cfrac{\afirst(at)}{\afirst(t)} -
      a^{\gamma-1}\Bigg)}{\asecond(t)} \qfirst(\bm x).
  \]
  Suppose there exist $a>0$, and sequences $t_n, t_n'$ such that
  \[
    \lim_{t_n \to \infty} \frac{\asecond(a t_n)}{\asecond(t_n)} = A,
    \lim_{t_n' \to \infty} \frac{\asecond(a t_n')}{\asecond(t_n')} =
    A', \lim_{t_n\to \infty} \cfrac{\Bigg(\cfrac{\afirst(at_n)}{\afirst(t_n)}
      - a^{\gamma-1}\Bigg)}{\asecond(t_n)}=B \quad \text{and}
    \lim_{t_n'\to \infty} \cfrac{\Bigg(\cfrac{\afirst(at_n')}{\afirst(t_n')} -
      a^{\gamma-1}\Bigg)}{\asecond(t_n')}=B',
  \] with $A\neq A'$ and $B\neq B'$.\ Then, for all
  $\bm x \in \Rstar$,
  \[
    0 = (A-A')a^\gamma \qsecond(\bm x) + a (B-B') \qfirst(\bm x)
  \]
  By assumption, $\qsecond$ is not equal to the product of $\qfirst$
  and a $0$--homogeneous functional.\ Hence, $A=A'$ and $B=B'$.\ This
  implies that there exists $\rho\leq 0$ and $c \in \RR$ such that
  \[
    \lim_{t\to\infty}\frac{\asecond(t x)}{\asecond(t)} = x^\rho
    \qquad \text{and}\qquad \lim_{t\to \infty}
    \cfrac{\Bigg(\cfrac{\afirst(at)}{\afirst(t)} -
      a^{\gamma-1}\Bigg)}{\asecond(t)} = c x^{\gamma-1}\frac{x^\rho -
      1}{\rho},\qquad \text{$\forall x > 0$,}
  \]
  for a proof see~\cite{de1996generalized}.\ Consequently, we arrive at
  the functional equation
  \begin{equation}
    \qsecond \left(\lVert \bm x \rVert \frac{\bm x}{\lVert \bm x
        \rVert} \right) = \lVert \bm x\rVert^{\gamma+\rho}
    \qsecond\left(\frac{\bm x}{\lVert \bm x \rVert}\right) + c
    \lVert \bm x
    \rVert^{\gamma}\frac{\lVert \bm x
      \rVert^\rho-1}{\rho}\qfirst\left(\frac{\bm x}{\lVert \bm x \rVert}\right), \qquad \bm x\in\Rstar.
    \label{eq:fnc_eqn}
  \end{equation}
  where $\qfirst$ is a $\gamma$--homogeneous functional on $\Rstar$
  and $\qsecond$ is an unknown function.\ Let
  \begin{equation}
    \qsecond^\star \left( \bm x \right) = \lVert \bm x\rVert^{\gamma+\rho}
    \widetilde{\lambda}_2\left(\frac{\bm x}{\lVert \bm x \rVert}\right) + c
    \lVert \bm x
    \rVert^{\gamma}\frac{\lVert \bm x
      \rVert^\rho-1}{\rho}\qfirst\left(\frac{\bm x}{\lVert \bm x \rVert}\right), \qquad \bm x\in\Rstar,
    \label{eq:fnc_eqn_sol}
  \end{equation}
  where $\widetilde{\lambda}_2$ be an arbitrary map from
  $\Rstar$ into $\RR$.\ A simple calculation shows that
  the function $\qsecond^\star$ satisfies the functional
  equation~\eqref{eq:fnc_eqn}.\ This proves that the general solution
  of the functional equation~\eqref{eq:fnc_eqn} is $\qsecond^\star$.\
  Last, by Assumption~\ref{ass:RV3} we have
  $\asecond(t)=\oo(\afirst(t))$ as
  $t\to\infty$.\ 


\end{proof}

\section{Convergence to gauge functions for $d$-dimensional copulas}\label{sec:limit-behav}
\subsection{Multivariate normal copula}
\label{sec:MVN_supplementary}
The negative logarithm of the probability density function of the
standard $d$-dimensional normal copula with standard Laplace margins
and positive-definite precision matrix $\mathsf{Q}$ is
\color{black}
standard $d$-dimensional normal copula with standard Laplace margins
and positive-definite precision matrix $\mathsf{Q}$ is
\begin{IEEEeqnarray}{rCl}
  -\log f(t \bm x) &=& -\frac{1}{2} \log |\mathsf{Q}| + d\log 2 + t\sum\limits_{i=1}^{d} \left|x_{t,i}\right| + \frac{1}{2} \bm H(t, \bm x)^\top (\mathsf{Q}-I)
  \bm H(t, \bm x) 
  \label{eq:MVN_logdens}
\end{IEEEeqnarray}
\color{black} where
$\bm H\,:\,\mathbb{R}_+\times \mathbb{R}^d \to \mathbb{R}^d$ is
defined by $\bm H(t, \bm x) = (H(t, x_i)\,:\, i=1,\dots, d)$, with
$H(t, y):=\Phi^{-1}\{F_L(t y)\}$, for $t>0$ and $y\in\mathbb{R}$.\
By Mill's ratio, each individual component of $\bm H(t, \bm x)$ admits the asymptotic
expansion
\begin{IEEEeqnarray}{rCl}
  H(t, x_i) &=&
  \begin{cases}
    0 & \text{$x_i = 0$}\\
    \text{sgn}(x_i)\left[2(\log 2 + \lvert t x_i\rvert)- \log\{4\pi (\log 2 + \lvert t x_i\rvert)\}\right]^{1/2}+o(t^{-1/2})  & \text{otherwise,} 
  \end{cases}
  \label{eq:asexp}
\end{IEEEeqnarray}
as $t\to \infty$, also implying that
\begin{IEEEeqnarray}{rCl}
  H(t, x_i)^2 &=&
  \begin{cases}
    0 & \text{$x_i = 0$}\\
    2(\log 2 + \lvert t x_i\rvert)- \log\{4\pi (\log 2 + \lvert t x_i\rvert)\}+o(1)  & \text{otherwise,}
  \end{cases}
  \label{eq:MVN_square}
\end{IEEEeqnarray}
as $t\to\infty$.

Suppose that $x_i\neq 0$ for all $i\in\{1,\dots, d\}$.\ Since
\[
  2(\log 2 + |t x_i|) > \log\{4\pi (\log 2 + |t x_i|)\}
\]
for all sufficiently large $t$, the binomial series gives that as
$t\to \infty$,
\begin{equation}
  H(t, x_i) = \text{sgn}(x_i) \sum_{k=0}^{\infty}  a_k(t, x_i)+o(t^{-1/2}).
  \label{eq:MVN_binomial_series}
\end{equation}
where
\begin{IEEEeqnarray*}{rCl}
  a_k(t, x_i) &=& \frac{{1/2 \choose k}(-1)^k [\log\{4\pi (\log 2 + t
    \lvert x_i \rvert)\}]^k }{\{2(\log 2 + t \lvert x_i
    \rvert)\}^{k-(1/2)}} 
\end{IEEEeqnarray*}
Let $\mathsf{A}=\mathsf{Q}-\mathsf{I}$ and consider the quadratic form
in expression~\eqref{eq:MVN_logdens}
\begin{IEEEeqnarray}{rCl}
  \lefteqn{ \frac{1}{2}\bm H(t, \bm x)^\top \mathsf{A}
    \bm H(t, \bm x)=}\label{eq:MVN_quadratic}\\
  &=& \frac{1}{2}\sum_{k=0}^\infty [\text{sgn}(\bm x)\bm a_k(t, \bm
  x)]^\top \, \mathsf{A} \, [\text{sgn}(\bm x)\bm a_k(t, \bm x)] + 
  \sum_{k=0}^\infty \sum_{k'=k+1}^\infty [\text{sgn}(\bm x)\bm
  a_{k}(t, \bm x)]^\top \, \mathsf{A} \, [\text{sgn}(\bm x)\bm
  a_{k'}(t, \bm x)] \nonumber
\end{IEEEeqnarray}
where $\bm a_k(t, \bm x)=(a_k(t, x_i)\,:\,i=1,\dots, d)^\top$ is a
column vector with $d$ rows, satisfying
\begin{IEEEeqnarray*}{rCl}
  \bm a_k(t, \bm x)
  & = & \left\{
    \begin{array}{lr}
      (2 t |\bm x|)^{1/2} + (2 t |\bm x|)^{-1/2}\log 2 + O(t^{-3/2}), & \text{$k=0$}\\\\
      -[\{(\log t) /(2t|\bm x|)^{1/2}\}+\{\log (4\pi |\bm x|)/(2t|\bm x|)^{1/2}\}]+ O(t^{-3/2}\log t),  & \text{$k=1$}\\\\
      (\log t)^k/(2t|\bm x|)^{k-(1/2)}+O((\log t)^{k-1}/t^{k-(1/2)})& \text{$k\geq 2$},
    \end{array}
  \right.
\end{IEEEeqnarray*}
as $t\to\infty$.\ Here, all $1$-dimensional real-valued terms are
recycled to match the length of the vector $\bm x$, and vector algebra
is interpreted as elementwise.

We consider each summand of the two terms in expression
\eqref{eq:MVN_quadratic} separately.\ The summand associated with
$k=0$ in the first term is
\begin{IEEEeqnarray*}{rCl}
  && \frac{1}{2}\left[\left\{(2 t |\bm x|)^{1/2} +\frac{\log 2}{ (2 t |\bm
        x|)^{1/2}}+ O(t^{-3/2})\right\}\text{sgn}(\bm x)\right]^\top
  \, \mathsf{A} \, \left[\left\{(2 t |\bm x|)^{1/2} + \frac{\log 2}{(2 t |\bm x|)^{1/2}}+O(t^{-3/2})\right\}\text{sgn}(\bm x)\right] = \\
  &=& t\, \left[\text{sgn}(\bm x) \vert \bm x\rvert^{1/2} +
    \text{sgn}(\bm x)\frac{\log 2}{ 2 t( \lvert \bm
      x\rvert)^{1/2}}\right]^\top \, \mathsf{A} \, \left[\text{sgn}(\bm
    x) \vert \bm x\rvert^{1/2} + \text{sgn}(\bm x)\frac{\log 2}{ 2 t(
      \lvert \bm x\rvert)^{1/2}}\right] + O(t^{-1/2}) \\
  &=& t \left[\text{sgn}(\bm x) \vert \bm
    x\rvert^{1/2}\right]^\top \mathsf{A} \, \left[\text{sgn}(\bm x)
    \vert \bm x\rvert^{1/2}\right] + \color{black}\color{black}
  \left\{\text{sgn}(\bm x) \vert \bm x\rvert^{1/2}\right\}^\top
  \mathsf{A} \, \left\{\text{sgn}(\bm x) \frac{\log2}{\lvert \bm x
      \rvert^{1/2}}\right\} + O(t^{-1/2})\\
  &=& t \left\{\text{sgn}(\bm x) \vert \bm
    x\rvert^{1/2}\right\}^\top \mathsf{Q} \, \left\{\text{sgn}(\bm x)
    \vert \bm x\rvert^{1/2}\right\} - t \sum\limits_{i=1}^{d} |x_i| + \color{black}\color{black}
  \left\{\text{sgn}(\bm x) \vert \bm x\rvert^{1/2}\right\}^\top
  \mathsf{A} \, \left\{\text{sgn}(\bm x) \frac{\log2}{\lvert \bm x
      \rvert^{1/2}}\right\} + O(t^{-1/2}).
\end{IEEEeqnarray*}
and a similar calculation shows that the terms associated with $k\geq 1$ decays to zero according to $O\{(\log t)^{2 k}/t^{2 k - 1}\}$.\ 
  Now, we consider the second term in \eqref{eq:MVN_quadratic}.\ Here,
  the summand associated with $(k,k')=(0,1)$
\begin{IEEEeqnarray*}{rCl}
  && \left[\{2(\log 2 + t \lvert \bm x\rvert)\}^{1/2}\text{sgn}(\bm
    x)\right]^\top \, \mathsf{A} \, \left[- \frac{\log\{4\pi (\log 2 +
      t \lvert \bm x \rvert)\}}{2 \{2 (\log 2
      + t \lvert \bm x\rvert)\}^{1/2}}\text{sgn}(\bm x)\right] = \\
  &=& -\log t \left\{\text{sgn}(\bm x) \lvert \bm
    x\rvert^{1/2}\right\}^\top \, \mathsf{A} \, \left\{\text{sgn}(\bm
    x)(2
    \lvert \bm x\rvert^{1/2})^{-1}\right\}-\\
  &&-\frac{1}{2} \left\{\text{sgn}(\bm x) \lvert \bm x\rvert^{1/2}\right\}^\top
  \, \mathsf{A} \, \left\{\text{sgn}(\bm x) 
    \log (4\pi \lvert \bm x \rvert) \lvert \bm x\rvert^{-1/2}\right\}+O(t^{-1})
\end{IEEEeqnarray*}
Thus,
\begin{equation}
  -\log f(t\bm x) = t g(\bm x) + (\log t)\, \qfirst(\bm x) + \qsecond
  (\bm x) + o(1) \quad \text{as $t\to\infty$},
  \label{eq:MVN_outer_expansion}
\end{equation}
where
\begin{IEEEeqnarray}{rCl}
  g(\bm x)&=& \left[\text{sgn}(\bm x) \vert \bm
    x\rvert^{1/2}\right]^\top \mathsf{Q} \, \left[\text{sgn}(\bm x)
    \vert \bm x\rvert^{1/2}\right]\label{eq:g_MVN_Laplace}\\
  \qfirst(\bm x) &=& -\frac{1}{2}\left[\text{sgn}(\bm x) \lvert \bm
    x\rvert^{1/2}\right]^\top \, (\mathsf{Q}-\mathsf{I}) \,
  \left[\text{sgn}(\bm x) \lvert \bm x\rvert^{-1/2}\right]\nonumber\\
    \qsecond(\bm x) &=& \left[\text{sgn}(\bm x) \vert \bm
      x\rvert^{1/2}\right]^\top (\mathsf{Q}-\mathsf{I}) \, \left[\text{sgn}(\bm
      x) \log 2 \lvert \bm x \rvert^{-1/2}\right]-
  \nonumber\\
  & & - 
  \frac{1}{2} \left[\text{sgn}(\bm x) \lvert \bm
    x\rvert^{1/2}\right]^\top \, (\mathsf{Q}-\mathsf{I}) \,
  \left[\text{sgn}(\bm x) \log (4\pi \lvert \bm x
      \rvert) \lvert \bm x\rvert^{-1/2}\right] - \frac{1}{2}\log
  |\mathsf{Q}| + d\log 2
\end{IEEEeqnarray}

The asymptotic expansion in \eqref{eq:MVN_outer_expansion} is
valid for large $t$ when all $x_i\neq 0$, and loses validity when at
least one $x_i$ is equal to 0 This is a result of the asymptotic
expansion of \eqref{eq:asexp} and \eqref{eq:MVN_square}, both
depending on $t$ and $x_i$, resulting in nonuniform behavior.\ When
$x_i = 0$, the expansions \eqref{eq:asexp} and \eqref{eq:MVN_square}
collapse to zero, while for $x_i \neq 0$, the leading term grows as
$t^{1/2}$ with logarithmic corrections.\ This difference in growth
rates reflects the nonuniformity in $x_i$, as the expansion behaves
differently near the origin and away from it, complicating the
asymptotic analysis of $-\log f(t \bm x)$ whenever $\bm x$ is near
an axis.\

In what follows, we show that a uniform asymptotic expansion exists,
by proving ${-\log f(t\bm{x})}/{t}$ converges locally uniformly to a
continuous gauge function via the method of continuous convergence
\citep{resnick2007heavy}.\ That is, we prove that
$-\log f(t\bm x)/t \to g(\bm x)$ as $t\to\infty$, uniformly on compact
sets in the variable $\bm x\in\mathbb{R}^d$, by showing that
${-\log f(t\bm{x}_t)}/{t}$ converges to $g(\bm x)$ whenever
$\bm x_t\to \bm x\in \mathbb{R}^d$.\ Consequently, because $\SSS^{d-1}$
is compact, the convergence is uniform on $\SSS^{d-1}$ and therefore,
the conditions of Proposition \ref{prop:RV1} are
satisfied.\ 
To prove this assertion, it suffices to consider three disjoint
index sets
$A,B,C\subseteq\left\{1,\dots,d\right\}\cup
\left\{\emptyset\right\}$ such that
$A\cup B\cup C = \left\{1,\dots,d\right\}$, and 
\begin{itemize}
\item[$(i)$] $ x_{t,i} \rightarrow x_{i}$ and
  $ \left|t x_{t,i}\right|\rightarrow \infty$, where $x_i\neq0$, as
  $t\rightarrow\infty$ for $i\in A$.
\item[$(ii)$] $x_{t,i} \rightarrow 0 $ and
  $ \left|t x_{t,i}\right|\rightarrow \infty$, as $t\rightarrow\infty$
  for $i\in B$.
\item[$(iii)$] $x_{t,i} \rightarrow 0 $ and
  $ \left|t x_{t,i}\right|\rightarrow c_i$, where $c_i \in\mathbb{R}$,
  as $t\rightarrow \infty$ for $i\in C$.
\end{itemize}

\noindent Decomposing the quadratic form in expression
\eqref{eq:MVN_logdens} based on this partition, leads to
\begin{align*}
&\bm H(t, \bm{x}_t)^\top \mathsf{A}
  \bm H(t, \bm{x}_t) \\
  & =   \bm H(t, \bm{x}_{A,t})^\top \mathsf{A}_{AA}
  \bm H(t, \bm{x}_{A,t}) +   \bm H(t, \bm{x}_{B,t})^\top \mathsf{A}_{BB}
  \bm H(t, \bm{x}_{B,t}) +    \bm H(t, \bm{x}_{C,t})^\top \mathsf{A}_{CC}
  \bm H(t, \bm{x}_{C,t}) \\
  & \hspace{0.8cm}+ 2 \bm H(t, \bm{x}_{A,t})^\top \mathsf{A}_{AB}
  \bm H(t, \bm{x}_{B,t}) + 2 \bm H(t, \bm{x}_{A,t})^\top \mathsf{A}_{AC}
  \bm H(t, \bm{x}_{C,t}) + 2 \bm H(t, \bm{x}_{B,t})^\top \mathsf{A}_{BC}
  \bm H(t, \bm{x}_{C,t}),
\end{align*}
where
$\mathsf{A}_{MM^{\prime}}=(\mathsf{A}_{ij})_{i\in M, j\in M^\prime}$
is the submatrix of $\mathsf{A}$ with rows indexed by $M$ and columns
indexed by $M^{\prime}$.\ In what follows, we derive the growth rates
of each term in this decomposition separately.\ 
By using equations \eqref{eq:asexp} and \eqref{eq:MVN_square} and the
binomial series, we obtain
\begin{align}
  \frac{1}{2}\bm H&(t, \bm{x}_{A,t})^\top \mathsf{A}_{AA}
  \bm H(t, \bm{x}_{A,t}) \\
  =& \frac{1}{2}\sum_{k=0}^\infty [\text{sgn}( \bm{x}_{A,t})\bm a_k(t,  \bm{x}_{A,t})]^\top \, \mathsf{A}_{AA} \, [\text{sgn}( \bm{x}_{A,t})\bm a_k(t,  \bm{x}_{A,t})] \label{eq:MVN_quadraticA} \\
  &+ 
  \sum_{k=0}^\infty \sum_{k'=k+1}^\infty [\text{sgn}( \bm{x}_{A,t})\bm
  a_{k}(t,  \bm{x}_{A,t})]^\top \, \mathsf{A}_{AA} \, [\text{sgn}( \bm{x}_{A,t})\bm
  a_{k'}(t,  \bm{x}_{A,t})] \nonumber
\end{align}
where
$\bm a_k(t, \bm{x}_{A,t})=(a_k(t, x_{t,i})\,:\,i=1,\dots, d)^\top$ is
a column vector with $|A|$ rows, defined in a similar manner to
$\bm a_k(t, \bm x)$ in expression \eqref{eq:MVN_quadratic}, restricted
to the subset $A$.\ 
The summand associated with $k=0$ in the first term of expression
\eqref{eq:MVN_quadraticA} is \scriptsize
\begin{align*}
    &\frac{1}{2}\left[\left\{(2 t |\bm{x}_{A,t}|)^{1/2} +\frac{\log 2}{ (2 t |\bm{x}_{A,t}|)^{1/2}}+ O(t^{-3/2})\right\}\text{sgn}(\bm{x}_{A,t})\right]^\top
  \, \mathsf{A}_{AA} \, \left[\left\{(2 t |\bm{x}_{A,t}|)^{1/2} + \frac{\log 2}{(2 t |\bm{x}_{A,t}|)^{1/2}}+O(t^{-3/2})\right\}\text{sgn}(\bm x)\right] \\
  &= t\, \left[\text{sgn}(\bm{x}_{A,t}) \vert \bm{x}_{A,t}\rvert^{1/2} +
    \text{sgn}(\bm{x}_{A,t})\frac{\log 2}{ 2 t( \lvert \bm{x}_{A,t}\rvert)^{1/2}}\right]^\top \, \mathsf{A}_{AA} \, \left[\text{sgn}(\bm{x}_{A,t}) \vert \bm{x}_{A,t}\rvert^{1/2} + \text{sgn}(\bm{x}_{A,t})\frac{\log 2}{ 2 t(
      \lvert \bm{x}_{A,t}\rvert)^{1/2}}\right] + O(t^{-1/2}) \\
      &=t \left[\text{sgn}(\bm{x}_{A,t}) \vert \bm{x}_{A,t}\rvert^{1/2}\right]^\top \mathsf{A}_{AA} \, \left[\text{sgn}(\bm{x}_{A,t})
    \vert \bm{x}_{A,t}\rvert^{1/2}\right] + 
  \left\{\text{sgn}(\bm{x}_{A,t}) \vert \bm{x}_{A,t}\rvert^{1/2}\right\}^\top
  \mathsf{A}_{AA} \, \left\{\text{sgn}(\bm{x}_{A,t}) \frac{\log2}{\lvert \bm{x}_{A,t}
      \rvert^{1/2}}\right\} + O(t^{-1/2})\\
     &= t \left\{\text{sgn}(\bm{x}_{A,t}) \vert \bm{x}_{A,t}\rvert^{1/2}\right\}^\top \mathsf{Q}_{AA} \, \left\{\text{sgn}(\bm{x}_{A,t})
    \vert \bm{x}_{A,t}\rvert^{1/2}\right\} - t \sum\limits_{i\in A} |x_{t,i}| +
  \left\{\text{sgn}(\bm{x}_{A,t}) \vert \bm{x}_{A,t}\rvert^{1/2}\right\}^\top
  \mathsf{A}_{AA} \, \left\{\text{sgn}(\bm{x}_{A,t}) \frac{\log2}{\lvert \bm{x}_{A,t}
      \rvert^{1/2}}\right\} + O(t^{-1/2}),
\end{align*}
\normalsize
and the $k$th summand in the first term of expression
\eqref{eq:MVN_quadraticA}, with $k \geq 1$, decays to zero with rate
$O\{(\log t)^{2 k}/t^{2 k - 1}\}$.\ 
Likewise, the first summand in the second term of expression
\eqref{eq:MVN_quadraticA}, associated with $k=0$ and $k^\prime =1$, is
\begin{align*}
  &\left[\{2(\log 2 + t \lvert \bm{x}_{A,t}\rvert)\}^{1/2}\text{sgn}(\bm{x}_{A,t})\right]^\top \, \mathsf{A}_{AA} \, \left[- \frac{\log\{4\pi (\log 2 +
             t \lvert \bm{x}_{A,t} \rvert)\}}{2 \{2 (\log 2
             + t \lvert \bm{x}_{A,t}\rvert)\}^{1/2}}\text{sgn}(\bm{x}_{A,t})\right]\\
  =&-\log t \left\{\text{sgn}(\bm{x}_{A,t}) \lvert \bm{x}_{A,t}\rvert^{1/2}\right\}^\top \, \mathsf{A}_{AA} \, \left\{\text{sgn}(\bm{x}_{A,t}) (2
     \lvert \bm{x}_{A,t}\rvert^{1/2})^{-1}\right\}\\
           & - \frac{1}{2} \left\{\text{sgn}(\bm{x}_{A,t}) \lvert \bm{x}_{A,t}\rvert^{1/2}\right\}^\top
             \, \mathsf{A}_{AA} \, \left\{\text{sgn}(\bm{x}_{A,t}) \lvert \bm{x}_{A,t}\rvert^{-1/2}  \log (4\pi \lvert \bm{x}_{A,t} \rvert)\right\}+O(t^{-1}).
\end{align*}
Putting this together, obtain
\begin{align*}
  \frac{1}{2}\bm H&(t, \bm{x}_{A,t})^\top \mathsf{A}_{AA}
                    \bm H(t, \bm{x}_{A,t}) \\
  =& t \left\{\text{sgn}(\bm{x}_{A,t}) \vert \bm{x}_{A,t}\rvert^{1/2}\right\}^\top \mathsf{Q}_{AA} \, \left\{\text{sgn}(\bm{x}_{A,t})
     \vert \bm{x}_{A,t}\rvert^{1/2}\right\} - t \sum\limits_{i\in A} |x_{t,i}| \\
                  &-\log t \left\{\text{sgn}(\bm{x}_{A,t}) \lvert \bm{x}_{A,t}\rvert^{1/2}\right\}^\top \, \mathsf{A}_{AA} \, \left\{\text{sgn}(\bm{x}_{A,t}) (2
                    \lvert \bm{x}_{A,t}\rvert^{1/2})^{-1}\right\}-\\
                  &+
                    \left\{\text{sgn}(\bm{x}_{A,t}) \vert \bm{x}_{A,t}\rvert^{1/2}\right\}^\top
                    \mathsf{A}_{AA} \, \left\{\text{sgn}(\bm{x}_{A,t}) \frac{\log2}{\lvert \bm{x}_{A,t}
                    \rvert^{1/2}}\right\} \\
                  &-\frac{1}{2} \left\{\text{sgn}(\bm{x}_{A,t}) \lvert \bm{x}_{A,t}\rvert^{1/2}\right\}^\top
                    \, \mathsf{A}_{AA} \, \left\{\text{sgn}(\bm{x}_{A,t}) \lvert \bm{x}_{A,t}\rvert^{-1/2}
                    \log (4\pi \lvert \bm{x}_{A,t} \rvert)\right\}
                    + O(t^{{1}/{2}})
\end{align*}
Since $\left|t x_{t,i}\right|\rightarrow\infty$ for all $i\in B$ as
$t\rightarrow\infty$, a similar analysis based the Mill's ratio
approximation yields
\begin{align*}
    \frac{1}{2}\bm H&(t, \bm{x}_{B,t})^\top \mathsf{A}_{BB}
  \bm H(t, \bm{x}_{B,t}) \\
  =& t \left\{\text{sgn}(\bm{x}_{B,t}) \vert \bm{x}_{B,t}\rvert^{1/2}\right\}^\top \mathsf{Q}_{BB} \, \left\{\text{sgn}(\bm{x}_{B,t})
    \vert \bm{x}_{B,t}\rvert^{1/2}\right\} - t \sum\limits_{i\in B} |x_{t,i}| \\
    &-\log t \left\{\text{sgn}(\bm{x}_{B,t}) \lvert \bm{x}_{B,t}\rvert^{1/2}\right\}^\top \, \mathsf{A}_{BB} \, \left\{\text{sgn}(\bm{x}_{B,t}) (2
    \lvert \bm{x}_{B,t}\rvert^{1/2})^{-1}\right\}-\\
    &+
  \left\{\text{sgn}(\bm{x}_{B,t}) \vert \bm{x}_{B,t}\rvert^{1/2}\right\}^\top
  \mathsf{A}_{BB} \, \left\{\text{sgn}(\bm{x}_{B,t}) \frac{\log2}{\lvert \bm{x}_{B,t}
      \rvert^{1/2}}\right\} \\
  &-\frac{1}{2} \left\{\text{sgn}(\bm{x}_{B,t}) \lvert \bm{x}_{B,t}\rvert^{1/2}\right\}^\top
  \, \mathsf{A}_{BB} \, \left\{\text{sgn}(\bm{x}_{B,t}) \lvert \bm{x}_{B,t}\rvert^{-1/2}
    \log (4\pi \lvert \bm{x}_{B,t} \rvert)\right\}
      + O(t^{{1}/{2}})
\end{align*}
\\
And finally note that 
$\bm H(t, \bm{x}_{C,t})^\top \mathsf{A}_{CC}\bm H(t, \bm{x}_{C,t}) \rightarrow \Phi^{-1}\{F_L(\bm{x}^C)\}^{\top}\mathsf{A}_{CC}\Phi^{-1}\{F_L(\bm{x}^C)\}$ as $t\rightarrow\infty$.
\\
For the mixed quadratic arguments of \eqref{eq:MVN_logdens}, start with the expansion using Mill's ratio.
\begin{align*}
  \bm H&(t, \bm{x}_{A,t})^\top \mathsf{A}_{AB}
  \bm H(t, \bm{x}_{B,t}) \\
  =& \sum_{k=0}^\infty [\text{sgn}( \bm{x}_{A,t})\bm a_k(t,  \bm{x}_{A,t})]^\top \, \mathsf{A}_{AB} \, [\text{sgn}( \bm{x}_{B,t})\bm a_k(t,  \bm{x}_{B,t})] \label{eq:MVN_quadratic} \\
  &+ 
  \underset{k,k^\prime \geq 0,k\neq k^\prime}{\sum\sum} [\text{sgn}( \bm{x}_{A,t})\bm
  a_{k}(t,  \bm{x}_{A,t})]^\top \, \mathsf{A}_{AB} \, [\text{sgn}( \bm{x}_{B,t})\bm
  a_{k'}(t,  \bm{x}_{B,t})] \nonumber
  \\
  =& t \left\{\text{sgn}(\bm{x}_{A,t}) \vert \bm{x}_{A,t}\rvert^{1/2}\right\}^\top \mathsf{Q}_{AB} \, \left\{\text{sgn}(\bm{x}_{B,t})
    \vert \bm{x}_{B,t}\rvert^{1/2}\right\} \\
    &- t \left\{\text{sgn}(\bm{x}_{A,t}) \vert \bm{x}_{A,t}\rvert^{1/2}\right\}^\top\, \left\{\text{sgn}(\bm{x}_{B,t})
    \vert \bm{x}_{B,t}\rvert^{1/2}\right\}\\
    & -\log t  \left\{\text{sgn}(\bm{x}_{A,t}) \vert \bm{x}_{A,t}\rvert^{-1/2}\right\}^\top \textsf{A}_{AB}\left\{\text{sgn}(\bm{x}_{B,t}) \vert \bm{x}_{B,t}\rvert^{1/2}\right\}\\
     &-\log t  \left\{\text{sgn}(\bm{x}_{A,t}) \vert \bm{x}_{A,t}\rvert^{1/2}\right\}^\top \textsf{A}_{AB}\left\{\text{sgn}(\bm{x}_{B,t}) \vert \bm{x}_{B,t}\rvert^{-1/2}\right\}\\
     &+\log 2  \left\{\text{sgn}(\bm{x}_{A,t}) \vert \bm{x}_{A,t}\rvert^{-1/2}\right\}^\top \textsf{A}_{AB}\left\{\text{sgn}(\bm{x}_{B,t}) \vert \bm{x}_{B,t}\rvert^{1/2}\right\}\\
      &+\log 2  \left\{\text{sgn}(\bm{x}_{A,t}) \vert \bm{x}_{A,t}\rvert^{1/2}\right\}^\top \textsf{A}_{AB}\left\{\text{sgn}(\bm{x}_{B,t}) \vert \bm{x}_{B,t}\rvert^{-1/2}\right\}\\
      &-\left\{\log\left(4\pi\left|\bm{x}_{A,t}\right|\right)\text{sgn}(\bm{x}_{A,t}) \vert \bm{x}_{A,t}\rvert^{-1/2}\right\}^\top \textsf{A}_{AB} \left\{\text{sgn}(\bm{x}_{B,t}) \vert \bm{x}_{B,t}\rvert^{1/2}\right\}\\
      &-\left\{\text{sgn}(\bm{x}_{A,t}) \vert \bm{x}_{A,t}\rvert^{1/2}\right\}^\top \textsf{A}_{AB} \left\{\log\left(4\pi\left|\bm{x}_{B,t}\right|\right)\text{sgn}(\bm{x}_{B,t}) \vert \bm{x}_{B,t}\rvert^{-1/2}\right\}
      + O(t^{{1}/{2}})
\end{align*}
We note that Mill's ratio is suitable here, as $\left|tx_{t,i}\right|\rightarrow\infty$ as $t\rightarrow\infty$ for $i\in A,B$.
Next, we consider the term $\bm H(t, \bm{x}_{A,t})^\top \mathsf{A}_{AC} \bm H(t, \bm{x}_{C,t})$, only applying Mill's ratio to components in $A$,
\begin{align*}
    \bm H&(t, \bm{x}_{A,t})^\top \mathsf{A}_{AC}
  \bm H(t, \bm{x}_{C,t}) \\
  =& \sum_{k=0}^\infty [\text{sgn}( \bm{x}_{A,t})\bm a_k(t,  \bm{x}_{A,t})]^\top \, \mathsf{A}_{AC}  H(t, \bm{x}_{C,t}) 
  \\
  =& \left\{\text{sgn}( \bm{x}_{A,t})(2 t |\bm{x}_{A,t}|)^{1/2} + \text{sgn}( \bm{x}_{A,t})\log 2 (2 t |\bm{x}_{A,t}|)^{-1/2} + O(t^{-3/2})\right\}^\top  \mathsf{A}_{AC}  H(t, \bm{x}_{C,t})  \\
  &+ \left\{-\text{sgn}( \bm{x}_{A,t})\frac{\log t}{(2t|\bm{x}_{A,t}|)^{1/2}}-\text{sgn}( \bm{x}_{A,t})\frac{\log (4\pi |\bm x|)}{(2t|\bm{x}_{A,t}|)^{1/2}}+ O(t^{-3/2}\log t)\right\}^\top \mathsf{A}_{AC}  H(t, \bm{x}_{C,t})\\
  =& \left\{\text{sgn}( \bm{x}_{A,t})(2 t |\bm{x}_{A,t}|)^{1/2} \right\}^\top\mathsf{A}_{AC}  H(t, \bm{x}_{C,t}) + \left\{\text{sgn}( \bm{x}_{A,t})\log 2 (2 t |\bm{x}_{A,t}|)^{-1/2}\right\}^\top\mathsf{A}_{AC}  H(t, \bm{x}_{C,t})\\
  &- t^{-1/2}\log t \left\{\text{sgn}( \bm{x}_{A,t})(2|\bm{x}_{A,t}|)^{-1/2}\right\}^\top \mathsf{A}_{AC}  H(t, \bm{x}_{C,t}) - \left\{\text{sgn}( \bm{x}_{A,t})\frac{\log (4\pi |\bm x|)}{(2t|\bm{x}_{A,t}|)^{1/2}}\right\}^\top \mathsf{A}_{AC}  H(t, \bm{x}_{C,t})\\
  &+ o(1)
\end{align*}
Similarly, we have
\begin{align*}
    \bm H&(t, \bm{x}_{B,t})^\top \mathsf{A}_{BC}
  \bm H(t, \bm{x}_{C,t}) \\
  =& \left\{\text{sgn}( \bm{x}_{B,t})(2 t |\bm{x}_{B,t}|)^{1/2} \right\}^\top\mathsf{A}_{BC}  H(t, \bm{x}_{C,t}) + \left\{\text{sgn}( \bm{x}_{B,t})\log 2 (2 t |\bm{x}_{B,t}|)^{-1/2}\right\}^\top\mathsf{A}_{BC}  H(t, \bm{x}_{C,t})\\
  &- t^{-1/2}\log t \left\{\text{sgn}( \bm{x}_{B,t})(2|\bm{x}_{B,t}|)^{-1/2}\right\}^\top \mathsf{A}_{BC}  H(t, \bm{x}_{C,t}) - \left\{\text{sgn}( \bm{x}_{B,t})\frac{\log (4\pi |\bm x|)}{(2t|\bm{x}_{B,t}|)^{1/2}}\right\}^\top \mathsf{A}_{BC}  H(t, \bm{x}_{C,t})\\
  &+ o(1)
\end{align*}
Combining all terms together leads to
\[
  -\frac{\log f(t \bm x_t)}{t} \to g(\bm x)\quad\text{whenever
    $\bm x_t \to \bm x$},
\]
where $g$ is defined in expression \eqref{eq:g_MVN_Laplace}.

\subsection{Multivariate Laplace copula}
\label{sec:MVL_supplementary}
The joint density of the $d$-variate Laplace distribution in standard
Laplace margins with positive definite precision matrix
$\QQ$ is 
\begin{align*}
  \fL(t\bm{x}) = \frac{2^{1+\frac{d}{2}}\lvert\QQ\rvert^{{1}/{2}}}{\left(2\pi\right)^{{d}/{2}}}t^v\left(\frac{1}{2^2}\bm{x}^\top \QQ\bm{x}\right)^{{v}/{2}}K_v \left\{t\left(\bm{x}^\top \QQ\bm{x}\right)^{{1}/{2}}\right\}
\end{align*}
for $t>0$, where $v={(2-d)}/{2}$ and $K_v$ is the modified Bessel
function of the second kind~\citep{kotz2001laplace}.\ Therefore, $-\log \fL(t\bm{x})=$
\begin{align*}
   =-\left({1+\frac{d}{2}}\right)\log 2 + \frac{d}{2}\log\left(2\pi\right) - \frac{1}{2}\log\left|\QQ\right| + v\log(2/t) - \frac{v}{2}\log\left(\bm{x}^\top \QQ\bm{x}\right) - \log \left[K_v \left\{t\left(\bm{x}^\top \QQ\bm{x}\right)^{{1}/{2}}\right\}\right]
\end{align*}
Asymptotically, we have $K_v(z) \sim \left(\pi/2z\right)^{{1}/{2}}e^{-z}\left(1+\OO(z^{-1})\right)$ as $z\rightarrow\infty$~\citep{gradshteyn2014table}.\ Applying the
negative logarithm, obtain
\begin{align*}
  -\log K_v(z) \sim -\frac{1}{2}\log\left(\frac{\pi}{2}\right) + \frac{1}{2}\log z + z + \OO(z^{-1}).
\end{align*}
Substituting this in the expression for $-\log \fL(t\bm{x})$, obtain 
\begin{align*}
  -\log \fL(t\bm{x}) \sim& -\left(1+\frac{d}{2}\right)\log 2 + \frac{d}{2}\log 2\pi - \frac{1}{2}\left|\QQ\right| -v\log t + v\log 2 -\frac{v}{2}\log \left(\bm{x}^\top \QQ\bm{x}\right)\\
                         &\hspace{0.3cm}-\frac{1}{2}\log\left(\frac{\pi}{2}\right) + \frac{1}{2}\log t + \frac{1}{4}\log \left(\bm{x}^\top \QQ\bm{x}\right) + t\left(\bm{x}^\top \QQ\bm{x}\right)^{{1}/{2}} + \OO(t^{-1})\\
  =& t\gG(\bm{x}) + \left(\log t\right)\qfirst(\bm{x}) + \qsecond (\bm{x}) + \OO(t^{-1}),
\end{align*}
for $t\rightarrow\infty$, where the gauge function is
$\gG(\bm{x}) = \left(\bm{x}^\top \QQ\bm{x}\right)^{{1}/{2}}$.\ The
higher order terms are given by $\qfirst(\bm{x}) = \frac{1}{2}-v$ and
$\qsecond(\bm{x}) = -\left(1+\frac{d}{2}-v\right)\log 2 +
\frac{d}{2}\log 2\pi -\frac{1}{2}\log\left|\QQ\right| -
\frac{1}{2}\log\left(\frac{\pi}{2}\right) +
\left(\frac{1}{4}-\frac{v}{2}\right)\log\left(\bm{x}^\top
  \QQ\bm{x}\right)$.



\subsection{Multivariate max-stable Logistic distribution with standard Fr{\'e}chet marginal distributions}\label{sec:ms-frechet}
For positive entries $\bm{x}\in\RR_+^d$, the joint density in standard
Fr{\'e}chet margins is
\begin{align*}
    f(\bm{x}) = \left(\sum\limits_{\pi\in\Pi}(-1)^{|\pi|}\prod\limits_{s\in\pi}V_{s}(\bm{x})\right) \exp\left\{-V(\bm{x})\right\}
\end{align*}
where
$\smash{V(\bm{x})=\left(\sum_{j=1}^d
    x_j^{-{1}/{\theta}}\right)^\theta}$ is a $-1$-homogeneous exponent
function with dependence parameter $\theta\in(0,1)$, and 
$$
    V_{s}(\bm{x}) = \left(-\theta\right)^{-|s|}\left[\prod_{k=0}^{|s|-1} (\theta-k)\right]\left(\prod_{k\in s}x_k\right)^{-\frac{1}{\theta}-1}\left(\sum_{j=1}^d
    x_j^{-{1}/{\theta}}\right)^{\theta-\left|s\right|}
$$
is the $-(\left|s\right|+1)$-homogeneous $\left|s\right|$-order
partial derivative of $V$ with respect to inputs whose indices are in
$s$.\ Using the intuition in Proposition \ref{prop:RV1} $(iii)$ and
setting $\psi(t)=t^{-(d+1)}$, we have the following convergence
\begin{align*}
    \frac{f(t\bm{x})}{t^{-(d+1)}}\rightarrow (-1)^{d}V_{\left\{1,\dots,d\right\}}(\bm{x}) = g(\bm{x})^{-(d+1)}
\end{align*}
as $t\rightarrow\infty$, where $V_{\left\{1,\dots,d\right\}}$ is the partial derivative of the exponent function with respect to all $d$ components.\ Note that in the context of Proposition \ref{prop:RV1} $(iii)$, this convergence implied that $\xi=1$.\ 
In the logistic setting, the expression of $g$ is therefore given by
$$
    g(\bm{x}) =  \left(-\theta\right)^{-d}\left[\prod_{k=0}^{d-1} (\theta-k)\right] \left(\prod_{k=1}^d x_k\right)^{-\frac{1}{\theta}-1}\left(\sum_{j=1}^d
    x_j^{-{1}/{\theta}}\right)^{\theta-d}
$$
Higher-order terms are given by
\begin{align*}
    u_k(\bm{w}) &= \frac{f(t\bm{x}) - t^{-(d+1)}g(\bm{x})^{-(d+1^{-1})} -  \sum_{i=1}^{k-1} t^{-\sum_{\ell=i+1}^{k}\sum_{j=1}^{\ell} (d+j)} u_i(\bm{x})}{t^{-\sum_{j=1}^{k+1} (d+j)} } \\ & \rightarrow \sum\limits_{\pi\in\Pi^{(d+1+k)}}(-1)^{|\pi|}\prod\limits_{s\in\pi}V_{s}(\bm{x}) \ ; \ t\rightarrow\infty
\end{align*}
for $k=1,\dots,d-1$,
where $\Pi^{(n)}\subset\Pi$ is such that for all $\pi\in\Pi^{(n)}$, $\sum_{s\in\pi}|s|=n$ for $n\in\left\{d+2,d+3,\dots,2d\right\}$.

\subsection{Multivariate inverted max-stable distributions, standard exponential margins}\label{sec:inv-ms-limit}
The class of multivariate inverted max-stable distributions is usually represented in exponential margins, and has a joint distribution function
$$
    F(\bm{x}) = e^{-\ell(\bm{x})}
$$
where $\ell$ is the 1-homogeneous \emph{stable tail dependence function}, and is defined by the $-1$-homogeneous exponent function $V$, introduced in Supplementary Material \ref{sec:ms-frechet}, through the relation $\ell(\bm{x}) = V({1}/{\bm{x}})$.\ The joint density in exponential margins is given by
$$
    f(\bm{x}) = \left(\sum\limits_{\pi\in\Pi}(-1)^{|\pi|}\prod\limits_{s\in\pi}\ell_{s}(\bm{x})\right) \exp\left\{-\ell(\bm{x})\right\}
$$
where $\ell_s(\bm{x})=V_s({1}/{\bm{x}})(-1)^{|s|}\prod_{j\in s}x_j^{-2}$ is $(1-|s|)$-homogeneous.\ Here, we use the limiting behaviour in Proposition \ref{prop:RV1} $(ii)$ to see limiting behaviour
\begin{align*}
    -\log f(t\bm{x})= tg(\bm{x}) + u(\bm{x})
\end{align*}
where the gauge function is given by $g(\bm{x})=\ell(\bm{x})$ and higher order term
$$
    u(\bm{x}) = -\log\left[(-1)^d\prod\limits_{j=1}^d\ell_{\left\{j\right\}}(\bm{x}) + o(1)\right]
$$
as $t\rightarrow\infty$.

\subsection{Multivariate max-stable and inverted max-stable copula, logistic dependence}
\label{sec:logistic_supplementary}
The joint density function in Fr\'{e}chet margins is 
\begin{align*}
    f_F(\bm{z}) = \left(\sum\limits_{\pi\in\Pi}(-1)^{|\pi|}\prod\limits_{s\in\pi}V_{s}(\bm{z})\right) \exp\left\{-V(\bm{z})\right\}
\end{align*}
where $\smash{V(\bm{z})=\left(\sum_{j=1}^d z_j^{-{1}/{\theta}}\right)^\theta}$ is a $-1$-homogeneous exponent function with dependence parameter $\theta\in(0,1)$, and $V_{s}$ is the $\left|s\right|$-order partial derivative of $V$ with respect to inputs whose indices are in $s$.\ Let $\Pi$ be the set of all partitions of the set of indices $\left\{1,\dots,d\right\}$, and let $\pi$ be the set of all partitions of an arbitraty element in $\Pi$.\ To obtain the joint density in Laplace margins, change of variables in implemented.\

Suppose $z_t(x_j) = z(tx_j)$ for $t>0$ and $j\in\left\{1,\dots,d\right\}$.\ If $x_j<0$ (or $z_t(x_j)<\left(\log 2\right)^{-1}$), then we perform the change of variables from Fr\'{e}chet to Laplace margins
\begin{align*}
    z_t(x_j) =& \left(-\log\left(\frac{1}{2}e^{tx_j}\right)\right)^{-1}=\left(-tx_j\right)^{-1}\left(1-\log 2 \left(-tx_j\right)^{-1} + O(t^{-2})\right)
\end{align*}
with derivative given by
\begin{align*}
    \frac{d}{d(tx_j)}z_t(x_j) = \left(-tx_j\right)^{-2}\left(1-2\log 2 \left(-tx_j\right)^{-1} + O(t^{-2})\right).
\end{align*}
If $x_j>0$ (or $z_t(x_j)>\left(\log 2\right)^{-1}$), then
\begin{align*}
    z_t(x_j) =& \left(-\log\left(1-\frac{1}{2}e^{-tx_j}\right)\right)^{-1}=2e^{tx_j} - \frac{1}{2} + O\left(e^{-tx_j}\right)
\end{align*}
with derivative given by
\begin{align*}
    \frac{d}{d(tx_j)}z_t(x_j) =  2e^{tx_j} + O\left(e^{-tx_j}\right).
\end{align*}
Lastly, if $x_j=0$, then $z_t(x_j)=(\log 2)^{-1}$.\ By the inverse function and chain rules,
\begin{align*}
    \frac{d}{d(tx_j)} z_t(x_j) &= \frac{f_L(tx_j)}{f_F\left(F^{-1}_F\left(F_L(tx_j)\right)\right)}=\begin{cases}
        \frac{e^{-t|x_j|}}{e^{tx_j}\left(\log\left(\frac{1}{2}e^{tx_j}\right)\right)^2}&;\;\;x_j<0\\
        \frac{\frac{1}{2}e^{-t|x_j|}}{\left(1-\frac{1}{2}e^{-tx_j}\right)\left(\log\left(1-\frac{1}{2}e^{-tx_j}\right)\right)^2}&;\;\;x_j>0\end{cases}\xrightarrow{x_j\rightarrow0} (\log2)^{-2}
\end{align*}
For a vector $\bm{x}=\left(x_1,\dots,x_d\right)^\top$, let $A,B,C\subset\left\{1,\dots,d\right\}$ be the set of indices such that $x_j$ is positive, negative, and zero for $j\in A,B,C$, respectively such that $|A|+|B|+|C|=d$.\ By change of variables, the joint density for the max-stable distribution with logistic dependence in Laplace margins is 
\begin{align*}
    f_L(t\bm{x}) =&\left|\prod\limits_{j=1}^d\frac{d}{d(tx_j)}z_t(x_j)\right|f_F\left(z(tx_1),\dots,z(tx_d)\right)\\
    =&(-1)^{d+1} \left\{\prod\limits_{\ell=0}^{d-1} \left(1-\frac{\ell}{\theta}\right)\right\} 2^{-\frac{|A|}{\theta}}t^{\left(\frac{1}{\theta}-1\right)|B| +1 -\frac{d}{\theta}}(\log 2)^{-2|C|}\left(\prod\limits_{k\in B} \left(-x_k\right)^{^{\frac{1}{\theta}-1}}\right)\left(\sum\limits_{k\in B} \left(-x_k\right)^{{1}/{\theta}}\right)^{\theta-d}\\
    &\times \exp\left\{-t\left[\frac{1}{\theta}\sum\limits_{j\in A}x_j + \left(\sum\limits_{k\in B} \left(-x_k\right)^{{1}/{\theta}}\right)^\theta \left(1 +O\left(e^{-\frac{t}{\theta}\min_{j\in A}x_j}\right)+ O(t^{-1})\right)\right]\right\}(1+o(1))
\end{align*}
Applying the negative logarithm, obtain the following expression
\begin{align*}
    -\log f_L(t\bm{x}) =& -\log\left[(-1)^{d+1} \left\{\prod\limits_{\ell=0}^{d-1} \left(1-\frac{\ell}{\theta}\right)\right\} 2^{-\frac{|A|}{\theta}}(\log 2)^{-2|C|}\right]-\left(\frac{1}{\theta}-1\right)\sum\limits_{k\in B}\log(-x_k) \\&- (\theta-d)\log\left(\sum\limits_{k\in B} \left(-x_k\right)^{{1}/{\theta}}\right)+\left\{-\left(\frac{1}{\theta}-1\right)|B| -1 +\frac{d}{\theta}\right\}\log t\\
    &+t\left[\frac{1}{\theta}\sum\limits_{j\in A}x_j + \left\{\sum\limits_{k\in B} \left(-x_k\right)^{{1}/{\theta}}\right\}^\theta \left\{1 +O\left(e^{-\frac{t}{\theta}\min_{j\in A}x_j}\right)+ O(t^{-1})\right\}\right] + o(1)\\
    =& t g(\bm{x}) + (\log t) \qfirst(\bm{x}) + \qsecond(\bm{x})
\end{align*}
where, when letting $t\rightarrow\infty$, the gauge function is
\begin{align*}
    g(\bm{x}) = \frac{1}{\theta}\sum\limits_{j\in A}x_j + \left\{\sum\limits_{k\in B} \left(-x_k\right)^{{1}/{\theta}}\right\}^\theta
\end{align*}
the higher order terms are given by
\begin{align*}
    \qfirst(\bm{x}) =-\left(\frac{1}{\theta}-1\right)|B| -1 +\frac{d}{\theta}
\end{align*}
and
\begin{align*}
    \qsecond(\bm{x}) =& -\log\left[\frac{(-1)^{d+1}2^{-\frac{|A|}{\theta}}}{(\log 2)^{-2|C|}}\left\{\prod\limits_{\ell=0}^{d-1} \left(1{-}\frac{\ell}{\theta}\right)\right\} \right]{-}\left(\frac{1}{\theta}{-}1\right)\sum\limits_{k\in B}\log(-x_k) {-} (\theta-d)\log\left\{\sum\limits_{k\in B} \left(-x_k\right)^{{1}/{\theta}}\right\}
\end{align*}
There are 2 special cases to consider:
\begin{itemize}
    \item \textbf{special case 1:} Suppose $x_j>0\;\forall\;j\in\left\{1,\dots,d\right\}$ and let $x_{(d)}=\min_{j=1,\dots,d}x_j$.\ Here, the joint log-density is 
    \begin{align*}
        -\log f_L(t\bm{x}) =& -\log\left[2^{-1} (-1)^{d} \left\{\prod\limits_{\ell=1}^{d-1} \left(1-\frac{\ell}{\theta}\right)\right\}\right] + t\left\{\frac{1}{\theta}\sum\limits_{j=1}^d x_j + \left(1-\frac{d}{\theta}\right)x_{(d)}\right\}+\\
        & +2^{-1}e^{-tx_{(d)}}\left(1+o(1)\right) + o(1)= tg(\bm{x}) + \qfirst(\bm{x}) + o(1)
    \end{align*}
    where the gauge function is 
    $$
        g(\bm{x})=\frac{1}{\theta}\sum\limits_{j=1}^d x_j + \left(1-\frac{d}{\theta}\right)\min_{k=1,\dots,d}x_k
    $$
    and the higher order term is 
    $$
        \qfirst(\bm{x}) = -\log\left[2^{-1} (-1)^{d+1} \left\{\prod\limits_{\ell=1}^{d-1} \left(1-\frac{\ell}{\theta}\right)\right\}\right],\qquad \text{as $t\rightarrow\infty$.}
    $$
    \item \textbf{special case 2:} Suppose $x_j<0\;\forall\;j\in\left\{1,\dots,d\right\}$.\ Here, we have
    \begin{align*}
        -\log f_L(t\bm{x}) =& tg(\bm{x}) + (\log t)\qfirst(\bm{x}) + \qsecond(\bm{x}) + o(1)
    \end{align*}
    where the gauge function is 
    $$
        g(\bm{x}) = \left\{\sum\limits_{j=1}^d \left(-x_j\right)^{{1}/{\theta}}\right\}^\theta
    $$
    the higher order terms are given by
    $$
        \qfirst(\bm{x}) = d-1
    $$
    and
    $$
        \qsecond(\bm{x})=-\log\left[(-1)^{d+1} \left\{\prod\limits_{\ell=0}^{d-1} \left(1-\frac{\ell}{\theta}\right)\right\}\left\{\prod\limits_{j=1}^d \left(-x_j\right)\right\}^{\frac{1}{\theta}-1}\left\{\sum\limits_{j=1}^d \left(-x_j\right)^{{1}/{\theta}}\right\}^{\theta-d}\right], \qquad \text{as $t\rightarrow\infty$.}
    $$
    
\end{itemize}
To study the case of \textbf{inverted logistic dependence}, first recall the definition of the joint distribution function in Fr\'{e}chet margins
$$
    F_F(\bm{z}) = \exp\left\{-V(\bm{z})\right\}
$$
By evaluating at the univariate Frechet quantile function, get the Logistic copula
\begin{align*}
    C_{\textit{log.}}(\qfirst,\dots,u_d) =& F_F\left(-\frac{1}{\log \qfirst},\dots,-\frac{1}{\log u_d}\right)=\exp\left\{-V\left(-\frac{1}{\log \qfirst},\dots,-\frac{1}{\log u_d}\right)\right\}
\end{align*}
where $V$ is specified to be the exponent function commonly associated with logistic dependence.\ The survival copula associated with the inverted logistic distribution is therefore given by
\begin{align*}
    \bar{C}_{\textit{inv.log.}}(\qfirst,\dots,u_d)=&C_{\textit{log.}}(1-\qfirst,\dots,1-u_d)=\exp\left\{-V\left(-\frac{1}{\log (1-\qfirst)},\dots,-\frac{1}{\log (1-u_d)}\right)\right\}
\end{align*}
So $x_j>0$ (or $u_j>{1}/{2}$) in the logistic case corresponds to $x_j<0$ (or $u_j<{1}/{2}$) in the inverted logistic case.\ This leads to the decomposition of the log-joint density in Laplace margins $-\log f_L(t\bm{x}) = tg(\bm{x}) + (\log t)\qfirst(\bm{x}) + g^\ast_0(\bm{x})$
where the gauge function is
\begin{align*}
    g(\bm{x}) =&  \left(\sum\limits_{j\in A} x_j^{{1}/{\theta}}\right)^\theta + \frac{1}{\theta}\sum\limits_{k\in B}(-x_k)
\end{align*}
the higher-order terms are given by
\begin{align*}
    \qfirst(\bm{x}) = -\left(\frac{1}{\theta}-1\right)|A| -1 +\frac{d}{\theta}
\end{align*}
and
\begin{align*}
    \qsecond(\bm{x}) =& -\log\left[(-1)^{d+1} \left\{\prod\limits_{\ell=0}^{d-1} \left(1{-}\frac{\ell}{\theta}\right)\right\} 2^{-\frac{|B|}{\theta}}(\log 2)^{-2|C|}\right]{-}\left(\frac{1}{\theta}{-}1\right)\sum\limits_{j\in A}\log(x_j)- (\theta-d)\log\left(\sum\limits_{j\in A} x_j^{{1}/{\theta}}\right)
\end{align*}
In the case where $x_j>0$ $\forall j\in\left\{1,\dots,d\right\}$ we have $-\log f_L(t\bm{x}) = tg(\bm{x}) + (\log t)\qfirst(\bm{x}) + \qsecond(\bm{x}) + o(1)$
where the gauge function is
$$
    g(\bm{x}) = \left(\sum\limits_{j=1}^d x_j^{{1}/{\theta}}\right)^\theta
$$
with higher-order terms given by
$$
    \qfirst(\bm{x}) = d-1
$$
and
$$
    \qsecond(\bm{x})=-\log\left[(-1)^{d+1} \left\{\prod\limits_{\ell=0}^{d-1} \left(1-\frac{\ell}{\theta}\right)\right\}\left(\prod\limits_{j=1}^d x_j\right)^{\frac{1}{\theta}-1}\left(\sum\limits_{j=1}^d x_j^{{1}/{\theta}}\right)^{\theta-d}\right]
$$
as $t\rightarrow\infty$.\ In the case where $x_j<0$ $\forall j\in\left\{1,\dots,d\right\}$ we have
\begin{align*}
        -\log f_L(t\bm{x}) =& tg(\bm{x}) + \qfirst(\bm{x}) + o(1)
\end{align*}
where the gauge function is 
$$
    g(\bm{x})=\frac{1}{\theta}\sum\limits_{j=1}^d (-x_j) + \left(1-\frac{d}{\theta}\right)\min_{k=1,\dots,d}\left|x_k\right|
$$
and the higher order term is 
$$
    \qfirst(\bm{x}) = -\log\left[2^{-1} (-1)^{d} \left\{\prod\limits_{\ell=1}^{d-1} \left(1-\frac{\ell}{\theta}\right)\right\}\right], \qquad \text{as $t\rightarrow\infty$}.
$$


The joint density function in Fr\'{e}chet margins is 
\begin{align*}
    f_F(\bm{z}) = \left(\sum\limits_{\pi\in\Pi}(-1)^{|\pi|}\prod\limits_{s\in\pi}V_{s}(\bm{z})\right) \exp\left\{-V(\bm{z})\right\}
\end{align*}
where $\smash{V(\bm{z})=\left(\sum_{j=1}^d z_j^{-{1}/{\theta}}\right)^\theta}$ is a $-1$-homogeneous exponent function with dependence parameter $\theta\in(0,1)$, and $V_{s}$ is the $\left|s\right|$-order partial derivative of $V$ with respect to inputs whose indices are in $s$.\ Let $\Pi$ be the set of all partitions of the set of indices $\left\{1,\dots,d\right\}$, and let $\pi$ be the set of all partitions of an arbitraty element in $\Pi$.\ To obtain the joint density in Laplace margins, change of variables in implemented.\

Suppose $z_t(x_j) = z(tx_j)$ for $t>0$ and $j\in\left\{1,\dots,d\right\}$.\ If $x_j<0$ (or $z_t(x_j)<\left(\log 2\right)^{-1}$), then we perform the change of variables from Fr\'{e}chet to Laplace margins
\begin{align*}
    z_t(x_j) =& \left(-\log\left(\frac{1}{2}e^{tx_j}\right)\right)^{-1}=\left(-tx_j\right)^{-1}\left(1-\log 2 \left(-tx_j\right)^{-1} + O(t^{-2})\right)
\end{align*}
with derivative given by
\begin{align*}
    \frac{d}{d(tx_j)}z_t(x_j) = \left(-tx_j\right)^{-2}\left(1-2\log 2 \left(-tx_j\right)^{-1} + O(t^{-2})\right).
\end{align*}
If $x_j>0$ (or $z_t(x_j)>\left(\log 2\right)^{-1}$), then
\begin{align*}
    z_t(x_j) =& \left(-\log\left(1-\frac{1}{2}e^{-tx_j}\right)\right)^{-1}=2e^{tx_j} - \frac{1}{2} + O\left(e^{-tx_j}\right)
\end{align*}
with derivative given by
\begin{align*}
    \frac{d}{d(tx_j)}z_t(x_j) =  2e^{tx_j} + O\left(e^{-tx_j}\right).
\end{align*}
Lastly, if $x_j=0$, then $z_t(x_j)=(\log 2)^{-1}$.\ By the inverse function and chain rules,
\begin{align*}
    \frac{d}{d(tx_j)} z_t(x_j) &= \frac{f_L(tx_j)}{f_F\left(F^{-1}_F\left(F_L(tx_j)\right)\right)}=\begin{cases}
        \frac{e^{-t|x_j|}}{e^{tx_j}\left(\log\left(\frac{1}{2}e^{tx_j}\right)\right)^2}&;\;\;x_j<0\\
        \frac{\frac{1}{2}e^{-t|x_j|}}{\left(1-\frac{1}{2}e^{-tx_j}\right)\left(\log\left(1-\frac{1}{2}e^{-tx_j}\right)\right)^2}&;\;\;x_j>0\end{cases}\xrightarrow{x_j\rightarrow0} (\log2)^{-2}
\end{align*}
For a vector $\bm{x}=\left(x_1,\dots,x_d\right)^\top$, let $A,B,C\subset\left\{1,\dots,d\right\}$ be the set of indices such that $x_j$ is positive, negative, and zero for $j\in A,B,C$, respectively such that $|A|+|B|+|C|=d$.\ By change of variables, the joint density for the max-stable distribution with logistic dependence in Laplace margins is 
\begin{align*}
    f_L(t\bm{x}) =&\left|\prod\limits_{j=1}^d\frac{d}{d(tx_j)}z_t(x_j)\right|f_F\left(z(tx_1),\dots,z(tx_d)\right)\\
    =&(-1)^{d+1} \left\{\prod\limits_{\ell=0}^{d-1} \left(1-\frac{\ell}{\theta}\right)\right\} 2^{-\frac{|A|}{\theta}}t^{\left(\frac{1}{\theta}-1\right)|B| +1 -\frac{d}{\theta}}(\log 2)^{-2|C|}\left(\prod\limits_{k\in B} \left(-x_k\right)^{^{\frac{1}{\theta}-1}}\right)\left(\sum\limits_{k\in B} \left(-x_k\right)^{{1}/{\theta}}\right)^{\theta-d}\\
    &\times \exp\left\{-t\left[\frac{1}{\theta}\sum\limits_{j\in A}x_j + \left(\sum\limits_{k\in B} \left(-x_k\right)^{{1}/{\theta}}\right)^\theta \left(1 +O\left(e^{-\frac{t}{\theta}\min_{j\in A}x_j}\right)+ O(t^{-1})\right)\right]\right\}(1+o(1))
\end{align*}
Applying the negative logarithm, obtain the following expression
\begin{align*}
    -\log f_L(t\bm{x}) =& -\log\left[(-1)^{d+1} \left\{\prod\limits_{\ell=0}^{d-1} \left(1-\frac{\ell}{\theta}\right)\right\} 2^{-\frac{|A|}{\theta}}(\log 2)^{-2|C|}\right]-\left(\frac{1}{\theta}-1\right)\sum\limits_{k\in B}\log(-x_k) \\&- (\theta-d)\log\left(\sum\limits_{k\in B} \left(-x_k\right)^{{1}/{\theta}}\right)+\left\{-\left(\frac{1}{\theta}-1\right)|B| -1 +\frac{d}{\theta}\right\}\log t\\
    &+t\left[\frac{1}{\theta}\sum\limits_{j\in A}x_j + \left\{\sum\limits_{k\in B} \left(-x_k\right)^{{1}/{\theta}}\right\}^\theta \left\{1 +O\left(e^{-\frac{t}{\theta}\min_{j\in A}x_j}\right)+ O(t^{-1})\right\}\right] + o(1)\\
    =& t g(\bm{x}) + (\log t) \qfirst(\bm{x}) + \qsecond(\bm{x})
\end{align*}
where, when letting $t\rightarrow\infty$, the gauge function is 
\begin{align*}
    g(\bm{x}) = \frac{1}{\theta}\sum\limits_{j\in A}x_j + \left\{\sum\limits_{k\in B} \left(-x_k\right)^{{1}/{\theta}}\right\}^\theta
\end{align*}
the higher order terms are given by
\begin{align*}
    \qfirst(\bm{x}) =-\left(\frac{1}{\theta}-1\right)|B| -1 +\frac{d}{\theta}
\end{align*}
and
\begin{align*}
    \qsecond(\bm{x}) =& -\log\left[\frac{(-1)^{d+1}2^{-\frac{|A|}{\theta}}}{(\log 2)^{-2|C|}}\left\{\prod\limits_{\ell=0}^{d-1} \left(1{-}\frac{\ell}{\theta}\right)\right\} \right]{-}\left(\frac{1}{\theta}{-}1\right)\sum\limits_{k\in B}\log(-x_k) {-} (\theta-d)\log\left\{\sum\limits_{k\in B} \left(-x_k\right)^{{1}/{\theta}}\right\}
\end{align*}
There are 2 special cases to consider:
\begin{itemize}
    \item \textbf{special case 1:} Suppose $x_j>0\;\forall\;j\in\left\{1,\dots,d\right\}$ and let $x_{(d)}=\min_{j=1,\dots,d}x_j$.\ Here, the joint log-density is 
    \begin{align*}
        -\log f_L(t\bm{x}) =& -\log\left[2^{-1} (-1)^{d} \left\{\prod\limits_{\ell=1}^{d-1} \left(1-\frac{\ell}{\theta}\right)\right\}\right] + t\left\{\frac{1}{\theta}\sum\limits_{j=1}^d x_j + \left(1-\frac{d}{\theta}\right)x_{(d)}\right\}+\\
        & +2^{-1}e^{-tx_{(d)}}\left(1+o(1)\right) + o(1)= tg(\bm{x}) + \qfirst(\bm{x}) + o(1)
    \end{align*}
    where the gauge function is 
    $$
        g(\bm{x})=\frac{1}{\theta}\sum\limits_{j=1}^d x_j + \left(1-\frac{d}{\theta}\right)\min_{k=1,\dots,d}x_k
    $$
    and the higher order term is 
    $$
        \qfirst(\bm{x}) = -\log\left[2^{-1} (-1)^{d+1} \left\{\prod\limits_{\ell=1}^{d-1} \left(1-\frac{\ell}{\theta}\right)\right\}\right],\qquad \text{as $t\rightarrow\infty$.}
    $$
    \item \textbf{special case 2:} Suppose $x_j<0\;\forall\;j\in\left\{1,\dots,d\right\}$.\ Here, we have
    \begin{align*}
        -\log f_L(t\bm{x}) =& tg(\bm{x}) + (\log t)\qfirst(\bm{x}) + \qsecond(\bm{x}) + o(1)
    \end{align*}
    where the gauge function is 
    $$
        g(\bm{x}) = \left\{\sum\limits_{j=1}^d \left(-x_j\right)^{{1}/{\theta}}\right\}^\theta
    $$
    the higher order terms are given by
    $$
        \qfirst(\bm{x}) = d-1
    $$
    and
    $$
        \qsecond(\bm{x})=-\log\left[(-1)^{d+1} \left\{\prod\limits_{\ell=0}^{d-1} \left(1-\frac{\ell}{\theta}\right)\right\}\left\{\prod\limits_{j=1}^d \left(-x_j\right)\right\}^{\frac{1}{\theta}-1}\left\{\sum\limits_{j=1}^d \left(-x_j\right)^{{1}/{\theta}}\right\}^{\theta-d}\right], \qquad \text{as $t\rightarrow\infty$.}
    $$
    
\end{itemize}
To study the case of \textbf{inverted logistic dependence}, first recall the definition of the joint distribution function in Fr\'{e}chet margins
$$
    F_F(\bm{z}) = \exp\left\{-V(\bm{z})\right\}
$$
By evaluating at the univariate Frechet quantile function, get the Logistic copula
\begin{align*}
    C_{\textit{log.}}(\qfirst,\dots,u_d) =& F_F\left(-\frac{1}{\log \qfirst},\dots,-\frac{1}{\log u_d}\right)=\exp\left\{-V\left(-\frac{1}{\log \qfirst},\dots,-\frac{1}{\log u_d}\right)\right\}
\end{align*}
where $V$ is specified to be the exponent function commonly associated with logistic dependence.\ The survival copula associated with the inverted logistic distribution is therefore given by
\begin{align*}
    \bar{C}_{\textit{inv.log.}}(\qfirst,\dots,u_d)=&C_{\textit{log.}}(1-\qfirst,\dots,1-u_d)=\exp\left\{-V\left(-\frac{1}{\log (1-\qfirst)},\dots,-\frac{1}{\log (1-u_d)}\right)\right\}
\end{align*}
So $x_j>0$ (or $u_j>{1}/{2}$) in the logistic case corresponds to $x_j<0$ (or $u_j<{1}/{2}$) in the inverted logistic case.\ This leads to the decomposition of the log-joint density in Laplace margins $-\log f_L(t\bm{x}) = tg(\bm{x}) + (\log t)\qfirst(\bm{x}) + g^\ast_0(\bm{x})$
where the gauge function is
\begin{align*}
    g(\bm{x}) =&  \left(\sum\limits_{j\in A} x_j^{{1}/{\theta}}\right)^\theta + \frac{1}{\theta}\sum\limits_{k\in B}(-x_k)
\end{align*}
the higher-order terms are given by
\begin{align*}
    \qfirst(\bm{x}) = -\left(\frac{1}{\theta}-1\right)|A| -1 +\frac{d}{\theta}
\end{align*}
and
\begin{align*}
    \qsecond(\bm{x}) =& -\log\left[(-1)^{d+1} \left\{\prod\limits_{\ell=0}^{d-1} \left(1{-}\frac{\ell}{\theta}\right)\right\} 2^{-\frac{|B|}{\theta}}(\log 2)^{-2|C|}\right]{-}\left(\frac{1}{\theta}{-}1\right)\sum\limits_{j\in A}\log(x_j)- (\theta-d)\log\left(\sum\limits_{j\in A} x_j^{{1}/{\theta}}\right)
\end{align*}
In the case where $x_j>0$ $\forall j\in\left\{1,\dots,d\right\}$ we have $-\log f_L(t\bm{x}) = tg(\bm{x}) + (\log t)\qfirst(\bm{x}) + \qsecond(\bm{x}) + o(1)$
where the gauge function is
$$
    g(\bm{x}) = \left(\sum\limits_{j=1}^d x_j^{{1}/{\theta}}\right)^\theta
$$
with higher-order terms given by
$$
    \qfirst(\bm{x}) = d-1
$$
and
$$
    \qsecond(\bm{x})=-\log\left[(-1)^{d+1} \left\{\prod\limits_{\ell=0}^{d-1} \left(1-\frac{\ell}{\theta}\right)\right\}\left(\prod\limits_{j=1}^d x_j\right)^{\frac{1}{\theta}-1}\left(\sum\limits_{j=1}^d x_j^{{1}/{\theta}}\right)^{\theta-d}\right]
$$
as $t\rightarrow\infty$.\ In the case where $x_j<0$ $\forall j\in\left\{1,\dots,d\right\}$ we have
\begin{align*}
        -\log f_L(t\bm{x}) =& tg(\bm{x}) + \qfirst(\bm{x}) + o(1)
\end{align*}
where the gauge function is 
$$
    g(\bm{x})=\frac{1}{\theta}\sum\limits_{j=1}^d (-x_j) + \left(1-\frac{d}{\theta}\right)\min_{k=1,\dots,d}\left|x_k\right|
$$
and the higher order term is 
$$
    \qfirst(\bm{x}) = -\log\left[2^{-1} (-1)^{d} \left\{\prod\limits_{\ell=1}^{d-1} \left(1-\frac{\ell}{\theta}\right)\right\}\right], \qquad \text{as $t\rightarrow\infty$}.
$$

\subsection{Multivariate Student-$t_{\nu}$ copula, Student-$t_{\nu}$ margins, $\nu>0$}
Suppose $\mathsf{Q}$ positive definite and $\nu>0$.
The joint density can be expressed as $f(\bm{x})=f_0(\gG(\bm{x}))$, 
where the homothetic function $f_0$ and gauge function is fiven by
\begin{alignat*}{2}    
    & f_0(s) &&= k_{\nu,\mathsf{Q}}(1+\nu^{-1}s^2)^{-\frac{1}{2}\left(\nu+d\right)},\\
    & \gG(\bm{x}) &&= \left(\bm{x}^\top \QQ\bm{x}\right)^{{1}/{2}},
\end{alignat*}
where $k_{\nu,\mathsf{Q}}={\Gamma\left(\frac{\nu+d}{2}\right)}/{\left\{\Gamma\left(\frac{\nu}{2}\right)\nu^{{d}/{2}}\pi^{{d}/{2}}\left|\mathsf{Q}\right|^{-{1}/{2}}\right\}}$.\ 
By the third condition of Proposition~\ref{prop:RV1}, setting $\psi(t)=k_{\nu,\mathsf{Q}}\nu^{\frac{1}{2}(\nu+d)}t^{-(\nu+d)}$, and for $t>0$ large, the higher-order term is 
\begin{align*}
    u(\bm{x})=& \frac{\frac{f(t\bm{x})}{\psi(t)} - g(\bm{x})^{-(\nu+d)}}{t^{-2}}\\
    =& t^2\left[\psi(t)^{-1}k_{\nu,\mathsf{Q}} \nu^{\frac{1}{2}(\nu+d)}t^{-(\nu+d)}\gG(\bm{x})^{-(\nu+d)}\left\{1+\nu t^{-2}\gG(\bm{x})^{-2}\right\}^{-\frac{1}{2}(\nu+d)} - \gG(\bm{x})^{-(\nu+d)}\right]\\
    =& t^2\left[\gG(\bm{x})^{-(\nu+d)}\left\{1-\frac{1}{2}(\nu+d)\nu t^{-2}\gG(\bm{x})^{-2} + O(t^{-4})\right\} - \gG(\bm{x})^{-(\nu+d)}\right]\\
    =&-\frac{1}{2}(\nu+d)\nu \gG(\bm{x})^{-2-(\nu+d)} + O(t^{-2})\\
    =&-\frac{1}{2}(\nu+d)\nu \gG(\bm{x})^{-2-(\nu+d)} + o(1)
\end{align*}

\subsection{Multivariate Student-$t_{\nu}$ copula, Student-$t_{\nu}$ margins, $\nu<0$}
Suppose $\mathsf{Q}$ positive definite and $\nu<0$ and $\nu<2-d$ \citep{paptawn13}.
For support $\left\{\bm{x}\in\mathbb{R}^d \,:\, 1+\nu^{-1}\bm{x}^\top\QQ\bm{x}>0\right\}$, the joint density is given by
\begin{alignat*}{2}    
    & f(\bm{x}) &&= k_{\nu,\mathsf{Q}}\left(1+\nu^{-1}\bm{x}^\top\QQ\bm{x}\right)^{-\frac{1}{2}\left(\nu+d\right)}
\end{alignat*}
where
$k_{\nu,\mathsf{Q}}=\pi^{-{d}/{2}}\left|\QQ\right|^{{1}/{2}}\left|\nu\right|^{1-\frac{d}{2}}\left(\left|\nu\right|-d\right)^{-1}\Gamma\left(\frac{\left|\nu\right|}{2}\right)\Gamma\left(\frac{\left|\nu\right|-d}{2}\right)^{-1}$.\
Note the homothetic form $f(\bm{x})=f_0(\gG(\bm{x}))$, where
$f_0(s)=k_{\nu,\mathsf{Q}}(1-s^2)^{-\frac{1}{2}\left(\nu+d\right)}$
for $s\in(-1,1)$.\ Therefore, we extract the gauge function
$$
    \gG(\bm{x})=\left(-\nu^{-1}\bm{x}^\top\QQ\bm{x}\right)^{{1}/{2}}
$$
Therefore, $\rG(\cdot)=\gG(\cdot)^{-1}$ is a $-1$-homogeneous radial
function, and $\rG(\bm{x})\bm{x}$ lies on the boundary of the support
(i.e., $f\left\{\rG(\bm{x})\bm{x}\right\}=0$).\ \color{black} Using
the transformation of the density function incondition in
Proposition~\ref{prop:RV1} $(i)$, we have the rate of convergence
{black}
\begin{align*}
    f\left\{\rG(\bm{x})\left(\bm{x}-t^{-1}\bm{1}\right)\right\} =& k_{\nu,\mathsf{Q}}\left[1-\rG\left\{\rG(\bm{x})\left(\bm{x}-t^{-1}\bm{1}\right)\right\}^{-2}\right]_+^{-\frac{1}{2}(\nu+d)}\\
    =& k_{\nu,\mathsf{Q}}\left[1-\rG(\bm{x})^2\rG\left(\bm{x}-t^{-1}\bm{1}\right)^{-2}\right]_+^{-\frac{1}{2}(\nu+d)}\\
    =& k_{\nu,\mathsf{Q}}\left[1-\rG(\bm{x})^2\left\{\rG(\bm{x}) -\nabla \rG(\bm{x}) t^{-1} + O(t^{-2})\right\}^{-2}\right]_+^{-\frac{1}{2}(\nu+d)}\\
    =& k_{\nu,\mathsf{Q}}t^{\frac{1}{2}(\nu+d)}\left\{-\frac{2}{\rG(\bm{x})}\nabla \rG(\bm{x}) + O(t^{-1}) \right\}\\
\end{align*}
Therefore,
\begin{align*}
    \frac{f\left\{\rG(\bm{x})\left(\bm{x}-t^{-1}\bm{1}\right)\right\}}{k_{\nu,\mathsf{Q}}t^{\frac{1}{2}(\nu+d)}} =& -\frac{2}{\rG(\bm{x})}\nabla \rG(\bm{x}) + o(1)
\end{align*}

\subsection{Multivariate Student-$t_{\nu}$ copula, $\nu>0$, Laplace margins}
\label{sec:Student-t_Lap}
The multivariate t-distribution with positive definite precision
matrix $\mathsf{Q}=\left(q_{ij}\right)_{i,j=1}^d$ and with
univariate t-distribution margins with $\nu$ degrees of freedom is
given has joint density
\begin{align*}
    f(\bm{z}) =& k_{\nu,\mathsf{Q}}\left(1+\frac{1}{\nu}\sum\limits_{j=1}^d q_{jj}z_j^2 + \frac{2}{\nu}\sum\limits_{1\leq j<k\leq d}q_{jk}z_jz_k\right)^{-\frac{1}{2}\left(\nu+d\right)}
\end{align*}
where $k_{\nu,\mathsf{Q}}={\Gamma\left(\frac{\nu+d}{2}\right)}/{\left\{\Gamma\left(\frac{\nu}{2}\right)\nu^{{d}/{2}}\pi^{{d}/{2}}\left|\mathsf{Q}\right|^{-{1}/{2}}\right\}}$.
Perform the change of variables to the standard Laplace distribution, where we take advantage of a univariate t-distribution analoge of Mill's ratio~\citep{soms1976asymptotic}.\ Let $f_{t_\nu}$ and $F_{t_\nu}$ be the density and distribution functions of the univariate t-distribution with $\nu$ degrees of freedom, respectively.\ 
\\ \\
Suppose $z=z(tx)>0$ (or $x>0$), then for $t>0$ large, the change of variables from the t-distribution to the standard Laplace distribution
\begin{align*}
    tx =& 
    -\log\left[2\left\{1-F_{t_\nu}(z(tx))\right\}\right]=
    -\log \left[\frac{2\Gamma\left(\frac{\nu+1}{2}\right)\nu^{\left(\frac{1-\nu}{2}\right)}}{\Gamma\left(\frac{\nu}{2}\right)\sqrt{\nu\pi}}\right] + \nu\log z(tx)  + O\left(z(tx)^{-2}\right).
\end{align*}
Inverting this transformation, obtain
\begin{align*}
    z(tx) =& c_\nu e^{\frac{t}{\nu}x}\left(1+O\left(e^{\frac{-2t}{\nu}}\right)\right),
\end{align*}
with partial derivative with respect to $tx$ given by
\begin{align*}
    \frac{d}{dtx}z(tx)=&\frac{c_\nu}{\nu} e^{\frac{t}{\nu}x}\left(1+O\left(e^{\frac{-2t}{\nu}}\right)\right),
\end{align*}
where $c_\nu=\left\{{2\Gamma\left(\frac{\nu+1}{2}\right)\nu^{\left(\frac{1-\nu}{2}\right)}}{\Gamma\left(\frac{\nu}{2}\right)^{-1}{(\nu\pi)^{-{1}/{2}}}}\right\}^{{1}/{\nu}}$.\ When $z(tx)<0$ (or $x<0$), the transformation to Laplace margins is 
\begin{align*}
    tx =& \log\left[2F_{t_\nu}(z(tx))\right]=\log2 +\log\left[1-F_{t_\nu}(-z(tx))\right].
\end{align*}
Negating,
\begin{align*}
    -tx =& -\log2 -\log\left(1-F_{t_\nu}(-z(tx))\right)
    = -\log \left[\frac{2\Gamma\left(\frac{\nu+1}{2}\right)\nu^{\left(\frac{1-\nu}{2}\right)}}{\Gamma\left(\frac{\nu}{2}\right)\sqrt{\nu\pi}}\right] + \nu\log (-z(tx))  + O\left(z(tx)^{-2}\right).
\end{align*}
Inverting this transformation,
\begin{align*}
    z(tx)=& -c_\nu e^{\frac{t}{\nu}\left|x\right|}\left(1+O\left(e^{\frac{-2t}{\nu}}\right)\right),
\end{align*}
with partial derivative
\begin{align*}
    \frac{d}{dtx}z(tx)=&\frac{c_\nu}{\nu} e^{\frac{t}{\nu}\left|x\right|}\left(1+O\left(e^{\frac{-2t}{\nu}}\right)\right).
\end{align*}
Therefore, by change of variables, obtain the joint density in Laplace margins
\begin{align*}
    f_L(t\bm{x}) =& \left|\prod\limits_{j=1}^d\frac{d}{d(tx_j)}z(tx_j)\right|f(z(tx_1),\dots,z(tx_d))\\
    =&\frac{c_\nu^d k_{\nu,\mathsf{Q}}}{\nu^d}\exp\left\{\frac{t}{\nu}\sum\limits_{j=1}^d \left|x_j\right|\right\}\frac{q_{j^\star j^\star}c_{\nu}^2}{\nu}\exp\left\{-t\left(1+\frac{d}{\nu}\right)\max_{j=1,\dots,d}\left|x_j\right|\right\}(1+o(1)),
\end{align*}
where $j^\star$ is the index such that $\left|x_{j^\star}\right|=\max_{j=1,\dots,d}\left|x_j\right|$.\ Taking the negative logarithm, we obtain 
\begin{align*}
    -\log f_L(t\bm{x}) =& tg(\bm{x}) + \qfirst(\bm{x}) + o(1),
\end{align*}
where the gauge function is 
\begin{align*}
    g(\bm{x}) =& -\frac{1}{\nu}\sum\limits_{j=1}^d \left|x_j\right| + \left(1+\frac{d}{\nu}\right)\max_{j=1,\dots,d}\left|x_j\right|,
\end{align*}
and the higher order term is 
\begin{align*}
    \qfirst(\bm{x}) =& -\log\left( \frac{c_\nu^{d+2} k_{\nu,\mathsf{Q}}q_{j^\star j^\star}}{\nu^{d+1}}\right).
\end{align*}

\subsection{Wishart distribution}
\label{sec:limit_Wishart}
For a $d\times d$ positive definite matrix $\bf{X}$, the Wishart
distribution with $\nu$ degrees of freedom and positive definite scale
matrix $\bf{V}$ has density function
$$
    f_{\mathbf{X}}(\mathbf{X}) = 2^{-{(\nu d)}/{2}}\det(\mathbf{V})^{-{\nu}/{2}}\Gamma_d({\nu}/{2})^{-1}\det(\mathbf{X})^{\frac{1}{2}(\nu - d - 1)}e^{-\frac{1}{2}\text{tr} \left(\mathbf{V}^{-1}\mathbf{X}\right)}
$$
Thus,
$$
    -\log f_{\mathbf{X}}(t\mathbf{X}) = tg(\mathbf{X}) + (\log t)u_1(\mathbf{X}) + u_2(\mathbf{X})
$$
where
$g(\mathbf{X})=\frac{1}{2}\text{tr}
\left(\mathbf{V}^{-1}\mathbf{X}\right)$ is 1-homogeneous, and
$u_1(\mathbf{X})=\left(-\frac{d}{2}(\nu - d - 1)\right)$ and
$u_2(\mathbf{X})=-\frac{1}{2}(\nu-d-1)\log\det(\mathbf{X}) +
\log\left(2^{-{(\nu
      d)}{2}}\det(\mathbf{V})^{-{\nu}/{2}}\Gamma_d({\nu}/{2})^{-1}\right)$
are higher-order terms.\ 

\newpage
\section{Directional densities}
\label{supp:angle}
\subsection{Multivariate normal distribution, standard normal margins}
\label{sec:angle_MVN}

The joint density can be written in the form $f(\bm{x})=f_0(\gG(\bm{x}))$, where $f(s)=(2\pi)^{-{d}/{2}}\left|\QQ\right|^{{1}/{2}}\exp\left\{-{s^{2}}/{2}\right\}$.\ By~\eqref{eq:angle-density-derivation},
\begin{align*}
    f_{\bm{W}}(\bm{w}) =& \gG(\bm{x})^{-d} \int\limits_{0}^{\infty} s^{d-1}f_0(s)ds=(2\pi)^{-{d}/{2}}\left|\QQ\right|^{{1}/{2}}\left(\bm{w}^\top\QQ\bm{w}\right)^{-{d}/{2}}\int\limits_{0}^{\infty} s^{d-1}e^{-\frac{1}{2}s^2}ds\\
    =& \Gamma({d}/{2})2^{\frac{d}{2}-1}(2\pi)^{-{d}/{2}}\left|\QQ\right|^{{1}/{2}}\left(\bm{w}^\top\QQ\bm{w}\right)^{-{d}/{2}}=\Gamma({d}/{2})2^{-1}\pi^{-{d}/{2}}\left|\QQ\right|^{{1}/{2}}\left(\bm{w}^\top\QQ\bm{w}\right)^{-{d}/{2}}
\end{align*}

\begin{figure}[h!]
\centering
\includegraphics[width=0.4\textwidth]{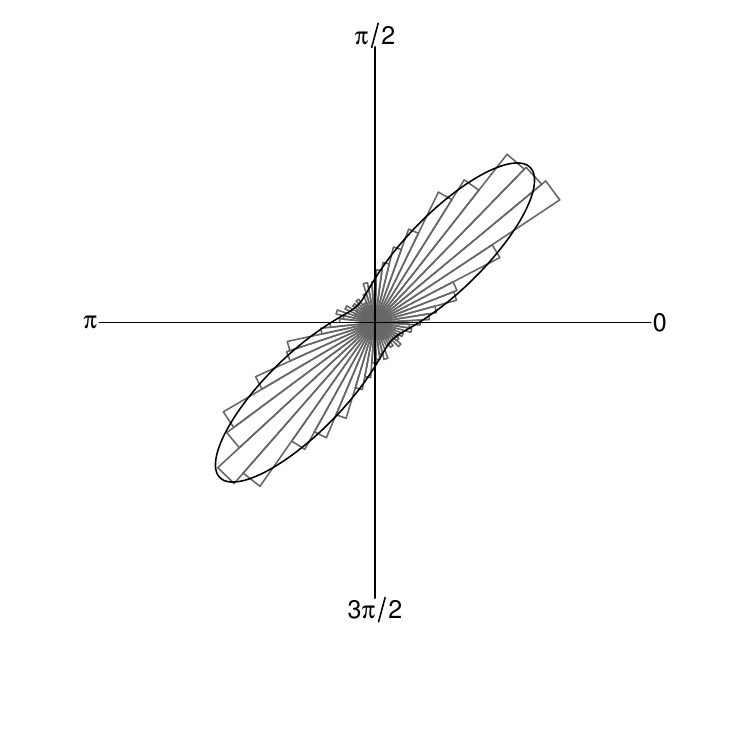}
\caption{Directional density for the bivariate Gaussian distribution, standard Gaussian margins, with $\QQ^{-1}_{11}=\QQ^{-1}_{22}=1$, $\QQ^{-1}_{12}=\QQ^{-1}_{21}=0.8$, plotted over an empirical sample}
\label{fig:supp-angle-mvn}
\end{figure}




\subsection{Multivariate Laplace distribution, standard Laplace margins}
\label{sec:angle_MVL}

We have $\fL(\bm{x})=f_0(\gG(\bm{x}))$, where $f_0(s)=|\QQ|^{{1}/{2}}(2\pi)^{-{d}/{2}}s^v K_v(s)$.\ By~\eqref{eq:angle-density-derivation}, the density of angles is 
\begin{align*}
    f_{\bm{W}}(\bm{w}) =& \int\limits_0^\infty r^{d-1}f_0(r\gG(\bm{w}))dr=\gG(\bm{w})^{-d} \int\limits_0^\infty s^{d-1}f_0(s)ds\\
    =& |\QQ|^{{1}/{2}}(2\pi)^{-{d}/{2}} \left(\bm{w}^\top\QQ\bm{w}\right)^{-{d}/{2}} \int\limits_0^\infty s^{d+v-1}K_v(s)ds=\Gamma({d}/{2})2^{\frac{d}{2}-1}(2\pi)^{-{d}/{2}}\left|\QQ\right|^{{1}/{2}}\left(\bm{w}^\top\QQ\bm{w}\right)^{-{d}/{2}}\\
    =& \Gamma({d}/{2})2^{-1}\pi^{-{d}/{2}}\left|\QQ\right|^{{1}/{2}}\left(\bm{w}^\top\QQ\bm{w}\right)^{-{d}/{2}}
\end{align*}
where the expression for the integral can be found in Sections 6.5--6.7 of~\cite{gradshteyn2014table}.

\begin{figure}[h!]
\centering
\includegraphics[width=0.4\textwidth]{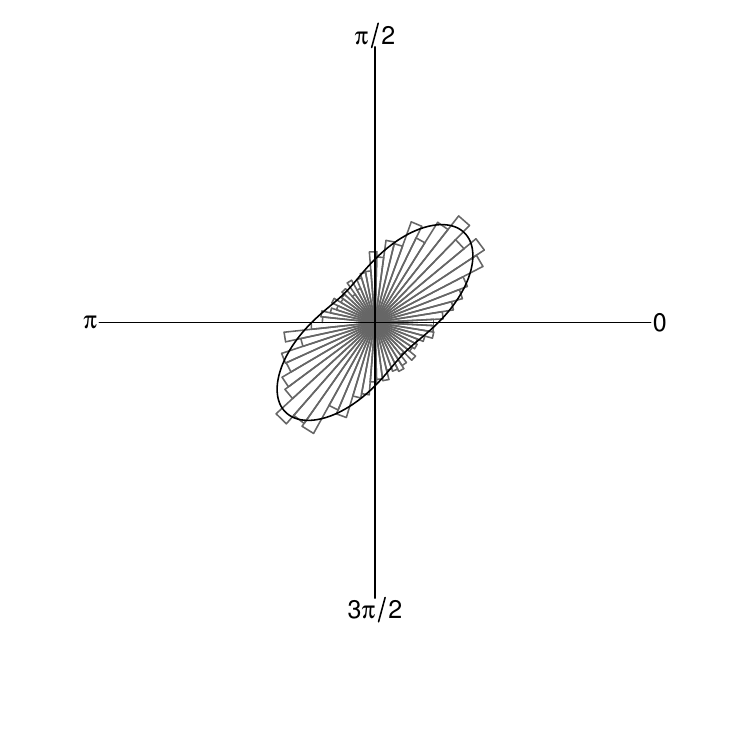}
\caption{Directional density for the bivariate Laplace distribution, standard Laplace margins, with  $\QQ^{-1}_{11}=\QQ^{-1}_{22}=1$, $\QQ^{-1}_{12}=\QQ^{-1}_{21}=0.5$, plotted over an empirical sample}
\label{fig:supp-angle-t}
\end{figure}

\subsection{Multivariate Student $t_{\nu}$ distribution, Student $t_{\nu}$ margins, $\nu>0$}
\label{sec:angle_t}

We have $f(\bm{x})=f_0(\gG(\bm{x}))$, where $f_0(s)=k_{\nu,\QQ}\left(1+\nu^{-1}s^{2}\right)^{-\frac{1}{2}\left(\nu+d\right)}$.\ By~\eqref{eq:angle-density-derivation}, the density of angles is 
\begin{align*}
    f_{\bm{W}}(\bm{w}) =& \int\limits_0^\infty r^{d-1}f_0(r\gG(\bm{w}))dr=\gG(\bm{w})^{-d} \int\limits_0^\infty s^{d-1}f_0(s)ds\\
    =& k_{\nu,\QQ}\left(\bm{w}^\top\QQ\bm{w}\right)^{-{d}/{2}}\int\limits_{0}^\infty s^{d-1}\left(1+\nu^{-1}s^{2}\right)^{-\frac{1}{2}\left(\nu+d\right)} ds\\
    =& \Gamma\left(\frac{\nu+d}{2}\right) \Gamma\left(\frac{\nu}{2}\right)^{-1}\nu^{-{d}/{2}}\pi^{-{d}/{2}}\left|\mathsf{Q}\right|^{{1}/{2}} \left(\bm{w}^\top\QQ\bm{w}\right)^{-{d}/{2}} \nu^{-1+{d}/{2}} \Gamma\left(\frac{\nu+d}{2}\right)^{-1} \Gamma\left(\frac{\nu}{2}\right)\Gamma\left(1+\frac{\nu}{2}\right)\\
    =& \Gamma\left(\frac{d}{2}\right)2^{-{1}/{2}}\pi^{-{d}/{2}}\left|\QQ\right|^{{1}/{2}}\left(\bm{w}^\top\QQ\bm{w}\right)^{-{d}/{2}}
\end{align*}

\begin{figure}[h!]
\centering
\includegraphics[width=0.4\textwidth]{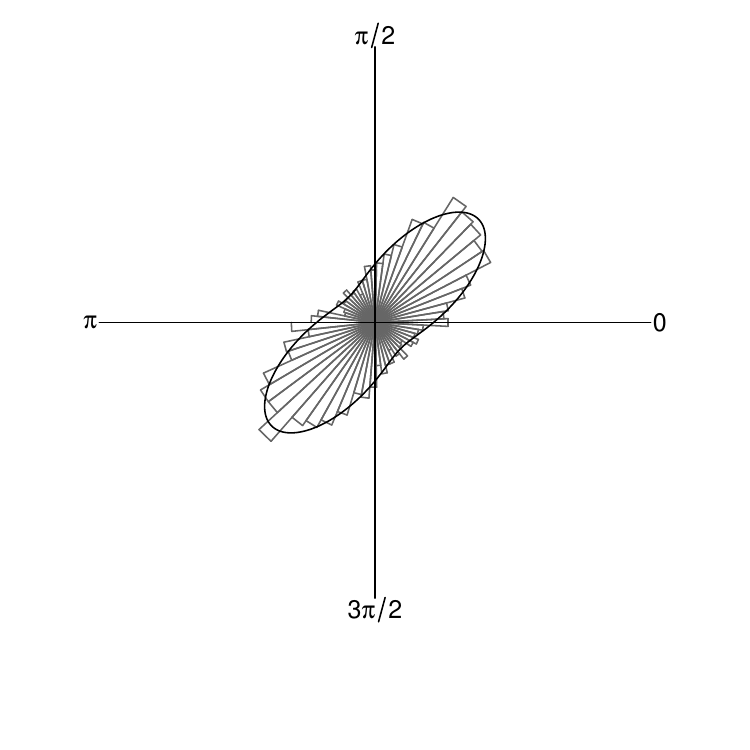}
\caption{Directional density for the bivariate Student $t_\nu$ distribution, $t_\nu$ margins, with $\nu=5$ and $\QQ^{-1}_{11}=\QQ^{-1}_{22}=1$, $\QQ^{-1}_{12}=\QQ^{-1}_{21}=0.6$, plotted over an empirical sample}
\label{fig:supp-angle-t}
\end{figure}

\subsection{Multivariate max-stable distribution, standard Fr\'{e}chet margins}
\label{sec:angle_log}

For a general $-1$-homogeneous exponent function,
\begin{align*}
    f_{\bm{W}}(\bm{w}) =& \int\limits_{0}^\infty r^{d-1}f(r\bm{w})dr\\
    =& \int\limits_0^\infty  r^{d-1}\left[\sum\limits_{\pi\in\Pi}(-1)^{|\pi|}\prod_{s\in\pi} V_s(r\bm{w})\right]e^{-V(r\bm{w})}dr\\
    =& \sum\limits_{\pi\in\Pi}(-1)^{|\pi|}\left[\prod_{s\in\pi} V_s(\bm{w})\right]\int\limits_0^\infty r^{d - \sum_{s\in\pi}(1+|s|) - 1}e^{-r^{-1} V(\bm{w})}dr\\
    =& \sum\limits_{\pi\in\Pi}(-1)^{|\pi|}\left[\prod_{s\in\pi} V_s(\bm{w})\right] \Gamma\left\{\sum_{s\in\pi}(1+|s|) - d\right\}V(\bm{w})^{\sum_{s\in\pi}(1+|s|) - d}
\end{align*}

\noindent Suppose we fix $d=2$, and we assume we have logistic dependence and Fr\'{e}chet margins, then the exponent function is given by $V(x,y)=(x^{-{1}/{\theta}} + y^{-{1}/{\theta}})^{\theta}$, and the true directional density $f_{\bm{W}}$ can be plotted against an empirical sample of angles (see Figure \ref{fig:supp-angle-ms}).

\begin{figure}[h!]
\centering
\includegraphics[width=0.4\textwidth]{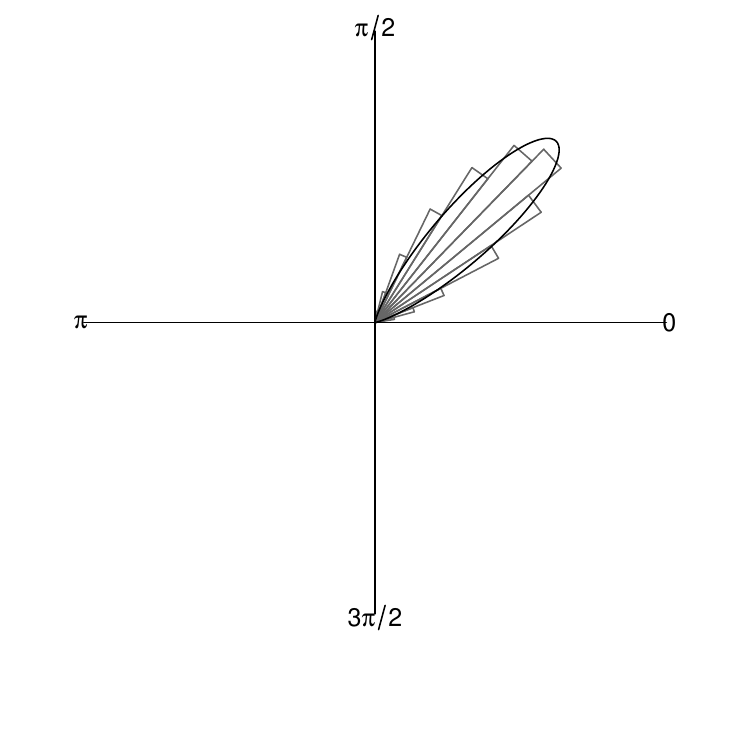}
\caption{Directional density for the bivariate max-stable logistic distribution, standard Fr\'{e}chet margins, with dependence parameter $\theta=0.3$, plotted over an empirical sample}
\label{fig:supp-angle-ms}
\end{figure}



\subsection{Multivariate inverted max-stable distributions, standard exponential margins}\label{sec:inv-ms-angle}
Recall, from \autoref{sec:inv-ms-limit} that the first order gauge function for inverted max-stable distributions in exponential margins is given by the 1-homogeneous stable tail dependence function, $g(\bm{x})=\ell(\bm{x})$.\ With this in mind, the density of angles in this setting is given by
\begin{align*}
    f_{\bm{W}}(\bm{w}) =& \int\limits_{0}^\infty r^{d-1}f(r\bm{w})dr\\
    =& \int\limits_0^\infty  r^{d-1}\left[\sum\limits_{\pi\in\Pi}(-1)^{|\pi|}\prod_{s\in\pi} \ell_s(r\bm{w})\right]e^{-\ell(r\bm{w})}dr\\
    =& \sum\limits_{\pi\in\Pi}(-1)^{|\pi|}\left[\prod_{s\in\pi} \ell_s(\bm{w})\right]\int\limits_0^\infty r^{d + \sum_{s\in\pi}(1-|s|) - 1}e^{-r \ell(\bm{w})}dr\\
    =& \sum\limits_{\pi\in\Pi}(-1)^{|\pi|}\left[\prod_{s\in\pi} \ell_s(\bm{w})\right] \Gamma\left\{d+\sum_{s\in\pi}(1-|s|) \right\}\ell(\bm{w})^{d+\sum_{s\in\pi}(1-|s|)}\\
    =& \left[\sum\limits_{\pi\in\Pi}(-1)^{|\pi|}\left[\prod_{s\in\pi} \ell_s(\bm{w})\right] \Gamma\left\{d+\sum_{s\in\pi}(1-|s|) \right\}g(\bm{w})^{\sum_{s\in\pi}(1-|s|)}\right] g(\bm{w})^d ,
\end{align*}
thus following the expression given in equation \eqref{eq:W_from_L} from Section \ref{sec:density_angles}.
Suppose we fix $d=2$, and we assume we have logistic dependence and standard exponential margins, then the exponent function is given by $\ell(x,y)=(x^{{1}/{\theta}} + y^{{1}/{\theta}})^{\theta}$, and the true directional density $f_{\bm{W}}$ can be plotted against an empirical sample of angles (see Figure \ref{fig:supp-angle-invms}).
\begin{figure}[h!]
\centering
\includegraphics[width=0.4\textwidth]{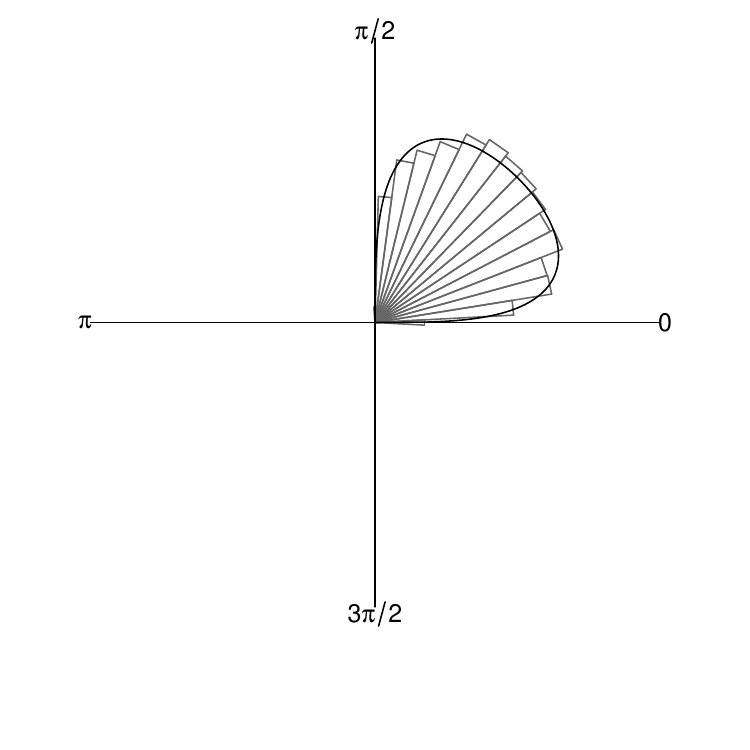}
\caption{Directional density for the bivariate inverted max-stable logistic distribution, standard exponential margins, with dependence parameter $\theta=0.7$, plotted over an empirical sample}
\label{fig:supp-angle-invms}
\end{figure}
\\
\subsection{Wishart distribution}
\label{sec:angle_Wishart}
The density of angles is computed using 
\begin{align*}
    f_{\mathbf{W}}(\mathbf{W}) =& \int\limits_{0}^\infty r^{d-1}f_{\mathbf{X}}(r\mathbf{W})dr\\
    =& 2^{-{(\nu d)}/{2}}\det(\mathbf{V})^{-{\nu}/{2}}\Gamma_d({\nu}/{2})^{-1} \det(\mathbf{W})^{\frac{1}{2}(\nu - d - 1)}\Gamma\left(\frac{d}{2}\left(1+\nu - d\right)\right)\left(\frac{1}{2}\text{tr}\left(\mathbf{V}^{-1}\mathbf{W}\right)\right)^{-\frac{d}{2}\left(1+\nu - d\right)}
\end{align*}
where in this setting, radii and angles are defined using a matrix norm; $R=\|\mathbf{X}\|$, $\mathbf{W}={\mathbf{X}}/{\|\mathbf{X}\|}$.

\section{Statistical inference}
\label{sup:Statistical_inference}
\subsection{Inference for latent variables}
\label{supp:Posterior}
Given $n$ independent
observations $\dat:=\{r_1\bm w_1,\ldots,r_n\bm w_n\}$ from
$\bm X=R\bm W\in \RR^d$, the latent quantities of interest are the set
$\G$, the quantile set $\QS_q$ and the set $\LL$.\ In what follows, we
detail the procedure to obtain realisations from the joint posterior
distribution of these quantities.\

We model the logarithms of the random radial functions $\rquant$,
${\rG}_q$ and $r_{\LL}$ as Mat\'ern (Gaussian) fields on $\SSS^{d-1}$
using the stochastic partial differential equation (SPDE) approach by
\cite{lindetal11}, using $\alpha=2$ (see their Equation~(2)) which is
also the default option in the \texttt{R-INLA} package
(\texttt{www.r-inla.org}).\ We denote the respective random intercepts
$\beta_{\QS},\beta_{\G},\beta_{\LL}\in\RR$ as well as their stochastic
weights in the finite element representation
\citep[Equation~(9)]{lindetal11} as
$\bm z_{\QS},\bm z_{\G},\bm z_{\LL}\in\RR^p$.
We denote by $\bm \theta=(\beta_{\QS},\bm z_{\QS},\beta_{\G},\bm z_{\G},\beta_{\LL},\bm z_{\LL})\in\RR^{3p+3}$ the set of all latent variables for our models, and it follows that the joint posterior distribution of $\bm \theta$ fully determines that of $\QSq$, $\G$ and $\LL$.\ Due to the hierarchical structure of all proposed models detailed in Sections~\ref{sec:QR} and~\ref{sec:likelihood}, the joint posterior
density of $\bm \theta$ factorises according to $\pi[\bm \theta \mid \dat]= \pi[\beta_{\G},\bm z_{\G}, \beta_{\LL},\bm z_{\LL} \mid \beta_{\QS},\bm z_{\QS},\dat]\pi[\beta_{\QS},\bm z_{\QS}\mid \dat]$.\
We first fit the Bayesian Gamma log-linear quantile regression model
described in Section~\ref{sec:QR} for a fixed probability $q$
to all observations $\dat$.\ Samples
$\{(\beta_{\QS,i},\bm z_{\QS,i}):i=1,\ldots,n_{\QS}\}$ from the posterior density
$\pi[\beta_{\QS},\bm z_{\QS}\mid \dat]$ map to a set of radial functions $\{\rquants{i}:i=1,\ldots,n_{\QS}\}$ which are then interpreted as
candidate radial functions of a $q$-th quantile set.\ 
Given a radial
function $\rquants{i}$, we define a set of exceedances through
\begin{equation}
\label{eq:set_Exc}
  \setExc_i = \left\{\left(r_j,\bm w_j\right)\in (0,\infty)\times\SSS^{d-1}\;:\; r_j>\rquants{i}(\bm w_j),\, r_j\bm w_j \in \RR^{d},\,j=1,\ldots,n\right\}.
\end{equation}
We then fit model $\Mh$, $\Mg$, or $\Mgh$ to each collection of
exceedances $\setExc_i$, as detailed in Section~\ref{sec:likelihood},
and also possibly include angles from non-exceedance data in the
likelihood.\ This procedure yields a conditional posterior density
$\pi[\beta_{\G},\bm z_{\G}, \beta_{\LL},\bm z_{\LL} \mid
\beta_{\QS,i},\bm z_{\QS,i},\dat]$ for each ${i=1,\ldots,n_{\QS}}$.\
Sampling $n_{\G \LL}$ realisations jointly from each
$\pi[\beta_{\G},\bm z_{\G}, \beta_{\LL},\bm z_{\LL} \mid
\beta_{\QS,i},\bm z_{\QS,i},\dat]$ provides an assembled sample of
$n_{\QS}\cdot n_{\G \LL}$ realisations from the joint posterior
distribution of $\bm \theta$,
\begin{equation}
\label{eq:samp_theta_ij}
    \left\{\bm \theta_{i,j}=\left(\beta_{\QS,i},\bm z_{\QS,i},\beta_{\G,(i,j)},\bm z_{\G,(i,j)},\beta_{\LL,(i,j)},\bm z_{\LL,(i,j)}\right)\in\bm \Theta\;:\;i=1,\ldots,n_{\QS},
    j=1,\ldots,n_{\G \LL}\right\}.
\end{equation}
For simplicity and without loss, we re-index the
sample~\eqref{eq:samp_theta_ij} from the posterior distribution of
$\bm \theta$ to $\{\bm \theta_{i}\,:\,i=1,\ldots,n_{\bm \theta}\}$
with $n_{\bm \theta}:=n_{\QS}\cdot n_{\G \LL}$ and use this notation
in the next sections.\ We shall also refer to the sampled latent
functions $\rquants{i}$, $\rGs{i}$, and $r_{\LL,i}$ constructed from
$\bm \theta_i$.\ 

Interval estimation for the latent fields $\rquant$, ${\rG}_q$, and $r_{\LL}$ is
accomplished via prediction
intervals~\citep[see][]{Bolin_excursions_2018}.\ A $(1-\alpha)$
prediction-interval for the value of a random field
$Y\,:\,\Omega \to \smash{\RR^{\SSS^{d-1}}}$ at an angle $\bm w\in\SSS^{d-1}$
is the closed line segment
$[q_\alpha(\bm w) \bm w : q_{1-\alpha}(\bm w) \bm w]$ 
where $q_{\alpha}(\bm w)$ is the $\alpha$-quantile of the distribution
of $Y(\bm w)$.\ A $(1-\alpha)$ prediction interval for the process $Y$
defined on $\SSS^{d-1}$ consists of the strip 
$R_{1-\alpha}:=\cup_{\bm w\in \SSS^{d-1}}R_{1-\alpha}(\bm w)$ defined through $R_{1-\alpha}(\bm w):=[q_\rho(\bm w) \bm w :
q_{1-\rho}(\bm w) \bm w]$
for some $\rho$ such that $q_{\rho}(\bm w)$ and $q_{1-\rho}(\bm w)$ satisfy
\[
  \PR\left[ q_\rho(\bm w) \leq Y(\bm w)\leq q_{1-\rho}(\bm w),\bm w\in\SSS^{d-1}\right] =
  1-\alpha.
\]
In the context of our latent fields $\rquant$, ${\rG}_1$, and $r_{\LL}$,
$(1-\alpha)$ prediction intervals consist of sets withing within which
the true functions lie entirely with probability $1-\alpha$.\
Prediction intervals can be obtained from a sample from the posterior
distribution of a parameter of interest using the \texttt{excursions}
package in \sfR.\
\citep{Bolin_excursions_2015,Bolin_excursions_2017}.\ \color{black}

\subsection{Rare event probability estimation}
\label{sup:Prediction}

Given $n$ independent observations
$\dat:=\{r_1\bm w_1,\ldots,r_n\bm w_n\}$ from
$\bm X=R\bm W\in \RR^d$, interest lies in the estimation of the
probability that a new draw from $R\bm W$ falls within some Borel set
$B\in\Rstar$ far enough from the origin ($\Orig\in\RR^d$).\
We provide a statistical inference framework for
${\PR_{B\mid\dat}:=\PR[R\bm W\in B\mid \dat]}$.\ In practice, common
types of sets of interest are boxes
$\{\bm x \in \RR^d: \bm a \leq \bm x \leq \bm b,\, \bm a,\bm
b\in\RR^d\}$ and sets of the form $\{r\bm w :r>h(\bm w)>0,r\in(0,\infty),\bm w \in \SSS^{d-1}\}$ for some positive function $h$ defined on $\SSS^{d-1}$, both
starshaped at $\Orig$.\ 
This is useful in our setting as it allows for an exact probability
calculation with respect to our model specification for exceedances.\ 
Simple adaptations can be made to consider more complex sets.\

For a set $B\in\Rstar$ starshaped at $\Orig$, define $S_B:=\{\bm w \in \SSS^{d-1}: r\bm w\in B, r\in(0,\infty]\} \subseteq
\SSS^{d-1}$ and consider, for any $\bm w \in S_B$, the
partition $I_{\inf}(\bm w)\cup I_B(\bm w)\cup I_{\sup} (\bm w)$ of
$(0,\infty)$, where
$I_{\inf}(\bm w):=(0,r_{B_{\inf}}(\bm w))$,
$I_B(\bm w):=[r_{B_{\inf}}(\bm w),r_{B}(\bm w)]$, and
$I_{\sup}(\bm w):=(r_{B}(\bm w),\infty)$ for the radial function
$r_{B}(\bm w)$ of $B$ and the function $r_{B_{\inf}}(\bm w)=\inf\{r>0:r\bm w \in B\}$.\ Then,
\begin{IEEEeqnarray}{rCl}
  \label{eq:P_B}
  \PR_{B\mid \dat}&=& \PR[R\in I_{B}(\bm W),\bm W \in S_B\mid \dat],
\end{IEEEeqnarray}
The posterior predictive distribution in~\eqref{eq:P_B} is given from the posterior density of $\bm \theta$ via
\begin{IEEEeqnarray}{rCl}
  \label{eq:P_B_from_posterior}
  \PR_{B\mid \dat}&=& \int_{\RR^{3p+3}}\PR[R\in I_{B}(\bm W),\bm W \in S_B\mid \bm \theta]\pi[\bm \theta\mid\dat]\,d\bm\theta,
\end{IEEEeqnarray}
where $\pi[\bm \theta\mid\dat]$ is obtained following the procedure described in Section~\ref{supp:Posterior}.\ Since $\QSq$ is a latent variable, $\rquant(\bm w)$ intersects with $I_{\inf}(\bm w)$, $I_B(\bm w)$, or $I_{\sup}(\bm w)$ with non-zero probability for all $\bm w\in S_B$.\ Hence, from total probability, $\PR_{B\mid \bm \theta}:=\PR[R\in I_{B}(\bm W),\bm W \in S_B\mid \bm \theta]$ of~\eqref{eq:P_B_from_posterior} satisfies
\begin{IEEEeqnarray}{rCl}
  \label{eq:cases}
  \PR_{B\mid \bm \theta}= \PR[R\bm W \in B, R>\rquant(\bm W)\mid \bm \theta] + \PR[R\bm W \in B, R\leq \rquant(\bm W)\mid \bm \theta],
\end{IEEEeqnarray}
and our model is best suited for sets $B$ such that the second term in the
equation~\eqref{eq:cases} is small.\ The first term in equation~\eqref{eq:cases} decomposes into
\begin{IEEEeqnarray}{rCl}
  \label{eq:post_pred_L}
  &&\PR[R\bm W \in B, R>\rquant(\bm W)\mid \bm \theta]= \PR[R\in I_B(\bm W)\mid \bm W \in S_B, R>\rquant(\bm W),\bm \theta]\times\nonumber\\
  &&\qquad \qquad \qquad \times \; \PR[\bm W \in S_B \mid
  R>\rquant(\bm W),\bm \theta]\PR[R>\rquant(\bm W)\mid \bm \theta].\ \quad
\end{IEEEeqnarray}
By the assumptions detailed in Section~\ref{sec:QR}, the last term $\PR[R>\rquant(\bm W)\mid \bm \theta]$ in
equation~\eqref{eq:post_pred_L} equals $1-q$ for all $\bm \theta\sim\pi[\bm \theta\mid\dat]$.\ 
Following our model formulation for $\bm W$, the second term in
equation~\eqref{eq:post_pred_L} corresponds to the predictive
distribution of the angles given $\bm \theta$ and is given by
\begin{IEEEeqnarray}{rCl}
    \label{eq:post_pred_angles}
  \PR[\bm W \in S_B \mid R>\rquant(\bm W),\bm \theta]&=&
  \int_{S_B}\widetilde{f}_{\bm W\mid \bm\theta}(\bm w\mid \bm\theta)d\bm w\bigg/\int_{\SSS^{d-1}}\widetilde{f}_{\bm W\mid \bm\theta}(\bm w\mid \bm\theta)d\bm w,\quad
\end{IEEEeqnarray}
where $\widetilde{f}_{\bm W}$ is chosen from $\Mh$, $\Mg$, or $\Mgh$
(see Section~\ref{sec:density_angles}).\ The integrals in~\eqref{eq:post_pred_angles} are computed
efficiently via numerical integration for $d=2$ or $3$.\ The first
term in~\eqref{eq:post_pred_L} is obtained through
\begin{IEEEeqnarray}{rCl}
  \label{eq:Expec_angles}
  &&\PR[R\in I_B(\bm W)\mid \bm W \in S_B, R>\rquant(\bm W),\bm \theta]=\EE_{\bm W\mid \bm W\in S_B,\dat}(\PR[R\in I_B(\bm W)\mid R>\rquant(\bm W),\bm \theta]),\quad~~ 
\end{IEEEeqnarray}
where $\PR[R\in I_B(\bm w)\mid \bm W = \bm w, R>\rquant(\bm w),\bm \theta]$, needed for the expectation calculation in equation~\eqref{eq:Expec_angles}, equals
\begin{IEEEeqnarray}{rCl}
  \label{eq:exact_prob}
  &&F_{R_E\mid \bm W}\left[\frac{\max\left\{r_{B}(\bm w),\rquant(\bm w)\right\}-\rquant(\bm w)}{{\rG}_q(\bm w)}\bigg\lvert \bm w\right]-F_{R_E\mid \bm W}\left[\frac{\max\left\{r_{B_{\inf}}(\bm w),\rquant(\bm w)\right\}-\rquant(\bm w)}{{\rG}_q(\bm w)}\bigg\lvert \bm w\right],\quad
\end{IEEEeqnarray}
with 
$F_{R_E\mid \bm W}(z\mid \bm w) = 1-[1+\xi(\bm w) z]_+^{-1/\xi(\bm w)}$.\
We approximate 
expectation~\eqref{eq:Expec_angles} through Monte-Carlo integration by
sampling angles
$\{\widetilde{\bm w}_1,\ldots,\widetilde{\bm w}_{n_w}\}$ from
$f_{\bm W}$ restricted to the set $S_B$ via
\begin{IEEEeqnarray}{rCl}
  \label{eq:MC_angles}
  n_{w}^{-1}\sum_{i=1}^{n_w}\PR[R\in I_B(\widetilde{\bm w}_i)\mid R>\rquant(\widetilde{\bm w}_i),\bm \theta] &\to& \PR[R\in I_B(\bm W)\mid \bm W \in S_B, R>\rquant(\bm W),\bm \theta],
\end{IEEEeqnarray}
as $n_w\to\infty$.\ The second term in equation~\eqref{eq:cases},
$\PR[R\bm W \in B, R\leq \rquant(\bm W)\mid \bm \theta]=: p$,
corresponds to a probability over a region in which our model is not
specified.\ We hence consider it as the success probability of the binomial random variable describing the number of observations
falling in the set $\QS_q \cap B$.\ For each
$i\in\{1,\ldots,n_{\bm \theta}\}$, we impose a binomial likelihood and
beta prior for the parameter $p$, leading to a closed form beta
posterior distribution.\ When $\QS_q \cap B\neq \varnothing$, we
sample a probability from the posterior distribution of $p$ and add it
accordingly in equation~\eqref{eq:cases}.\

In practice, inference for $\PR_{B\mid \dat}$ is performed by first collecting a sample $\{\bm \theta_1,\ldots,\bm \theta_{n_{\bm\theta}}\}$ from $\pi[\bm \theta\mid\dat]$ through the procedure described in Section~\ref{supp:Posterior}, and then calculating $\PR_{B\mid \bm \theta_i}$ for each $i\in\{1,\ldots,n_{\bm \theta}\}$.\ \textit{A posteriori} mean estimates corresponding to the Monte-Carlo integration
\begin{IEEEeqnarray}{rCl}
  \label{eq:P_B_MonteCarlo}
  n_{\bm \theta}^{-1}\sum_{i=1}^{n_{\bm \theta}}\PR_{B\mid \bm \theta_i}&\to& \PR_{B\mid \dat}, \quad \text{as } n_{\bm \theta} \to \infty,
\end{IEEEeqnarray}
approximating equation~\eqref{eq:P_B_from_posterior}.\ Equal-tailed credible intervals from the sample $\{\PR_{B\mid \bm\theta_i}:i=1,\ldots,n_{\theta}\}$ hence complete the inference procedure for $\PR_{B\mid \dat}$.\

\subsection{Model selection and validation}
\label{sup:Validation}
Our model formulation gives rise to various modelling choices needing
to be assessed and validated.\ For instance, as discussed in
Section~\ref{sec:likelihood}, selecting the most appropriate model
within the set of candidate models $\Mh$, $\Mg$, and $\Mgh$ given
observed data amounts to analysing the properties of the distribution
of $R,\bm W\mid R>\rquant(\bm W)$ with respect to the association
between $\gG$ and $f_{\bm W}$.\ Another modelling choice requiring
assessment is that of the set of angles contributing to
likelihood~\eqref{eq:likelihood-v2} since it translates into a
bias-variance trade-off for $\G$ and $\LL$.\ Further, hyper-parameters
values imposed on the prior distributions of $\QSq$, $\G$, and $\LL$
imply \textit{a priori} information on the differentiability
properties of their boundaries (see Section~\ref{supp:Posterior}), and
their impact on the posterior may need assessment.\ Finally, the usual
concern of the sensitivity of the posterior distribution of the latent
variables to the return period for the latent threshold function
$\rquant$ (Section~\ref{sec:QR}) remains.\ We detail possible methods
for model selection and validation below.\

The nested structure of model $\Mg$
within the parameter space of model $\Mgh$ raises the question of
whether model $\Mg$ can serve as a sensible simplification of model
$\Mgh$ given observed data
$\dat = \{r_1\bm w_1,\ldots, r_n\bm w_n\}$.\ As discussed in
Section~\ref{sec:likelihood}, evidence for a constant $r_{\B}$ in model $\Mgh$ points to evidence that $\Mg$ can serve as an
sensible simplification of $\Mgh$, or equivalently that the observed
data $\dat$ may come from a process following a homothetic density with respect to its gauge function $\gG$.\
Given a sample of random functions $\{\psi(\bm w)^\top\bm z_{\B,i}:i=1,\ldots,n_{\theta}\}$
from the posterior distribution of the latent random field
$\psi(\bm w)^\top\bm z_{\B}$, we obtain a $0.95$
prediction interval, $R_{0.95}$, for $\psi(\bm W)^\top\bm z_{\B}$ through the procedure described in Section~\ref{supp:Posterior}.\ 
We say there is evidence against a homothetic underlying desnity whenever there is some non-empty subset $S_0\subset\SSS^{d-1}$ such that $0\notin R_{0.95}(\bm w)$ for $\bm w \in S_0$.\

The sensitivity of the posterior distributions of the latent
parameters to the remaining modelling choices discussed in the
beginning of this section can be assessed through the quality of the
calibration of the posterior predictive distribution to the observed
data.\ Given a set of observed exceedances $\setExc_i$ of a sampled
latent function $\rquants{i}$ (see expression~\eqref{eq:set_Exc}), we
wish to assess the calibration of the predictive distribution for the
excess variable $R-\rquants{i}(\bm W)\mid R>\rquants{i}(\bm W)$.\ We thus
compare each observed excess $r_j-\rquants{i}(\bm w_j)$ with the
quantile of the Exponential$({\rG}_q(\bm w_j))$, with
$\log{\rG}_q(\bm w)=\beta_{\G} + \bm \psi(\bm w)^\top \bm z_{\G}$ and
$(\beta_{\G},\bm z_{\G})\sim\pi[\beta_{\G}, \bm
z_{\G}\mid\beta_{\QS,i},\bm z_{\QS,i},\dat]$.\ We draw exceedance
radii $\{\widetilde r_1,\ldots,\widetilde r_{n_s}\}$ from their
predictive distribution Exp$({\rG}_q(\bm w_j))$, and define the empirical
distribution function
$F_{\bm w_j,n_s}(r):={n_s}^{-1}\sum_{k=1}^{n_s}\mathbbm{1}[\widetilde
r_k\leq r]$, for $r\in(0,\infty)$.\ A probability-probability (PP)
plot for the exceedances $\setExc_i$ is then given by
\begin{IEEEeqnarray}{rCl}
\label{eq:pp_exc}
  \left\{\left(\frac{j}{\lvert \setExc_i \rvert +1} ,p_{(j)}\right):j=1,\ldots,\lvert\setExc_i\vert\right\},
\end{IEEEeqnarray}
where $p_{(j)}$ denotes the $j$-th order statistic of the sample $\{F_{\bm w_j,n_s}(r_j-\rquants{i}(\bm w_j)):j=1,\ldots,\lvert\setExc_i\vert\}$.\ A quantile-quantile (QQ) plot in unit exponential margins for the exceedances $\setExc_i$ is then easily obtained via
probability-probability (PP) plot for the exceedances $\setExc_i$ via
\begin{IEEEeqnarray}{rCl}
\label{eq:qq_exc}
  \left\{\left(-\log\left(1-\frac{j}{\lvert \setExc_i \rvert +1}\right) ,-\log\left(1-p_{(j)}\right)\right):j=1,\ldots,\lvert\setExc_i\vert\right\}.
\end{IEEEeqnarray}
In a similar fashion, we can obtain PP plots for the predictive
distribution of $\bm W\mid R>\rquants{i}(\bm W)$.\ This is achieved by
sampling angles
$\{\widetilde{\bm w}_1,\ldots,\widetilde{\bm w}_{n_s}\}$ from the
predictive density
$f_{\bm W\mid \beta_{\G},\bm z_{\G}, \beta_{\B},\bm z_{\B},\dat}$
where
$(\beta_{\G},\bm z_{\G}, \beta_{\B},\bm z_{\B})\sim\pi[\beta_{\G},\bm
z_{\G}, \beta_{\B},\bm z_{\B} \mid \beta_{\QS,i},\bm
z_{\QS,i},\dat]$.\ We transform the sampled angles and the observed
angles from $\setExc$ to spherical coordinates, respectively denoted
$\{\widetilde{\bm
  \varphi}_j=(\widetilde{\varphi}_{1,j},\ldots,\widetilde{\varphi}_{d-1,j})\in\Phi:j=1,\ldots,n_s\}$
and
$\{\bm
\varphi_j=(\varphi_{1,j},\ldots,\varphi_{d-1,j})\in\Phi:j=1,\ldots,\lvert\setExc_i\rvert\}$,
and consider the empirical distribution function
$F_{n_s}(\bm
\varphi):={n_s}^{-1}\sum_{k=1}^{n_s}\mathbbm{1}[\widetilde{\bm
  \varphi}_k\leq \bm \varphi]$.\ A PP plot for the calibration of the
predictive distribution of angles to the observed angles in
$\setExc_i$ is then given by expression~\eqref{eq:pp_exc} with
$p_{(j)}$ denoting the $j$-th order statistic of the sample
$\{F_{n_s}(\bm \varphi_j):j=1,\ldots,\lvert\setExc_i\vert\}$.\

\newpage
\section{Simulation study}
\label{sec:sim_study}

\subsection{Setting}
\label{sec:Setting}
We here describe the setting of a simulation study performed using
samples from classical parametric bivariate copulas.\ We provide a
brief summary of the properties of these
distributions.\ 

A random variable $\bm X_L$ is said to follow a bivariate Laplace
distribution with precision matrix $\mathsf{Q}=\mathsf{\Sigma}^{-1}$
for a positive definite covariance $\mathsf{\Sigma}$---with unit
diagonal terms and a correlation value $\rho\in[-1,1]$ on the
remaining off-diagonal terms
The parameter $\rho$ hence fully determines a member of the family of
the bivariate normal distribution as well as its extremal properties
and in the subsequent sections, we use $\rho=0.5$.\

A random vector $\bm X\in \RR^2$ has a bivariate normal distribution
with Laplace marginal distributions and precision matrix
$\mathsf{Q}=\Sigma^{-1}$ for a positive definite covariance
$\mathsf{\Sigma}$---with unit diagonal terms and a correlation value
$\rho\in[-1,1]$ on the remaining off-diagonal
terms
.\ In what follows, we use the bivariate normal distribution as an
instance of asymptotically independent distribution and draw samples
using the value $\rho=0.8$.\

A random vector $\bm X\in \RR^2$ has a bivariate max-stable logistic
distribution with Laplace margins 
parametrised by an exponent function
$V(x_1,\dots,x_d)=\left(\sum_{j=1}^d
  x_j^{-{1}/{\theta}}\right)^\theta$ 
where the dependence parameter $\theta\in(0,1]$.\ The max-stable
logistic distribution belongs the class of asymptotic dependent
distributions provided $\theta \neq 1$.\ For any $\theta\in(0,1)$,
$\alpha_{\mid i}=1$ and $\eta=1$.\ We carry the simulation study using
$\theta=0.3$.\ Table~\ref{tbl:true_extr_coeff} presents, for each
family of distributions, the values of extremal coefficients
associated with fixed parameters.\
\begin{table}[htbp!]
  \centering
  \begin{tabular}{llll}
    \toprule
    Extremal coefficient & Laplace  & Normal & Max-stable logistic \\
    \toprule
    $\alpha_{\mid i}$& $\rho=0.5$ & $\rho^2=0.64$ & 1 \\
    $\eta$& $\sqrt{(1-\rho^2)/(2-2\rho)}\approx 0.866$ & $(1+\rho)/2=0.9$ & 1\\
    \toprule
  \end{tabular}
  \label{tbl:true_extr_coeff}
  \caption{Values of the extremal coefficients $\alpha_{\mid i}$ and $\eta$ for the bivariate Laplace, normal, and max-stable logistic distributions with specified parameters.\ \label{tbl:true_extr_coeff}}
\end{table}

\subsection{Posterior distribution of latent variables}
\label{sec:sim_study_posterior}
We perform a simulation study to obtain realisations from the
posterior distribution of the latent sets $\QS_q$,
$\LL$, and $\G$, as well as from the posterior predictive distribution
of $R$ and $\bm W$.\ We centre the simulation study on the
bivariate Laplace, normal, and max-stable logistic families of
distributions.\ The bivariate Laplace distribution is chosen for the
homothetic structure of its density with respect to its gauge function $\gG$, implying
that its marginal density of the angles is exactly that assumed by model
$\Mg$.\ The bivariate max-stable logistic and normal distributions are
chosen respectively for their asymptotic dependence and independence properties.\ From each of these distributions, we
generate $100$ samples of size $n_1=1000$ and $100$ samples of size
$n_2=10000$, resulting in six cases (600 samples) with which we carry
our analyses.\
\begin{figure}[htbp!]
  \centering
  \includegraphics[width=\textwidth]{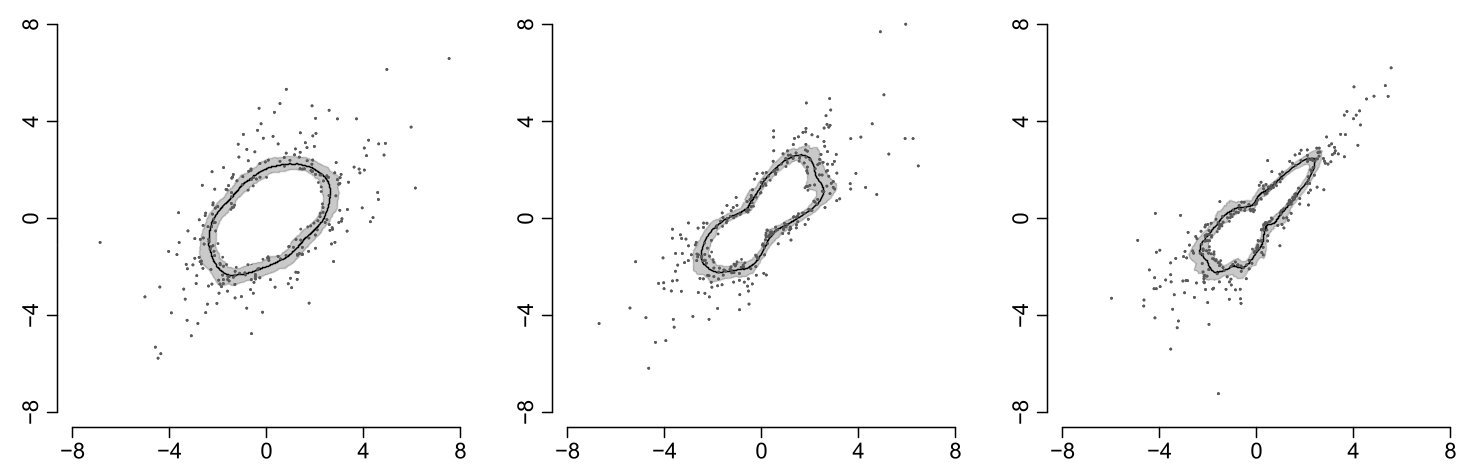}
  \caption{Angle-wise posterior mean and $0.95$ prediction intervals for
    $100$ samples from the posterior distribution of
    $r_{\QS_{0.8}}$ obtained from $n_1$ draws from a
    bivariate Laplace (left), normal (centre), max-stable logistic
    (right) distributions in Laplace margins.\ Points correspond to
    exceedances of at least one sampled posterior
    $r_{\QS_{0.8}}$.\ \label{fig:r_0q}}
  \includegraphics[width=\textwidth]{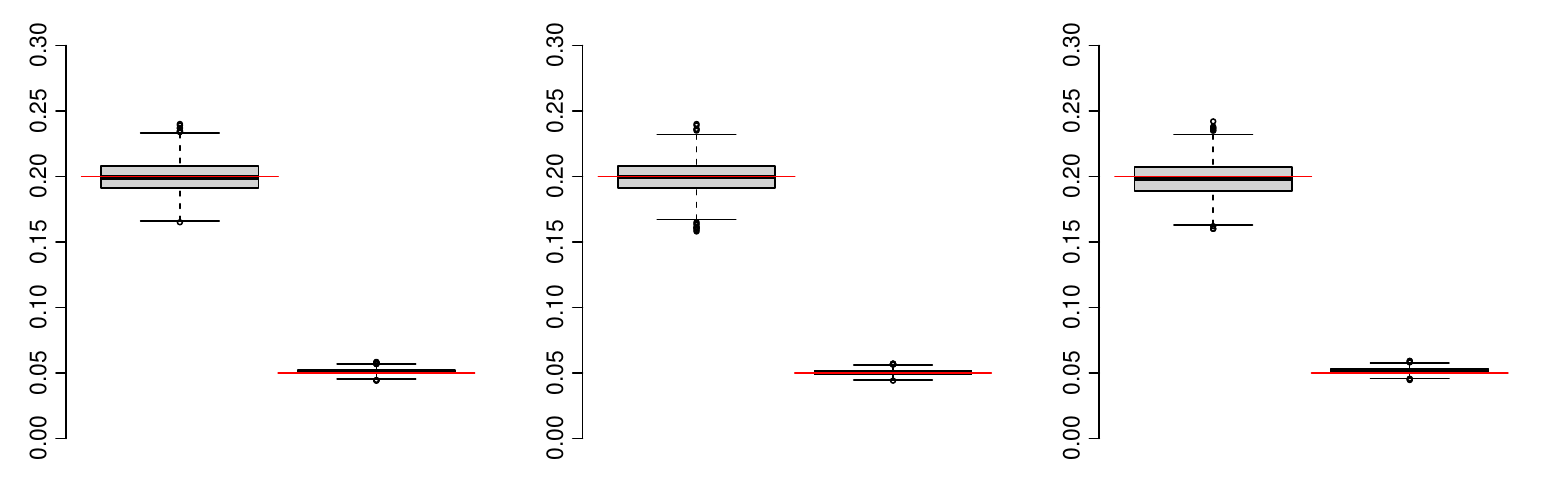}
  \caption{Boxplots of the proportion of data assigned as
    exceedances by each realisation from the posterior distribution
    of $\rquant$ for $n$ draws from the bivariate Laplace
    (left), normal (centre), max-stable logistic (right)
    distributions in Laplace margins.\ Within panels:\ $n_1$,
    $1-q_1=0.2$ (left), $n_2$, $1-q_2=0.05$ (right).\ Red segments:
    values of $1-q_k$ used in the fitting
    procedure.\ \label{fig:QR_prop_excesses}}
\end{figure}

Following the posterior sampling procedure described in
Section~\ref{supp:Posterior}, we first fit a bayesian Gamma quantile
regression to each of the $600$ samples described above using
$q_1=0.8$ for each sample of size $n_1$ and $q_2=0.98$ for each
samples of size $n_2$.\ Here, $q_1$ and $q_2$ are chosen such that the
expected number of exceedances of $\QS_{q_k}$,
$n_k\cdot\PR[R\bm W \in \QS_{q_k}^\prime\mid \dat]$, equals $200$ for
${k=\{1,2\}}$.\ Again for each sample, we obtain $20$ realisations $\{r_{\QS_{{q_k},i}}:i=1,\ldots,20\}$ 
from the posterior distribution of $r_{\QS_{q_k}}$ and define $20$
sets of exceedances accordingly.\ Figure~\ref{fig:r_0q} displays the
angle-wise posterior mean and $0.95$ prediction intervals
for $r_{\QS_{0.8}}$ resulting from fitting quantile regressions to
one sample of $n_1$ draws from each of the three distributions of
interest.\ Figure~\ref{fig:r_0q} also displays the data points defined
as exceedances for at least one realisation from the posterior
distribution of $r_{\QS_{0.8}}$.\ A brief assessment of the
performance of the quantile regression is made through
Figure~\ref{fig:QR_prop_excesses}, which consists of boxplots of the
number of exceedances of each realisation from the posterior
distribution of $r_{\QS_{q_k}}$ for each of the $600$ samples
described above grouped by family of distributions and sample size.\
\begin{figure}[htbp!]
  \centering
  \includegraphics[width=\textwidth]{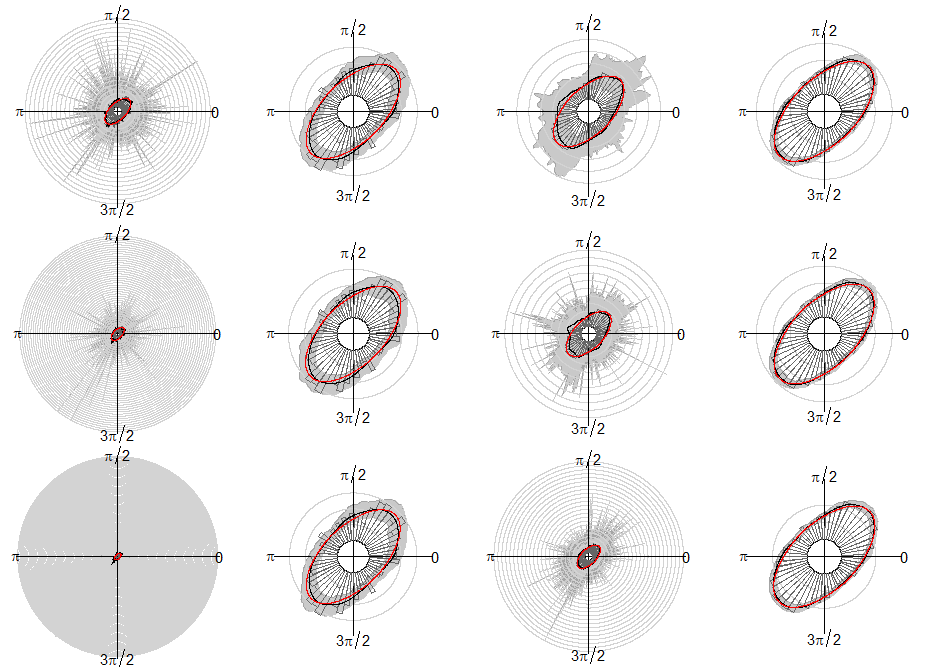}
  \caption{ Histogram of the angles $\{\varphi_j\in[0,2\pi]:\bm w_j = (\cos\varphi_j,\sin\varphi_j), j=1,\ldots,n_k\}$
    from a sample of $n$ observations from a bivariate Laplace
    distribution with angle-wise posterior mean of $\pi[\bm W \mid \G,\LL,\dat]$ (black line),
    $0.95$ prediction interval for
    $\pi[\bm W \mid \G,\LL,\dat]$ (grey band) output from model
    $\mathcal{M}_i$, and the true density $f_{\bm W}$ for
    the bivariate Laplace (red line).\ Rows from top to bottom:
    $\Mh$, $\Mg$, and $\Mgh$.\ First column:\ $\pi[\bm W \mid \G,\LL,\dat]$ fitted using exceedances of
    $\QS_{0.8}$ only (sample size $n_1$).\ Second column:\ 
    $\pi[\bm W \mid \G,\LL,\dat]$ fitted using all sampled
    angles (sample size $n_1$).\ Third column:\ $\pi[\bm W \mid \G,\LL,\dat]$
    fitted using exceedances of $\QS_{0.98}$ only (sample size $n_2$).\ Fourth column:\ $\pi[\bm W \mid \G,\LL,\dat]$ fitted using
    all sampled angles (sample size $n_2$).\ \label{fig:MVL_f_W}}
\end{figure}

Given the satisfying properties of the posterior distribution of
$r_{\QS_{q_k}}$, we proceed with the investigation of the joint model
for $R$ and $\bm W$.\ Following the procedure described in
Section~\ref{supp:Posterior}, we obtain realisations from the
posterior distributions of $r_{\G}$ and $r_{\LL}$---say
$\{r_{\G_i}:i=1,\ldots,n_\theta\}$ and
$\{r_{\LL_i}:i=1,\ldots,n_\theta\}$ for
$n_\theta=n_{\QS}n_{\G\LL}=20\cdot50$---using each model $\Mh$, $\Mg$,
and $\Mgh$ on each sample described above.\ Due to
expression~\eqref{cor:equality_of_W_distributions}, we consider the
two cases of including only the angles of exceedances of the sampled
quantile sets $\QS_{q_k,i}$ and including all observed angles
$(\bm w_1,\ldots,\bm w_n)$ in the model likelihood.\
Figure~\ref{fig:MVL_f_W} provides a comparison of the posterior
predictive density of angles $\pi[\bm W \mid \G,\LL,\dat]$ output by
models $\Mh$, $\Mg$, and $\Mgh$ fitted on samples of sizes $n_1$ and
$n_2$ from a multivariate Laplace distribution in Laplace margins.\ As
discussed in Section~\ref{sec:density_angles}, the marginal density of
angles for the multivariate Laplace distribution is known exactly
given the homothetic structure of its density with respect to its
gauge function $\gG$.\ This allows for validation of the performance
of $\Mh$, $\Mg$, and $\Mgh$.\ It is clear from
Figure~\ref{fig:MVL_f_W} that all three models recover the true or the
empirical marginal density of angles of the multivariate Laplace
distribution well.\ As expected, including all angles in the fitting
procedure and increasing the sample size both lead to reduction in the
uncertainty of the posterior predictive density
$\pi[\bm W \mid \G,\LL,\dat]$.\
\begin{figure}[htbp!]
  \centering
  \includegraphics[width=\textwidth]{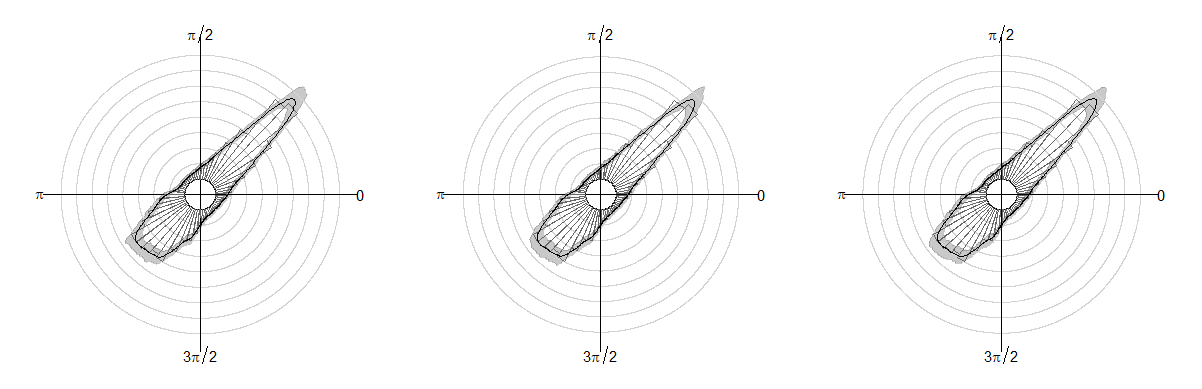}
  \caption{Histogram of the angles
    $\{\varphi_j\in[0,2\pi]:\bm w_j = (\cos\varphi_j,\sin\varphi_j), j=1,\ldots,n_2\}$ of $n_2$ draws from a
    bivariate max-stable logistic distribution, with angle-wise
    posterior mean (black line) and $0.95$ prediction intervals
    (grey) for the marginal distribution of angles
    $\pi[\bm W \mid \G,\LL,\dat]$.\ From left to right:\ outputs of
     $\Mh$, $\Mg$, and $\Mgh$.\ \label{fig:f_W_Log}}
  \includegraphics[width=\textwidth]{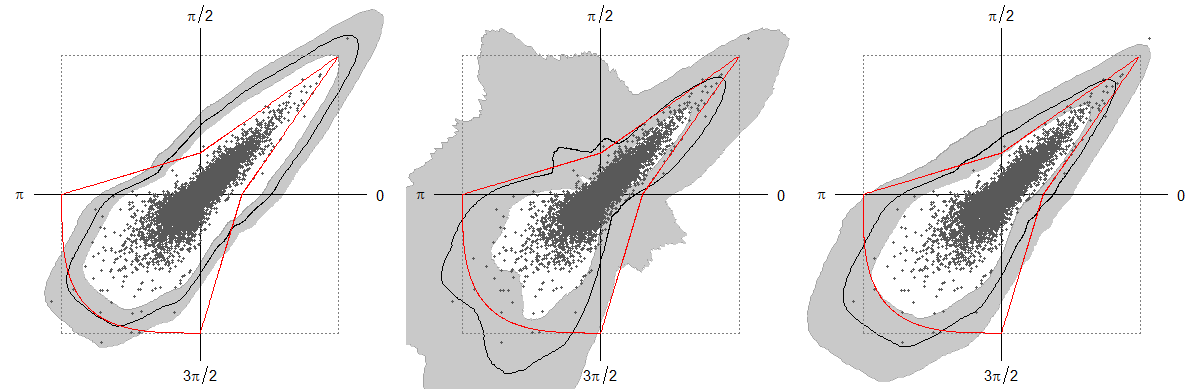}
  \caption{Estimated limit sets from $n_2$ draws from a scaled
    bivariate max-stable logistic distribution, with posterior mean
    (black line), $0.95$ prediction interval (grey), and true
    limit set (red line).\ From left to right:\ outputs of $\Mh$,
    $\Mg$, and $\Mgh$.\ \label{fig:limit_sets_Log}}
\end{figure}

The max-stable logistic distribution is an instance of family of
distributions for which $\Mg$ is known to be a wrong model for the
marginal density of the angles.\ However, it is observed from
Figure~\ref{fig:f_W_Log}---which displays the mean of the posterior
predictive density $\pi[\bm W \mid \G,\LL,\dat]$ as well as its $0.95$
prediction interval---that all three models possess the ability to
capture the empirical distribution of angles when incorporating all
sampled angles in the model likelihood.\
Figure~\ref{fig:limit_sets_Log} depicts some consequences of this.\
When the number of data used to fit $\pi[\bm W \mid \G,\LL,\dat]$ is
greater than the number of exceedances of $\QS_{q_k}$, the posterior
distribution of the gauge function $\gG$ obtained from $\Mg$ is
allowed to over-fit to the directional component of the model.\ This
generates bias and over-confidence in the posterior distribution of
the gauge function $\gG$ and of its associated limit boundary
$\partial \G$.\ Under $\Mh$, the decoupling of the dependence of the
density of angles $\pi[\bm W \mid \G,\LL,\dat]$ on the gauge function
$\gG$ (so that $\pi[\bm W \mid \G,\LL,\dat]=\pi[\bm W\mid \dat]$) is
expected to remove potential bias in $\gG$ coming from the observed
angles.\ A clear consequence, however, is the increase in the variance
of the posterior distribution of $\gG$.\

Figure~\ref{fig:MSD_limit_sets}
provides, for each of the modelling choices detailed in the beginning of this Section, a boxplot of the mean of
the mean integrated square distance (MISD) from each realisation of
the posterior of $\partial \G$.\
Model $\Mgh$, in all cases, seems to be the best---or approximately as
valid as the best---model in terms of MISD.\
\begin{figure}[htbp!]
  \centering
  \includegraphics[width=\textwidth]{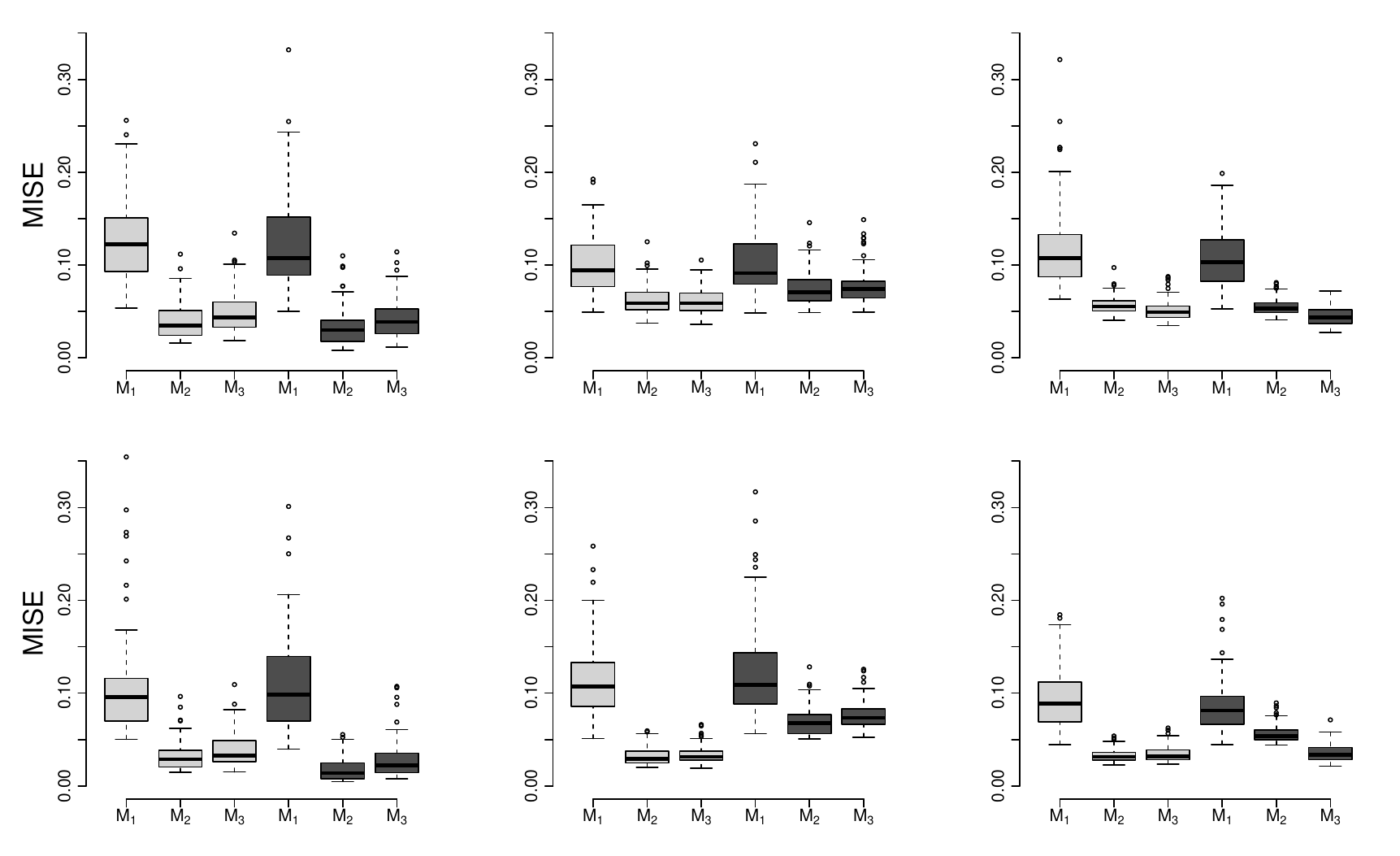}
  \caption{Boxplots of the means of the mean integrated squared error
    (MISE) from each posterior realisations of $\partial \G$ from
    each of the $600$ samples to the true limit sets.\ From top to
    bottom rows:\ $n_1$, $n_2$.\ From left to right columns:
    Bivariate Laplace, normal, and max-stable logistic
    distributions.\ Within panels:\ Results for models $\Mh$,
    $\Mg$, and $\Mgh$ using only angles of exceedances
    (light grey) and using all angles (dark
    grey).\ \label{fig:MSD_limit_sets}}
\end{figure}
Figure~\ref{fig:limit_sets_sim_study} presents the angle-wise mean of
the posterior distribution of the limit boundaries $\partial \G$
estimated from $\Mgh$ with the associated $0.95$ prediction
interval.\ As expected, increasing the sample size decreases uncertainty
in the estimated limit boundaries.\ One may also infer that when it is
reasonable to assume that observed extremes come from the same data
generating mechanism as all other observations, including all sampled
angles into the likelihood provides smoother estimates of $\partial \G$
than a procedure using only the exceedances of $\QS_{q_k}$.\ We argue that while there may be valid reasons to select $\Mh$ or
$\Mg$---respectively for evidence of a homothetic data generating
mechanism or of a poor matching between the density of angles
$\pi[\bm W \mid \G,\LL,\dat]$ and the homothetic density $\G^d/\vol{\G^d}$---, $\Mgh$
generally provides a better-behaved and more balanced choice.\
\begin{figure}[htbp!]
  \centering
  \includegraphics[width=\textwidth]{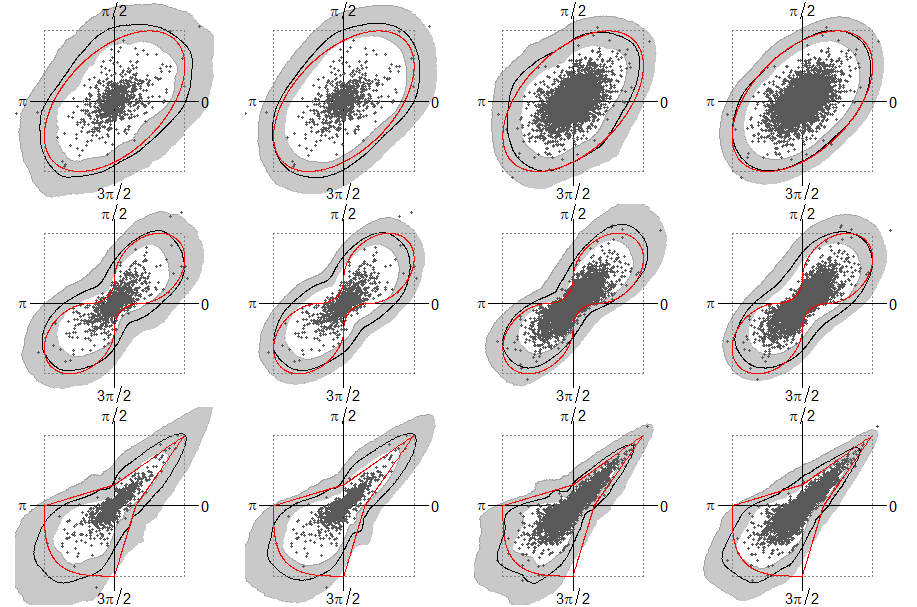}
  \caption{Limit set estimated from samples of size $n$ by $\Mgh$
    via the posterior mean of $\partial \G$ (black line) and $0.95$ prediction interval (grey bands) with true limit set
    (red line) and unit square (dotted grey).\ Rows from top to
    bottom rows:\ Draws from the bivariate normal, max stable
    logistic, and multivariate Laplace distributions.\ Columns from
    left to right:\ $\Mgh$ fitted only exceedances of
    $r_{\QS_{0.8}}$ (sample size $n_1$), $\Mgh$ fitted using all angles
    (sample size $n_1$), $\Mgh$ fitted only exceedances of
    $r_{\QS_{0.98}}$ (sample size $n_2$), $\Mgh$ fitted using all angles
    (sample size $n_2$).\ \label{fig:limit_sets_sim_study}}
\end{figure}

\subsection{Estimating rare event probabilities}
\label{sec:sim_study_predictive}
Based on the results of Figures~\ref{fig:MSD_limit_sets}
and~\ref{fig:limit_sets_sim_study}, we carry the remainder of this
simulation study under the framework of model $\Mgh$ and use all
observed angles for the cases of the bivariate Laplace and max-stable
logistic distributions and only those of exceedances of $\QS_{q_k}$ for the bivariate normal distribution.\ For each
of the $600$ samples defined in Section~\ref{sec:sim_study_posterior},
we obtain realisations from the posterior distribution of
$\PR_{B_i\mid \dat}$ according to the procedure described in
Section~\ref{supp:Posterior} to estimate the probability that a new
observation $R\bm W$ falls within extreme sets $B_1$, $B_2$, and $B_3$
as defined in Figure~\ref{fig:B_regions}.\ Figure~\ref{fig:prob_estim}
displays boxplots of the log-mean of the realisations from the
posterior predictive distribution of each set $B_i$ and each sample.\
\begin{figure}[htbp!]
  \centering
  \includegraphics[width=\textwidth]{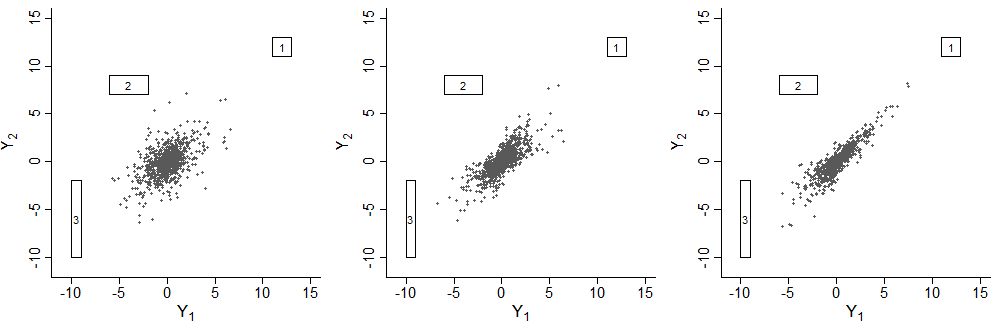}
  \caption{Regions $B_1=[11,13]\times[11,13]$,
    $B_2=[-6,-2]\times[7,9]$, and $B_3=[-10,-9]\times[-2,-10]$ for
    probability estimation superposed with samples of size $n_1$
    for the Gaussian (left), logistic (centre), and multivariate
    Laplace (right) distributions in Laplace
    margins.\ \label{fig:B_regions}}
\end{figure}
It is clear that the mean of predictive posterior distribution is
well-behaved for the sets $B_i$ corresponding to directional subsets of
$\SSS$ where the posterior distribution of the limit sets presented in
Section~\ref{sec:sim_study_posterior} is itself well-behaved.\
Table~\ref{tbl:coverage_probs} details the coverage obtained from
$100$ samples for the probability that a new observation $R\bm W$
falls within the sets $B_1$, $B_2$, and $B_3$ based on $0.95$ prediction
intervals based on the quantile method for the posterior predictive
distributions.\ The coverage properties of the prediction intervals seem
to be satisfying, except for the set $B_2$ where the true probability
is on the order of $10^{-16}$ and $10^{-14}$ for the bivariate normal
and max-stable logistic distributions respectively.\ We again note
that these correspond to regions in which the posterior distributions
of $\rG$ are biased in the case of the bivariate
normal distribution.\ We argue that reasonable coverage properties are
obtained in other cases.\
\begin{figure}[htbp!]
    \centering
    \includegraphics[width=0.8\textwidth]{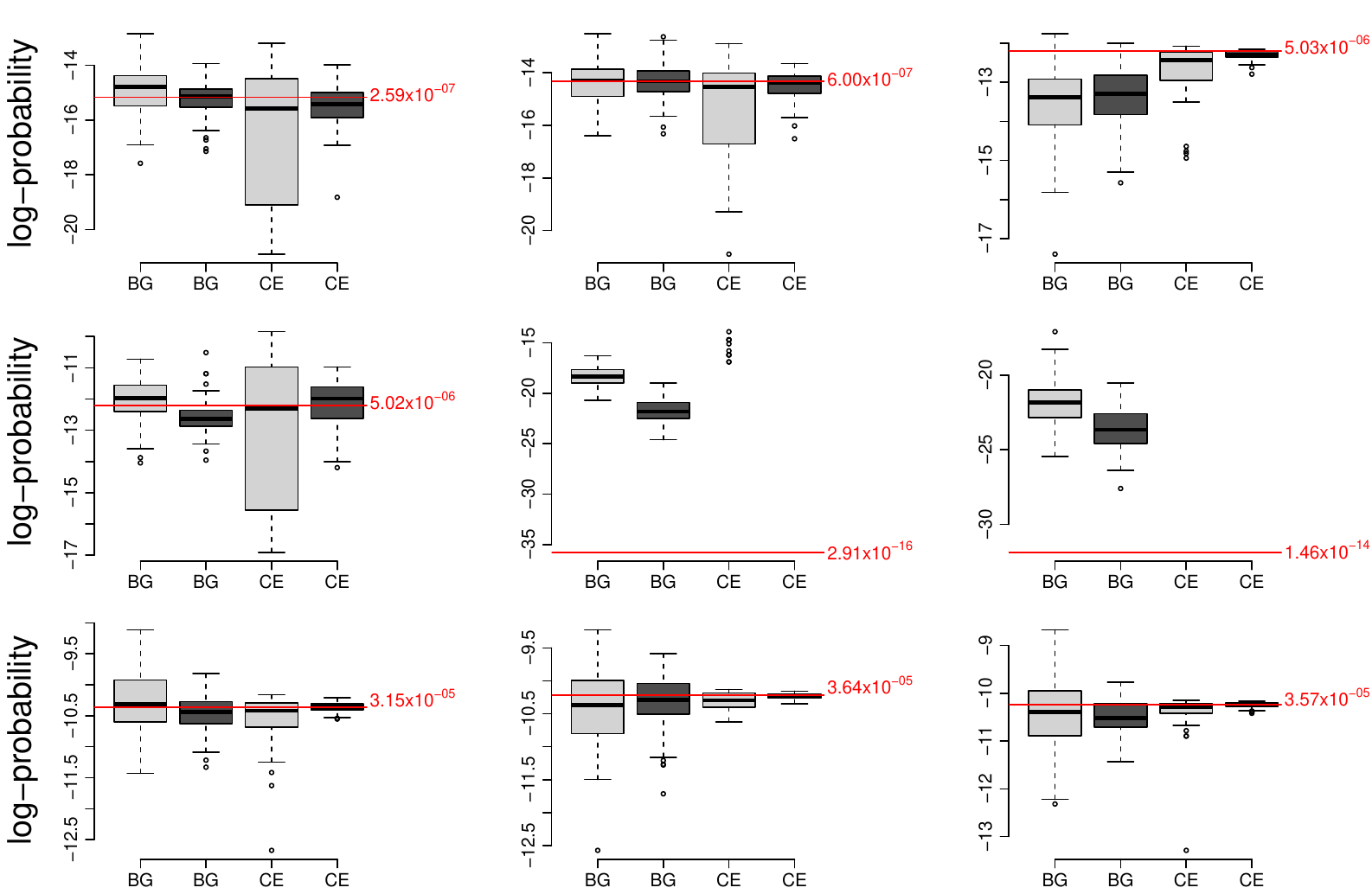}
    \caption{Boxplots of $100$ log-means of posterior predictive
    samples for regions $B_i$.\ Panels from left to right:\ Output of
    $\Mgh$ fitted on draws from the bivariate Laplace, Gaussian, and
    max-stable logistic distributions in Laplace margins.\ Panels
    from top to bottom:\ $\log \PR_{B_1\mid \dat}$,
    $\log \PR_{B_2\mid \dat}$, and $\log \PR_{B_3\mid \dat}$.\ Within
    panels, light grey:\ sample size $n_1$, dark grey:\ sample size $n_2$.\label{fig:prob_estim}}
\end{figure}
\begin{table}[htbp!]
  \centering
  \begin{tabular}{lllllll}
    \toprule
    Bivariate distribution & Laplace &   & Normal &  & Max-stable logistic & \\
    \midrule
    $n$ & 1000 & 10000 & 1000 & 10000 & 1000 & 10000\\
    \toprule
    Coverage for $\PR_{B_1\mid \dat}$& 0.97 & 0.94 & 0.99 & 0.97 & 0.56 & 0.57 \\
    Coverage for $\PR_{B_2\mid \dat}$& 0.96 & 0.88 & 0.13 & 0.70 & 0.89 & 0.85 \\
    Coverage for $\PR_{B_2\mid \dat}$& 0.98 & 0.97 & 0.97 & 0.94 & 0.92 & 0.92 \\
    \toprule
  \end{tabular}
  \label{tbl:coverage_probs}
  \caption{Coverage study of the $0.95$ prediction intervals from the posterior predictive distribution $\PR_{B_i\mid \dat}$ for the bivariate Laplace, normal, and max-stable logistic distributions.\ \label{tbl:coverage_probs}}
\end{table}

\subsection{Return sets}
We obtain realisations from the posterior distribution of the radial
function of the return set $\QS_{1- 1/T}$ using
equation~\eqref{eq:rs_stability_exp} on each realisation from the
posterior distribution of the radial functions of $\QS_q$ and $\G$.\
The angle-wise posterior mean of $r_{\QS_{1- 1/T}}$ hence defines a
the boundary of a return set with period $T$ as well as its associated
$0.95$ simultaneous prediction interval.\ Figure~\ref{fig:Ret_sets}
displays return level-sets return sets of period
$T\in\{100,1000,10000\}$ estimated from one sample of $n_2$ draws from
the bivariate Laplace, normal, and max-stable logistic distributions.\
\begin{figure}[htbp!]
  \centering
  \includegraphics[width=0.8\textwidth]{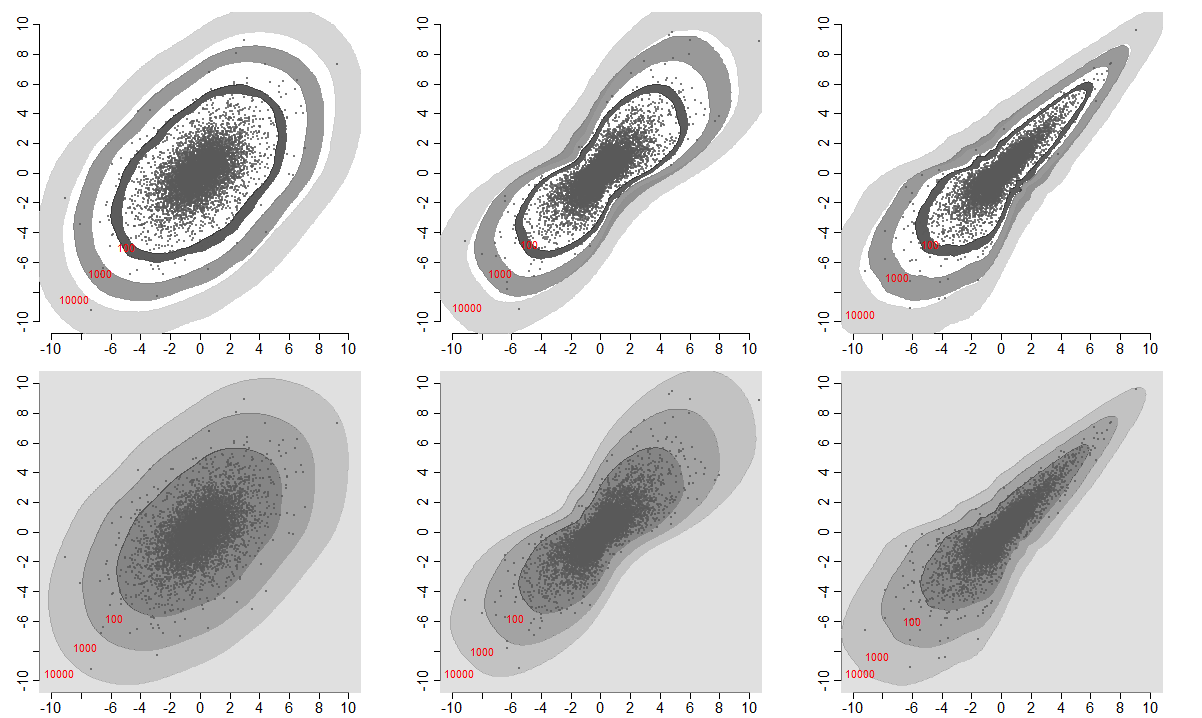}
  \caption{\textit{Top row}:\ $0.95$ prediction intervals for the
    boundaries $r_{\QS_{1- 1/T}}$ of the canonical return sets with
    associated return period $T$ displayed in red.\ \textit{Bottom row}:\ Posterior mean of the return
    sets with associated return period (red) defined by the angle-wise posterior
    mean of $r_{\QS_{1- 1/T}}$.\ Every return set contains all lighter-grey
    sets.\ Columns from left to right:\ Bivariate Laplace, normal, and
    max-stable logistic distributions in Laplace
    margins.\ \label{fig:Ret_sets}}
\end{figure}

\subsection{Estimation of extremal coefficients}
Figure~\ref{fig:alphas} displays boxplots of the means of the samples
from the posterior distribution of $\alpha_{\mid 1}$ obtained from
both estimation procedures on each of the $600$ samples.\
Table~\ref{tbl:coverage_alphas} presents a coverage study of the $0.95$ 
credible intervals for $\alpha_{\mid 1}$.\ 
\begin{table}[htbp!]
  \centering
  \begin{tabular}{lllllll}
    \toprule
    Bivariate distribution & Laplace &   & Normal &  & Max-stable  & \\
    \midrule
    $n$ & 1000 & 10000 & 1000 & 10000 & 1000 & 10000\\
    \toprule
    Coverage for $\alpha_{\mid 1}$, Method 1&  1.00&  1.00& 1.00&  1.00&  0.80 & 0.64\\
    Coverage for $\alpha_{\mid 1}$, Method 2&  1.00 &  1.00&  1.00&  1.00&  0.73 & 0.49 \\
    \toprule
  \end{tabular}
  \label{tbl:coverage_alphas}
  \caption{Coverage study of the $0.95$ credible intervals from the posterior distribution of $\alpha_{\mid 1}$ under both estimation methods for the bivariate Laplace, normal, and max-stable logistic distributions.\ \label{tbl:coverage_alphas}}
\end{table}
\begin{figure}[htbp!]
  \centering
  \includegraphics[width=\textwidth]{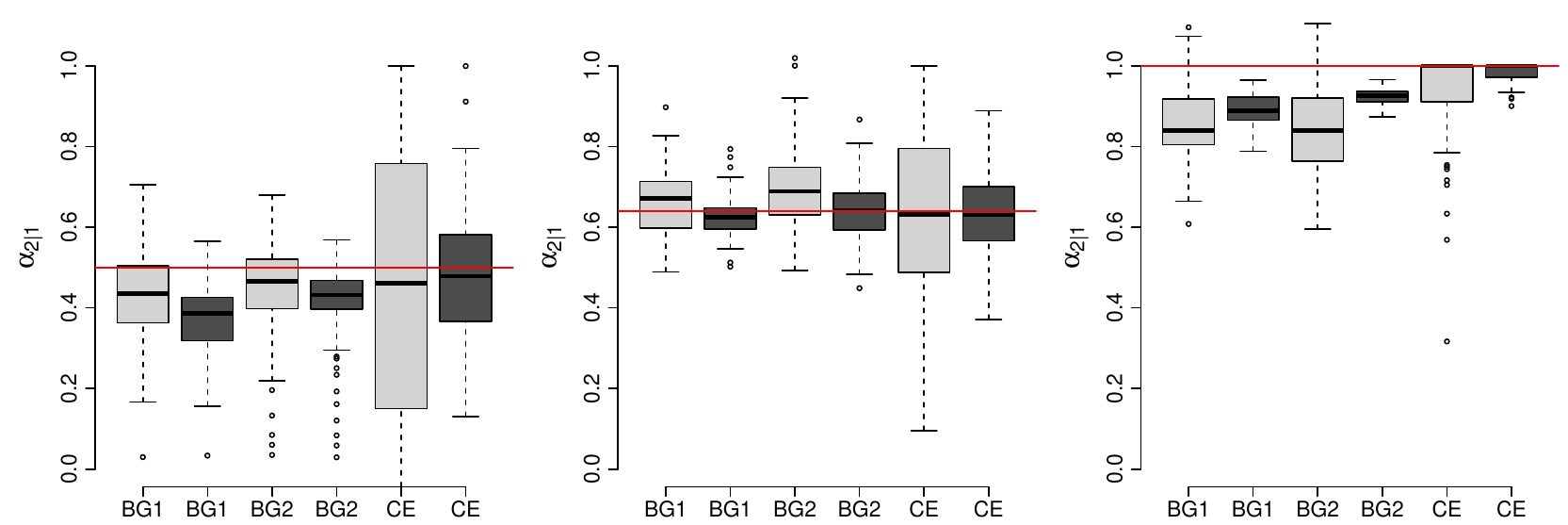}
  \caption{Boxplots of $100$ means of realisations from the
    posterior of $\alpha_{2\mid 1}$.\ Columns
    from left to right:\ Output of $\Mgh$ fitted on draws from the
    bivariate Laplace, Gaussian, and logistic distributions in
    Laplace margins.\ Within panels, 1:\ method 1 and $n_1$, 2:
    method 2 and $n_1$, 1:\ method 1 and $n_2$, 2:\ method 2
    and $n_2$.\ The true values of $\alpha_{\mid 1}$ are given
    by the red line.\ \label{fig:alphas}}
\end{figure}
\begin{figure}[htbp!]
  \centering
  \includegraphics[width=\textwidth]{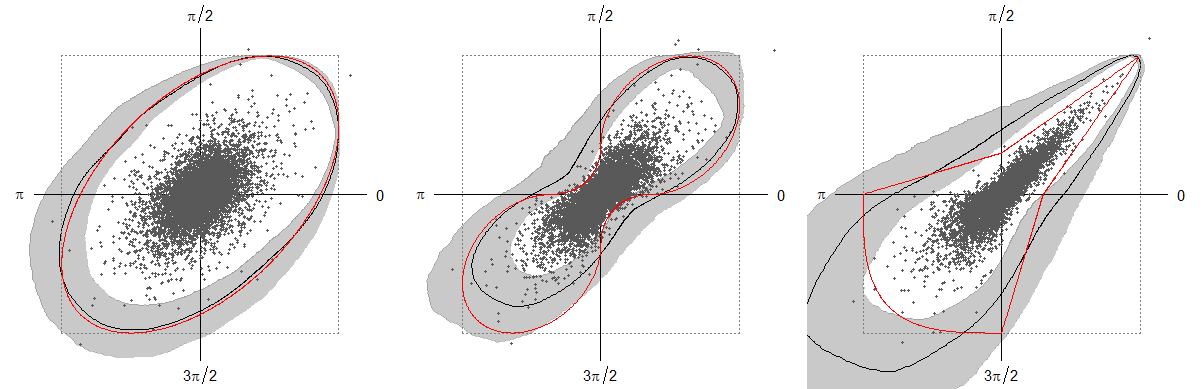}
  \caption{Posterior mean (black line) and associated $0.95$ prediction bands (grey) for the transformed posterior
    distribution $r_{\G_T}$ of the bivariate Laplace (left),
    normal (centre), and max-stable logistic (right) distributions with true limit set (red
    line).\ \label{fig:transformed_partialG}}
\end{figure}

Instances of transformed $r_{\G_T}$ are displayed in Figure~\ref{fig:transformed_partialG}.\ As can be observed, the transformation does not lead to a posterior $r_{\G_T}$ that respects the limit boundary conditions everywhere, but it is better-behaved in the directional region between the angles $\bm w^{(1)}$ and $\bm w^{(2)}$.\ The improvement in the estimation of $\eta$ through this transformation procedure is observed in the results of the simulation study.\ Figure~\ref{fig:etas} contains boxplots of the posterior means and 0.95 prediction intervals for $\eta$ resulting from $\partial \G$ and $\partial \G_T$.\ Table~\ref{tbl:coverage_etas} presents a coverage study of the $0.95$ credible intervals for $\eta$.\ We argue that our model, especially under the consideration of the transformed $\partial \G_T$, can provide meaningful estimates of $\eta$.
\begin{table}[htbp!]
  \centering
  \begin{tabular}{lllllll}
    \toprule
    Bivariate distribution & Laplace &   & Normal &  & Max-stable logistic & \\
    \midrule
    $n$ & 1000 & 10000 & 1000 & 10000 & 1000 & 10000\\
    \toprule
    Coverage for $\eta$, Method 1& 0.93 & 0.97  & 1.00 & 0.98 & 0.95 & 0.85 \\
    Coverage for $\eta$, Method 2& 0.95 & 0.96  & 1.00 & 1.00 & 0.99 & 0.67 \\
    \toprule
  \end{tabular}
  \label{tbl:coverage_etas}
  \caption{Coverage study of the $0.95$ credible intervals from the posterior distribution of $\eta$ under both estimation methods for the bivariate Laplace, normal, and max-stable logistic distributions.\ \label{tbl:coverage_etas}}
\end{table}
\begin{figure}[htbp!]
  \centering
  \includegraphics[width=\textwidth]{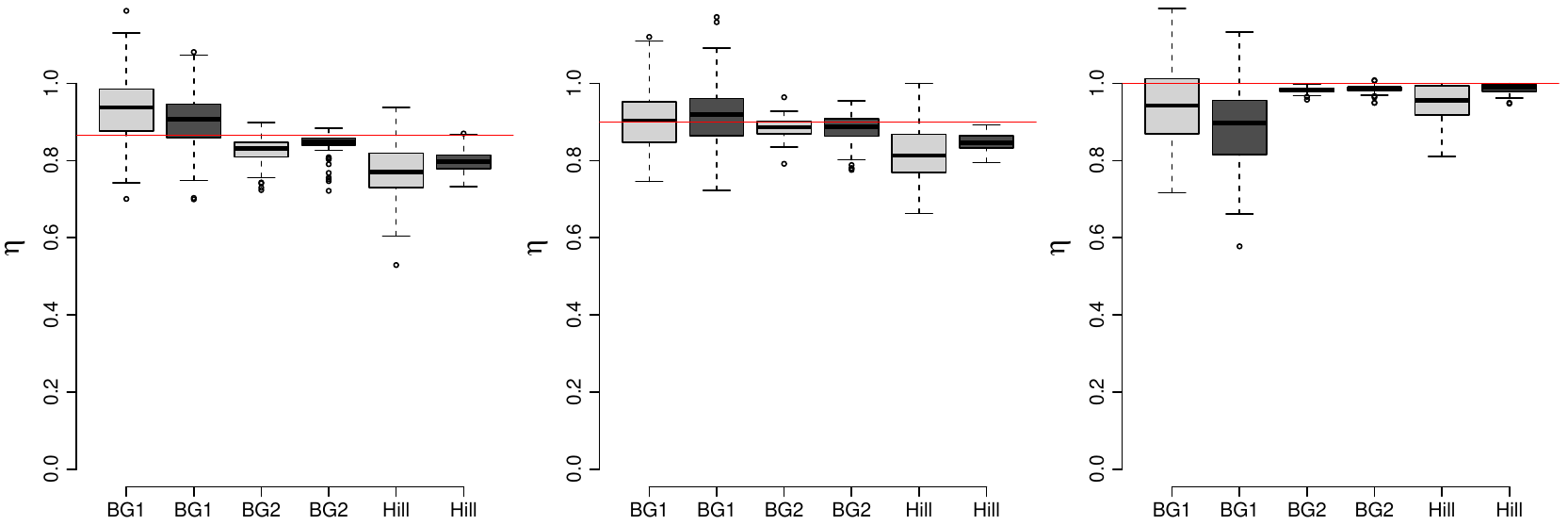}
  \caption{Boxplots of $100$ means of posterior predictive samples
    for $\eta$.\ Panels from left to right:\ Draws from the bivariate
    Laplace, Gaussian, and logistic distributions in Laplace
    margins.\ Within panels, 1:\ method 1, 2:\ method 2.\ The true
    values of $\eta$ are given by the red line.\ \label{fig:etas}}
\end{figure}

\newpage
\section{Case studies}\label{app:E}
\subsection{Posterior model fits on river data}
\label{sec:river-g-unit-level-sets}

On the bivariate Thames tributary data, we fit models using the three architectures $\Mh$, $\Mg$, and $\text{M}_3$ using data exceeding a high posterior quantile estimate and on all available data.\ Figures \ref{fig:river-gauge-unit-level-sets} and \ref{fig:river-angle-densities} show posterior mean limit set boundary and posterior directional densities for each of the six fitted models.\ In these figures, we see that the fitted model associated with the $\Mg$ architecture and fitted on exceedance data only is best.\ This model's posterior limit set agrees most with $\log n$-scaled data, while the posterior mean directional density is in good agreement with an empirical sample of angles.\ The posterior QQ plots in Figure \ref{fig:river-qq} corresponding to this optimal model shows good agreement with the radial exceedance model and the underlying exponential distribution, while the PP plots in \ref{fig:river-pp-angle} show general agreement between the posterior and the underlying directional models.

\begin{figure}[h!]
  \centering
  \begin{subfigure}{.35\textwidth}
    \centering
    \includegraphics[width=.82\linewidth]{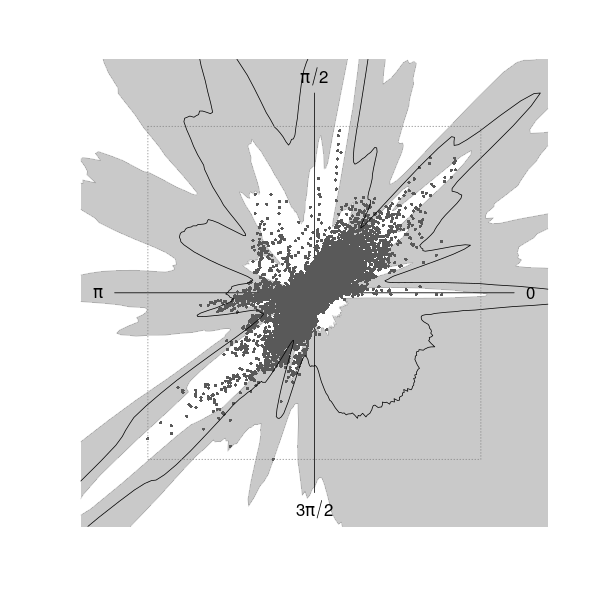}
    \caption{$\Mh$, exc.\ only}
    \label{fig:river-gauge-m1-exc}
  \end{subfigure}
  \begin{subfigure}{.35\textwidth}
    \centering
    \includegraphics[width=.82\linewidth]{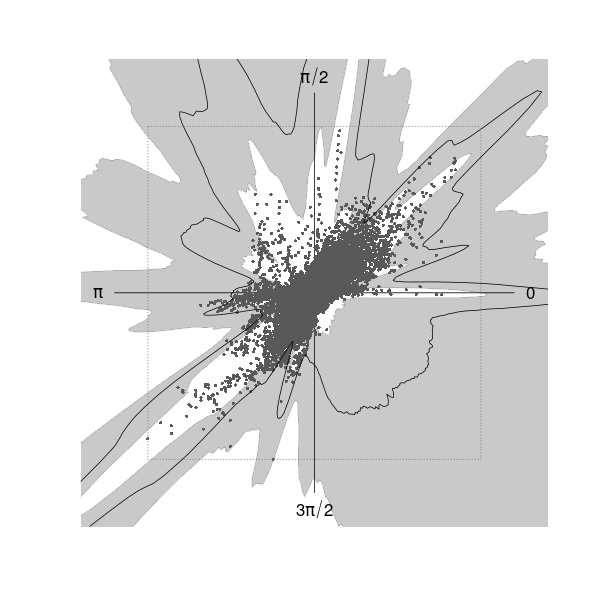}
    \caption{$\Mh$, all angles}
    \label{fig:river-gauge-m1-all}
  \end{subfigure}
  \begin{subfigure}{.35\textwidth}
    \centering
    \includegraphics[width=.82\linewidth]{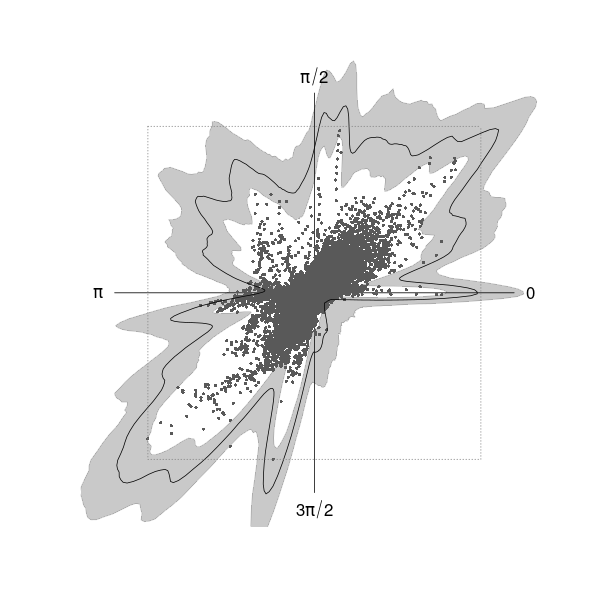}
    \caption{$\Mg$, exc.\ only}
    \label{fig:river-gauge-m2-exc}
  \end{subfigure}
  \begin{subfigure}{.35\textwidth}
    \centering
    \includegraphics[width=.82\linewidth]{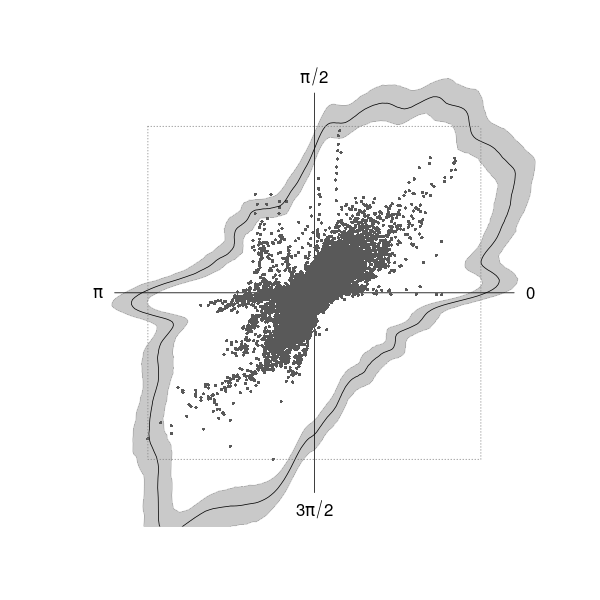}
    \caption{$\Mg$, all angles}
    \label{fig:river-gauge-m2-all}
  \end{subfigure}
  \begin{subfigure}{.35\textwidth}
    \centering
    \includegraphics[width=.82\linewidth]{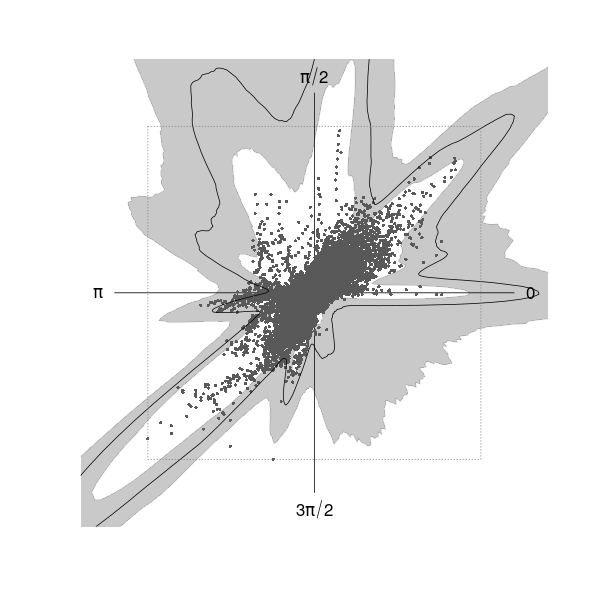}
    \caption{$\Mgh$, exc.\ only}
    \label{fig:river-gauge-m3-exc}
  \end{subfigure}
  \begin{subfigure}{.35\textwidth}
    \centering
    \includegraphics[width=.82\linewidth]{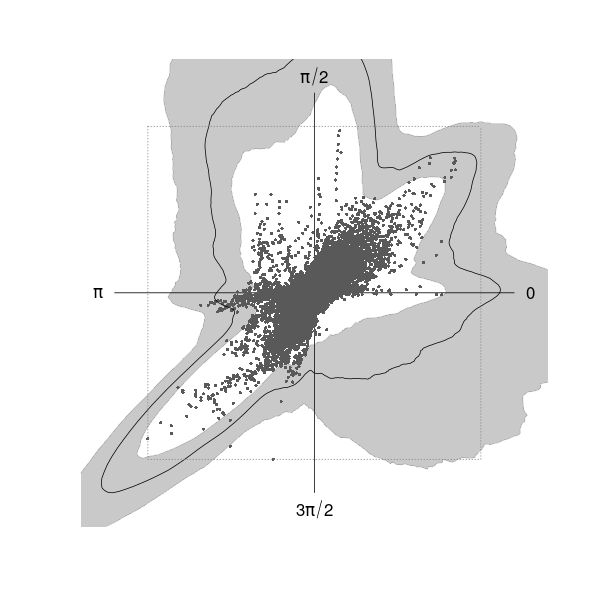}
    \caption{$\Mgh$, all angles}
    \label{fig:river-gauge-m3-all}
  \end{subfigure}
  \caption{Posterior estimates of the unit level set $\gG(\bm{x})=1$
    for the river flow dataset.\ The black line corresponds to the
    posterior mean, with the 0.95 prediction interval shaded in grey.\
    Black points are the original data in Laplace margins scaled by
    $\log({n}/{2})$.\ Dashed border line is the unit box.\ $\Mh$,
    $\Mg$, and $\Mgh$ define the angle density kernel, as described in
    section~\ref{sec:likelihood}.}
  \label{fig:river-gauge-unit-level-sets}
\end{figure}

\begin{figure}[h!]
  \centering
  \begin{subfigure}{.4\textwidth}
    \centering
    \includegraphics[width=.8\linewidth]{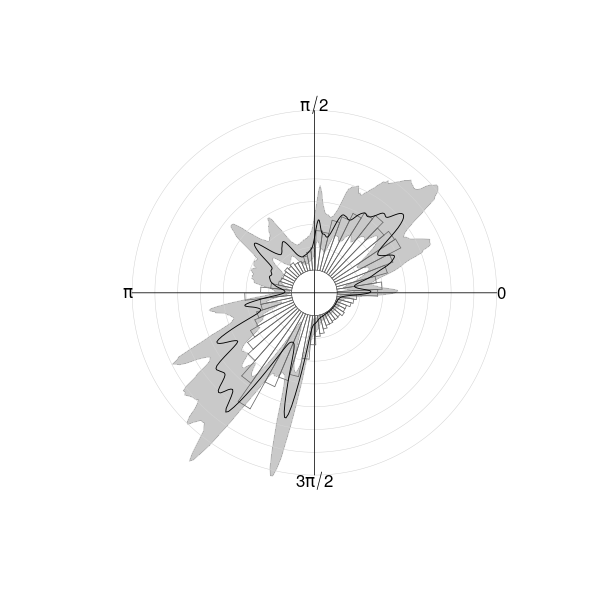}
    \caption{$\Mh$, exc.\ only}
    \label{fig:river-angle-m1-exc}
  \end{subfigure}
  \begin{subfigure}{.4\textwidth}
    \centering
    \includegraphics[width=.8\linewidth]{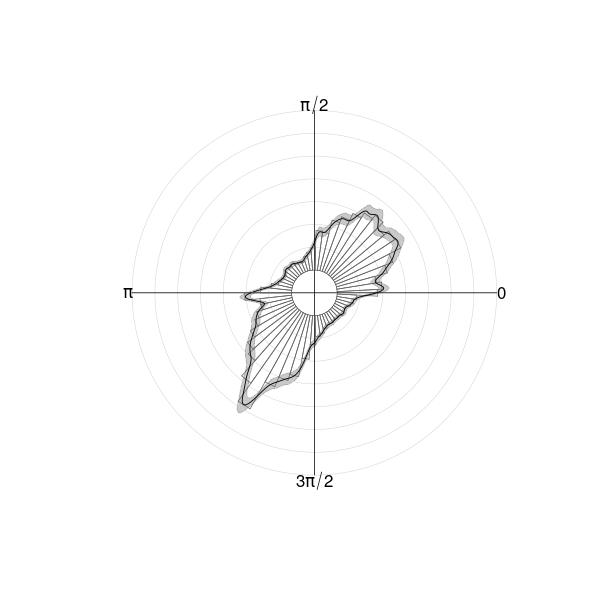}
    \caption{$\Mh$, all angles}
    \label{fig:river-angle-m1-all}
  \end{subfigure}
  \begin{subfigure}{.4\textwidth}
    \centering
    \includegraphics[width=.8\linewidth]{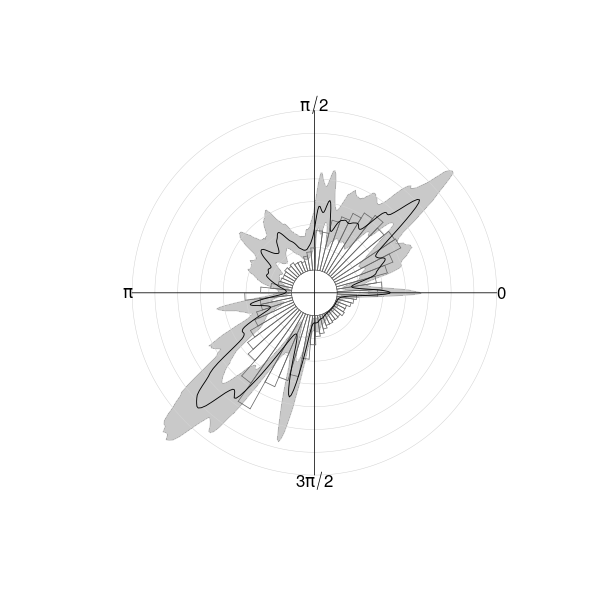}
    \caption{$\Mg$, exc.\ only}
    \label{fig:river-angle-m2-exc}
  \end{subfigure}
  \begin{subfigure}{.4\textwidth}
    \centering
    \includegraphics[width=.8\linewidth]{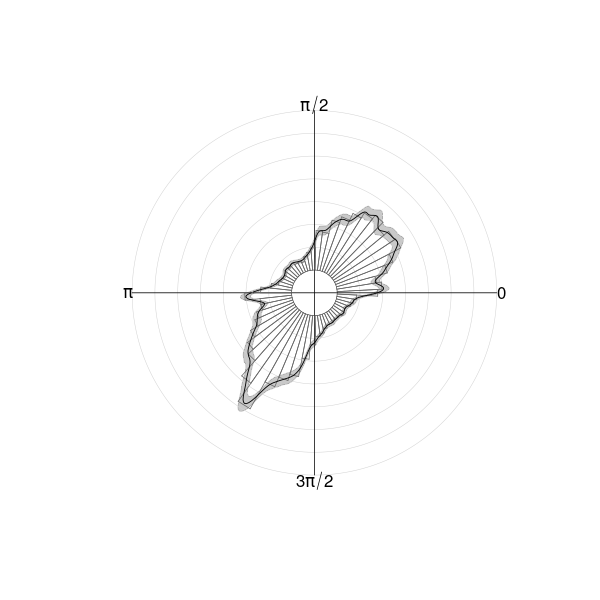}
    \caption{$\Mg$, all angles}
    \label{fig:river-angle-m2-all}
  \end{subfigure}
  \begin{subfigure}{.4\textwidth}
    \centering
    \includegraphics[width=.8\linewidth]{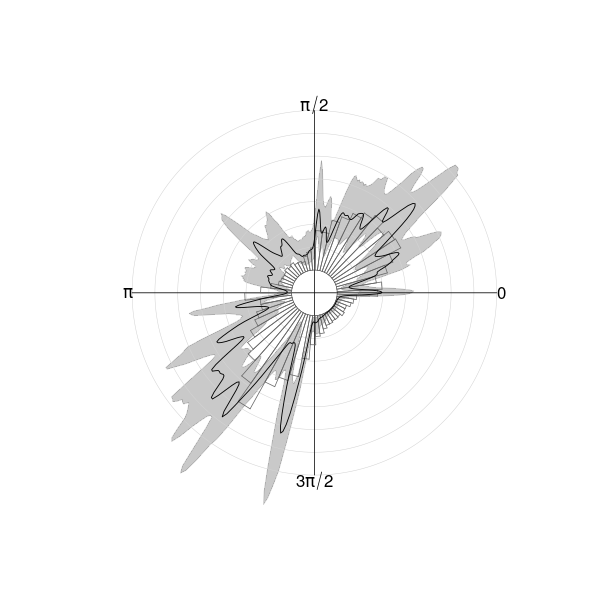}
    \caption{$\Mgh$, exc.\ only}
    \label{fig:river-angle-m3-exc}
  \end{subfigure}
  \begin{subfigure}{.4\textwidth}
    \centering
    \includegraphics[width=.8\linewidth]{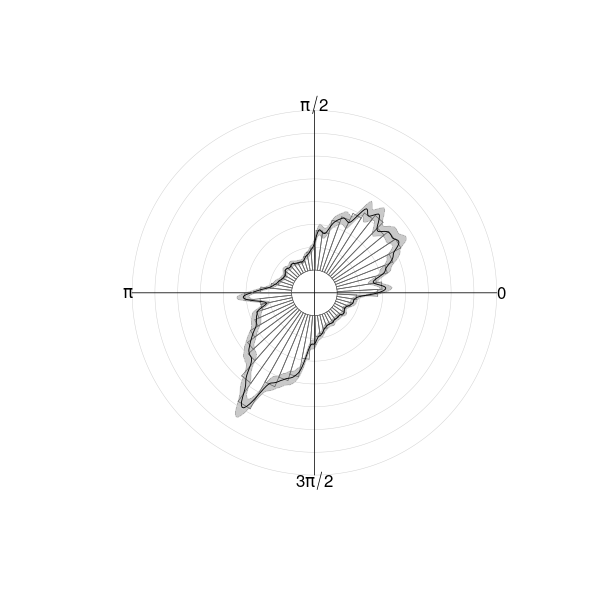}
    \caption{$\Mgh$, all angles}
    \label{fig:river-angle-m3-all}
  \end{subfigure}
  \caption{Posterior estimates of the mean angle density for the river flow dataset with 0.95 prediction intervals.\ The empirical density of angles is given by the underlying histogram.\ $\Mh$, $\Mg$, and $\Mgh$ define the angle density kernel, as described in section~\ref{sec:likelihood}.}
  \label{fig:river-angle-densities}
\end{figure}

\begin{figure}[h!]
    \centering
    \includegraphics[width=0.8\textwidth]{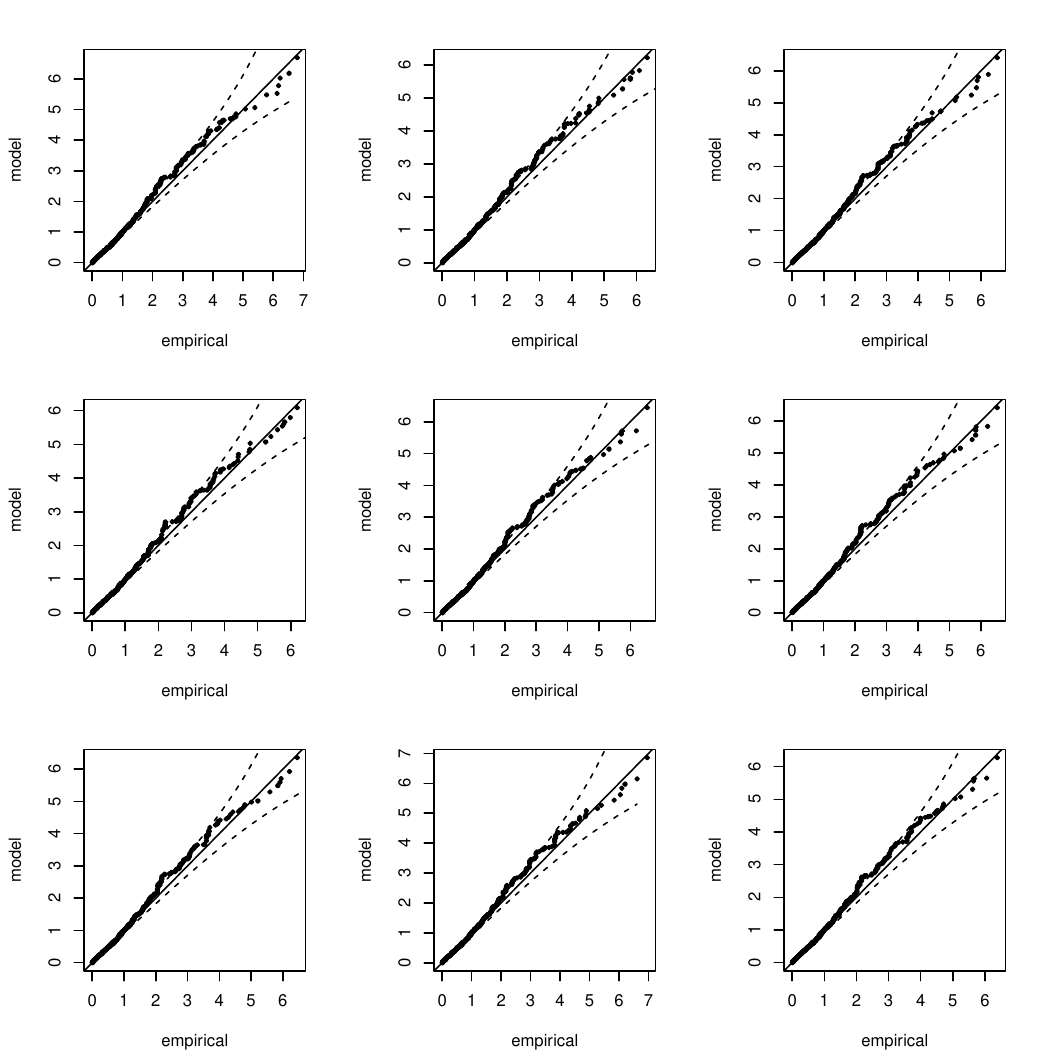}
    \caption{QQ plots for the Thames tributary river flow dataset exceedance model for 9 sampled posterior thresholds $\rquant(\bm{w})$, with 95\% confidence intervals.}
    \label{fig:river-qq}
\end{figure}
\begin{figure}[h!]
    \centering
    \includegraphics[width=0.3\textwidth]{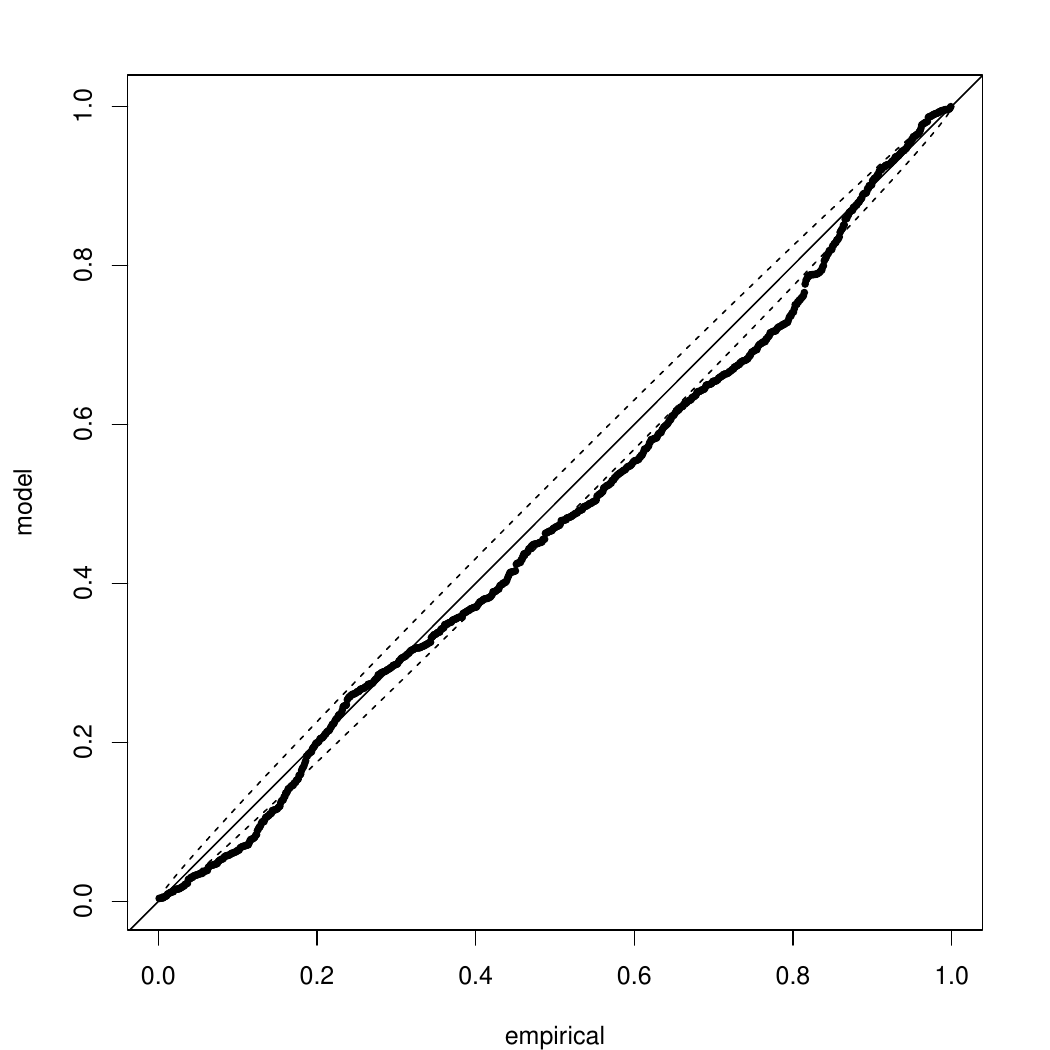}
    \caption{PP plots for the Thames tributary river flow dataset directional model, with 95\% confidence intervals.}
    \label{fig:river-pp-angle}
\end{figure}

\newpage
\clearpage
\subsection{Additional Newlyn wave height figures}
\label{sec:wave-additional-plots}
On the three-dimensional Newlyn wave data, we fit models using the
three architectures $\Mh$, $\Mg$, and $\text{M}_3$ using data
exceeding a high posterior quantile estimate and on all available
data.\ Figure \ref{fig:wave-gauge-unit-level-sets} show posterior mean
limit set boundaries for each of the six fitted models.\ In these
figures, we see that the fitted model associated with the $\Mg$
architecture and fitted on exceedance data only agrees most with
$\log n$-scaled data, though all models fit well.\ The posterior QQ
plots in Figure \ref{fig:wave-qq} corresponding to this optimal model
shows good agreement with the radial exceedance model and the
underlying exponential distribution, while the PP plots in
\ref{fig:wave-pp-angle} show good agreement between the posterior and
the underlying directional models.\ In Section \ref{sec:wave} of the main
body of this manuscript we see that this model is also very accurate
in estimating other diagnostics.\ One diagnostic not mentioned in
Section \ref{sec:wave} of the main body is plots of $\chi_q(A)$,
defined in equation \eqref{eq:chi-def} in Section \ref{sec:intro} of
the main body.\ Figure~\ref{fig:wave-chi} shows posterior mean
estimates of $\chi_q(A)$ for $A\in\left\{HP, HS,PS, HPS\right\}$ and
for $q\in(0.9,1)$ with posterior 95\% confidence intervals.\
\cite{wadsworth2022statistical} establish that the variable groups
$HP$, $PS$, and $HPS$ are asymptotically independent, and that $HS$ is
asymptotically dependent.\ With this in mind, we conclude that our
posterior joint model is able to accurately describe the extremal
dependence structure presented in this dataset.\
Figure~\ref{fig:wave-chi} correctly shows posterior estimates of
$\chi_q(A)$ tending to 0 as $q$ tends to $1$ for the asymptotically
independent groups, and tending to a nonzero value close to the
empirical estimate for the asymptotically dependent pair $HS$.

\begin{figure}[h!]
  \centering
    \begin{subfigure}{.34\textwidth}
    \centering
    \includegraphics[width=.85\linewidth]{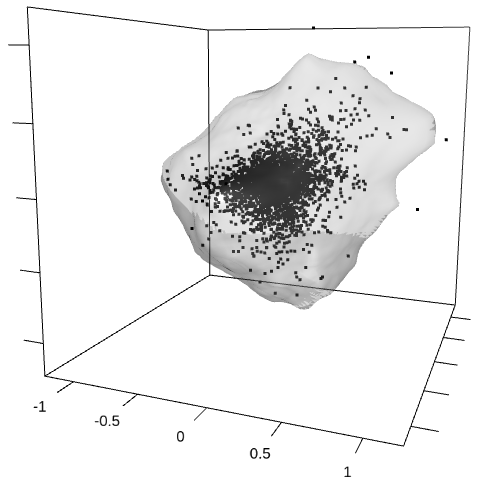}
    \caption{$\Mh$, exc.\ only}
    \label{fig:wave-gauge-m2-exc}
  \end{subfigure}
  \begin{subfigure}{.34\textwidth}
    \centering
    \includegraphics[width=.85\linewidth]{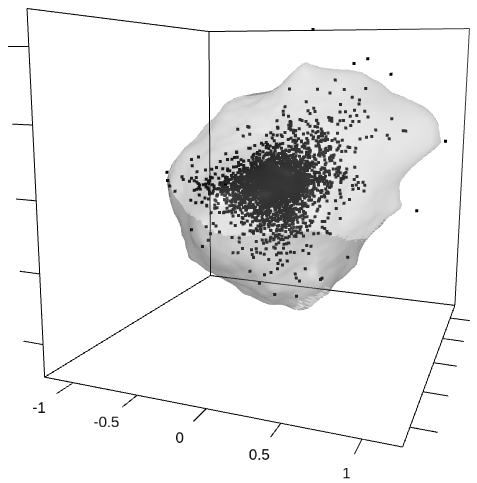}
    \caption{$\Mh$, all angles}
    \label{fig:wave-gauge-m2-all}
  \end{subfigure}
  \begin{subfigure}{.34\textwidth}
    \centering
    \includegraphics[width=.85\linewidth]{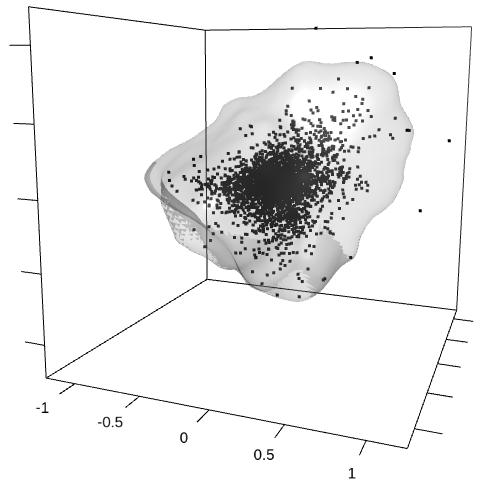}
    \caption{$\Mg$, exc.\ only}
    \label{fig:wave-gauge-m1-exc}
  \end{subfigure}
  \begin{subfigure}{.34\textwidth}
    \centering
    \includegraphics[width=.85\linewidth]{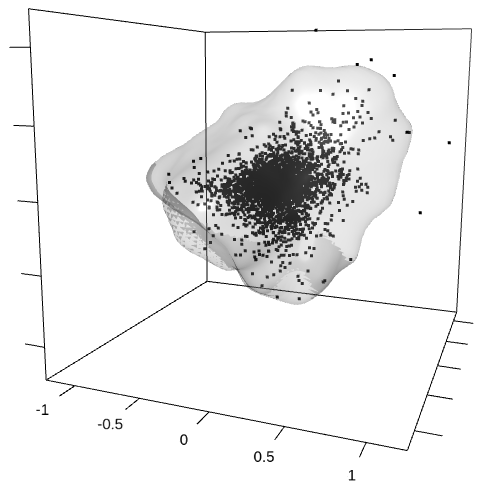}
    \caption{$\Mg$, all angles}
    \label{fig:wave-gauge-m1-all}
  \end{subfigure}
    \begin{subfigure}{.34\textwidth}
    \centering
    \includegraphics[width=.85\linewidth]{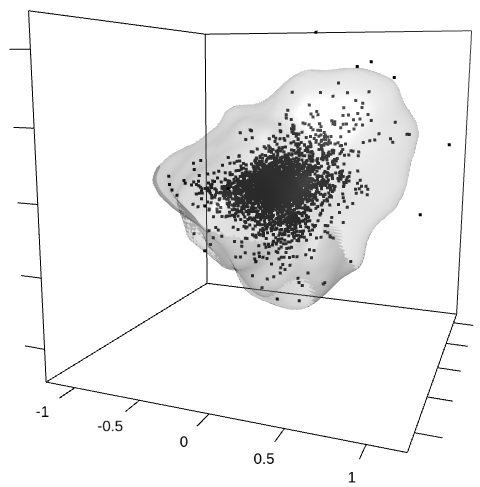}
    \caption{$\Mgh$, exc.\ only}
    \label{fig:wave-gauge-m3-exc}
  \end{subfigure}
  \begin{subfigure}{.34\textwidth}
    \centering
    \includegraphics[width=.85\linewidth]{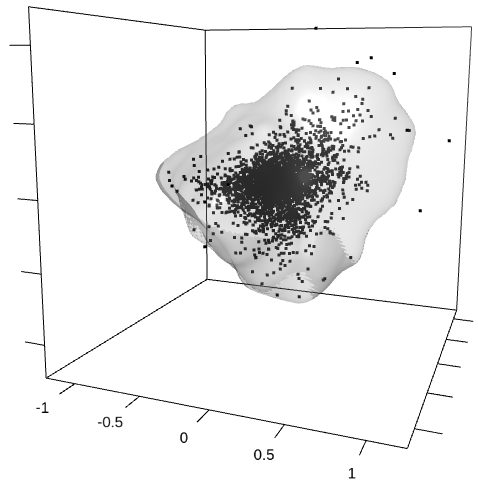}
    \caption{$\Mgh$, all angles}
    \label{fig:wave-gauge-m3-all}
  \end{subfigure}
  \caption{Posterior mean estimates of the unit level set
    $\gG(\bm{x})=1$ for the wave dataset.\ Black points are the original data in
    Laplace margins scaled by $\log({n}/{2})$.\ $\Mh$, $\Mg$, and $\Mgh$ define the angle
    density kernel, as described in section~\ref{sec:likelihood}.}
  \label{fig:wave-gauge-unit-level-sets}
\end{figure}

\begin{figure}[h!]
    \centering
    \includegraphics[width=0.8\textwidth]{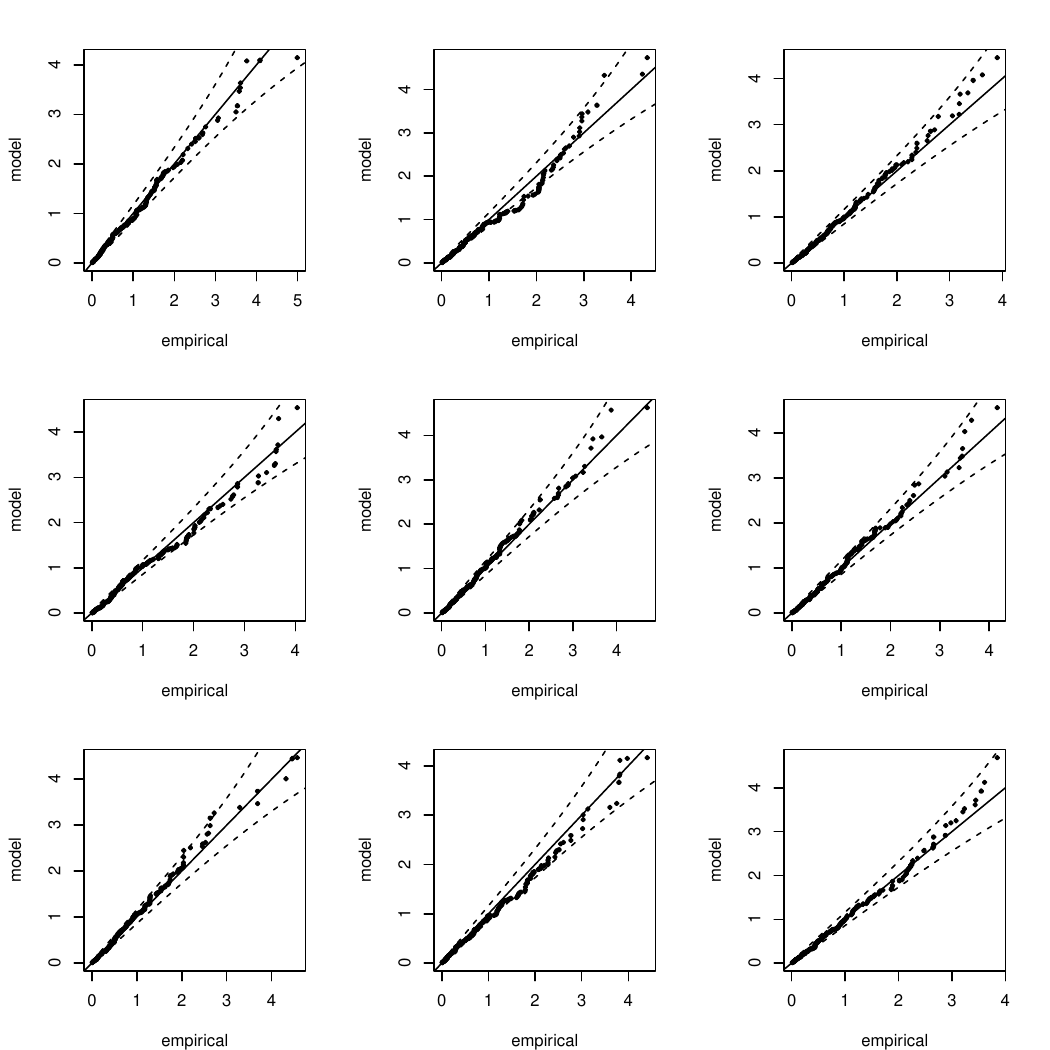}
    \caption{QQ plots for the Newlyn wave dataset exceedance model for 9 sampled posterior thresholds $\rquant(\bm{w})$, with 95\% confidence intervals.}
    \label{fig:wave-qq}
\end{figure}
\begin{figure}[h!]
    \centering
    \includegraphics[width=0.3\textwidth]{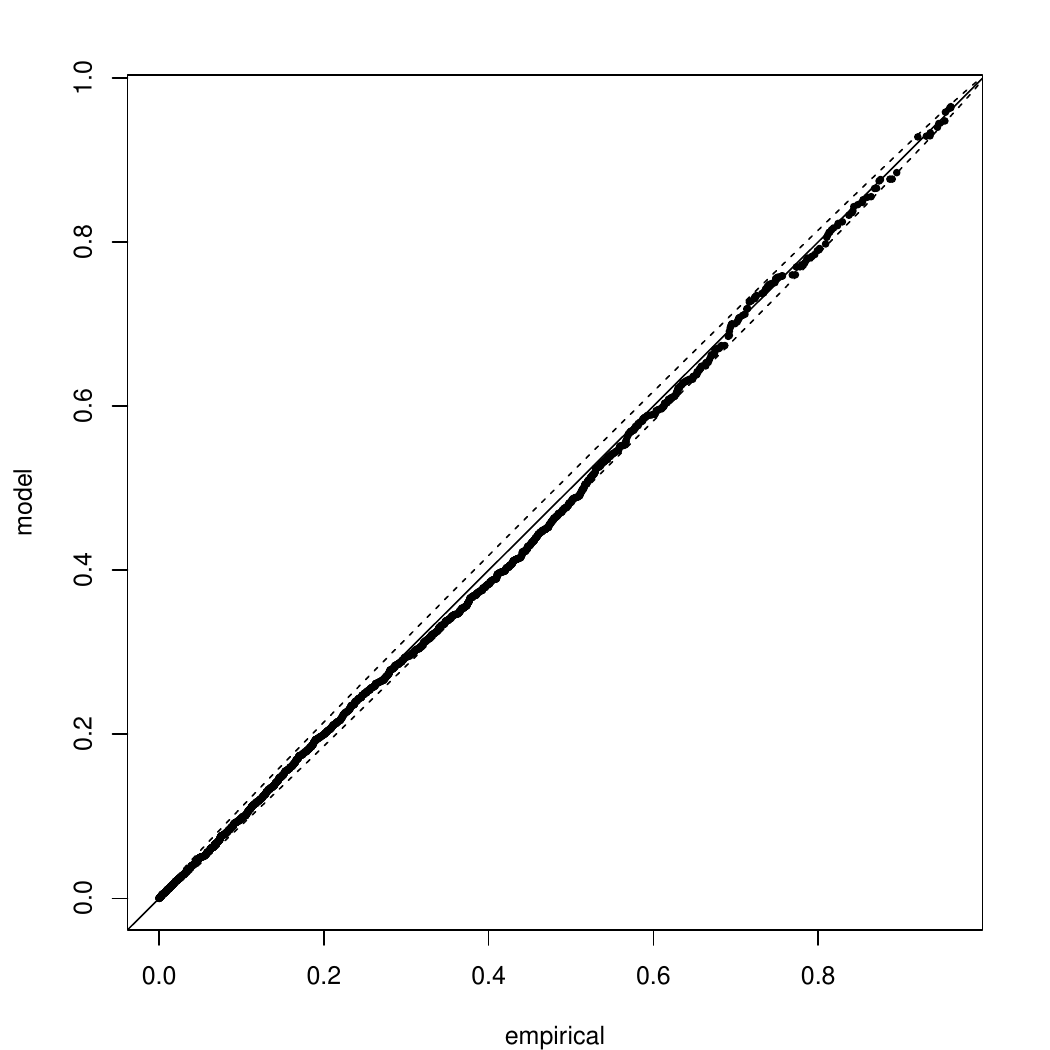}
    \caption{PP plots for the Newlyn wave dataset directional model, with 95\% confidence intervals.}
    \label{fig:wave-pp-angle}
\end{figure}

\begin{figure}[t!]
    \centering
    \includegraphics[width=0.24\textwidth]{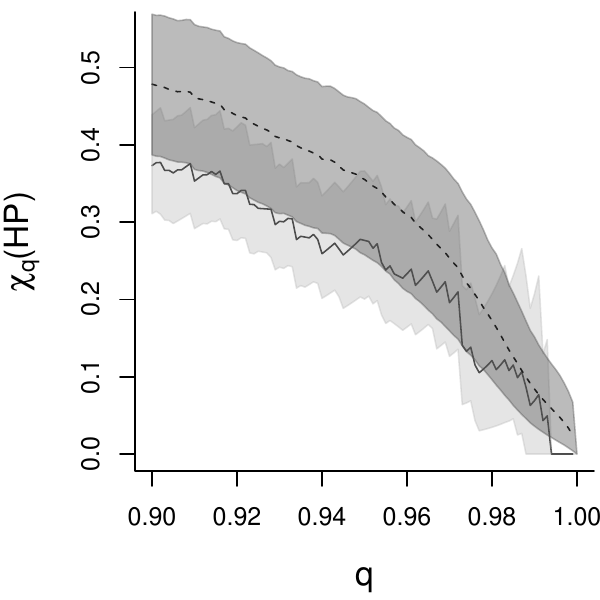}
    \includegraphics[width=0.24\textwidth]{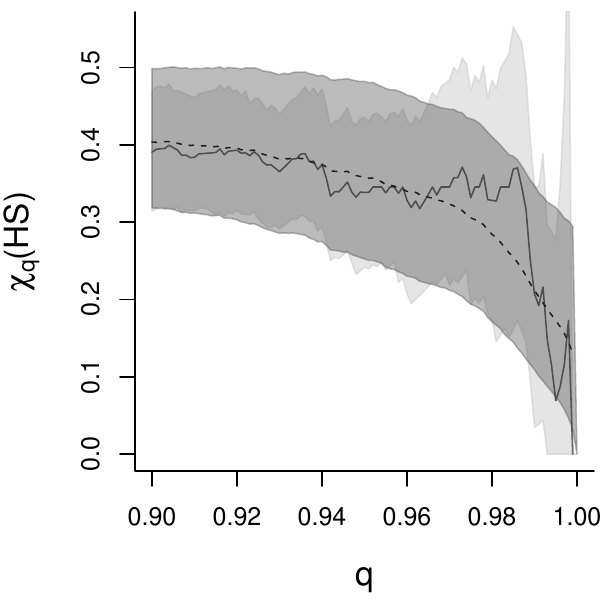}
    \includegraphics[width=0.24\textwidth]{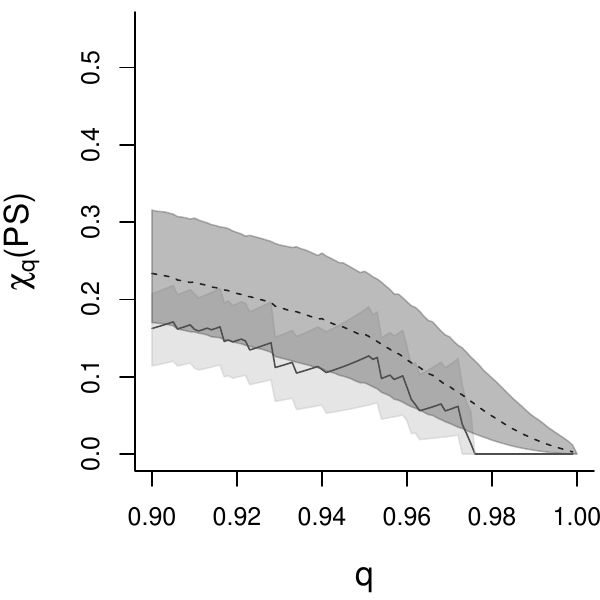}
    \includegraphics[width=0.24\textwidth]{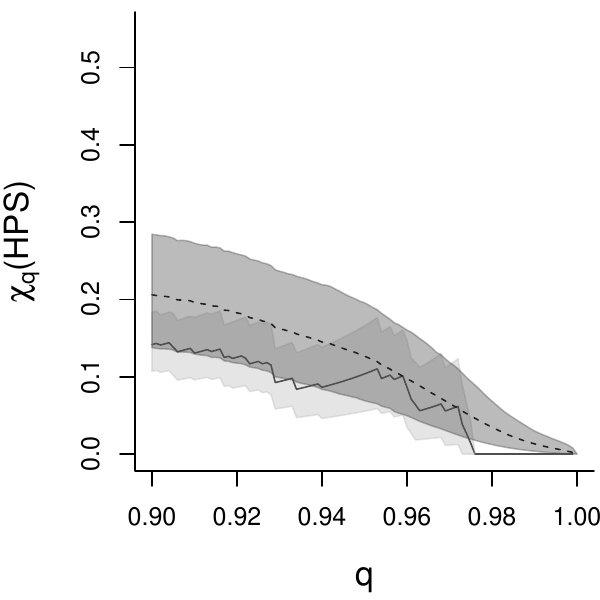}
    \caption{$\chi_q$ plots for all pairwise combinations of marginal variables and for all three marginal variables.\ Solid line is an empirical estimate with 95\% bootstrap confidence intervals in light grey, and the dashed line is a model-based estimate with 95\% posterior confidence intervals in dark grey.}
    \label{fig:wave-chi}
\end{figure}

\clearpage
\newpage

\section{Bivariate Gaussian mixture model fits}
\label{sec:biv-gauss-mix}
Here, we recreate the task of return level set boundary estimation in
the case of bivariate Gaussian mixture data presented in Section 4 of
\cite{hallin2021distribution}.\ Since our approach is non-empirical,
we are able to extrapolate to higher levels than observed in the data
and introduce uncertainty in our estimates.\ In
Figure~\ref{fig:d2-gauss-mix}, the return level curves, along with
0.95 credible intervals, are presented for return periods $T=20$ (or
$q=0.95$) and $T=10^5$ (or $q=1-10^{-5}$).\ For this each mixture
distribution, we generated $n=5,000$ datapoints and fit a quantile
regression model with threshold at the $q=0.9$ level.\ For the joint
model, we used the $\text{M}_3$ hierarchical structure.\ For display
purposes, 20 posterior observations of $\rquant$ were drawn and a
joint radial-directional model was fit at each of these observations.\
There were 50 draws of posterior return level curves drawn for each of
these models for a total of 1,000 return level curves for which to
obtain posterior mean and credible intervals for each plot in
Figure~\ref{fig:d2-gauss-mix}.
\begin{figure}[h!]
    \centering
    \includegraphics[width=\textwidth]{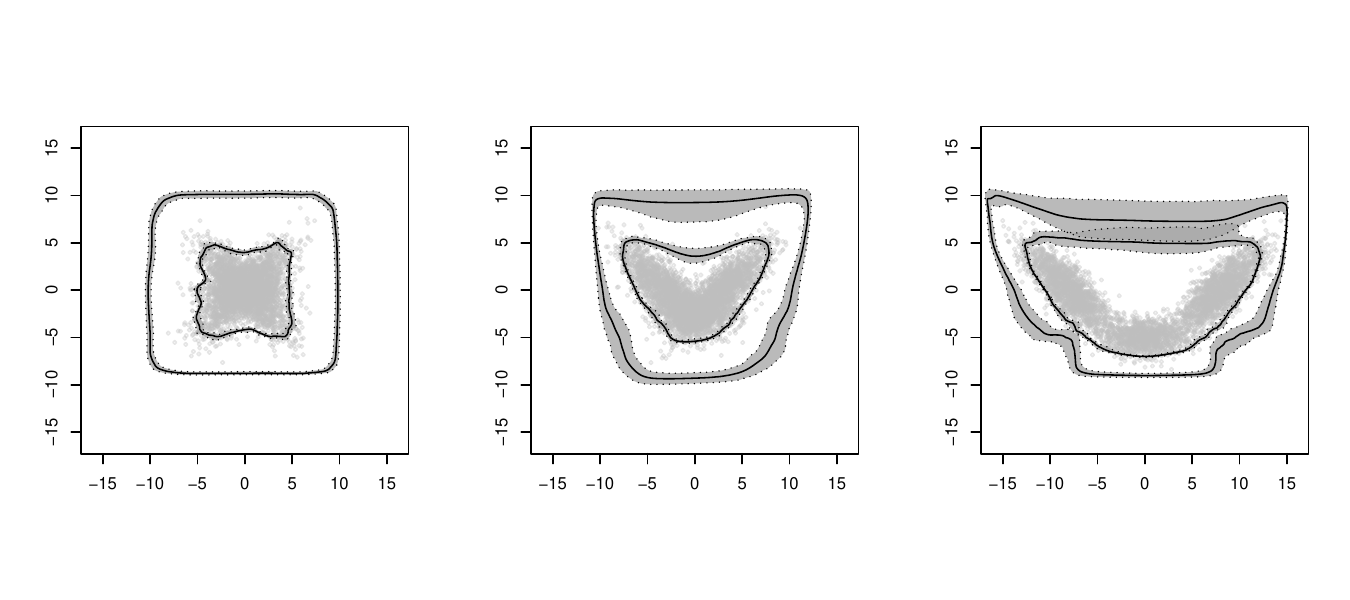}
    \caption{Posterior return level set boundaries for return periods $T=20$ (inner) and $T=10^5$ (outer) for 3 different bivariate Gaussian mixture distributions.\ Grey points are the data in original margins, solid black lines indicate the posterior mean return level set boundaries, with dark grey regions indicating the corresponding 0.95 credible bands, bordered by dotted lines.}
    \label{fig:d2-gauss-mix}
\end{figure}

\end{document}